\documentclass[letterpaper,12pt]{article}
\usepackage{amsfonts}
\usepackage{amsmath, amsthm, astron, amssymb}
\usepackage{bbm}
\usepackage{xcolor}

\usepackage{multirow}
\usepackage{float}
\floatstyle{plaintop}
\restylefloat{table}

\usepackage{endnotes}
\let\footnote=\endnote

\usepackage[english]{babel}
\usepackage[utf8]{inputenc}
\usepackage{fancyhdr}
\fancyheadoffset{\textwidth}

\newtheorem{assumption}{Assumption}
\newtheorem{example}{Example}
\newtheorem{theorem}{Theorem}[section]

\newtheorem{definition}{Definition}
\newtheorem{lemma}[theorem]{Lemma}
\input{tcilatex}

\newcommand{\Tr}{{\rm Tr}}

\begin{document}

\title{\textbf{Dynamic Linear Panel Regression Models \\ with Interactive Fixed Effects}\thanks{This is the last working paper version of the paper published in \textit{Econometric Theory} \textbf{33}(1), 158--195, 2017; doi:10.1017/S0266466615000328. This paper is based on an unpublished manuscript of the authors that was
circulated under the title ``Likelihood Expansion for Panel Regression Models with Factors''
but is now completely assimilated by the current paper and Moon and Weidner~(2015).
We greatly appreciate comments
from the participants in the Far Eastern Meeting of the Econometric Society 2008,
the SITE 2008 Conference, the All-UC-Econometrics Conference 2008,
the July 2008 Conference in Honour of Peter Phillips in Singapore,
the International Panel Data Conference 2009,
the North American Summer Meeting of the Econometric Society 2009,
and from seminar participants at Penn State, UCLA, and USC.
 We are also grateful for the comments and suggestions of Guido Kuersteiner, Peter Phillips, and anonymous referees.
Moon is grateful for the financial support from the NSF via grant SES 0920903 and the faculty development award from USC. Weidner acknowledges support from the Economic and Social Research Council through the ESRC Centre for Microdata Methods and Practice grant RES-589-28-0001.}
 }

\author{\setcounter{footnote}{2}
Hyungsik Roger Moon\footnote{
Department of Economics and USC Dornsife INET, University of Southern California,
Los Angeles, CA 90089-0253.
Email: {\tt moonr@usc.edu}.
Department of Economics, Yonsei University, Seoul, Korea.
}
\and Martin Weidner\footnote{
 Department of Economics,
 University College London,
 Gower Street,
 London WC1E~6BT, U.K.,
 and CeMMaP.
 Email: {\tt m.weidner@ucl.ac.uk}.
}}

\date{August 2015}

\maketitle
 
\vspace{-0.8cm} 

\begin{abstract}
\noindent 
We analyze linear
panel regression models with interactive fixed effects
and predetermined regressors, for example lagged-dependent variables.
The first-order asymptotic theory of the
least squares (LS) estimator of the regression coefficients is worked out
in the limit where both the cross-sectional dimension and the number of time periods become large.
We find two sources of asymptotic bias of the LS estimator:
bias due to correlation or heteroscedasticity of the idiosyncratic error term,
and bias due to predetermined (as opposed to strictly exogenous) regressors.
We provide a bias-corrected LS estimator.
We also present 
bias-corrected versions of the three classical test statistics (Wald, LR, and LM test)
and show their asymptotic distribution is a $\chi^2$-distribution.
Monte Carlo simulations show the bias correction of the LS estimator and of the test statistics
also work well for finite sample sizes.
\end{abstract}

\pagebreak

\section{Introduction}

In this paper, we study a linear panel regression model in which the individual fixed
effects~$\lambda_{i}$, called factor loadings, interact with common time-specific
effects $f_{t}$, called factors. This interactive fixed effect specification contains the conventional 
individual specific effects
and time-specific effects as special cases but is significantly more flexible because it allows the
factors $f_t$ to affect each individual with a different loading $\lambda_i$.

Factor models have been widely studied in various economics disciplines, for example, in asset pricing, forecasting, empirical macro, 
and empirical labor economics.\footnote{See, e.g.,
Chamberlain and Rothschild \cite*{CR1983}, Ross \cite*{Ross1976}, and Fama and French \cite*{FamaFrench1993} for asset pricing;
Stock and Watson \cite*{StockWatson2002} and Bai and Ng \cite*{BaiNg2006} for forecasting;
 Bernanke, Boivin, and Eliasz \cite*{BernankeBoivinEliasz2005} for empirical macro; and Holtz-Eakin, Newey, and Rosen \cite*{HoltzEakin-Newey-Rosen1988} for empirical labor economics.}
The panel literature often uses factor models to represent time-varying individual effects (or heterogenous time effects), so-called interactive fixed effects. For panels with a large cross-sectional dimension ($N$) but a short time dimension ($T$), Holtz-Eakin, Newey, and Rosen \cite*{HoltzEakin-Newey-Rosen1988} (hereafter HNR) study a linear panel regression model with interactive fixed effects and lagged dependent variables. To solve the incidental parameter problem caused by the $\lambda_i$'s, they estimate a quasi-differenced version of the model using appropriate lagged variables as instruments, and treating $f_t$'s as a fixed number of parameters to estimate. Ahn, Lee, and Schmidt \cite*{AhnLeeSchmidt2001} also consider large $N$ but short $T$ panels. Instead of eliminating the individual effects $\lambda_i$ by transforming the panel data, they impose various second-moment restrictions including the correlated random effects $\lambda_i$, and derive moment conditions to estimate the regression coefficients. The more recent literature considers panels with comparable size of $N$ and $T$. The interactive fixed effect panel regression model of
Pesaran~\cite*{Pesaran2006} allows heterogenous regression coefficients. Pesaran's estimator is the common correlated effect (CCE) estimator that uses the cross-sectional averages of the dependent variable and the independent variables as control functions for the interactive fixed effects.\footnote{%
The theory of the CCE estimator was further developed in, e.g.,
Harding and Lamarche~\cite*{HardingLamarche2009,HardingLamarche2011}, 
Kapetanios, Pesaran, and Yamagata~\cite*{KapetaniosPesaranYamagata2011},
Pesaran and Tosetti~\cite*{PesaranTosetti2011},
Chudik, Pesaran, and Tosetti~\cite*{ChudikPesaranTosetti2011},
and Chudik and Pesaran~\cite*{ChudikPesaran2015}.
}

Among the interactive fixed effect panel literature, most closely related to our paper is Bai~\cite*{Bai2009}. Bai assumes the regressors are  \textit{strictly} exogenous and the number of factors is known. The estimator he investigates is the least squares (LS) estimator, which minimizes the sum of squared residuals of the model
jointly over the regression coefficients and
the fixed effect parameters $\lambda_i$ and $f_t$.\footnote{%
The LS estimator is sometimes called ``concentrated'' least squares estimator in the literature, and in an earlier version of the paper, we referred to it as the ``Gaussian Quasi Maximum Likelihood Estimator'', because LS estimation is equivalent to maximizing a conditional Gaussian likelihood function.}
Using alternative asymptotics where $N,T\rightarrow \infty$ at the same rate,\footnote{%
Hahn and Kuersteiner~\cite*{HahnKuersteiner2002}
introduced the
alternative asymptotics to
characterize the asymptotic bias due to incidental parameter problems in 
fixed effect dynamic  panel data models.
See also
Arellano and Hahn~\cite*{ArellanoHahn2007} 
and Moon, Perron, and Phillips~\cite*{MoonPerronPhillips2014}
and references therein.
}
Bai shows the LS estimator is $\sqrt{NT}$-consistent and asymptotically normal, but may have an asymptotic bias. The bias in the normal limiting distribution occurs when the regression errors are correlated or heteroscedastic. Bai also shows how to estimate the bias, and proposes a bias-corrected estimator.

Following the methodology in Bai~\cite*{Bai2009}, we investigate the LS estimator for a linear panel regression with a known number of interactive fixed effects. 
The main difference from Bai is that we consider \textit{predetermined} regressors,
thus allowing feedback of past outcomes to future regressors. 
One of the main findings of the present paper is that 
the limit distribution of the LS estimator has two types of  biases: one type of bias due to  correlated or heteroscedastic errors (the same bias as in Bai) and the other type of bias due to the predetermined regressors. This additional bias term is analogous to the incidental parameter bias of Nickell~\cite*{Nickell1981} in finite $T$ and the bias in Hahn and Kuersteiner~\cite*{HahnKuersteiner2002} in large $T$. 

In addition to allowing for predetermined regressors, we also extend Bai's results to
models in which both ``low-rank regressors'' (e.g., time-invariant and common regressors, or interactions of those two) and ``high-rank-regressors'' (almost all other regressors that vary across individuals and over time) are present simultaneously, wheras Bai (2009) only considers the low-rank regressors separately and in a restrictive setting (in particular, not allowing for
regressors that are obtained by interacting time-invariant and common variables).
A general treatment of low-rank regressors  is desirable because
they often occur in applied work, for example,  Gobillon and Magnac~\cite*{GobillonMagnac2013}.
The analysis of those regressors
is challenging, however, 
because the unobserved interactive fixed effects also represent a low-rank $N \times T$ matrix, thus posing a non-trivial
identification problem for low-rank regressors, which needs to be addressed. We provide conditions under which
the different types of regressors are identified jointly, and under which they can be estimated consistently as $N$
and $T$ grow large.

Another contribution of this paper is to establish the asymptotic theory of the three classical test
statistics (Wald test, LR test, and LM (or score) test)
for testing restrictions on the regression coefficients in a large $N$, $T$ panel framework.\footnote{%
The ``likelihood ratio'' and the score used in the tests are 
based on the LS objective function, which can be interpreted as the
(misspecified) conditional Gaussian likelihood function.
}
Regarding testing for  coefficient restrictions, Bai~\cite*{Bai2009} investigates the Wald test based on the
bias-corrected LS estimator, and HNR consider the LR test in their 2SLS estimation framework with fixed $T$.\footnote{%
Another type of widely studied tests in the interactive fixed effect panel literature are panel unit root tests,
e.g., Bai and Ng~\cite*{BaiNg2004}, Moon and Perron~\cite*{MoonPerron2004}, and Phillips and Sul~\cite*{PhillipsSul2003}.
}
What we show is that the conventional LR and LM test statistics based on the LS profile objective function have non-central chi-square limits due to incidental parameters in the interactive fixed effects. We therefore propose modified LR and LM tests whose asymptotic distributions are  conventional chi-square distributions.

To establish the asymptotic theories of the LS estimator and the three classical tests, we use 
the quadratic approximation of the profile LS objective function
derived  in Moon and Weidner~\cite*{MoonWeidner2015}. This method is different from Bai~\cite*{Bai2009}, who uses the first-order condition of the LS optimization problem as the starting point of his analysis.
One advantage of our methodology is that it can also directly be applied to derive the asymptotic properties of
the LR and LM test statistics.

In this paper, we assume the regressors are not endogenous and the number of factors is known,
which might be restrictive in some applications.
In other papers, we study how to relax these restrictions. Moon and Weidner~\cite*{MoonWeidner2015} investigates the asymptotic properties of the LS estimator of the linear panel regression model with factors when the number of factors is unknown and extra factors are included unnecessarily
in the estimation. 
We find that under suitable conditions,\footnote{In Moon and Weidner \cite*{MoonWeidner2015}
we do not consider low-rank regressors or testing problems,
and we impose more restrictive assumptions on the error term of the model implying that some leading bias terms of the LS estimator are not present.}
the limit distribution of the LS estimator is unchanged when the number of 
factors is overestimated. 
The extension to allow for endogenous regressors is very briefly discussed in section 6 of the current paper,
and
is closely related to the results in
Moon, Shum, and Weidner~\cite*{MoonShumWeidner2012} (hereafter MSW). MSW's main purpose is to extend the random coefficient multinomial logit demand model (known as the BLP demand model from  Berry, Levinsohn, and Pakes~\cite*{BerryLevinsohnPakes1995}) by allowing for interactive product and market specific fixed effects. Although the main model of interest is quite different from the linear panel regression model of the current paper, MSW's econometrics framework is directly applicable to the model of the current paper with endogenous regressors.\footnote{Lee, Moon, and Weidner~\cite*{LeeMoonWeidner2012} also apply the MSW estimation method to estimate a simple dynamic panel regression with interactive fixed effect and classical measurement errors.}

Comparing the different estimation approaches
for interactive fixed effect panel regressions proposed in the literature, it seems fair to say that the LS estimator in
Bai~\cite*{Bai2009} and our paper, the CCE estimator of Pesaran~\cite*{Pesaran2006}, and the IV estimator based on quasi-differencing in HNR all have their own relative advantages and disadvantages.
These three estimation methods handle the interactive fixed effects quite differently. The LS method concentrates out the interactive fixed effects by taking out the principal components. The CCE method controls the factor (or time effects) using the cross-sectional averages of the dependent and independent variables. The HNR's approach quasi-differences out the individual effects, treating the remaining time effects as parameters to estimate. 
The IV estimator of HNR should work well when $T$ is short, but is expected to also suffer
from an incidental parameter problem when $T$ becomes large,
because then many factors need to be estimated as parameters that enter the model non-linearly.
Pesaran's CCE estimation method does not require the number of factors to be known and does not require the strong
factor assumption that we will impose below, but for the CCE estimator to work, not only the DGPs of the dependent variable (e.g., the regression model) but also the DGPs of the explanatory variables need to be restricted such that their cross-sectional average can control for unobserved factors. The LS estimator and its bias-corrected version perform well under relatively weak restrictions on the regressors, but requires that $T$ should not be too small and that the factors should be sufficiently strong to be correctly picked up as the leading principal components.

The paper is organized as follows.
In section~\ref{sec:model}, we introduce the interactive fixed effect model and provide conditions for identifying the regression coefficients in the presence of the interactive fixed effects. In section~\ref{sec:estimator}, we define the LS estimator of the regression parameters and provide a set of assumptions that are sufficient to show
consistency of the LS estimator. In section~\ref{sec:limdist}, we work out the asymptotic distribution of the LS estimator under alternative asymptotics.
We also provide a consistent estimator for  the asymptotic bias and a bias-corrected LS estimator.
In section~\ref{sec:testing}, we consider the Wald, LR, and LM tests
for testing restrictions on the regression coefficients of the model.
We present bias-corrected versions of these tests and show that they have
chi-square limiting distributions.
In section~\ref{sec:Endogenous Regression}, we briefly discuss how to estimate the interactive fixed effect linear panel regression when the regressors are endogenous.
In section~\ref{sec:MC}, we present Monte Carlo simulation results for an ${\rm AR}(1)$ model with interactive fixed effects. The simulations show the LS estimator for the ${\rm AR}(1)$ coefficient is biased, and 
the tests based on it can have severe size distortions and power asymmetries, wheras the bias-corrected LS estimator and test statistics have better properties.
We conclude in section~\ref{sec:conclusion}. We present all proofs of theorems and some technical details
 in the appendix or  supplementary material.

A few words on notation are due. For a column vector $v$, the Euclidean norm is
defined by $\| v \| = \sqrt{v^{\prime}v}$. For the $n$-th largest
eigenvalues (counting multiple eigenvalues multiple times) of a symmetric
matrix $B$, we write $\mu_n(B)$. For an $m\times n$ matrix $A$,
the Frobenius norm is $\| A \|_{F} = \sqrt{{\rm Tr}%
(AA^{\prime})}$, and the spectral norm is $\| A \| = \max_{0 \neq v \in
\mathbb{R}^n} \, \frac{ \| A v \|} {\| v\|}$, or equivalently $\| A \| =
\sqrt{ \mu_1(A^{\prime}A) }$. Furthermore, we define $P_A = A
(A^{\prime}A)^{\dagger} A'$ and $M_A = \mathbb{I} - A (A^{\prime}A)^{\dagger} A'$,
where $\mathbb{I}$ is the $m\times m$ identity matrix,
and $(A^{\prime}A)^{\dagger}$ is the Moore-Penrose pseudoinverse, to allow for the case that
$A$ is not of full column rank. For
square matrices $B$, $C$, we write $B>C$ (or $B\geq C$) to indicate $B-C$ is positive (semi) definite.
For a positive definite symmetric matrix $A$, we write $A^{1/2}$ and $A^{-1/2}$ for
the unique symmetric matrices that satisfy $A^{1/2} A^{1/2}=A$ and $A^{-1/2} A^{-1/2} =A^{-1}$.
We use $\nabla$ for the gradient of a function; that is, $\nabla f(x)$ is the 
column vector of partial derivatives
of $f$ with respect to each component of $x$. We use ``wpa1'' for ``with probability approaching one''.

\section{Model and Identification}
\label{sec:model}

We study the following panel regression model with
cross-sectional size $N$, and $T$ time periods:
\begin{align}
Y_{it}&=\beta^{0\prime } X_{it}+\lambda_{i}^{0\prime }f_{t}^{0}+e_{it}, &
i=1\ldots N,\quad t=1\ldots T,  \label{model0}
\end{align}%
where $X_{it}$ is a $K\times 1$ vector of observable regressors, $\beta^{0}$
is a $K\times 1$ vector of regression coefficients, $\lambda_{i}^{0}$ is an
$R\times 1$ vector of unobserved factor loadings, $f_{t}^{0}$ is an $R\times1$ vector of
unobserved common factors, and $e_{it}$ are unobserved errors.
The superscript zero indicates the true parameters.  
We write $f_{tr}^{0}$ and $\lambda^0_{ir}$, where $r=1,\ldots,R$, for the components of $\lambda_{i}^{0}$
and $f_{t}^{0}$, respectively.
$R$ is the number of factors.
Note that we can have $f_{tr}^{0}=1$ for all $t$ and a particular  $r$, in which case
the corresponding $\lambda^0_{ir}$ become standard individual-specific effects. 
Analogously, we can have $\lambda^0_{ir}=1$ for all $i$ and a particular   $r$, so that the corresponding
$f_{tr}^{0}$ become standard time-specific effects.

Throughout this paper, we assume the true number of factors $R$ is known.\footnote{%
To remove this restriction, one could estimate $R$ consistently in the presence of the regressors. In the literature so far, however, consistent estimation procedures for $R$ are established mostly in pure factor models (e.g., Bai and Ng \cite*{BaiNg2002},
Onatski \cite*{Onatski2010} and Harding \cite*{Harding2007}). Alternatively, one could rely on Moon and Weidner \cite*{MoonWeidner2015} who consider a regression model with interactive
fixed effects when only an upper bound on the number of factors is known --- but 
extending those results to the more general setup considered here
 is mathematically challenging.}
We introduce the notation $\beta^0 \cdot X = \sum_{k=1}^{K}\,\beta_{k}^{0}\,X_{k}$.
In matrix notation, the model can then be written as
\begin{align*}
Y= \beta^0 \cdot X \,+\,\lambda^0 f^{0\prime }+e\;,
\end{align*}
where  $Y$, $X_{k}$, and $e$ are $N\times T$ matrices, $\lambda^0$ is an $N\times R$ matrix,
and $f^0$ is a $T\times R$ matrix. The elements of $X_k$ are denoted by $X_{k,it}$.

We separate the $K$ regressors into $K_1$
``low-rank regressors'' $X_{l}$, $l=1,\ldots,K_1$,
and $K_2=K-K_1$ ``high-rank regressors'' $X_{m}$, $m=K_1+1,\ldots,K$.
Each low-rank regressor $l=1,\ldots,L$ is assumed to satisfy ${\rm rank}(X_l) = 1$.
Therefore, we can write $X_l=w_l v_l'$, where $w_l$ is an  $N$-vector
and $v_l$ is a $T$-vector, and
we also define the $N \times K_1$ matrix $w=(w_1,\ldots,w_{K_1})$
and the $T \times K_1$ matrix $v=(v_1,\ldots,v_{K_1})$.

Let $l=1,\ldots,K_1$.
The two most prominent types of low-rank regressors are time-invariant regressors, which satisfy
$X_{l,it}=Z_{i}$ for all $i,t$, and common (or cross-sectionally invariant) regressors, in which case
$X_{l,it}=W_{t}$ for all $i,t$. Here, $Z_{i}$ and $W_{t}$ are some observed variables, which only vary over
$i$ or $t$, respectively. A more general low-rank regressor can be obtained by interacting $Z_i$ and $W_t$
multiplicatively, namely, $X_{l,it}=Z_{i} W_t$, an empirical example of which is given in 
Gobillon and Magnac~\cite*{GobillonMagnac2013}. 
In these examples, and probably for the vast majority of applications,
the low-rank regressors all satisfy ${\rm rank}(X_{l})=1$, but our results can easily be extended to more general
low-rank regressors.\footnote{%
If we have low-rank regressors with rank larger than one, then
we write $X_l=w_l v_l'$, where $w_l$ is an  $N\times {\rm rank}(X_l)$ matrix
and $v_l$ is a $T \times {\rm rank}(X_l)$ matrix, and we define
$w=(w_1,\ldots,w_{K_1})$ as a
$N \times \sum_{l=1}^L {\rm rank}(X_l)$ matrix,
and $v=(v_1,\ldots,v_{K_1})$ ae a $T \times \sum_{l=1}^L {\rm rank}(X_l)$ matrix.
All our results are then unchanged, as long as
${\rm rank}(X_l)$ is a finite constant for all $l=1,\ldots,K_1$, and we replace
$2 R + K_1$ by
 $ 2 R
       +   {\rm rank}\left(  w \right)$
    in Assumption~\ref{ass:id}$(v)$ and Assumption~\ref{ass:A4}$(ii)(a)$.
}

High-rank regressors are those whose distribution guarantees they have high rank
(usually full rank) when considered as an $N \times T$ matrix. For example, a regressor whose entries satisfy
$X_{m,it} \sim iid \, {\cal N}(\mu,\sigma)$,
with $\mu \in \mathbb{R}$ and $\sigma>0$, satisfies
${\rm rank}(X_m)=\min(N,T)$ with probability one.

This separation of the regressors into low- and high-rank regressors is  important to formulate our assumptions
for identification and consistency, but actually plays no role in the estimation and inference procedures for $\widehat \beta$
discussed below.

\begin{samepage}
\newtheorem{IDassumption}{Assumption}
\renewcommand{\theIDassumption}{ID}
\begin{IDassumption}[\bf Assumptions for Identification]~
   \label{ass:id}
    \begin{itemize}
      \item[(i)] {\bf Existence of Second Moments:} \\
        The second moments of $X_{k,it}$ and
                 $e_{it}$ conditional on $\lambda^0$, $f^0$, $w$ exist for all $i$, $t$, $k$.

      \item[(ii)]  {\bf Mean Zero Errors and Exogeneity:} \\
      $\mathbb{E}\left( e_{it} | \lambda^0, f^0,w \right)=0$,
      \, and \,
     $\mathbb{E}(X_{k, it} e_{it}  |  \lambda^0, f^0,w)=0$, \,  a.s., \,
   for all $i$, $t$, $k$.

\end{itemize}
The following two assumptions only need to be imposed if $K_1>0$, that is,
if low-rank regressors are present:

\begin{itemize}

      \item[(iii)]  {\bf Non-collinearity of Low-Rank Regressors:} \\
         Consider linear combinations
         $\alpha \cdot X_{{\rm low}} = \sum_{l=1}^{K_1} \alpha_l X_l$ of the low-rank regressors $X_l$
with $\alpha \in \mathbb{R}^{K_1}$. For all $\alpha \neq 0$, we assume
\begin{align*}
     \mathbb{E}\left[ (\alpha \cdot X_{{\rm low}}) M_{f^0} (\alpha \cdot X_{{\rm low}})'  \big|  \lambda^0, f^0,w \right]
       \neq 0 \, ,
       \qquad \text{a.s.}
\end{align*}

    \item[(iv)]  {\bf  No Collinearity between Factor Loadings and Low-Rank Regressors:} \\
    ${\rm rank}( M_{w} \lambda^0 ) =  {\rm rank}(\lambda^0)$.\footnote{%
       Note that ${\rm rank}(\lambda^0) = R$ if $R$ factors are present. Our identification results
       are consistent with the possibility that ${\rm rank}(\lambda^0) < R$, i.e., that $R$ only represents an upper bound
       on the number of factors, but later we assume  ${\rm rank}(\lambda^0) = R$ to show consistency.
    }
\end{itemize}
The following assumption only needs to be imposed if $K_2>0$, that is, if high-rank
regressors are present:

\begin{itemize}

\item[(v)] {\bf Non-collinearity of High-Rank Regressors:}  \\
 Consider linear combinations
$\alpha \cdot X_{{\rm high}} = \sum_{m=K_1+1}^K \alpha_m X_m$ of the high-rank regressors $X_m$
for $\alpha \in \mathbb{R}^{K_2}$,
where the components of the $K_2$-vector $\alpha$ are
denoted by $\alpha_{K_1+1}$ to $\alpha_{K}$.
 For all $\alpha \neq 0$, we assume
\begin{align*}
      {\rm rank}\left\{ \mathbb{E}\left[ (\alpha \cdot X_{{\rm high}}) (\alpha \cdot X_{{\rm high}})'
       \big| \lambda^0, f^0,w \right] \right\}
       > 2 R
       +  K_1 \, ,
       \qquad \text{a.s.}      
\end{align*}

\end{itemize}

\end{IDassumption}
\end{samepage}

All expectations in the assumptions are conditional on $\lambda^0$, $f^0$, and $w$; in particular,
$e_{it}$ is not allowed to be correlated with $\lambda^0$, $f^0$, and $w$. However, $e_{it}$
is allowed to be correlated with $v$
(i.e., predetermined low-rank regressors are allowed).
If desired, one can interchange the role
of $N$ and $T$ in the assumptions, by using the formal symmetry of the model under exchange
of the panel dimensions ($N \leftrightarrow T$,
$\lambda^0 \leftrightarrow f^0$, $Y \leftrightarrow Y'$, $X_k \leftrightarrow X_k'$,
$w \leftrightarrow v$).

Assumptions~\ref{ass:id}$(i)$ and $(ii)$ have standard interpretations, but the other assumptions
require some further discussion.

Assumption~\ref{ass:id}$(iii)$ states
 the low-rank regressors are non-collinear even after projecting out all variation that is explained by
the true factors $f^0$. This assumption would, for example, be violated if $v_l = f^0_r$ for some $l=1,\ldots,K_1$
and $r=1,\ldots,R$, because
then $X_l M_{f^0} = 0$ and we can choose $\alpha$ such that $X_{\rm low} = X_l$.
Similarly, Assumption~\ref{ass:id}$(iv)$ rules out, for example, that $w_l = \lambda^0_r$  for some $l=1,\ldots,K_1$
and $r=1,\ldots,R$, because then ${\rm rank}( M_{w} \lambda^0 ) <  {\rm rank}(\lambda^0)$, in general.
It ought  to be expected that $\lambda^0$ and $f^0$ have to feature in the identification conditions
for the low-rank regressors, because the interactive fixed effects structure 
and the low-rank regressors 
represent similar types of low-rank $N \times T$ structures.

Assumption~\ref{ass:id}$(v)$ is a generalized non-collinearity assumption
for the high-rank regressors,
which guarantees  any
 linear combination $\alpha \cdot X_{{\rm high}}$ 
 of the high-rank regressors is sufficiently different from the low-rank regressors and from the interactive fixed effects. 
A standard non-collinearity assumption can be formulated by demanding
 the $N \times N$ matrix
 $\mathbb{E}\big[  \allowbreak (\alpha \cdot X_{{\rm high}}) (\alpha \cdot X_{{\rm high}})'
       \big|  \allowbreak \lambda^0, f^0,w \big]$ is non-zero for all 
       non-zero 
       $\alpha \in \mathbb{R}^{K_2}$, which can be equivalently expressed as
${\rm rank}\big\{ \allowbreak \mathbb{E}\left[ (\alpha \cdot X_{{\rm high}}) (\alpha \cdot X_{{\rm high}})'
       \big| \lambda^0, f^0,w \right] \big\} \allowbreak > 0$
 for all non-zero $\alpha \in \mathbb{R}^{K_2}$.      
Assumption~\ref{ass:id}$(v)$ strengthens this standard non-collinearity
assumption by demanding the rank not only to be positive, but  larger than $2R+K_1$.
 This also explains the name ``high-rank regressors,'' because their
 rank has to be sufficiently large to satisfy 
 this assumption. Note also that only the number
 of factors $R$, but not $\lambda^0$ and $f^0$, features in Assumption~\ref{ass:id}$(v)$.
 The sample version of this assumption
 is given by Assumption~\ref{ass:A4}$(ii)(a)$
 below, which is also very closely related to Assumption~A
 in Bai~\cite*{Bai2009}.

\begin{theorem}[\bf Identification]
   \label{th:id}
    Suppose the Assumptions~\ref{ass:id} are satisfied.
    Then, the minima of the expected objective
    function
     $\mathbb{E}\left(\left\| Y \, - \, \beta \cdot X \, - \, \lambda \, f'
    \right\|^2_F \Big| \lambda^0, f^0,w \right)$
    over   $(\beta,\lambda,f) \in \mathbb{R}^{K + N \times R + T \times R}$
    satisfy
     $\beta=\beta^0$ and $\lambda f' = \lambda^0 f^{0 \prime}$.
     This shows that $\beta^0$ and $\lambda^0 f^{0 \prime}$ are identified.

\end{theorem}

The theorem shows the true parameters are identified as minima of the expected value
of $\left\| Y \, - \, \beta \cdot X \, - \, \lambda \, f'
    \right\|^2_F = \sum_{i,t} (  Y_{it} \beta^{\prime } - X_{it}- \lambda_{i}^{\prime }f_{t} )^2$,
which is the sum of squared residuals. We use the same objective function, to define the estimators
$\widehat \beta$, $\widehat \lambda$ and $\widehat f$ below.
Without further normalization conditions, the parameters $\lambda^0$ and $f^0$ are not separately identified,
because the outcome variable $Y$ is invariant under transformations $\lambda^0 \rightarrow \lambda^0 A'$
and $f^0 \rightarrow f^0 A^{-1}$, where $A$ is a non-singular $R\times R$ matrix.
However, the product $\lambda^0 f^{0 \prime}$ is uniquely identified according to the theorem.
Because our focus is on identification and estimation of $\beta^0$, we do not need to discuss those additional
normalization conditions for $\lambda^0$ and $f^0$ in this paper.

\section{Estimator and Consistency}
\label{sec:estimator}

The objective function of the model
is simply the sum of squared residuals, which in matrix notation can be expressed as
\begin{align}
   {\cal L}_{NT}(\beta,\lambda,f) &=
    \frac{1}{NT} \; \left\|  Y- \beta  \cdot X -\lambda f^{\prime } \right\|^2_F
   \nonumber  \\
    &=
   \frac{1}{NT} \; {\rm Tr}
    \left[ \left( Y- \beta  \cdot X -\lambda f^{\prime }
       \right)^{\prime }\left( Y- \beta  \cdot X -\lambda f^{\prime } \right) \right] \; .
   \label{DefCalL}
\end{align}
The estimator we consider is the LS estimator that jointly minimizes ${\cal L}_{NT}(\beta,\lambda,f)$
over $\beta$, $\lambda$ and~$f$. Our main objects of interest are the regression parameters
$\beta=\left( \beta_{1},...,\beta_{K}\right)^{\prime}$, whose estimator is given by
\begin{align}
   \widehat \beta &= \limfunc{argmin}_{\beta \in \mathbb{B}}\,L_{NT}\left( \beta \right) \;,
   \label{DefQMLE}
\end{align}
where $\mathbb{B} \subset \mathbb{R}^{K}$ is a compact parameter set that contains the true parameter, 
namely, $\beta^0 \in \mathbb{B}$,
and the objective function is the profile objective function
\begin{align}
   L_{NT}\left( \beta \right) &= \min_{\lambda ,f} \; {\cal L}_{NT}(\beta,\lambda,f)
      \nonumber \\
            &= \min_{f} \; \frac{1}{NT} \; {\rm Tr}%
  \left[ \left( Y-\beta  \cdot X \right) M_{f}
  \left( Y- \beta  \cdot X \right)^{\prime} \right]
      \nonumber \\
            &= \frac{1}{NT} \; \sum_{r=R+1}^T \;
   \mu_r \left[ \left( Y- \beta  \cdot X \right)^{\prime}
   \left(Y- \beta  \cdot X \right) \right] \; .
  \label{LNT123}
\end{align}
Here, the first expression for $L_{NT}(\beta)$ is its definition as
the minimum value of ${\cal L}_{NT}(\beta,\lambda,f)$ over $\lambda$ and $f$.
We denote the minimizing incidental parameters
by $\widehat \lambda(\beta)$ and $\widehat f(\beta)$,
and we define the estimators $\widehat \lambda = \widehat \lambda(\widehat \beta)$
and $\widehat  f = \widehat  f(\widehat \beta)$.
Those minimizing incidental parameters are not uniquely
determined  -- for the same reason that $\lambda^0$ and $f^0$ are non uniquely identified -- but the product
$\widehat \lambda(\beta) \widehat f'(\beta)$ is unique.

The second expression for $L_{NT}\left( \beta \right)$ in equation \eqref{LNT123} is obtained by concentrating out $\lambda$ (analogously, one can concentrate out $f$ to obtain a formulation whereby only the parameter $\lambda$ remains). The optimal $f$ in the second expression is given by the $R$ eigenvectors that correspond to the $R$ largest
eigenvalues of the $T\times T$ matrix
$\left( Y-\beta  \cdot X \right)^{\prime}\left( Y- \beta  \cdot X \right)$. This insight leads to the third line that presents the
profile objective function as the sum over the $T-R$ smallest
eigenvalues of this $T\times T$ matrix. Lemma~\ref{lemma:Optimization} in the appendix shows equivalence of the three expressions for $L_{NT}(\beta)$
given above.

Multiple local minima of $L_{NT}(\beta)$ may exist, and one should use
multiple starting values for the numerical optimization of $\beta$ to guarantee
the true global minimum $\widehat \beta$ is found.

To show consistency of the LS estimator $\widehat \beta$ of the interactive fixed effect model,
and also later for our first-order asymptotic theory,
we consider the limit $N,T\rightarrow \infty$.
In the following we present assumptions on $X_k$, $e$, $\lambda$, and $f$ that guarantee consistency.\footnote{%
We could write $X_k^{(N,T)}$, $e^{(N,T)}$, $\lambda^{(N,T)}$, and $f^{(N,T)}$,
because all these matrices, and even their dimensions, are functions on $N$ and $T$, but we suppress this
dependence throughout the paper.}
\begin{assumption}
\label{ass:A1} (i)~$\limfunc{plim}_{N,T \rightarrow \infty}\left(\lambda^{0\prime} \lambda^0/N\right) > 0$,
(ii)~$\limfunc{plim}_{N,T \rightarrow \infty} \left( f^{0\prime} f^0 / T \right) > 0$.
\end{assumption}
\begin{assumption}
\label{ass:A2}
   $\limfunc{plim}_{N,T \rightarrow \infty} \left[
      (NT)^{-1} {\rm Tr}(X_k \, e^{\prime}) \right] = 0$, for all $k=1,\ldots,K$.
\end{assumption}
\begin{assumption}
\label{ass:A3} $\limfunc{plim}_{N,T \rightarrow \infty} \left( \| e \| / \sqrt{NT} \right) = 0$.
\end{assumption}

Assumption \ref{ass:A1} guarantees the matrices $f^0$ and $\lambda^0$ have full rank, that is, that
$R$ distinct factors and factor loadings exist asymptotically, and that the norm of each factor and factor loading
 grows at a rate of $\sqrt{T}$ and $\sqrt{N}$, respectively.
Assumption \ref{ass:A2}
demands the regressors are weakly exogenous.
Assumption \ref{ass:A3} restricts the spectral norm of the $N \times T$ error matrix $e$.
We discuss this assumption in more detail in the next section,
and we give examples of error
distributions that satisfy this condition in section~\ref{app:error}
of the supplementary material.
The final assumption needed for consistency is an assumption on the regressors $X_k$.
We already introduced the distinction between the
$K_1$ ``low-rank regressors'' $X_{l}$, $l=1,\ldots,K_1$,
and the $K_2=K-K_1$ ``high-rank regressors'' $X_{m}$, $m=K_1+1,\ldots,K$ above.

\begin{assumption}
\label{ass:A4}
$\phantom{a}$
\begin{itemize}
\item[(i)] $\limfunc{plim}_{N,T \rightarrow \infty} \left[
     (NT)^{-1} \, \sum_{i=1}^N \, \sum_{t=1}^T \, X_{it} X_{it}' \right] > 0$.
\item[(ii)] The two types of regressors satisfy:
\begin{itemize}
\item[(a)] Consider linear combinations
$\alpha \cdot X_{{\rm high}}=\sum_{m=K_1+1}^K \alpha_m X_m$ of the high-rank regressors $X_m$
for $K_2$-vectors $\alpha$ with $\|\alpha\|=1$,
 where the components of the $K_2$-vector $\alpha$ are
denoted by $\alpha_{K_1+1}$ to $\alpha_{K}$. We assume a constant $b>0$ exists such that
\begin{align*}
   \min_{\{\alpha \in \mathbb{R}^{K_2}, \|\alpha\|=1\}}   \,
   \sum_{r=2R+K_1+1}^N \,
    \mu_{r} \left[ \frac{ (\alpha \cdot X_{{\rm high}}) (\alpha \cdot X_{{\rm high}})'  } {NT} \right] \;  &\geq b \; \qquad
   \text{wpa1.}
\end{align*}

\item[(b)] For the low-rank regressors, we assume ${\rm rank}(X_l)=1$, $l=1,\ldots,K_1$; that is,
they can be written as $X_l=w_l v_l'$ for $N$-vectors $w_l$ and
$T$-vectors $v_l$, and we define the $N\times K_1$ matrix $w=(w_1,\ldots,w_{K_1})$
and the $T\times K_1$ matrix $v=(v_1,\ldots,v_{K_1})$.
We assume a constant $B>0$ exists such that
$N^{-1} \, \lambda^{0\prime} \, M_{w} \, \lambda^0 \, > \, B \, \mathbb{I}_R$
and $T^{-1} \, f^{0\prime} \, M_{v} \, f^0 \, > \, B \, \mathbb{I}_R$, wpa1.
\end{itemize}
\end{itemize}
\end{assumption}

Assumption~\ref{ass:A4}$(i)$ is a standard non-collinearity condition for all the regressors.
Assumption~\ref{ass:A4}$(ii)(a)$ is an appropriate sample analog of the identification Assumption~\ref{ass:id}$(v)$.
If the sum in Assumption~\ref{ass:A4}$(ii)(a)$ were to start from $r=1$, we would have
$\sum_{r=1}^N
    \mu_{r} \left[ \frac{ (\alpha \cdot X_{{\rm high}}) (\alpha \cdot X_{{\rm high}})'  } {NT} \right]
    = \frac 1 {NT} {\rm Tr}[ (\alpha \cdot X_{{\rm high}}) (\alpha \cdot X_{{\rm high}})'  ]$, so that the
    assumption would become a standard non-collinearity condition. Not including the first $2R+K_1$
  eigenvalues in the sum implies the $N \times N$ matrix $(\alpha \cdot X_{{\rm high}}) (\alpha \cdot X_{{\rm high}})'$
  needs to have rank larger than $2R + K_1$.

Assumption~\ref{ass:A4}$(ii)(b)$ is closely related to  the identification Assumptions~\ref{ass:id}$(iii)$
and $(iv)$.  The appearance of the factors and factor loadings in this assumption on the low-rank regressors
is inevitable to guarantee consistency. For example, consider a low-rank regressor
that is cross-sectionally independent and proportional to the $r$'th unobserved factor, for example,
$X_{l,it}=f_{tr}$. The corresponding regression coefficient $\beta_l$ is then not identified,
because the model is invariant under a shift $\beta_l \mapsto \beta_l + a$, $\lambda_{ir} \mapsto \lambda_{ir} - a$,
for an arbitrary $a\in \mathbb{R}$. This phenomenon is well known from ordinary fixed effect models, where
the coefficients of time-invariant regressors are not identified.
Assumption \ref{ass:A4}(ii)(b) therefore guarantees for $X_l=w_l v_l'$ that $w_l$
is sufficiently different from $\lambda^0$, and $v_l$ is sufficiently
different from $f^0$.

\begin{theorem}[\bf Consistency]
\label{th:consistency}
Let Assumptions~\ref{ass:A1}, \ref{ass:A2}, \ref{ass:A3}, and \ref{ass:A4}
be satisfied; let the parameter set $\mathbb{B}$ be compact;  and let $\beta^0 \in \mathbb{B}$.
In the limit $N,T \rightarrow \infty$, we then have
\begin{align*}
\widehat \beta \; \limfunc{\longrightarrow}_p \; \beta^0 \; .
\end{align*}
\end{theorem}

We assume compactness of $\mathbb{B}$ to guarantee existence of the minimizing $\widehat \beta$.
We also use boundedness of $\mathbb{B}$ in the consistency proof, but only for those parameters
$\beta_l$, $l=1\ldots K_1$,
that correspond to low-rank regressors, that is, if only high-rank regressors ($K_1=0$) are present,
the compactness assumption can be omitted, as long as existence of $\widehat \beta$ is guaranteed
(e.g., for $\mathbb{B}=\mathbb{R}^K$).

Bai \cite*{Bai2009} also proves consistency of the LS estimator of the interactive fixed effect model, but under somewhat different assumptions.
He also employs what we call Assumptions~\ref{ass:A1} and \ref{ass:A2},
and he uses a low-level version of Assumption~\ref{ass:A3}.
He demands the regressors to be strictly exogenous.
Regarding consistency, the main difference between our assumptions and
his is the treatment of high- and low-rank regressors. He first gives a condition on the regressors (his Assumption~A) that rules out low-rank regressors, and
later discusses the case in which all regressors
are either time-invariant or common regressors (i.e., are all low rank).
By contrast, our Assumption~\ref{ass:A4} allows for a combination of high- and low-rank regressors, and for low-rank regressors that are more general than time-invariant and common regressors.

\section{Asymptotic Distribution and Bias Correction}
\label{sec:limdist}

Because we have already shown consistency of the LS estimator $\widehat \beta$, it is sufficient
to study the local properties of the objective function $L_{NT}(\beta)$
around $\beta^0$ to derive the first-order asymptotic theory
of $\widehat \beta$. Moon and Weidner~\cite*{MoonWeidner2015} derived a useful approximation of $L_{NT}(\beta)$ around $\beta^0$, and we briefly summarize
the ideas and results of this approximation in the following subsection.
We then apply those results to derive the asymptotic distribution
of the LS estimator, including working out the asymptotic bias, which was not done previously.
Afterward, we discuss bias correction and inference.

\subsection{Expansion of the Profile Objective Function}

The last expression in equation \eqref{LNT123} for the profile objective function
is convenient because it does not involve any minimization over
the parameters $\lambda$ or $f$. On the other hand,
this expression cannot be easily discussed by
analytic means, because in general, no explicit formula exists for the
eigenvalues of a matrix. The conventional method that involves a Taylor
series expansion in the regression parameters $\beta$
{\it alone} seems infeasible here.
In Moon and Weidner \cite*{MoonWeidner2015},
we showed how to overcome this problem by
expanding the profile objective function {\it jointly}
in $\beta$ and $\|e\|$.
The key idea is the following decomposition:
\begin{align*}
  Y -  \beta \cdot X
    &= \underbrace{ \lambda^0 f^{0\prime} }
   _{\begin{minipage}{1.3cm} \centering \scriptsize leading \\ term   \end{minipage}}
    \underbrace{ - \left( \beta  - \beta^0  \right) \cdot X
    + e }_{\text{perturbation term}}.
\end{align*}
If the perturbation term is zero, the profile objective
$L_{NT}(\beta)$ is also zero, because the leading term
$\lambda^0 f^{0\prime}$ has rank $R$, so that the $T-R$ smallest eigenvalues
of $f^0 \lambda^{0\prime} \lambda^0 f^{0\prime}$ all vanish. One may thus expect
that small values of the perturbation term should correspond to small values
of $L_{NT}(\beta)$. This idea can indeed be made mathematically precise.
By using the perturbation theory of linear operators
(see, e.g., Kato \cite*{Kato}), one can work
out an expansion of $L_{NT}(\beta)$ in the perturbation term,
and one can show this expansion is convergent
as long as the spectral norm of the perturbation term
is sufficiently small.

The assumptions on the model made so far are
in principle already sufficient to apply this
expansion of the profile objective function,
but to truncate the expansion
at an appropriate order and to provide a bound on the remainder term that is
sufficient to derive the first-order asymptotic theory of the LS estimator,
we need to strengthen Assumption~\ref{ass:A3} as follows.

\newtheorem*{assumption3p}{Assumption \ref{ass:A3}$^*$}
\begin{assumption3p}
   $\|e\|=o_p(N^{2/3})$.
\end{assumption3p}

In the rest of the paper, we only consider asymptotics in which $N$ and $T$ grow at the
same rate; that is, we could equivalently write $o_p(T^{2/3})$
instead of $o_p(N^{2/3})$ in Assumption \ref{ass:A3}$^*$.
In section~\ref{app:error}
of the supplementary material, we provide examples of error distributions that satisfy Assumption~\ref{ass:A3}$^*$. In fact, for these examples, we
have $\|e\|={\cal O}_p(\sqrt{\max(N,T)})$. A large literature studies the asymptotic behavior of the spectral norm of random matrices; see, for example, Geman \cite*{Geman1980},
Silverstein \cite*{Silverstein1989}, Bai, Silverstein, and Yin \cite*{BaiSilvYin1988},
Yin, Bai, and Krishnaiah \cite*{BaiKrishYin1988},
and Latala \cite*{Latala2006}.
Loosely speaking, we expect the result $\|e\|={\cal O}_p(\sqrt{\max(N,T)})$
to hold as long as the errors $e_{it}$ have mean zero,
uniformly bounded fourth moment, and weak time-serial and cross-sectional correlation
(in some well-defined sense, see the examples).

We can now present the quadratic approximation of the profile objective function $L_{NT}(\beta)$ that 
we derived in Moon and Weidner \cite*{MoonWeidner2015}.

\begin{theorem}[\bf  Expansion of  Profile Objective Function]
\label{th:ass_expand}
Let Assumption~\ref{ass:A1}, \ref{ass:A3}$^*$,
and \ref{ass:A4}(i) be satisfied, and consider the limit
$N,T \rightarrow \infty$ with $N/T \rightarrow \kappa^2$, $0<\kappa<\infty$.
Then, the profile objective function
satisfies $L_{NT}(\beta) = L_{q,NT}(\beta) + (NT)^{-1} \, R_{NT}(\beta)$,
where the remainder $R_{NT}(\beta)$ is such that for any sequence
$\eta_{NT}\rightarrow 0$, we have
\begin{align*}
  \sup_{\{\beta :\left\| \beta -\beta^{0} \right\| \leq \eta_{NT}\}} \frac{ \left|
  R_{NT}(\beta) \right| } { \left( 1 + \sqrt{NT} \, \left\| \beta -\beta^{0} \right\| \right)^2 } = o_{p}\left( 1 \right) ,
\end{align*}
and $L_{q,NT}(\beta)$ is a second-order polynomial in $\beta$; namely,
\begin{align*}
   L_{q,NT}(\beta) \, &= L_{NT}(\beta^0)
                       \, - \, \frac 2 {\sqrt{NT}}  \, (\beta-\beta^0)' \, C_{NT}
                        + \, (\beta-\beta^0)' \, W_{NT} \, (\beta-\beta^0) \; ,
\end{align*}
with $K\times K$ matrix $W_{NT}$ defined by
$W_{NT,k_1 k_2} = (NT)^{-1} \, {\rm Tr}(M_{f^0} \, X^{\prime}_{k_1} \, M_{\lambda^0} \, X_{k_2})$, and $K$-vector $C_{NT}$ with entries
$C_{NT,k} = C^{(1)}\left(\lambda^0 \, ,f^0 \, ,X_k \, e \right)
           +C^{(2)}\left(\lambda^0 \, ,f^0 \, ,X_k \, e \right)$,
where
\begin{align*}
  C^{(1)}\left(\lambda^0,\, f^0,\, X_{k},\, e \right) &= \frac 1 {\sqrt{NT}} \, {\rm Tr}(M_{f^0}
                       \, e^{\prime}\, M_{\lambda^0} \, X_k) \; ,  \notag \\
  C^{(2)}\left(\lambda^0,\, f^0,\, X_{k},\, e \right) &= - \, \frac 1 {\sqrt{NT}} \, \bigg[
       {\rm Tr}\left(e M_{f^0} \, e' \, M_{\lambda^0} \, X_k \,
              f^0 \, (f^{0\prime}f^0)^{-1} \, (\lambda^{0\prime}\lambda^0)^{-1} \, \lambda^{0\prime} \right)
    \nonumber \\ & \qquad \qquad \quad
       +{\rm Tr}\left(e^{\prime}M_{\lambda^0} \, e \, M_{f^0} \, X^{\prime}_k \,
              \lambda^0 \, (\lambda^{0\prime}\lambda^0)^{-1} \, (f^{0\prime}f^0)^{-1} \, f^{0\prime} \right)
    \nonumber \\ & \qquad \qquad \quad
       +{\rm Tr}\left(e^{\prime}M_{\lambda^0} \, X_k \, M_{f^0} \, e^{\prime}
                \, \lambda^0 \, (\lambda^{0\prime}\lambda^0)^{-1} \, (f^{0\prime}f^0)^{-1} \, f^{0\prime} \right)
                        \bigg]  \; .
\end{align*}

\end{theorem}

We refer to $W_{NT}$ and $C_{NT}$ as the approximated Hessian and
the approximated score (at the true parameter $\beta^0$).
The exact Hessian and the exact score (at the true parameter $\beta^0$) contain higher-order expansion terms in $e$,
but the expansion up to the particular order above is sufficient
to work out the first-order asymptotic theory of the LS estimator, as the following
corollary shows.

\begin{corollary}
  \label{cor:limit}
Let the assumptions of Theorem~\ref{th:consistency} and \ref{th:ass_expand} hold;
let $\beta^0$ be an interior point of the parameter set $\mathbb{B}$;
and  assume $C_{NT}={\cal O}_p(1)$.
We then have
 $\sqrt{NT} \big( \widehat \beta - \beta^0 \big) =  W^{-1}_{NT} C_{NT} + o_p(1) = {\cal O}_p(1)$.
\end{corollary}

Combining consistency of the LS estimator and the expansion of the profile objective function in Theorem~\ref{th:ass_expand}, one obtains
$\sqrt{NT} \, W_{NT} \big( \widehat \beta - \beta^0 \big) =   C_{NT} + o_p(1)$;
see, for example, Andrews \cite*{Andrews1999}.
To obtain the corollary, one needs in addition that
$W_{NT}$ does not become degenerate as $N,T \rightarrow \infty$; that is,
the smallest eigenvalue of $W_{NT}$ should be bounded from below by a positive constant. Our assumptions
already guarantee existence of such a lower bound, as is shown in the supplementary material.

Analogous to the expansions of the profile objective function $L_{NT}(\beta)$, one can also derive 
expansions of the projectors $M_{\widehat \lambda}$ and $M_{\widehat f}$, and those can be used to show 
consistency of $\widehat \lambda$ and $\widehat f$, up to normalization; see Lemma~\ref{lemma:lambdafINV} in the 
supplementary material.

\subsection{Asymptotic Distribution}

We now apply Corollary \ref{cor:limit} to work out the asymptotic distribution of the LS estimator $\widehat \beta$.
For this purpose, we need more specific assumptions on $\lambda^0$, $f^0$, $X_k$, and $e$.

\begin{assumption}
   \label{ass:A5}
   A sigma algebra ${\cal C} = {\cal C}_{NT}$ (which in the following we will refer
   to as the conditioning set) exists that contains
      the sigma algebra generated by $\lambda^0$ and $f^0$, such that
      
      \begin{itemize}

          \item[(i)]  
         $\mathbb{E}\left[ e_{it} \, \big| \, {\cal C} \vee \sigma( \{ (X_{is}, e_{i,s-1}), s \leq t \}) \right] = 0$, 
        for all $i,t$.\footnote{%
          Here and in the following, we write $\sigma(A)$ for the sigma algebra generated by the
          (collection of) random variable(s) $A$, and we write ${\cal A} \vee {\cal B}$ for the sigma algebra generated
          by the unions of all elements in the sigma algebra ${\cal A}$ and ${\cal B}$, so that in the conditional expectation
          in Assumption~\ref{ass:A5}(ii), we condition jointly on ${\cal C}$ and $\{ (X_{is}, e_{i,s-1}), s \leq t \}$.
          }

          \item[(ii)] $e_{it}$ is independent over $t$, conditional on ${\cal C}$, for all $i$.

           \item[(iii)] $\{ (X_{it}, e_{it}) , t=1,\ldots,T \}$ is independent across $i$,
          conditional on ${\cal C}$.

          \item[(iv)]   $\frac 1 {NT} \sum_{i=1}^N
            \sum_{t,s=1}^T  \left| {\rm Cov}\left( X_{k,it},  X_{\ell,is} \Big| \, {\cal C} \right) \right|
            = {\cal O}_p(1)$, for all $k,\ell=1,\ldots,K$.

          \item[(v)] $\frac 1 {N T^2} \sum_{i=1}^N  \sum_{t,s,u,v=1}^T
              \left| {\rm Cov}\left( e_{it} \widetilde X_{k,is},  \, e_{iu} \widetilde X_{\ell,iv}   \Big| \, {\cal C} \right) \right|
              = {\cal O}_p(1)$, where $\widetilde X_{k,it} = X_{k,it} - \mathbb{E}\left[ X_{k,it} \big| {\cal C} \right]$, \\
              for all $k,\ell=1,\ldots,K$.

          \item[(vi)] 
          An $\epsilon>0$ exists such that
          $\mathbb{E}\left( e_{it}^8  \big| \, {\cal C} \right)$ and $\mathbb{E}\left( \| X_{it} \|^{8+\epsilon} \big| \, {\cal C} \right)$
          and $\mathbb{E} \| \lambda^0_i \|^4$ and
          $\mathbb{E} \| f^0_t \|^{4+\epsilon}$
           are bounded by a non-random constant,
                   uniformly over $i,t$ and $N,T$.

          \item[(vii)] $\beta^0$ is an interior point of the compact parameter set $\mathbb{B}$.

   \end{itemize}

\end{assumption}

\subsubsection*{Remarks on Assumption~\ref{ass:A5}}

\begin{itemize}
    \item[(1)] Part $(i)$ of Assumption~\ref{ass:A5} imposes
       that $e_{it}$ is a martingale difference sequence over time
       for a particular filtration.
       Conditioning on ${\cal C}$, the time series of $e_{it}$ is independent
       over time (part $(ii)$ of the assumption) and
       the error term $e_{it}$ and regressors $X_{it}$ are
       cross-sectionally independent (part $(iii)$
       of the assumption), but
       unconditional correlation is allowed.
       Part $(iv)$ imposes weak time-serial correlation of $X_{it}$.
       Part $(v)$ demands weak time-serial correlation of
       $\widetilde X_{k,it} =  X_{k,it} - \mathbb{E}\left[ X_{k,it} \big| {\cal C} \right]$
       and $e_{it}$.
       Finally, parts $(vi)$ and $(vii)$ require 
        bounded higher moments of the error term, regressors,
        factors and factor loadings,
          and a compact parameter set with an interior true parameter.

    \item[(2)] Assumption~\ref{ass:A5}$(i)$  
        implies
              $\mathbb{E} \left(X_{k,it}e_{it} | \mathcal{C} \right)
        = 0$
        and 
        $\mathbb{E} \left(X_{k,it}e_{it} X_{\ell,is}e_{is} | \mathcal{C} \right) = 0$
           for $t \neq s$. Thus, the assumption guarantees
            $X_{it}e_{it}$ is mean zero
            and  uncorrelated over~$t$, and independent across~$i$, conditional on ${\cal C}$.
            Notice the conditional mean independence restriction in Assumption~\ref{ass:A5}$(i)$ is weaker than Assumption~D of Bai (2009), besides sequential exogeneity. Bai imposes independence between $e_{it}$ and $(\{ X_{js},\lambda_j,f_s \}_{j,s})$.  
	   
    \item[(3)]  Assumption~\ref{ass:A5} is sufficient for
    Assumption~\ref{ass:A2}.
      To see this, notice
      $ {\rm Tr}(X_k \, e^{\prime}) = \sum_{i,t}  X_{k,it}e_{it}$,
      and also that the sequential exogeneity and the cross-sectional independence assumption imply
        $\mathbb{E}\left[ \left( (NT)^{-1}  \sum_{i,t}   X_{k,it}e_{it} \right)^2 \Big| {\cal C} \right]
        =(NT)^{-2} \sum_{i,t} \mathbb{E}\left[ \left(     X_{k,it}e_{it} \right)^2\Big| {\cal C} \right]$.
      Then, together with the assumption of bounded moments, we have 
      $(NT)^{-1}  \sum_{i,t}   X_{k,it}e_{it} = o_p(1)$.

    \item[(4)]   Assumption~\ref{ass:A5} is also sufficient for  Assumption~\ref{ass:A3}$^*$
    (and thus for  Assumption~\ref{ass:A3}), because $e_{it}$ is assumed independent over $t$ and
    across $i$ and has a bounded fourth moment,
    conditional on ${\cal C}$,
     which by using results in Latala \cite*{Latala2006},
    implies the spectral norm satisfies $\| e \| = \sqrt{\max(N,T)}$
    as $N$ and $T$ become large; see the supplementary material.

    \item[(5)] Examples of regressor processes, which satisfy Assumptions~\ref{ass:A5}$(iv)$ and $(v)$, are discussed in the following.
    These examples  also illuminate the role of the conditioning
    sigma field~${\cal C}$.

\end{itemize}

\subsubsection*{Examples of DGPs for $X_{it}$}

Here we provide examples of the DGPs of the regressors $X_{it}$ that satisfy the conditions  in
Assumption~\ref{ass:A5}.  Proofs for these examples are provided in the supplementary material.

\begin{example}
    The first example is a simple AR(1) interactive fixed effect regression:
\[
    Y_{it} = \beta^0 Y_{i,t-1} + \lambda_{i}^{0\prime} f_{t}^0 + e_{it},
\]
where $e_{it}$ is mean zero, independent across $i$ and $t$, and independent of $\lambda^0$ and $f^0$.
Assume $|\beta^0| < 1$ and that
$e_{it}$, $\lambda_{i}^{0}$, and $f_t^0$ all possess uniformly bounded moments of order $8+\epsilon$.
In this case, the regressor  is
$X_{it} = Y_{it-1} = \lambda_{i}^{0\prime} F_{t}^{0} + U_{it}$,
where $F_{t}^{0} = \sum_{s=0}^{\infty} (\beta^{0})^{s} f_{t-1-s}^{0}$ and $U_{it} = \sum_{s=0}^{\infty} (\beta^{0})^{s} e_{i,t-1-s}$.
For the conditioning sigma field $\mathcal{C}$ in Assumption~\ref{ass:A5}, we choose $\mathcal{C} = \sigma \left( \{ \lambda^0_{i} : 1 \leq i \leq N \}, \{ f^0_t : 1 \leq t \leq T \} \right) $.
Conditional on ${\cal C}$, the only variation in
$X_{it}$ stems from $U_{it}$, which is independent across $i$ and weakly correlated over $t$, so that
Assumption~\ref{ass:A5}$(iv)$ holds.
Furthermore, we have
 $\mathbb{E} \left( X_{it} | \mathcal{C} \right) = \lambda_{i}^{0\prime} F_{t}^{0}$ and $\widetilde X_{it} = U_{it}$,
 which allows us to verify Assumption~\ref{ass:A5}$(v)$.

 This example can be generalized to a ${\rm VAR}(1)$ model as follows:
\begin{align}
   {Y_{it} \choose Z_{it}} &=  {\cal B} \, 
   \underbrace{ {Y_{i,t-1} \choose Z_{i,t-1}} }_{=X_{it}}
                             + {\lambda^{0 \prime}_i f^{0}_t \choose d_{it}} +
                           \underbrace{   
                              {e_{it} \choose u_{it}} }_{=E_{it}} \; ,
   \label{VAR}
\end{align}
where $Z_{it}$ is an $m \times 1$ vector of additional variables and ${\cal B}$ is an
$(m+1)\times(m+1)$ matrix of VAR parameters
whose eigenvalues lie within the unit circle.
The $m\times 1$ vector $d_{it}$ and the factors $f^0_t$ and factor loadings
 $\lambda^0_i$ are assumed to be independent of the $(m+1) \times 1$ vector of innovations $E_{it}$.
 Suppose our interest is to estimate the first
row in equation \eqref{VAR}, which corresponds exactly to 
our interactive fixed effects model with regressors $Y_{i,t-1}$ and $Z_{i,t-1}$.
Choosing $\mathcal{C}$ to be the sigma field generated by all $f^0_t$, $\lambda^0_i$, $d_{it}$,
we obtain $\widetilde X_{it} = \sum_{s=0}^{\infty} {\cal B}^{s} E_{i,t-1-s}$.
Analogous to the AR(1) case, we then find 
Assumption~\ref{ass:A5}$(iv)$ and $(v)$  are satisfied in this example if the innovations $E_{it}$
are independent across $i$ and over $t$, and have appropriate bounded moments. 
\end{example}

\begin{example}
Consider a scalar $X_{it}$ for simplicity, and let
$X_{it}=g\left( v_{it}, \delta_{i},h_{t}\right)$.
We assume (i) $\left\{ \left( e_{it},v_{it}\right)
_{i=1,...,N;t=1,...,T}\right\} \bot \left\{ \left( \lambda_{i}^{0},\delta
_{i}\right)_{i=1,...,N},\left( f_{t}^{0},h_{t}\right)_{t=1,...,T}\right\} ,
$ (ii) $\left( e_{it},v_{it},\delta_{i}\right) $ are independent across $i$
for all $t,$ and (iii) $v_{is}\perp e_{it}$ for $s\leq t$ and all $i$.
Furthermore, assume $\sup_{it}\mathbb{E} |X_{it}|^{8+\epsilon} < \infty$ for some positive $\epsilon$.
For the conditioning sigma field $\mathcal{C}$ in Assumption~\ref{ass:A5}, we choose
$\mathcal{C}=\sigma \left( \left\{ \lambda_{i}^{0}:1\leq i\leq N\right\} ,%
\text{ }\left\{ \delta_{i}:1\leq i\leq N\right\} ,\text{ }\left\{
f_{t}^{0}: -\infty \leq t\leq \infty \right\} ,\text{ }\left\{ h_{t}: -\infty \leq t\leq \infty\right\}
\right) .$
Furthermore, as in Hahn and Kuersteiner~\cite*{HahnKuersteiner2011},
let
 $\mathcal{F}_{\tau}^{t}(i) = 
 {\cal C} \vee  \sigma \left( \{ (e_{is},v_{is}) : \tau \leq s \leq t \} \right)$, and define the conditional 
$\alpha$-mixing coefficient on $\mathcal{C}$:
\[
    \alpha_{m}(i) = \sup_{A \in \mathcal{F}_{-\infty}^{t}(i), B \in \mathcal{F}_{t+m}^{\infty}(i)} \left[ \mathbb{P} \left(A \cap B \right)
     - \mathbb{P} \left(A \right) \mathbb{P} \left(B \right) | \mathcal{C} \right].
\]
Let $\alpha_m = \sup_{i} \alpha_{m}(i)$, and assume
$  \alpha_m = O \left( m^{-\zeta} \right), \text{ where } \zeta > \frac{12 \, p}{4p-1}$ for $p>4$.
Then, Assumptions~\ref{ass:A5}$(iv)$ and $(v)$  are satisfied.

In this example, the  shocks $h_t$ (which may contain the factors $f^0_t$),
$\delta_i$ (which may contain the factor loadings $\lambda^0_i$),
and $v_{it}$ (which may contain past values of $e_{it}$)
 can enter in a general non-linear way into the regressor  $X_{it}$.

\end{example}

The following assumption guarantees the limiting variance and the asymptotic bias converge to constant values.

\begin{assumption}
   \label{ass:A6}
   Let ${\cal X}_k=M_{\lambda^0} \, X_{k} \, M_{f^0}$, which is an $N \times T$ matrix with
   entries ${\cal X}_{k,it}$.
    For each $i$ and $t$, define the $K$-vector ${\cal X}_{it}=({\cal X}_{1,it},\ldots,{\cal X}_{K,it})'$.
   We assume existence of the following probability limits
   for all $k=1,\ldots,K$:
     \begin{align*}
       W &= \limfunc{plim}_{N,T \rightarrow \infty}
          \frac 1 {NT} \, \sum_{i=1}^N \, \sum_{t=1}^T \, {\cal X}_{it} \, {\cal X}_{it}' \, ,
       \nonumber \\
       \Omega &=
       \limfunc{plim}_{N,T \rightarrow \infty}
    \frac 1 {NT} \, \sum_{i=1}^N \, \sum_{t=1}^T
            \, e_{it}^2  {\cal X}_{it} \, {\cal X}_{it}'  \; ,  
       \nonumber \\
       B_{1,k} &= \limfunc{plim}_{N,T \rightarrow \infty}
       \frac 1 {N} \, {\rm Tr} \left[ P_{f^0} \mathbb{E}\left( e' X_k \, \big| \, {\cal C} \right) \right] \; ,
       \nonumber \\
       B_{2,k} &= \limfunc{plim}_{N,T \rightarrow \infty}
       \frac 1 T
         {\rm Tr}\left[ \mathbb{E}\left( e e'  \, \big| \, {\cal C}  \right)   \, M_{\lambda^0} \, X_k \,
              f^0 \, (f^{0\prime}f^0)^{-1} \, (\lambda^{0\prime}\lambda^0)^{-1} \, \lambda^{0\prime} \right] \; ,
       \nonumber \\
       B_{3,k} &=
       \limfunc{plim}_{N,T \rightarrow \infty}
       \frac 1 N
       {\rm Tr}\left[ \mathbb{E}\left( e' e \, \big| \, {\cal C}  \right) \, M_{f^0} \, X^{\prime}_k \,
              \lambda^0 \, (\lambda^{0\prime}\lambda^0)^{-1} \, (f^{0\prime}f^0)^{-1} \, f^{0\prime} \right]
                   \; ,
    \end{align*}
    where ${\cal C}$ is the same conditioning set that appears in Assumption~\ref{ass:A5}.
\end{assumption}
Here, $W$ and $\Omega$ are $K\times K$ matrices,
and we define the $K$-vectors $B_1$, $B_2$, and $B_3$ with components
$B_{1,k}$, $B_{2,k}$ and $B_{3,k}$, $k=1,\ldots,K$.

\begin{theorem}[\bf Asymptotic Distribution]
   \label{th:limdis}
   Let Assumptions~\ref{ass:A1}, \ref{ass:A4}, \ref{ass:A5}, and \ref{ass:A6} be satisfied,\footnote{%
    Assumption~\ref{ass:A2} and~\ref{ass:A3}$^*$ are implied by Assumption~\ref{ass:A5} and therefore
    need not be explicitly assumed here.}
   and consider the limit
   $N,T \rightarrow \infty$ with $N/T \rightarrow \kappa^2$, where
                  $0<\kappa<\infty$.
   Then we have
   \begin{align*}
      \sqrt{NT} \left( \widehat \beta - \beta^0 \right) \, \limfunc{\rightarrow}_d
              \, {\cal N} \left( W^{-1} B , \; W^{-1} \, \Omega \, W^{-1} \right) \; ,
   \end{align*}
   where $B=-\kappa B_1-\kappa^{-1} B_2 - \kappa B_3$.
\end{theorem}

From Corollary \ref{cor:limit}, we already know the limiting distribution of $\widehat \beta$
is given by the limiting distribution of $W_{NT}^{-1} C_{NT}$.
Note $W_{NT} =   \frac 1 {NT} \, \sum_{i=1}^N \, \sum_{t=1}^T \, {\cal X}_{it} \, {\cal X}_{it}'$; that is, $W$ is simply defined as the probability limit of $W_{NT}$. Assumption \ref{ass:A4} guarantees $W$ is positive definite,
as shown in the supplementary material.

Thus, the main task in proving
Theorem~\ref{th:limdis} is to show the approximated score at the true parameter satisfies
$C_{NT} \limfunc{\rightarrow}_d  \, {\cal N} \left( B , \Omega \right)$.
We find the
 asymptotic variance $\Omega$ and the asymptotic bias $B_1$ originate from the $C^{(1)}$ term,
whereas the two further bias terms $B_2$ and $B_3$ originate from the $C^{(2)}$ term of $C_{NT}$.

The bias $B_1$ is due to correlation of the errors $e_{it}$ and the regressors $X_{k,i\tau}$ in the time direction
(for $\tau>t$). This bias term generalizes the  Nickell~\cite*{Nickell1981} bias that occurs in dynamic models
with standard fixed effects, and it is not present in Bai \cite*{Bai2009}, where only strictly exogenous regressors are considered.

The other two bias terms $B_2$ and $B_3$ are already described in Bai \cite*{Bai2009}.
If $e_{it}$ is homoscedastic, that is, if $\mathbb{E}(e_{it} | {\cal C}) = \sigma^2$,
then $\mathbb{E}\left( e e' | {\cal C} \right) = \sigma^2 \mathbb{I}_N$ and
$\mathbb{E}\left( e' e | {\cal C} \right) = \sigma^2 \mathbb{I}_T$, so that $B_2 = 0$ and $B_3=0$
(because the trace is cyclical and $ f^{0\prime}  M_{f^0} =0$ and $\lambda^{0\prime} M_{\lambda^0} = 0$).
Thus, $B_2$ is only non-zero if $e_{it}$ is heteroscedastic across $i$, and $B_3$ is only non-zero
if $e_{it}$ is heteroscedastic over $t$. Correlation in $e_{it}$ across $i$ or over $t$
would also generate non-zero bias terms of exactly the form $B_2$ and $B_3$, but is ruled out by our
assumptions.

\subsection{Bias Correction}
\label{sec:BiasCorrection}

Estimators for $W$, $\Omega$, $B_1$, $B_2$, and $B_3$ are obtained by forming suitable sample
analogs and replacing the unobserved $\lambda^0$, $f^0$, and $e$ by the estimates
$\widehat \lambda$, $\widehat f$, and the residuals $\widehat e$.

\begin{definition}
   \label{def:estimators}
    Let $\widehat{{\cal X}}_k=M_{\widehat \lambda} \, X_{k} \, M_{\widehat f}$. For each $i$ and $t$, define the $K$-vector
    $\widehat{{\cal X}}_{it}=(\widehat{{\cal X}}_{1,it},\ldots,\widehat{{\cal X}}_{K,it})'$.
Let $\Gamma: \mathbb{R} \rightarrow \mathbb{R}$ be the truncation kernel defined by $\Gamma(x)=1$ for $|x|\leq 1$, and $\Gamma(x)=0$
otherwise. Let $M$ be a bandwidth parameter that depends on $N$ and $T$.
    We define the $K\times K$ matrices $\widehat W$ and $\widehat \Omega$, and the
    $K$-vectors $\widehat B_{1}$, $\widehat B_{2}$, and $\widehat B_{3}$ as follows:
    \begin{align*}
       \widehat W &=
  \frac 1 {NT} \, \sum_{i=1}^N \, \sum_{t=1}^T \, \widehat{{\cal X}}_{it} \, \widehat{{\cal X}}_{it}' \; ,
       \nonumber \\
       \widehat \Omega &= \frac 1 {NT} \, \sum_{i=1}^N \, \sum_{t=1}^T \, (\widehat e_{it})^2
                                                     \, \widehat{{\cal X}}_{it} \, \widehat{{\cal X}}_{it}' \; ,
       \nonumber \\
       \widehat B_{1,k} &= \frac 1 N \, 
       \sum_{i=1}^N
       \sum_{t=1}^{T-1}
       \sum_{s=t+1}^T
        \Gamma\left( \frac{s-t} M  \right)
       \, \big[ P_{\widehat f} \big]_{t s} 
       \;   \widehat e_{it} \; X_{k,is}    \; ,
       \nonumber \\
       \widehat B_{2,k} &=
       \frac 1 T
        \sum_{i=1}^N
       \sum_{t=1}^T
        (\widehat e_{it})^2 
        \left[ M_{\widehat \lambda} \, X_k \,
              \widehat f \, (\widehat f^{\prime} \widehat f)^{-1} \,
                             (\widehat \lambda^{\prime}\widehat \lambda)^{-1} \, \widehat \lambda^{\prime} \right]_{ii} \; ,
       \nonumber \\
       \widehat B_{3,k} &=
       \frac 1 N
       \sum_{i=1}^N
       \sum_{t=1}^T
        (\widehat e_{it})^2 
        \left[ M_{\widehat f} \, X^{\prime}_k \,
              \widehat \lambda \, (\widehat \lambda^{\prime} \widehat \lambda)^{-1} \,
                             (\widehat f^{\prime}\widehat f)^{-1} \, \widehat f^{\prime} \right]_{tt}
                   \; ,
    \end{align*}
    where $\widehat e
    = Y \, - \,  \widehat \beta \cdot X \, - \, \widehat \lambda  \, \widehat f'$,
    and $\widehat e_{it}$ denotes the elements of 
    $\widehat e$,
    $\left[ A \right]_{ts}$ denotes
                             the (t,s)th element of the matrix $A$.  
\end{definition}

Notice the estimators $\widehat\Omega$, $\widehat B_2$, and $\widehat B_3$ are similar to White's standard error estimator under heteroskedasticity, and the estimator $\widehat B_1$ is similar to the HAC estimator with a kernel.
To show consistency of these estimators, we impose some additional assumptions.

\begin{samepage}
\begin{assumption} $\phantom{a}$
    \label{ass:bc}
    \begin{itemize}
        \item[(i)] $\| \lambda^0_i \|$ and $\|f^0_t\|$ are uniformly bounded
        over $i$, $t$, and $N$, $T$.

        \item[(ii)]
          There exist $c>0$ and $\epsilon >0$ such that for all $i,t,m,N$, and $T$, we have \\
         $ \left| \frac 1 N  \sum_{i=1}^N  \mathbb{E} (e_{it}  X_{k,it+m} \, \big| \, {\cal C}) \right|
                            \leq  c \, m^{- (1+ \epsilon)}$.
    \end{itemize}
\end{assumption}
\end{samepage}

Assumption~\ref{ass:bc}$(i)$ is made for convenience to simplify the consistency proof
for the estimators in Definition~\ref{def:estimators}. Weakening this assumption is possible by only
assuming suitable bounded moments of $\| \lambda^0_i \|$ and $\|f^0_t\|$.
To  show consistency
of $\widehat B_1$, we need to control how strongly $e_{it}$ and $X_{k,i \tau}$, $t<\tau$, are allowed to be correlated,
which is done by Assumption~\ref{ass:bc}$(ii)$.
It is straightforward to verify Assumption~\ref{ass:bc}$(ii)$ is satisfied in the two examples of regressor processes presented after Assumption~\ref{ass:A5}.

\begin{theorem}[\bf Consistency of Bias and Variance Estimators]
\label{th:biascorrection}
   Let Assumptions~\ref{ass:A1}, \ref{ass:A4}, \ref{ass:A5}, \ref{ass:A6},  and \ref{ass:bc} hold,
   and consider a limit $N,T \rightarrow \infty$ with $N/T \rightarrow \kappa^2$,
                  $0<\kappa<\infty$,
   such that  the bandwidth $M=M_{NT}$ satisifies $M \rightarrow \infty$ and $M^5 / T \rightarrow 0$.
  We then have
  $\widehat W=W + o_p(1)$, $\widehat \Omega=\Omega + o_p(1)$,
  $\widehat B_1=B_1 + o_p(1)$, $\widehat B_2=B_{2} + o_p(1)$,
  and $\widehat B_3=B_{3} + o_p(1)$.
\end{theorem}

The assumption $M^5/T \rightarrow 0$ can be relaxed if additional higher- moment restrictions on
$e_{it}$ and $X_{k,it}$ are imposed. Note also that for the construction of the estimators $\widehat W$,
$\widehat \Omega$, and $\widehat B_i$, $i=1,2,3$, knowing whether the regressors are strictly exogenous or
predetermined is unnecessary; in both cases, the estimators for $W$, $\Omega$, and $B_i$, $i=1,2,3$, are consistent.
We can now present our bias-corrected estimator and its limiting distribution.

\begin{corollary}
   \label{corr:biascorrected}
   Under the assumptions of Theorem~\ref{th:biascorrection}, the bias-corrected estimator
   \begin{align*}
      \widehat \beta^* &= \widehat \beta + \widehat W^{-1}
            \left( T^{-1} \widehat B_1 + N^{-1} \widehat B_2  + T^{-1} \widehat B_3 \right)
   \end{align*}
   satisfies $\sqrt{NT} \left( \widehat \beta^* - \beta^0 \right) \, \limfunc{\rightarrow}_d
              \, {\cal N} \left(0 , \; W^{-1} \, \Omega \, W^{-1} \right)$.
\end{corollary}

According to Theorem~\ref{th:biascorrection}, a consistent estimator of the asymptotic variance
of $\widehat \beta^*$ is given by $\widehat W^{-1} \, \widehat \Omega \, \widehat W^{-1}$.

An alternative to the analytical bias-correction result given by Corollary~\ref{corr:biascorrected} is
to use Jackknife bias correction to eliminate the asymptotic bias.
For panel models with incidental parameters only in the cross-sectional dimensions,
one typical finds a large $N,T$ leading  incidental parameter bias of order $1/T$ for the parameters of interest.
To correct for this $1/T$ bias, one can use the delete-one Jackknife bias correction if observations are iid over $t$
\cite{HahnNewey2004} and the split-panel Jackknife bias-correction if observations are correlated over $t$ \cite{DhaeneJochmans2015}.
In our current model, we have incidental parameters in both panel dimensions ($\lambda^0_i$ and $f^0_t$),
resulting in leading bias terms of order $1/T$ (bias term $B_1$ and $B_3$) and of order $1/N$ (bias term $B_2$).
Fern{\'a}ndez-Val and Weidner~\cite*{FernandezValWeidner2013} discuss the generalizations of the split-panel Jackknife bias-correction to that case.

The corresponding bias-corrected split-panel Jackknife
estimator reads $\widehat \beta^J = 3 \widehat \beta_{NT} - \overline \beta_{N,T/2} - \overline \beta_{N/2,T}$,
where $\widehat \beta_{NT} = \widehat \beta$ is the LS estimator obtained from the full sample,
$\overline \beta_{N,T/2}$ is the average of the two LS estimators that leave out the first and second halves of the time
periods, and $\overline \beta_{N/2,T}$ is the average of the two LS estimators that leave out half of the individuals.
Jackknife bias correction is convenient because only the order of the bias, and not the structure of the terms $B_1$, $B_2$, and $B_3$, needs not be known
in detail. However, one requires
additional stationarity assumptions over $t$ and homogeneity assumptions across $i$ to justify the Jackknife
correction and to show that
$\widehat \beta^J$ has the same limiting distribution as $\widehat \beta^*$ in Corollary~\ref{corr:biascorrected};
see Fern{\'a}ndez-Val and Weidner~\cite*{FernandezValWeidner2013} for more details. They also observe
through Monte Carlo simulations that the finite
sample variance of the Jackknife-corrected estimator is often larger
than of the analytically corrected estimator. We do not explore
Jackknife bias-correction further in this paper.

\section{Testing Restrictions on $\beta^0$}
\label{sec:testing}

In this section, we discuss the three classical test statistics for
testing linear restrictions on $\beta^0$. The null hypothesis is $H_0: \;   H \beta^0 = h$,
and the alternative is $H_a: \;   H \beta^0 \neq h$,
where $H$ is an $r \times K$ matrix of rank $r \leq K$, and $h$ is an $r \times 1$ vector. We restrict the presentation to testing a linear hypothesis for ease of exposition. One can generalize the discussion
to the testing of non-linear hypotheses, under conventional regularity
conditions.
Throughout this subsection, we assume $\beta^0$
is an interior point of $\mathbb{B}$; that is, no local restrictions are on $\beta$
as long as the null hypothesis is not imposed. Using the expansion of $L_{NT}(\beta)$,
one could also discuss testing when the true parameter is on the boundary, as shown in
Andrews \cite*{Andrews2001}.

The restricted estimator is defined by
\begin{align}
   \widetilde \beta &= \limfunc{argmin}_{\beta \in \widetilde{\mathbb{B}}}\,L_{NT}\left( \beta \right) \; ,
   \label{DefBtilde}
\end{align}
where $\widetilde{\mathbb{B}} = \{ \beta \in \mathbb{B}| \, H\beta = h \}$ is the restricted parameter set. Analogous to Theorem \ref{th:limdis} for the
unrestricted estimator $\widehat \beta$, we can use the
expansion of the profile objective function to derive the
limiting distribution of the restricted estimator.
Under the assumptions of Theorem~\ref{th:limdis}, we have
\begin{align*}
  \sqrt{NT}(\widetilde \beta - \beta^0)
           \; \;  &\limfunc{\longrightarrow}_d \; \;
{\cal N} \left( \mathfrak{W}^{-1} B , \; \mathfrak{W}^{-1} \, \Omega \, \mathfrak{W}^{-1} \right) \; ,
\end{align*}
where $\mathfrak{W}^{-1} = W^{-1} - W^{-1} H' (H W^{-1} H')^{-1} H W^{-1}$.
The $K\times K$ covariance matrix in the limiting distribution of $\widetilde \beta$ is not full rank,
but satisfies
${\rm rank}(\mathfrak{W}^{-1} \, \Omega \, \mathfrak{W}^{-1})=K-r$,
because $H \mathfrak{W}^{-1}=0$ and thus ${\rm rank}(\mathfrak{W}^{-1})=K-r$.
The asymptotic distribution of $\sqrt{NT}(\widetilde \beta - \beta^0)$
is therefore $K-r$ dimensional, as it should be for the restricted estimator.

\subsubsection*{Wald Test}

Using the result of Theorem~\ref{th:limdis}, we find that under the null hypothesis,
$\sqrt{NT} \left( H \widehat \beta - h \right)$ is asymptotically distributed as
${\cal N} \left( H W^{-1} B , \; H W^{-1} \, \Omega \, W^{-1} H' \right)$.
Thus, due to the presence of the bias~$B$, the standard Wald test statistic
$WD_{NT} = NT \left( H \widehat \beta - h \right)'
                 \left( H \widehat W^{-1} \, \widehat \Omega \, \widehat W^{-1} H' \right)^{-1}
                 \left( H \widehat \beta - h \right)$
is not asymptotically $\chi^2_r$ distributed.
Using the estimator $\widehat B = - \sqrt{\frac N T}  \,  \widehat B_{1}
                                - \sqrt{\frac T N}  \,  \widehat B_{2}
                                 - \sqrt{\frac N T}  \,  \widehat B_{3}$
for the bias, we can define the bias-corrected Wald test statistic as
\begin{align}
   WD^*_{NT} &=    \left[ \sqrt{NT} \left( H \widehat \beta^* - h \right) 
   \right]'
                 \left( H \widehat W^{-1} \, \widehat \Omega \, \widehat W^{-1} H' \right)^{-1}
                 \left[ \sqrt{NT} \left( H \widehat \beta^* - h \right)  \right] ,
   \label{DefWDs}
\end{align}
where $\widehat \beta^* = \widehat \beta - \widehat W^{-1} \widehat B$
is the bias-corrected estimator. $WD^*_{NT}$ is just the standard Wald test
statistics applied to $\widehat \beta^*$.
Under the null hypothesis
and the Assumptions of Theorem~\ref{th:biascorrection},
we find
$WD^*_{NT} \,  \limfunc{\rightarrow}_d \, \chi^2_r$.

\subsubsection*{Likelihood Ratio Test}

To implement the LR test, we need the relationship between the
asymptotic Hessian $W$
and the asymptotic score variance $\Omega$ of the profile objective function
to be of the form $\Omega=c W$, where $c>0$ is a scalar constant.
This condition is satisfied in our interactive fixed effect model
if $\mathbb{E}(e_{it}^2 | {\cal C}) = c$, that is, if the error is homoskedastic.
A consistent estimator for $c$ is then given by
$\widehat c = (NT)^{-1} \sum_{i=1}^N \sum_{t=1}^T \widehat e_{it}^2$, where $\widehat e
    = Y \, - \,  \widehat \beta \cdot X \, - \, \widehat \lambda  \, \widehat f'$.
Because the likelihood function for the interactive fixed effect model is just the sum of squared residuals,
we have $\widehat c = L_{NT}( \widehat \beta )$.
The likelihood ratio test statistic is defined by
\begin{align*}
   LR_{NT} &=\widehat c^{-1} \, NT \left[ L_{NT}\left(\widetilde \beta \right) - L_{NT}\left(\widehat \beta \right) \right] \; .
\end{align*}
Under the assumption of Theorem~\ref{th:limdis}, we then have
\begin{align*}
   LR_{NT} \; \;  \limfunc{\longrightarrow}_d  \quad &
    \, c^{-1} C' W^{-1} H' (H W^{-1} H')^{-1} H W^{-1} C \; ,
\end{align*}
where $C \sim {\cal N}(B,\Omega)$,
i.e. $C_{NT} \limfunc{\rightarrow}_d C$.
It is the same limiting distribution that one finds for the Wald test if $\Omega=c W$
(in fact, one can show $WD_{NT}=LR_{NT}+o_p(1)$). Therefore, we need to do a bias-correction for the LR test
to achieve a $\chi^2$ limiting distribution.
We define
\begin{align}
   LR_{NT}^{*} &=  \widehat c^{-1} NT\left[
    \min_{\{ \beta \in \mathbb{B}| \, H\beta = h \}}
     L_{NT}\left( \beta + (NT)^{-1/2} \widehat W^{-1} \widehat B \right)
   - \min_{\beta \in \mathbb{B}}
     L_{NT}\left( \beta + (NT)^{-1/2} \widehat W^{-1} \widehat B \right) \right] ,
   \label{DefLRs}
\end{align}
where $\widehat B$ and $\widehat W$ do not depend on the parameter $\beta$
in the minimization problem.\footnote{%
Alternatively,
one could use $\widehat B(\widetilde \beta)$ and $\widehat W(\widetilde \beta)$ as estimates for $B$ and $W$,
and would obtain the same limiting distribution of $LR_{NT}^*$ under the null hypothesis $H_0$.
These alternative estimators are not consistent if $H_0$ is false, i.e.
the power-properties of the test would be different. The question of which specification
should be preferred is left for future research.
}
Asymptotically, we have
$\min_{\beta \in \mathbb{B}}
  L_{NT}\left( \beta + (NT)^{-1/2} \widehat W^{-1} \widehat B \right) = L_{NT}(\widehat \beta)$,
because $\beta \in \mathbb{B}$ does not impose local constraints; in other words, close to $\beta^0$,
whether one minimizes over $\beta$ or
over $\beta + (NT)^{-1/2} \widehat W^{-1} \widehat B$ does not matter for the value of the minimum. 
The correction to the LR test therefore originates from the
first term in $LR_{NT}^{*}$. For the minimization over the restricted parameter set, 
whether the argument of $L_{NT}$ is $\beta$ or $\beta + (NT)^{-1/2} \widehat W^{-1} \widehat B$ matters, because
generically, we have $H W^{-1} B \neq 0$ (otherwise, no correction would be necessary for the LR statistics). One can show that
\begin{align*}
   LR^*_{NT} \; \; \limfunc{\longrightarrow}_d  \quad &
   c^{-1} (C-B)' W^{-1} H' (H W^{-1} H')^{-1} H W^{-1} (C-B) \; ;
\end{align*}
that is, we obtain the same formula as for $LR_{NT}$, but 
the bias-corrected term $C-B$ replaces the limit of the score $C$.
Under the Assumptions of Theorem~\ref{th:biascorrection},
if $H_0$ is satisfied, and for homoscedastic errors $e_{it}$,
we have $LR^*_{NT} \,  \limfunc{\rightarrow}_d \, \chi^2_r$.
In fact, one can show $LR^*_{NT} = WD^*_{NT}+o_p(1)$.

\subsubsection*{Lagrange Multiplier Test}

Let $\widetilde \nabla {\cal L}_{NT}$ be the
gradient of the LS objective function \eqref{DefCalL}
with respect to $\beta$, evaluated at the restricted parameter estimates; that is,
\begin{align*}
   \widetilde \nabla {\cal L}_{NT} \, = \,
   \nabla {\cal L}_{NT}(\widetilde \beta,\widetilde \lambda,\widetilde f) &=
          \left(
          \frac{\partial {\cal L}_{NT}(\beta,\widetilde \lambda,\widetilde f)}
               {\partial \beta_1} \bigg|_{\beta=\widetilde \beta} , \ldots ,
          \frac{\partial {\cal L}_{NT}(\beta,\widetilde \lambda,\widetilde f)}
               {\partial \beta_K} \bigg|_{\beta=\widetilde \beta} \right)'
      \nonumber \\
          &= \, - \, \frac 2 {NT} \Big(  {\rm Tr}\left(X'_1 \widetilde e\right) , \ldots ,
                                          {\rm Tr}\left(X'_K \widetilde e\right)  \Big)' \; ,
\end{align*}
where $\widetilde \lambda = \widehat \lambda(\widetilde \beta)$,
$\widetilde f = \widehat f(\widetilde \beta)$,
and  $\widetilde e =
   Y \, - \,  \widetilde \beta \cdot X \, - \, \widetilde \lambda  \, \widetilde f'$.
Under the assumptions of Theorem~\ref{th:limdis},
 and if the null hypothesis $H_0: \;   H \beta^0 = h$ is satisfied,
one finds that\footnote{The proof of the statement is given in the 
supplementary material as part of the proof
of Theorem~\ref{th:testing}.}
\begin{align}
   \sqrt{NT} \, \widetilde \nabla {\cal L}_{NT}
                   &=  \sqrt{NT} \, \nabla L_{NT}(\widetilde \beta) + o_p(1).
   \label{EquivGrads}
\end{align}
Due to this equation, one can base
the Lagrange multiplier test on the gradient of ${\cal L}_{NT}(\widetilde \beta,\widetilde \lambda,\widetilde f)$,
or on the gradient of the profile quasi-likelihood function $L_{NT}(\widetilde \beta)$, and obtain the same limiting distribution.

Using the bound on the remainder $R_{NT}(\beta)$ given in Theorem~\ref{th:ass_expand},
one cannot infer any properties of the score function, that is, of the gradient $\nabla L_{NT}(\beta)$,
because nothing is said about $\nabla R_{NT}(\beta)$. The following theorem gives
a bound
on $\nabla R_{NT}(\beta)$ that is sufficient to derive the limiting distribution of the Lagrange multiplier.

\begin{theorem}
   \label{th:gradient}
   Under the assumptions of Theorem~\ref{th:ass_expand}, and with $W_{NT}$ and $C_{NT}$ as
   defined there, the score function satisfies
   \begin{align*}
       \nabla L_{NT}(\beta) \, &= \, 2 \, W_{NT} \, (\beta-\beta^0)
                             \, - \, \frac 2 {\sqrt{NT}} \, C_{NT} \, + \frac 1 {NT} \, \nabla R_{NT}(\beta) \; ,
   \end{align*}
   where the remainder $\nabla R_{NT}(\beta)$ satisfies for any sequence $\eta_{NT}\rightarrow 0$:
  \begin{align*}
      \sup_{\{\beta :\left\| \beta -\beta^{0} \right\| \leq \eta_{NT}\}}
                  \frac{ \left\| \nabla R_{NT}(\beta) \right\| }
       { \sqrt{NT} \left(1 + \sqrt{NT} \, \left\| \beta -\beta^{0} \right\| \right) } = o_{p}\left( 1 \right) .
  \end{align*}
\end{theorem}
From this theorem, and the fact that $\widetilde \beta$ is $\sqrt{NT}$-consistent
under $H_0$,
we obtain
\begin{align*}
   \sqrt{NT} \, \widetilde \nabla {\cal L}_{NT}
                   &= \sqrt{NT} \, \nabla L_{q,NT}(\widetilde \beta) +  o_p(1)
      \nonumber \\
                   &= 2 \sqrt{NT} \, W_{NT} \, (\widetilde \beta - \beta^0) - 2 \, C_{NT} + o_p(1) \; .
\end{align*}
Remember $\widetilde \beta$ is the restricted estimator defined
in equation~\eqref{DefBtilde}.
Using this result and the known limiting distribution of $\widetilde \beta$, we now find
\begin{align}
   \label{limNablaL}
   \sqrt{NT} \, \widetilde \nabla {\cal L}_{NT}
               \;\; \limfunc{\longrightarrow}_d \;\; - 2 H' (H W^{-1} H')^{-1} H W^{-1} C \; .
\end{align}
Note also that $\sqrt{NT} H W^{-1} \nabla L_{NT}(\widetilde \beta) \, \limfunc{\rightarrow}_d \, - 2 H W^{-1} C$.
We define $\widetilde B$, $\widetilde W$, and $\widetilde \Omega$,
analogous to $\widehat B$, $\widehat W$,
and $\widehat \Omega$, but with unrestricted parameter estimates replaced
by restricted parameter estimates.
The LM test statistic is then given by
\begin{align*}
  LM_{NT} &= \frac{NT} 4 \,
   (\widetilde \nabla {\cal L}_{NT})'
      \widetilde W^{-1}  H' (H \widetilde W^{-1} \widetilde \Omega \widetilde W^{-1} H')^{-1} H \widetilde W^{-1} \widetilde \nabla {\cal L}_{NT}
     \; .
\end{align*}
One can show the LM test is asymptotically equivalent to the Wald test: $LM_{NT}=WD_{NT}+o_p(1)$; that is, 
again, bias-correction is necessary. We define the bias-corrected LM test statistic as
\begin{align}
   LM^*_{NT} &= \frac{1} 4 \, \left(\sqrt{NT} \, \widetilde \nabla {\cal L}_{NT} + 2 \widetilde B \right)'
                   \widetilde W^{-1} H' (H \widetilde W^{-1} \widetilde \Omega \widetilde W^{-1} H')^{-1} H \widetilde W^{-1}
                      \left(\sqrt{NT} \, \widetilde \nabla {\cal L}_{NT} + 2 \widetilde B \right) \; .
   \label{DefLMs}
\end{align}

The following theorem summarizes the main results of the present subsection.
\begin{theorem}[\bf Chi-Square Limit of Bias-Corrected Test Statistics]
   \label{th:testing}
   Let the assumptions of Theorem~\ref{th:biascorrection}
  and the null hypothesis $H_0: \;   H \beta^0 = h$ be satisfied.
   For the bias-corrected Wald and LM test statistics introduced
   in equation \eqref{DefWDs} and \eqref{DefLMs}, we then have
   \begin{align*}
      WD^*_{NT} \; \;  &\limfunc{\longrightarrow}_d \; \; \chi^2_r \; , &
      LM^*_{NT} \; \;  &\limfunc{\longrightarrow}_d \; \; \chi^2_r \; .
   \end{align*}
   If, in addition, we assume $\mathbb{E}(e_{it}^2 | {\cal C}) = c$, that is, the idiosyncratic
   errors are homoscedastic,
   and we use $\widehat c = L_{NT}(\widehat \beta)$ as an estimator for $c$,
   the LR test statistic defined in equation \eqref{DefLRs} satisfies
   \begin{align*}
      LR^{*}_{NT} \; \;  &\limfunc{\longrightarrow}_d \; \; \chi^2_r \; .
   \end{align*}
\end{theorem}

\section{Extension to Endogenous Regressors}
\label{sec:Endogenous Regression}

In this section, we briefly discuss how to estimate the regression coefficient $\beta^0$ of Model~\eqref{model0} when some of the regressors in $X_{it}$ are endogenous with respect to the regression error $e_{it}$.
The question is how instrumental variables can be used to estimate the regression coefficients of the endogenous regressor in the presence of the interactive fixed effects~$\lambda_{i}^{0\prime }f_{t}^{0}$.

The existing literature has already investigated similar questions under various setups.
Harding and Lamarche~\cite*{HardingLamarche2009,HardingLamarche2011} investigate the problem of estimating an endogenous panel (quantile) regression with interactive fixed effects, and show how to use IVs in the CCE estimation framework.
Moon, Shum, and Weidner~\cite*{MoonShumWeidner2012} (hereafter MSW) estimate a random coefficient multinomial demand model (as in Berry, Levinsohn, and Pakes~\cite*{BerryLevinsohnPakes1995}) when the unobserved product-market characteristics have interactive fixed effects. The IVs are required to identify the parameters of the random coefficient
distribution and to control for price endogeneity.
They suggested a multi-step  ``least squares-minimum distance'' (LS-MD) estimator.\footnote{%
Chernazhukov and Hansen~\cite*{ChernozhukovHansen2005} also
used a similar method for estimating endogenous quantile regression models.}
The LS-MD approach is also applicable to linear panel regression models with endogenous regressors and interactive fixed effects, as demonstrated in Lee, Moon, and Weidner~\cite*{LeeMoonWeidner2012}
for the case of a dynamic linear panel regression model with interactive fixed effects and measurement error.

We now discuss how to implement the LS-MD estimation in our setup.
Let $X_{it}^{\rm end}$ be the vectors of endogenous regressors, and let $X_{it}^{\rm exo}$ be the vector
of exogenous regressors, with respect to $e_{it}$, such that $X_{it}= (X_{it}^{\rm end \prime},X_{it}^{\rm exo \prime})'$.
The model then reads
\[
Y_{it}=\beta_{\rm end}^{0\prime }X_{it}^{\rm end}+\beta_{\rm exo}^{0\prime
}X_{it}^{\rm exo}+\lambda_{i}^{0\prime }f_{t}^{0}+e_{it},
\]%
where $X_{it}^{\rm exo}$ denotes the exogenous
and $X_{it}^{\rm end}$ denotes the endogenous 
regressors (wrt to $e_{it}$).
Suppose $Z_{it}$ is an additional $L$-vector of
exogenous instrumental variables (IVs),
but $Z_{it}$ may be correlated with $\lambda^0_i$ and $f^0_t$.
The LS-MD estimator of $\beta ^{0}=\left( \beta_{\rm end}^{0\prime },\beta
_{\rm exo}^{0\prime }\right) ^{\prime }$ can then be calculated by the following
three steps:

\begin{itemize}
\item[(1)] For given $\beta_{\rm end}$, we run the least squares regression of $%
Y_{it}-\beta_{\rm end}^{\prime }X_{it}^{\rm end}$ on the included exogeneous
regressors $X_{it}^{\rm exo}$, the interactive fixed effects $\lambda
_{i}^{\prime }f_{t}$, and the IVs $Z_{it}:$%
\begin{align*}
& \left( \widetilde{\beta}_{\rm exo}\left( \beta_{\rm end}\right) ,\widetilde{\gamma}\left(
\beta_{\rm end}\right) ,\widetilde{\lambda}\left( \beta_{\rm end}\right) ,\widetilde{f}%
\left( \beta_{\rm end}\right) \right) \\
& \qquad = \; \limfunc{argmin}_{\{ \beta_{\rm exo},\gamma ,\lambda
,f \} } \; \sum_{i=1}^{N}\sum_{t=1}^{T}\left( Y_{it}-\beta_{\rm end}^{\prime
}X_{it}^{\rm end}-\beta_{\rm exo}^{\prime }X_{it}^{\rm exo}-\gamma ^{\prime
}Z_{it}-\lambda_{i}^{\prime }f_{t}\right) ^{2}.
\end{align*}

\item[(2)] We estimate $\beta_{\rm end}$ by finding
 $\widetilde{\gamma}\left( \beta
_{\rm end}\right) $, obtained by step (1), that is closest to zero.
 To do so, we choose a symmetric positive definite
 $L \times L$  weight matrix $W_{NT}^{\gamma }$ and compute
\[
\widehat{\beta}_{\rm end}=\limfunc{argmin}_{\beta_{\rm end}}\widetilde{\gamma}\left( \beta
_{\rm end}\right) ^{\prime } \, W_{NT}^{\gamma } \, \widetilde{\gamma}\left( \beta
_{\rm end}\right) .
\]

\item[(3)] We estimate $\beta_{\rm exo}$ (and $\lambda ,f)$ by running the
least squares regression of $Y_{it}-\widehat{\beta}_{\rm end}^{\prime }X_{it}^{\rm end}$
on the included exogeneous regressors $X_{it}^{\rm exo}$ and the interactive fixed
effects $\lambda_{i}^{\prime }f_{t}$:%
\begin{align*}
 \left( \widehat{\beta}_{\rm exo}  ,\widehat \lambda  ,\widehat{f}  \right)
  &= \; \limfunc{argmin}_{\{ \beta_{\rm exo},\gamma ,\lambda
,f \} } \; \sum_{i=1}^{N}\sum_{t=1}^{T}\left( Y_{it}- \widehat \beta_{\rm end}^{\prime
}X_{it}^{\rm end}-\beta_{\rm exo}^{\prime }X_{it}^{\rm exo} -\lambda_{i}^{\prime }f_{t}\right) ^{2}.
\end{align*}
\end{itemize}
The idea behind this estimation procedure is that valid instruments are excluded from the model
for $Y_{it}$, so that their first-step regression coefficients $\widetilde{\gamma}\left( \beta_{\rm end}\right)$
should be close to zero if $\beta_{\rm end}$ is close to its true value $\beta_{\rm end}^0$.
Thus, as long as $X_{it}^{\rm exo}$ and $Z_{it}$ jointly satisfy the assumptions of the current paper,
we obtain $\widetilde{\gamma}\left( \beta_{\rm end}^0 \right) = o_p(1)$ for the first-step LS estimator,
and we also obtain the asymptotic
distribution of $\widetilde{\gamma}\left( \beta_{\rm end}^0 \right)$ from the results derived in section~\ref{sec:limdist}.

However, to justify the second-step minimization formally, one needs to study
the properties of $\widetilde{\gamma}\left( \beta_{\rm end} \right)$ also for $ \beta_{\rm end} \neq \beta_{\rm end}^0$.
To do so, we refer to MSW.
Our $\beta_{\rm end},\beta_{\rm exo},$ and $Y_{it}-\beta_{\rm end}^{\prime
}X_{it}^{\rm end}$ correspond to their
$\alpha$, $\beta$, and $\delta_{jt}\left( \alpha \right)$, respectively.
Assumptions 1 - 5 in MSW can be translated accordingly, and
the results in MSW show large $N,T$ consistency and asymptotic normality
of the LS-MD estimator.

The final step of the LS-MD estimation procedure is essentially a repetition of the first step, but without including
$Z_{it}$ in the set of regressors, which results in some efficiency gains for $\widehat{\beta}_{\rm exo} $ compared
to the first step.

\section{Monte Carlo Simulations}
\label{sec:MC}

We consider an ${\rm AR}(1)$ model with $R=1$ factors:
\begin{align*}
   Y_{it} \, &= \, \rho^0 \, Y_{i,t-1} \, + \, \lambda^0_{i} \, f^0_{t} \, + \, e_{it} \; .
\end{align*}
We estimate the model as an interactive fixed effect model; that is, no distributional
assumptions on $\lambda^0_{i}$ and $f^0_{t}$ are made in estimation.
 The parameter of interest is $\rho^0$. The estimators we consider are the
OLS estimator (which completely ignores the presence of the factors), the least squares estimator with
interactive fixed effects (denoted FLS in this section to differentiate from OLS) defined in equation
\eqref{DefQMLE},\footnote{%
  Here we can either use $\mathbb{B}=(-1,1)$, or $\mathbb{B}=\mathbb{R}$. In the present model, we only
  have high-rank regressors; i.e., the parameter space need not be bounded to show consistency.}
and its bias-corrected version (denoted BC-FLS), defined in Theorem~\ref{corr:biascorrected}.

For the simulation, we draw the $e_{it}$ independently and identically distributed from a t-distribution with
five degrees of freedom,
the $\lambda^0_{i}$ independently distributed from ${\cal N}(1,1)$, and we generate the factors
from an ${\rm AR}(1)$ specification, namely, $f^0_{t}=\rho_f \, f^0_{t-1} + u_{t}$,
where $u_{t} \sim {\rm iid} {\cal N}(0,(1-\rho_f^2)\sigma_f^2)$, and
$\sigma_f$ is the standard deviation of $f^0_t$.
  For all simulations, we generate 1,000 initial time periods for $f^0_t$ and $Y_{it}$ that are not used for estimation.
  This approach guarantees the simulated data used for estimation are distributed according to the stationary distribution
  of the model.

This setup contains no correlation and heteroscedasticity in $e_{it}$; that is, only the bias
term $B_1$ of the FLS estimator is non-zero, but we ignore this information in the estimation; that is, 
we correct for
all three bias terms ($B_1$, $B_2$, and $B_3$, as introduced in Assumption~\ref{ass:A6})
in the bias-corrected FLS estimator.

Table~\ref{tab:T1} shows the simulation results for the bias, standard error, and root mean square error
of the three different estimators
for the case $N=100$, $\rho_f=0.5$, and $\sigma_f=0.5$, and different values of $\rho^0$ and $T$. 
The OLS estimator, the FLS estimator (computed with correct $R=1$),
                           and the corresponding bias-corrected FLS estimator with factors (BC-FLS)
                            were computed for 10,000 simulation runs.
                           The table lists the mean bias, the standard deviation (std),
                           and the square root of the mean square error (rmse) for the three estimators.
As expected, the OLS
estimator is biased because of the factor structure and its bias does not vanish (it actually increases) as $T$ increases.
The FLS estimator is also biased, but as theory predicts its bias vanishes as $T$ increases.
The bias-corrected FLS estimator performs better than the non-corrected FLS estimator, in particular, its bias vanishes faster.
Because we only correct for the first-order bias of the FLS estimator, we could not expect the bias-corrected FLS estimator to be unbiased.
However, as $T$ gets larger, more and more of the FLS estimator bias is corrected for; for example,
for $\rho^0=0.3$, we find that at $T=5$, the bias correction only corrects for
about half of the bias, whereas at $T=80$, it already corrects for about $90 \%$ of it.

Table~\ref{tab:T2} is similar to Table~\ref{tab:T1}, with the only difference being that we allow for misspecification in the
number of factors $R$, namely, the true number of factors is assumed to be $R=1$ (i.e., same DGP as for Table~\ref{tab:T1}), but we incorrectly use $R=2$ factors when calculating the FLS and BC-FLS estimator.
By comparing Table~\ref{tab:T2} with Table~\ref{tab:T1}, we find this type
of misspecification of the number of factors increases the bias and the standard deviation of both the FLS
and the BC-FLS estimator in finite samples. That increase, however, is comparatively small once both $N$ and $T$
are large. According to the results in
Moon and Weidner~\cite*{MoonWeidner2015}, we expect the limiting distribution of the correctly
specified ($R=1$) and incorrectly specified ($R=2$) FLS estimator to be identical when $N$ and $T$ grow at the same
rate. Our simulations suggest the same is true for the BC-FLS estimator.
The remaining simulation all assume correctly specified $R=1$.

An import issue is the choice of bandwidth $M$ for the bias correction.
Table~\ref{tab:T3} gives the fraction of the FLS estimator bias that is captured by the estimator for the bias
in a model with $N=100$, $T=20$, $\rho_f=0.5$, $\sigma_f=0.5$ and different values for $\rho$ and $M$.
The table shows the optimal bandwidth (in the sense that most of the bias is corrected for) depends on $\rho^0$: it is
$M=1$ for $\rho=0$, $M=2$ for $\rho=0.3$, $M=3$ and $\rho=0.6$, and $M=5$ for $\rho=0.9$.
Choosing too large or too small a bandwidth results in a smaller fraction of the bias to be corrected.
Table~\ref{tab:T4} also reports the properties of the BC-FLS estimator for different values of $\rho^0$, $T$, and $M$.
It shows the effect of the bandwidth choice on the standard deviation of the BC-FLS estimator
is relatively small at $T=40$, but is more pronounced at $T=20$.
The issue of optimal bandwidth choice is therefore an important topic for future research.
In the simulation results presented here, we tried to choose reasonable values for $M$, but made no attempt to optimize the bandwidth.

In our setup, we have $\|\lambda^0 f^{0 \prime}\| \approx \sqrt{2 N T} \sigma_f$ and $\|e\| \approx \sqrt{N} + \sqrt{T}$.\footnote{%
To be precise, we have $\|\lambda^0 f^{0 \prime}\| / (\sqrt{2 N T} \sigma_f) \, \rightarrow_p \, 1$,
and $\|e\|/(\sqrt{N} + \sqrt{T}) \, \rightarrow_p \, 1$.}
Assumptions \ref{ass:A1} and \ref{ass:A3} imply $\|\lambda^0 f^{0 \prime}\| \gg \|e\|$ asymptotically.
We can therefore only be sure our asymptotic results for the FLS estimator distribution are a good approximation of the
finite sample properties if $\|\lambda^0 f^{0 \prime}\| \gtrsim \|e\|$, that is, if
$\sqrt{2 N T} \sigma_f \, \gtrsim \, \sqrt{N} + \sqrt{T}$.
To explore this further, we present
in Table~\ref{tab:T5}  simulation results
for $N=100$, $T=20$, $\rho^0=0.6$, and different values of $\rho_f$ and $\sigma_f$.
For $\sigma_f=0$, we have  $0=\|\lambda^0 f^{0 \prime}\| \ll \|e\|$, and this case is equivalent
to $R=0$ (no factor at all). In this case, the OLS estimator estimates the true model
and is almost unbiased, and correspondingly, the FLS estimator and the bias-corrected FLS estimator perform worse than OLS
in finite samples (though we expect all three estimators are asymptotically equivalent), but the
bias-corrected FLS estimator has a lower bias and a lower variance than the non-corrected FLS estimator.
The case $\sigma_f = 0.2$ corresponds to $\|\lambda^0 f^{0 \prime}\| \approx \|e\|$, and one finds the bias and the variance
of the OLS estimator and of the FLS estimator are of comparable size. However, the bias-corrected FLS estimator already
has much smaller bias and a bit smaller variance in this case.
Finally, in the case $\sigma_f = 0.5$, we have $\|\lambda^0 f^{0 \prime}\| > \|e\|$, and we expect our asymptotic results to be
a good approximation of this situation. Indeed, one finds that for $\sigma_f=0.5$, the OLS estimator is heavily biased and
very inefficient compared to the FLS estimator, whereas the bias-corrected FLS estimator performs even better in terms of bias
and variance.

In Table~\ref{tab:T6}, we present simulation results for the size of the various tests discussed in the last section
when testing the null hypothesis $H_0: \rho=\rho^0$.
We choose a nominal size of $5\%$, $\rho_f=0.5$, $\sigma_f=0.5$, and different values for $\rho^0$, $N$, and $T$.
In all cases, the size distortions of the uncorrected Wald, LR, and LM test are rather large,
and the size distortion of these tests do not vanish as $N$ and $T$ increase: the size for
$N=100$ and $T=20$ is about the same as for $N=400$ and $T=80$, and the size for
$N=400$ and $T=20$ is about the same as for $N=1600$ and $T=80$.
By contrast, the size distortions
for the bias-corrected Wald, LR, and LM test are much smaller, and tend toward zero (i.e., the size
becomes closer to $5\%$) as $N,T$ increase, holding the ratio $N/T$ constant.
For fixed $T$, an increase in $N$ results in a larger size distortion, whereas for fixed $N$, an increase in $T$
results in a smaller size distortion (both for the non-corrected and for the bias-corrected tests).

In Table~\ref{tab:T7} and \ref{tab:T8}, we present the power and the size-corrected power when testing the
left-sided alternative $H^{\rm left}_a: \rho = \rho^0-(NT)^{-1/2}$ and the right-sided alternative
$H^{\rm right}_a: \rho = \rho^0+(NT)^{-1/2}$. The model specifications are the same as for the size
results in Table \ref{tab:T4}. Because both the FLS estimator and the bias-corrected FLS estimator for $\rho$ have a negative bias,
one finds the power for the left-sided alternative to be much smaller than the power for the right-sided
alternative. For the uncorrected tests, this effect can be extreme and the size-corrected power of these
tests for the left-sided alternative is below $2 \%$ in all cases and does not improve as
$N$ and $T$ become large, holding $N/T$ fixed. By contrast, the power for the bias-corrected tests becomes
more symmetric as $N$ and $T$ become large, and the size-corrected power for the left-sided alternative is much larger
than for the uncorrected tests, whereas the size-corrected power for the right-sided alternative is about the same.

\section{Conclusions}
\label{sec:conclusion}

This paper studies the least squares estimator for dynamic linear panel regression models
with interactive fixed effects. We provide conditions under which the estimator
is consistent, allowing for predetermined regressors and for a general
combination of ``low-rank'' and ``high-rank'' regressors.
We then show how a quadratic approximation of the profile objective function
 $L_{NT}(\beta)$ can
be used to derive the first-order asymptotic theory of the LS estimator of $\beta$ under the alternative asymptotic $N,T \rightarrow \infty$. We find the asymptotic distribution of the LS estimator
can be asymptotically biased (i) because of weak exogeneity of the regressors and (ii) because of heteroscedasticity (and correlation) of the idiosyncratic errors $e_{it}$.
Consistent estimators for the asymptotic covariance matrix
and for the asymptotic bias of the LS estimator are provided, and thus a bias-corrected LS estimator is given.
We furthermore study the asymptotic distributions of the Wald, LR, and LM test statistics
for testing a general linear hypothesis on $\beta$.
The uncorrected test statistics are not asymptotically chi-square because of the asymptotic bias of the score and of the LS estimator, but
bias-corrected test statistics that are asymptotically chi-square distributed can be constructed.
We also discussed a possible extension of the estimation procedure to the case of endogeneous regressors.
The findings of our Monte Carlo simulations show our asymptotic results
on the distribution of the (bias-corrected) LS estimator and of the (bias-corrected) test statistics
provide a good approximation of their finite sample properties. Although the bias-corrected LS estimator
has a non-zero bias in finite samples, this bias is much smaller than that of the LS estimator. Analogously,
the size distortions and power asymmetries of the bias-corrected Wald, LR, and LM test are much smaller
than for the non-bias-corrected versions.

\appendix

\section*{Appendix}

\section{Proof of Consistency (Theorem \ref{th:consistency})}
\label{app:consistency}

The following theorem is useful for the consistency proof and beyond.

\begin{lemma}
   \label{lemma:Optimization}
   Let $N$, $T$, $R$, $R_1$, and $R_2$ be positive integers such that
   $R \leq N$, $R \leq T$, and $R=R_1+R_2$.
   Let $Z$ be an $N\times T$ matrix, $\lambda$ be an $N\times R$, $f$ be a $T\times R$ matrix,
   $\widetilde \lambda$ be an $N\times R_1$ matrix, and $\widetilde f$ be a $T\times R_2$ matrix.
   Then the following six expressions (that are functions of Z only) are equivalent:
   \begin{align*}
      \min_{f,\lambda} \, & {\rm Tr}\left[ \left(Z-\lambda f'\right) \left(Z'-f \lambda'\right)\right]
        = \min_f \, {\rm Tr}(Z \, M_f \, Z')
         = \min_\lambda \, {\rm Tr}(Z' \, M_\lambda \, Z)
    \nonumber \\
         &= \min_{\tilde \lambda,\tilde f} \, {\rm Tr}(M_{\widetilde \lambda} \, Z \, M_{\widetilde f} \, Z')
         = \sum_{i=R+1}^{T} {\mu}_i(Z'Z)
         = \sum_{i=R+1}^{N} {\mu}_i(ZZ') .
   \end{align*}
\end{lemma}

In the above minimization problems,
we do not have to restrict the matrices $\lambda$, $f$, $\widetilde \lambda$, and $\widetilde f$ to be
of full rank. If, for example, $\lambda$ is not of full rank, 
the generalized inverse $(\lambda^{\prime}\lambda)^{\dagger}$
is still well defined, and
the projector $M_\lambda$ still  
satisfies $M_\lambda \lambda=0$ and  ${\rm rank}(M_\lambda) = N - \limfunc{rank}(\lambda)$.
If ${\rm rank}(Z)\geq R$, the optimal $\lambda$, $f$, $\widetilde \lambda$, and $\widetilde f$ always have full rank.

Lemma~\ref{lemma:Optimization} shows the equivalence of the three different
versions of the profile objective function in equation \eqref{LNT123}. It also
considers minimization of ${\rm Tr}(M_{\widetilde \lambda} \, Z \, M_{\widetilde f} \, Z')$ over $\widetilde \lambda$
and $\widetilde f$, which will be used in the consistency proof below. The proof of the theorem is given in
the supplementary material. The following lemma is due to Bai \cite*{Bai2009}.
\begin{lemma}
  \label{op1lemma}
  Under the assumptions  of Theorem~\ref{th:consistency} we have
  \begin{align*}
  \sup_f \left| \frac { {\rm Tr}(X_k \, M_f \, e^{\prime}) } {NT} \right| &=
  o_p(1) \; ,  &
  \sup_f \left| \frac { {\rm Tr}(\lambda^0 \, f^{0\prime} \, M_f \,
  e^{\prime}) } {NT} \right| &= o_p(1) \; ,  &
  \sup_f \left| \frac { {\rm Tr}(e \, P_f \, e^{\prime}) } {NT} \right| &=
  o_p(1) \; ,
  \end{align*}
  where the parameters $f$ are $T\times R$ matrices with $\limfunc{rank}(f)=R$.
\end{lemma}
\begin{proof}[\bf Proof]
By Assumption~\ref{ass:A2}, we know the first equation in
Lemma \ref{op1lemma} is satisfied when replacing $M_f$ by the identity matrix.
So we are left to show $\max_f \left| \frac 1 {NT} {\rm Tr}(\Xi
\, e^{\prime}) \right| = o_p(1)$, where $\Xi$ is either $X_k P_f$, $\lambda^0
f^{0\prime} M_f$, or $e P_f$. In all three cases, we have $\| \Xi \|/\sqrt{NT} =
\mathcal{O}_p(1)$ by Assumption~\ref{ass:A1}, \ref{ass:A3}, and \ref{ass:A4},
respectively, and we have ${\rm rank}(\Xi) \leq R$.
We therefore find\footnote{%
   Here we use $\left| {\rm Tr}\left( C\right) \right|
                      \leq \left\| C\right\| \limfunc{rank}\left( C\right)$,
   which holds for all square matrices $C$;
   see the supplementary material.
}
\begin{align*}
\sup_f \left| \frac 1 {NT} {\rm Tr}(\Xi \, P_f \, e^{\prime}) \right|
   &\leq R \, \frac{\| e \|} {\sqrt{NT}} \, \frac{\| \Xi \|} {\sqrt{NT}} = o_p(1)
\; .
\end{align*}
\end{proof}

\begin{proof}[\bf Proof of Theorem \ref{th:consistency}]
For the second version of the profile objective function in equation~\eqref{LNT123}, we write
$L_{NT}(\beta) = \min_f \, S_{NT}(\beta,f)$, where
\begin{align*}
S_{NT}(\beta,f) &= \frac{1}{NT} \; {\rm Tr}\left[ \left( \lambda^0 \,
f^{0\prime} + \sum_{k=1}^{K} (\beta^0_k-\beta_k) X_{k} + e \right) \, M_f \,
\left( \lambda^0 \, f^{0\prime} + \sum_{k=1}^{K} (\beta^0_k-\beta_k) X_{k} +
e \right)^{\prime}\right] .
\end{align*}
We have $S_{NT}(\beta^0,f^0) = \frac{1}{NT} \, {\rm Tr}
\left(e \, M_{f^0} \, e^{\prime}\right)$. Using Lemma \eqref{op1lemma}, we find
\begin{align}
S_{NT}(\beta,f) &= S_{NT}(\beta^0,f^0) + \widetilde S_{NT}(\beta,f)  \notag \\
& \qquad \qquad + \frac{2}{NT} \, {\rm Tr}\left[ \left( \lambda^0 \,
f^{0\prime} + \sum_{k=1}^{K} (\beta^0_k-\beta_k) X_{k} \right) \, M_f \,
e^{\prime}\right] + \frac{1}{NT} \, {\rm Tr} \left(e \, (P_{f^0}-P_f) \,
e^{\prime}\right)  \notag \\
& = S_{NT}(\beta^0,f^0) + \widetilde S_{NT}(\beta,f) + o_p(\|\beta-\beta^0\|) +
o_p(1) \; ,
  \label{Bound1SNT}
\end{align}
where we defined
\begin{align*}
\widetilde S_{NT}(\beta,f) &= \frac{1}{NT} \; {\rm Tr}\left[ \left(
\lambda^0 \, f^{0\prime} + \sum_{k=1}^{K} (\beta^0_k-\beta_k) X_{k} \right)
\, M_f \, \left( \lambda^0 \, f^{0\prime} + \sum_{k=1}^{K}
(\beta^0_k-\beta_k) X_{k} \right)^{\prime}\right] .
\end{align*}
Up to this point, the consistency proof is almost equivalent to the one given in Bai \cite*{Bai2009}, but
the remainder of the proof differs from Bai, because we allow for more general low-rank regressors, and because
we allow for high-rank and low-rank regressors simultaneously.
We split $\widetilde S_{NT}(\beta,f)=\widetilde S^{(1)}_{NT}(\beta,f)+\widetilde S^{(2)}_{NT}(\beta,f)$, where
\begin{align*}
  \widetilde S^{(1)}_{NT}(\beta,f)
    &= \frac{1}{NT} \; {\rm Tr}\left[ \left(
       \lambda^0 \, f^{0\prime} + \sum_{k=1}^{K} (\beta^0_k-\beta_k) X_{k} \right)
       \, M_f \, \left( \lambda^0 \, f^{0\prime} + \sum_{k=1}^{K}
       (\beta^0_k-\beta_k) X_{k} \right)^{\prime}\, M_{(\lambda^0,w)} \right]
  \nonumber \\
   &= \frac{1}{NT} \; {\rm Tr}\left[ \left( \sum_{m=K_1+1}^{K}
            (\beta^0_m-\beta_m) X_{m} \right) \, M_f \, \left( \sum_{m=K_1+1}^{K}
            (\beta^0_m-\beta_m) X_{m} \right)^{\prime}\, M_{(\lambda^0,w)} \right] ,  \notag
   \nonumber \\
  \widetilde S^{(2)}_{NT}(\beta,f)
    &=   \frac{1}{NT} \; {\rm Tr}\left[ \left(
       \lambda^0 \, f^{0\prime} + \sum_{k=1}^{K} (\beta^0_k-\beta_k) X_{k} \right)
       \, M_f \, \left( \lambda^0 \, f^{0\prime} + \sum_{k=1}^{K}
       (\beta^0_k-\beta_k) X_{k} \right)^{\prime}\, P_{(\lambda^0,w)} \right] ,
\end{align*}
and $(\lambda^0,w)$ is the $N\times(R+K_1)$ matrix that is composed out of $\lambda^0$ and the $N\times K_1$ matrix
$w$ defined in Assumption~\ref{ass:A4}.
For $\widetilde S^{(1)}_{NT}(\beta,f)$, we can apply Lemma~\ref{lemma:Optimization} with
$\widetilde f=f$ and $\widetilde \lambda=(\lambda^0,w)$ (the $R$ in the theorem is now $2R+K_1$) to find
\begin{align}
  \widetilde S^{(1)}_{NT}(\beta,f)
    &\geq  \frac 1 {NT} \, \sum_{i=2R+K_1+1}^N \, {\mu}_{i} \left[ \left( \sum_{m=K_1+1}^{K}
            (\beta^0_m-\beta_m) X_{m} \right)  \left( \sum_{m=K_1+1}^{K}
            (\beta^0_m-\beta_m) X_{m} \right)^{\prime} \right]
 \nonumber \\
    &\geq  \, b \, \left\| \beta^{\rm high} - \beta^{0,{\rm high}} \right\|^2  \; , \qquad \text{wpa1,}
    \label{boundS1NT}
\end{align}
where in the last step, we used the existence of a constant $b>0$ guaranteed by Assumption~\ref{ass:A4}(ii)(a),
and we introduced $\beta^{\rm high}=(\beta_{K_1+1},\ldots,\beta_K)'$, which refers to the $K_2\times 1$
parameter vector corresponding to the high-rank regressors.
Similarly, we define $\beta^{\rm low}=(\beta_{1},\ldots,\beta_{K_1})'$
for the $K_1\times 1$ parameter vector of low-rank regressors.

Using $P_{(\lambda^0,w)}=P_{(\lambda^0,w)}P_{(\lambda^0,w)}$ and the cyclicality of the trace, we see
$\widetilde S^{(2)}_{NT}(\beta,f)$ can be written as the trace of a positive definite matrix, and therefore
$\widetilde S^{(2)}_{NT}(\beta,f) \geq 0$. Note also that we can choose $\beta=\beta^0$ and $f=f^0$ in the minimization problem over $S_{NT}(\beta,f)$;
that is, the optimal $\beta=\widehat \beta$ and $f=\widehat f$ must satisfy
$S_{NT}(\widehat \beta,\widehat f) \leq S_{NT}(\beta^0,f^0)$. Using this result, equation \eqref{Bound1SNT}, $\widetilde S^{(2)}_{NT}(\beta,f) \geq 0$,
and the bound in \eqref{boundS1NT}, we find
\begin{align*}
    0 \, &\geq \, b \left\| \widehat \beta^{\rm high} - \beta^{0,{\rm high}} \right\|^2
                 +o_{p}\left( \left\| \widehat \beta^{\rm high} - \beta^{0,{\rm high}} \right\|  \right)
                 +o_{p}\left( \left\| \widehat \beta^{\rm low} - \beta^{0,{\rm low}} \right\|  \right)
                 +o_{p}(1)    \; .
\end{align*}
Because we assume $\widehat \beta^{\rm low}$ is bounded, the last equation implies 
$\left\| \widehat \beta^{\rm high} - \beta^{0,{\rm high}} \right\| = o_p(1)$; that is,
$\widehat \beta^{\rm high}$ is consistent. What is left to show is that
$\widehat \beta^{\rm low}$ is consistent, too.
In the supplementary material, we show Assumption~\ref{ass:A4}(ii)(b) guarantees that  
finite positive constants $a_0$, $a_1$, $a_2$, $a_3$, and $a_4$ exist such that
\begin{align*}
  \widetilde S^{(2)}_{NT}(\beta,f)
    &\geq  \, \frac{a_0  \left\| \beta^{\rm low} - \beta^{0,{\rm low}} \right\|^2 }
                { \left\| \beta^{\rm low} - \beta^{0,{\rm low}} \right\|^2
                  + a_1 \left\| \beta^{\rm low} - \beta^{0,{\rm low}} \right\|
                  + a_2 }
        \nonumber \\ & \qquad \qquad
      -  a_3 \left\| \beta^{\rm high} - \beta^{0,{\rm high}} \right\|
  - a_4 \left\| \beta^{\rm high} - \beta^{0,{\rm high}} \right\| \, \left\| \beta^{\rm low} - \beta^{0,{\rm low}} \right\|
         \; , \qquad \text{wpa1.}
\end{align*}
Using consistency of $\widehat \beta^{\rm high}$ and again boundedness of $\beta^{\rm low}$, the
previous inequality implies
 $a>0$ exists such that
$\widetilde S^{(2)}_{NT}(\widehat \beta,f)
    \geq  a  \left\| \widehat \beta^{\rm low} - \beta^{0,{\rm low}} \right\|^2 + o_p(1)$.
With the same argument as for $\widehat \beta^{\rm high}$, we therefore find
$\left\| \widehat \beta^{\rm low} - \beta^{0,{\rm low}} \right\| = o_p(1)$; that is, $\widehat \beta^{\rm low}$ is consistent.
\end{proof}

\section{Proof of Limiting Distribution (Theorem \ref{th:limdis})}
\label{app:limdis}

Theorem~\ref{th:ass_expand} is from
Moon and Weidner~\cite*{MoonWeidner2015}, and the proof can be found there.
Note Assumption~\ref{ass:A4}(i) implies $\|X_k\|={\cal O}_p(\sqrt{NT})$,
which we assume in Moon and Weidner \cite*{MoonWeidner2015}.
There, we also assume that $\|e\|={\cal O}_p(\sqrt{\max(N,T)})={\cal O}_p(\sqrt{N})$,
whereas in the current paper we assume  $\|e\|=o_p(\|N^{2/3}\|)$. It is, however,
straightforward to verify that the proof of Theorem \ref{th:ass_expand}
is also valid under this weaker assumption.

Moon and Weidner~\cite*{MoonWeidner2015} also includes
the proof of Corollary~\ref{cor:limit}.
The proof requires consistency of $\widehat \beta$, which in the current paper
is derived under weaker assumptions than in Moon and Weidner~\cite*{MoonWeidner2015}, where no
low-rank regressors are considered. Corollary~\ref{cor:limit} is therefore stated under weaker assumptions here, 
 but the proof is unchanged.
 In the supplementary material, we show 
the assumptions of Corollary \ref{cor:limit} already guarantee 
$W_{NT}$ does not become singular as $N,T \rightarrow \infty$.

For each $k=1,\ldots,K$,
we define the $N \times T$ matrices $\overline X_k$, $\widetilde X_k$, and $\mathfrak{X}_k$ as follows:
 \begin{align*}
      \overline X_k &=  \mathbb{E}\left( X_k \big| \, {\cal C} \right) ,
      &
      \widetilde X_k &= X_k - \mathbb{E}\left( X_k \big| \, {\cal C} \right) ,
      &
      \mathfrak{X}_k &= M_{\lambda^0} \, \overline X_k \, M_{f^0}  +  \widetilde X_k .
 \end{align*}
 Note the difference between $\mathfrak{X}_k$ and ${\cal X}_k=M_{\lambda^0} \, X_{k} \, M_{f^0}$,
 which was defined in Assumption~\ref{ass:A6}. In particular, conditional on ${\cal C}$, the elements
  $\mathfrak{X}_{k,it}$ of $\mathfrak{X}_k$ are contemporaneously uncorrelated with the error term $e_{it}$,
 although the same is not true for ${\cal X}_k$.

To present the proof of
Theorem \ref{th:limdis}, it is convenient to first state two technical lemmas.
\begin{lemma}
   \label{lemma:vanishing}
   Under the assumptions of Theorem~\ref{th:limdis}, we have
   \begin{align*}
     (a) &&
      \frac 1 {\sqrt{NT}} {\rm Tr} \left( P_{f^0} \, e^{\prime}\, P_{\lambda^0} \, \widetilde  X_k \right)
            &= o_p(1) \; ,
   \\
     (b) &&
      \frac 1 {\sqrt{NT}} {\rm Tr} \left(  P_{\lambda^0} \, e \,\widetilde X_k' \right)     &= o_p(1) \; ,
      \nonumber \\ (c) &&
      \frac 1 {\sqrt{NT}} {\rm Tr} \left\{ P_{f^0} \, \left[ e^{\prime} \, \widetilde  X_k
                        - \mathbb{E}\left( e^{\prime} \, \widetilde  X_k \, \big| \, {\cal C}  \right) \right] \right\} &= o_p(1) \; ,
    \nonumber \\ (d) &&
       \frac 1 {\sqrt{NT}} {\rm Tr}\left(e P_{f^0} \, e' \, M_{\lambda^0} \,   X_k \,
              f^0 \, (f^{0\prime}f^0)^{-1} \, (\lambda^{0\prime}\lambda^0)^{-1} \, \lambda^{0\prime} \right) &= o_p(1) \; ,
    \nonumber \\ (e) &&
        \frac 1 {\sqrt{NT}} {\rm Tr}\left(e^{\prime} \, P_{\lambda^0} \, e \, M_{f^0} \,   X^{\prime}_k \,
              \lambda^0 \, (\lambda^{0\prime}\lambda^0)^{-1} \, (f^{0\prime}f^0)^{-1} \, f^{0\prime} \right) &= o_p(1) \; ,
    \nonumber \\ (f) &&
         \frac 1 {\sqrt{NT}} {\rm Tr}\left(e^{\prime}M_{\lambda^0} \,   X_k \, M_{f^0} \, e^{\prime}
                \, \lambda^0 \, (\lambda^{0\prime}\lambda^0)^{-1} \, (f^{0\prime}f^0)^{-1} \, f^{0\prime} \right)
                       &= o_p(1) \; ,
    \nonumber \\ (g) &&
       \frac 1 {\sqrt{NT}} {\rm Tr}\left\{ \left[ e e' - 
        \mathbb{E} \left( e e'  \, \big| \, {\cal C} \right)    
       \right] \, M_{\lambda^0} \,   X_k \, f^0 \, (f^{0\prime}f^0)^{-1} \, (\lambda^{0\prime}\lambda^0)^{-1} \, \lambda^{0\prime} \right\} &= o_p(1) \; ,
    \nonumber \\ (h) &&
        \frac 1 {\sqrt{NT}} {\rm Tr}\left\{ \left[ e^{\prime}  e - 
           \mathbb{E} \left( e^{\prime}  e \, \big| \, {\cal C} \right)  
          \right]
               \, M_{f^0} \,   X^{\prime}_k \,
              \lambda^0 \, (\lambda^{0\prime}\lambda^0)^{-1} \, (f^{0\prime}f^0)^{-1} \, f^{0\prime} \right\} &= o_p(1) \; ,
      \nonumber \\ (i) &&
          \frac 1 {NT} \, \sum_{i=1}^N \, \sum_{t=1}^T
             \left[   e_{it}^2   \, \mathfrak{X}_{it} \, \mathfrak{X}_{it}'
             -  \mathbb{E} \left(  e_{it}^2 \, \mathfrak{X}_{it} \, \mathfrak{X}_{it}' \, \big| \, {\cal C}
                  \right)    \right]  &= o_p(1) ,
      \nonumber \\ (j) &&
          \frac 1 {NT} \, \sum_{i=1}^N \, \sum_{t=1}^T
            \,     e_{it}^2   
            \left(  \mathfrak{X}_{it} \, \mathfrak{X}_{it}'   -  {\cal X}_{it} \, {\cal X}_{it}' \right)  &= o_p(1) .
      \end{align*}
\end{lemma}

\begin{lemma}
   \label{lemma:denCLT}
   Under the assumptions of Theorem~\ref{th:limdis}, we have
   \begin{align*}
      \frac 1 {\sqrt{NT}} \sum_{i=1}^N \sum_{t=1}^T e_{it} \mathfrak{X}_{it}  \, \limfunc{\rightarrow}_d \,
               {\cal N}\left( 0,  \Omega  \right)   .
   \end{align*}
\end{lemma}

The proofs of Lemma~\ref{lemma:vanishing} and Lemma~\ref{lemma:denCLT} are provided in the
supplementary material.
We briefly want to discuss why 
the asymptotic variance-covariance matrix in Lemma~\ref{lemma:denCLT}
turns out to be $\Omega$.
Note that because $e_{it} \mathfrak{X}_{it}$ is mean zero and
uncorrelated across both $i$ and~$t$, conditional on ${\cal C}$,
we have
\begin{align}
      {\rm Var}\left(
        \frac 1 {\sqrt{NT}} \sum_{i=1}^N \sum_{t=1}^T e_{it} \mathfrak{X}_{it} \, \bigg| \, {\cal C} \right)
        &=
        \frac 1 {NT} \, \sum_{i=1}^N \, \sum_{t=1}^T
            \, \mathbb{E} \left(  e_{it}^2 \, \mathfrak{X}_{it} \, \mathfrak{X}_{it}'
              \, \big| \, {\cal C}   \right)
  \nonumber \\
       &=
        \frac 1 {NT} \, \sum_{i=1}^N \, \sum_{t=1}^T
            \,  e_{it}^2    \, \mathfrak{X}_{it} \, \mathfrak{X}_{it}'       + o_p(1)
  \nonumber  \\
         &=        \frac 1 {NT} \, \sum_{i=1}^N \, \sum_{t=1}^T
            \,   e_{it}^2  \, {\cal X}_{it} \, {\cal X}_{it}'    + o_p(1)
  \nonumber  \\
         &= \Omega + o_p(1),
     \label{VarEqOmega}
\end{align}
where we also used part (i)  of Lemma~\ref{lemma:vanishing} 
for the second equality
and part (j)  of Lemma~\ref{lemma:vanishing} for the third equality,
and the  definition of $\Omega$ in Assumptions~\ref{ass:A6}
in the last step.  

Using those lemmas, we can now prove
the theorem on the limiting distribution of $\widehat \beta$ in the main text.

\begin{proof}[\bf Proof of Theorem \ref{th:limdis}]
    Assumption~\ref{ass:A5} implies $\|e\|={\cal O}_p(N^{1/2})$ as $N$ and $T$ grow at the same rate,
    as discussed in section~\ref{app:error} of the supplementary material; that is, Assumption~\ref{ass:A3}$^*$ is satisfied. 
    We can therefore apply Corollary~\ref{cor:limit} to calculate the limiting distribution of $\widehat \beta$.
   Note that
$M_{\lambda^0} X_k M_{f^0} = \mathfrak{X}_k -  \widetilde X_k \, P_{f^0} - P_{\lambda^0} \, \widetilde X_k + P_{\lambda^0} \, \widetilde X_k \, P_{f^0}$.
   Using Lemmas \ref{lemma:vanishing} and \ref{lemma:denCLT} and Assumption~\ref{ass:A6}, we find
  \begin{align*}
    \frac 1 {\sqrt{NT}} C^{(1)}\left(\lambda^0,\, f^0,\, X_{k},\, e \right) =&
       \frac 1 {\sqrt{NT}} {\rm Tr} \left( e^{\prime}\, M_{\lambda^0} \, X_k
       \,  M_{f^0} \right)
     \nonumber \\
       =& \frac 1 {\sqrt{NT}} {\rm Tr} \left( e^{\prime} \mathfrak{X}_k \right)
          - \frac 1 {\sqrt{NT}} {\rm Tr} \left[ P_{f^0} \, \mathbb{E}\left( e^{\prime} \, \widetilde  X_k  \, \big| \, {\cal C} \right) \right]
        \nonumber \\ & \qquad
         - \frac 1 {\sqrt{NT}} {\rm Tr} \left(  e^{\prime}\, P_{\lambda^0} \, \widetilde X_k \right)
          + \frac 1 {\sqrt{NT}} {\rm Tr} \left( P_{f^0} \, e^{\prime}\, P_{\lambda^0} \, \widetilde  X_k \right)
        \nonumber \\ & \qquad \qquad
        - \frac 1 {\sqrt{NT}} {\rm Tr} \left\{ P_{f^0} \, \left[ e^{\prime} \, \widetilde  X_k
                         - \mathbb{E}\left( e^{\prime} \, \widetilde  X_k \, \big| \, {\cal C} \right) \right] \right\}
     \nonumber \\
       =& \frac 1 {\sqrt{NT}} {\rm Tr} \left( e^{\prime} \mathfrak{X}_k \right)
          - \frac 1 {\sqrt{NT}}
               {\rm Tr} \left[ P_{f^0} \, \mathbb{E}\left( e^{\prime} \,  X_k \, \big| \, {\cal C} \right) \right] + o_p(1) \; .
     \nonumber \\
       \limfunc{\rightarrow}_d  & \; {\cal N}\left( - \kappa B_1 , \, \Omega \right)  \; ,
  \end{align*}
  where we also used that
  $ \mathbb{E}\left( e^{\prime} \, \widetilde  X_k \, \big| \, {\cal C} \right) =  \mathbb{E}\left( e^{\prime} \, X_k \, \big| \, {\cal C} \right)$.
  Using Lemma \ref{lemma:vanishing},
  we also find   
  \begin{align*}
  \frac 1 {\sqrt{NT}} C^{(2)}\left(\lambda^0,\, f^0,\, X_{k},\, e \right) =& - \, \frac 1 {\sqrt{NT}} \bigg[
       {\rm Tr}\left(e M_{f^0} \, e' \, M_{\lambda^0} \, X_k \,
              f^0 \, (f^{0\prime}f^0)^{-1} \, (\lambda^{0\prime}\lambda^0)^{-1} \, \lambda^{0\prime} \right)
    \nonumber \\ & \qquad \quad
       +{\rm Tr}\left(e^{\prime}M_{\lambda^0} \, e \, M_{f^0} \, X^{\prime}_k \,
              \lambda^0 \, (\lambda^{0\prime}\lambda^0)^{-1} \, (f^{0\prime}f^0)^{-1} \, f^{0\prime} \right)
    \nonumber \\ & \qquad \quad
       +{\rm Tr}\left(e^{\prime}M_{\lambda^0} \, X_k \, M_{f^0} \, e^{\prime}
                \, \lambda^0 \, (\lambda^{0\prime}\lambda^0)^{-1} \, (f^{0\prime}f^0)^{-1} \, f^{0\prime} \right)
                        \bigg]
     \nonumber \\
       =&    \frac 1 {\sqrt{NT}} {\rm Tr}\left(e P_{f^0} \, e' \, M_{\lambda^0} \,   X_k \,
              f^0 \, (f^{0\prime}f^0)^{-1} \, (\lambda^{0\prime}\lambda^0)^{-1} \, \lambda^{0\prime} \right)
    \nonumber \\
        & \quad - \frac 1 {\sqrt{NT}} {\rm Tr}\left\{ \left[ e e' - \mathbb{E} \left( e e' | {\cal C} \right) \right]  \, M_{\lambda^0} \,   X_k \,
              f^0 \, (f^{0\prime}f^0)^{-1} \, (\lambda^{0\prime}\lambda^0)^{-1} \, \lambda^{0\prime} \right\}
    \nonumber \\
        & \quad - \frac 1 {\sqrt{NT}} {\rm Tr}\left[ \mathbb{E} \left( e e' | {\cal C} \right) \, M_{\lambda^0} \,   X_k \,
              f^0 \, (f^{0\prime}f^0)^{-1} \, (\lambda^{0\prime}\lambda^0)^{-1} \, \lambda^{0\prime} \right]
    \nonumber \\
        & \quad + \frac 1 {\sqrt{NT}} {\rm Tr}\left(e^{\prime}P_{\lambda^0} \, e \, M_{f^0} \,  X^{\prime}_k \,
              \lambda^0 \, (\lambda^{0\prime}\lambda^0)^{-1} \, (f^{0\prime}f^0)^{-1} \, f^{0\prime} \right)
    \nonumber \\
        & \quad - \frac 1 {\sqrt{NT}} {\rm Tr}\left\{ \left[ e^{\prime}  e - \mathbb{E} \left( e^{\prime}  e | {\cal C} \right) \right]
               \, M_{f^0} \,  X^{\prime}_k \,
              \lambda^0 \, (\lambda^{0\prime}\lambda^0)^{-1} \, (f^{0\prime}f^0)^{-1} \, f^{0\prime} \right\}
    \nonumber \\
        & \quad - \frac 1 {\sqrt{NT}} {\rm Tr}\left[ \mathbb{E} \left( e^{\prime}  e | {\cal C} \right) \, M_{f^0} \,  X^{\prime}_k \,
              \lambda^0 \, (\lambda^{0\prime}\lambda^0)^{-1} \, (f^{0\prime}f^0)^{-1} \, f^{0\prime} \right]
    \nonumber \\
         & \quad + \frac 1 {\sqrt{NT}} {\rm Tr}\left(e^{\prime}M_{\lambda^0} \,  X_k \, M_{f^0} \, e^{\prime}
                \, \lambda^0 \, (\lambda^{0\prime}\lambda^0)^{-1} \, (f^{0\prime}f^0)^{-1} \, f^{0\prime} \right)
    \nonumber \\
       =&  - \frac 1 {\sqrt{NT}} {\rm Tr}\left[ \mathbb{E} \left( e e' | {\cal C} \right) \, M_{\lambda^0} \,  X_k \,
              f^0 \, (f^{0\prime}f^0)^{-1} \, (\lambda^{0\prime}\lambda^0)^{-1} \, \lambda^{0\prime} \right]
     \nonumber \\
      & \quad - \frac 1 {\sqrt{NT}} {\rm Tr}\left[ \mathbb{E} \left( e^{\prime}  e | {\cal C} \right) \, M_{f^0} \,  X^{\prime}_k \,
              \lambda^0 \, (\lambda^{0\prime}\lambda^0)^{-1} \, (f^{0\prime}f^0)^{-1} \, f^{0\prime} \right]
        + o_p(1) \; ,
    \nonumber \\
       =&  - \kappa^{-1} B_2 - \kappa B_3  + o_p(1) \; .
   \end{align*}
   Combining these results, we obtain
   \begin{align*}
        \sqrt{NT} \left( \widehat \beta - \beta^0 \right)  &= W_{NT}^{-1}
                   \frac{1} {\sqrt{NT}} C_{NT}  
            \nonumber \\
          \, &\limfunc{\rightarrow}_d
              \, {\cal N} \left( - \, W^{-1} \,
                   \left( \kappa B_1 + \kappa^{-1} B_2 +\kappa B_3 \right) , \; W^{-1} \, \Omega \, W^{-1} \right) \; ,
   \end{align*}
   which is what we wanted to show.
\end{proof}

\theendnotes

\clearpage
\setcounter{section}{0}
\setcounter{equation}{0}
\setcounter{theorem}{0}
\setcounter{assumption}{0}
\setcounter{example}{0}
\setcounter{definition}{0}
\setcounter{table}{0}
\setcounter{figure}{0}
\renewcommand{\thesection}{S.\arabic{section}}
\renewcommand{\thetable}{S.\arabic{table}}
\renewcommand{\thefigure}{S.\arabic{figure}}
\renewcommand{\theequation}{S.\arabic{section}.\arabic{equation}}
\begin{center}
{\Large\bfseries Supplementary Material}
\end{center}
\bigskip
\section{Proof of Identification (Theorem~\ref{th:id})}
\label{app:identification}

\begin{proof}[\bf Proof of Theorem~\ref{th:id}]
    Let  $Q(\beta,\lambda,f) \equiv  \mathbb{E}\left(\left\| Y \, - \, \beta \cdot X \, - \, \lambda \, f'
    \right\|^2_F \Big|  \lambda^0, f^0, w \right)$, where $\beta \in \mathbb{R}^K$,
    $\lambda \in \mathbb{R}^{N\times R}$ and
     $f \in  \mathbb{R}^{T\times R}$.
    We have
   \begin{align*}
        & Q(\beta,\lambda,f)
      \nonumber \\
         &= \mathbb{E} \left\{ \Tr \left[ \left( Y \, - \, \beta \cdot X \, - \, \lambda \, f' \right)'
                      \left( Y \, - \, \beta \cdot X \, - \, \lambda \, f' \right)
                      \right] \Big|  \lambda^0, f^0, w \right\}
      \nonumber \\
         &=  \mathbb{E} \left\{ \Tr \left[
          \left(  \lambda^0 f^{0 \prime} -  \lambda f' -  (\beta-\beta^0) \cdot X  + e \right)'
           \left(  \lambda^0 f^{0 \prime} -  \lambda f' -  (\beta-\beta^0) \cdot X  + e \right)
           \right] \Big|  \lambda^0, f^0, w \right\}
       \nonumber \\
          &=   \mathbb{E} \left[ \Tr \left( e' e \right) \Big|  \lambda^0, f^0, w \right]
       \nonumber \\
        & \qquad +
          \underbrace{ \mathbb{E} \left\{ \Tr \left[
          \left(  \lambda^0 f^{0 \prime} -  \lambda f' -  (\beta-\beta^0) \cdot X  \right)'
           \left(  \lambda^0 f^{0 \prime} -  \lambda f' -  (\beta-\beta^0) \cdot X  \right)
            \right] \Big|  \lambda^0, f^0, w \right\}
          }_{ \equiv Q^*(\beta,\lambda,f) }
              .
    \end{align*}
    In the last step we used Assumption~\ref{ass:id}$(ii)$.
    Because $ \mathbb{E} \left[ \Tr \left( e' e \right) \Big|  \lambda^0, f^0, w \right]$
    is independent of $\beta,\lambda,f$, we find
   minimizing $Q(\beta,\lambda,f)$ is equivalent to
   minimizing $Q^*(\beta,\lambda,f)$. We decompose $Q^*(\beta,\lambda,f)$
   as follows
   \begin{align*}
       & Q^*(\beta,\lambda,f)
     \nonumber \\
       &=\mathbb{E} \left\{ \Tr \left[
          \left(  \lambda^0 f^{0 \prime} -  \lambda f' -  (\beta-\beta^0) \cdot X  \right)'
           \left(  \lambda^0 f^{0 \prime} -  \lambda f' -  (\beta-\beta^0) \cdot X  \right)
            \right] \Big|  \lambda^0, f^0, w \right\}
      \nonumber \\
        &=      \mathbb{E} \left\{ \Tr \left[
          \left(  \lambda^0 f^{0 \prime} -  \lambda f' -  (\beta-\beta^0) \cdot X  \right)'
             M_{(\lambda,\lambda^0,w)}
           \left(  \lambda^0 f^{0 \prime} -  \lambda f' -  (\beta-\beta^0) \cdot X  \right)
            \right] \Big|  \lambda^0, f^0, w \right\}
        \nonumber \\
          & \quad +
              \mathbb{E} \left\{ \Tr \left[
          \left(  \lambda^0 f^{0 \prime} -  \lambda f' -  (\beta-\beta^0) \cdot X  \right)'
             P_{(\lambda,\lambda^0,w)}
           \left(  \lambda^0 f^{0 \prime} -  \lambda f' -  (\beta-\beta^0) \cdot X  \right)
            \right] \Big|  \lambda^0, f^0, w \right\}
      \nonumber \\
        &=   \underbrace{   \mathbb{E} \left\{ \Tr \left[
          \left(  (\beta^{\rm high}-\beta^{0, {\rm high}}) \cdot X_{\rm high}  \right)'
             M_{(\lambda,\lambda^0,w)}
           \left(    (\beta^{\rm high}-\beta^{0, {\rm high}}) \cdot X_{\rm high}  \right)
            \right] \Big|  \lambda^0, f^0, w \right\}
            }_{\equiv Q^{\rm high}(\beta^{\rm high},\lambda) }
        \nonumber \\
          & \quad +
             \underbrace{     \mathbb{E} \left\{ \Tr \left[
          \left(  \lambda^0 f^{0 \prime} -  \lambda f' -  (\beta-\beta^0) \cdot X  \right)'
             P_{(\lambda,\lambda^0,w)}
           \left(  \lambda^0 f^{0 \prime} -  \lambda f' -  (\beta-\beta^0) \cdot X  \right)
            \right] \Big|  \lambda^0, f^0, w \right\}
             }_{\equiv Q^{\rm low}(\beta,\lambda,f) }    ,
   \end{align*}
   where $ (\beta^{\rm high}-\beta^{0, {\rm high}}) \cdot X_{\rm high} =
   \sum_{m=K_1+1}^K (\beta_m-\beta^0_m) X_m$.
   A lower bound
   on $ Q^{\rm high}(\beta^{\rm high},\lambda)$ is given by
   \begin{align}
        & Q^{\rm high}(\beta^{\rm high},\lambda)
     \nonumber \\
        &\geq
        \min_{\widetilde \lambda \in \mathbb{R}^{N \times (R+R+{\rm rank}(w))}}
        \mathbb{E} \left\{ \Tr \left[
          \left(  (\beta^{\rm high}-\beta^{0, {\rm high}}) \cdot X_{\rm high}  \right)'
             M_{(\widetilde \lambda,\lambda,w)}
           \left(    (\beta^{\rm high}-\beta^{0, {\rm high}}) \cdot X_{\rm high}  \right)
            \right] \Big|  \lambda^0, f^0, w \right\}
     \nonumber \\
        &=
        \sum_{r=R+R+{\rm rank}(w)}^{\min(N,T)}
       \mu_r\left\{  \mathbb{E}\left[ \left( (\beta^{\rm high}-\beta^{0, {\rm high}}) \cdot X_{\rm high}  \right)
             \left(  (\beta^{\rm high}-\beta^{0, {\rm high}}) \cdot X_{\rm high}  \right)'
               \Big|  \lambda^0, f^0, w \right] \right\}.
          \label{LowerBoundQhigh}
   \end{align}

   Because $Q^*(\beta,\lambda,f)$,
   $Q^{\rm high}(\beta^{\rm high},\lambda)$,
   and
   $Q^{\rm low}(\beta,\lambda,f)$,
   are expectations
   of traces of positive semi-definite matrices we have  $Q^*(\beta,\lambda,f) \geq 0$,
   $Q^{\rm high}(\beta^{\rm high},\lambda) \geq 0$,
   and $Q^{\rm low}(\beta,\lambda,f) \geq 0$
   for all $\beta$, $\lambda$, $f$.
  Let $\bar \beta$, $\bar \lambda$ and $\bar f$ be the parameter values
   that minimize $Q(\beta,\lambda,f)$, and thus also $Q^*(\beta,\lambda,f)$.
  Because
    $Q^*(\beta^0,\lambda^0,f^0)=0$
   we have
    $Q^*(\bar \beta,\bar \lambda,\bar f) = \min_{\beta,\lambda,f} Q^*(\beta,\lambda,f) = 0$.
   This implies
   $Q^{\rm high}(\bar \beta^{\rm high},\bar \lambda) = 0$
   and
   $Q^{\rm low}(\bar \beta,\bar \lambda,\bar f) = 0$.
   Assumption~\ref{ass:id}$(v)$,
   the lower bound \eqref{LowerBoundQhigh},
   and $Q^{\rm high}(\bar \beta^{\rm high},\bar \lambda) = 0$
   imply $\bar \beta^{\rm high} =  \beta^{0, {\rm high}}$.
   Using this, we find
   \begin{align}
       & Q^{\rm low}(\bar \beta,\bar \lambda,\bar f)
     \nonumber \\
       &= \mathbb{E} \left\{ \Tr \left[
          \left(  \lambda^0 f^{0 \prime} -  \bar \lambda \bar f' -
           (\bar \beta^{\rm low}-\beta^{0, {\rm low}}) \cdot X_{\rm low}  \right)'
           \left(  \lambda^0 f^{0 \prime} -  \bar \lambda \bar f'
              -  (\bar \beta^{\rm low}-\beta^{0,{\rm low}}) \cdot X_{\rm low}  \right)
            \right] \Big|  \lambda^0, f^0, w \right\}   ,
     \nonumber \\
       &\geq \min_f \mathbb{E} \left\{ \Tr \left[
          \left(  \lambda^0 f^{0 \prime} -  \bar \lambda  f' -
           (\bar \beta^{\rm low}-\beta^{0, {\rm low}}) \cdot X_{\rm low}  \right)'
           \left(  \lambda^0 f^{0 \prime} -  \bar \lambda  f'
              -  (\bar \beta^{\rm low}-\beta^{0,{\rm low}}) \cdot X_{\rm low}  \right)
            \right] \Big|  \lambda^0, f^0, w \right\}
     \nonumber \\
      &=  \mathbb{E} \left\{ \Tr \left[
          \left(  \lambda^0 f^{0 \prime} -
           (\bar \beta^{\rm low}-\beta^{0, {\rm low}}) \cdot X_{\rm low}  \right)'
           M_{\bar \lambda}
           \left(  \lambda^0 f^{0 \prime}
              -  (\bar \beta^{\rm low}-\beta^{0,{\rm low}}) \cdot X_{\rm low}  \right)
            \right] \Big|  \lambda^0, f^0, w \right\} ,
        \label{QlowLowBound}
   \end{align}
   where
   $  (\bar \beta^{\rm low}-\beta^{0,{\rm low}}) \cdot X_{\rm low}
   =  \sum_{l=1}^{K_1} (\bar \beta_l-\beta^0_l) X_l$.
  Because
    $Q^{\rm low}(\bar \beta,\bar \lambda,\bar f) = 0$
    and the last expression in \eqref{QlowLowBound} is non-negative
    we must have
   \begin{align*}
        \mathbb{E} \left\{ \Tr \left[
          \left(  \lambda^0 f^{0 \prime} -
           (\bar \beta^{\rm low}-\beta^{0, {\rm low}}) \cdot X_{\rm low}  \right)'
           M_{\bar \lambda}
           \left(  \lambda^0 f^{0 \prime}
              -  (\bar \beta^{\rm low}-\beta^{0,{\rm low}}) \cdot X_{\rm low}  \right)
            \right] \Big|  \lambda^0, f^0, w \right\} &= 0.
   \end{align*}
   Using $M_{\bar \lambda} =M_{\bar \lambda} M_{\bar \lambda}$
   and  the cyclicality of the trace we
   obtain from the last equality:
   \begin{align*}
       \Tr \bigg\{
         M_{\bar \lambda}
          A M_{\bar \lambda}
                       \bigg\}  = 0,
   \end{align*}
   where
   $A =\mathbb{E} \left[
           \left(  \lambda^0 f^{0 \prime}
              -  (\bar \beta^{\rm low}-\beta^{0,{\rm low}}) \cdot X_{\rm low}  \right)
\left(  \lambda^0 f^{0 \prime} -
           (\bar \beta^{\rm low}-\beta^{0, {\rm low}}) \cdot X_{\rm low}  \right)' \Big|  \lambda^0, f^0, w \right]$.   
   The trace of a positive semi-definite matrix is only equal to zero
   if the matrix itself is equal to zero, so we find
   \begin{align*}
       M_{\bar \lambda}
          A M_{\bar \lambda}
           &= 0,
   \end{align*}
   This together with the fact that $A$ itself
   is positive semi definite implies
   (note that $A$ positive semi-definite implies $A=CC'$
   for some matrix $C$, and $M_{\bar \lambda}
          A M_{\bar \lambda}
           = 0$ then implies $M_{\bar \lambda} C = 0$,
           i.e., $C = P_{\bar \lambda} C$)
   \begin{align*}
      A &=  
          P_{\bar \lambda}
          A P_{\bar \lambda}  ,
   \end{align*}
   and therefore
   ${\rm rank}(A) \leq {\rm rank}( P_{\bar \lambda} )  \leq R$. 
   We have thus shown
    \begin{align*}
       {\rm rank} \left\{ \mathbb{E} \left[
          \left(  \lambda^0 f^{0 \prime} -
           (\bar \beta^{\rm low}-\beta^{0, {\rm low}}) \cdot X_{\rm low}  \right)
           \left(  \lambda^0 f^{0 \prime}
              -  (\bar \beta^{\rm low}-\beta^{0,{\rm low}}) \cdot X_{\rm low}  \right)'
            \Big|  \lambda^0, f^0, w \right]
          \right\} \leq R.
   \end{align*}
    We furthermore find
    \begin{align*}
        R &\geq {\rm rank} \left\{ \mathbb{E} \left[
          \left(  \lambda^0 f^{0 \prime} -
           (\bar \beta^{\rm low}-\beta^{0, {\rm low}}) \cdot X_{\rm low}  \right)
           \left(  \lambda^0 f^{0 \prime}
              -  (\bar \beta^{\rm low}-\beta^{0,{\rm low}}) \cdot X_{\rm low}  \right)'
            \Big|  \lambda^0, f^0, w \right] \right\}
       \nonumber \\
      &  \geq      {\rm rank} \left\{ M_{w} \mathbb{E} \left[
          \left(  \lambda^0 f^{0 \prime} -
           (\bar \beta^{\rm low}-\beta^{0, {\rm low}}) \cdot X_{\rm low}  \right)
           P_{f^0}
           \left(  \lambda^0 f^{0 \prime}
              -  (\bar \beta^{\rm low}-\beta^{0,{\rm low}}) \cdot X_{\rm low}  \right)' M_w
            \Big|  \lambda^0, f^0, w \right] \right\}
   \nonumber \\ & \quad
       +  {\rm rank} \left\{ P_{w} \mathbb{E} \left[
          \left(  \lambda^0 f^{0 \prime} -
           (\bar \beta^{\rm low}-\beta^{0, {\rm low}}) \cdot X_{\rm low}  \right)
           M_{f^0}
           \left(  \lambda^0 f^{0 \prime}
              -  (\bar \beta^{\rm low}-\beta^{0,{\rm low}}) \cdot X_{\rm low}  \right)' P_w
            \Big|  \lambda^0, f^0, w \right] \right\}
       \nonumber \\
      &  \geq      {\rm rank} \left[
      M_{w}   \lambda^0 f^{0 \prime} f^0 \lambda^{0 \prime} M_w
          \right]
   \nonumber \\ & \quad
       +  {\rm rank} \left\{   \mathbb{E} \left[
          \left(
           (\bar \beta^{\rm low}-\beta^{0, {\rm low}}) \cdot X_{\rm low}  \right)
           M_{f^0}
           \left(
              (\bar \beta^{\rm low}-\beta^{0,{\rm low}}) \cdot X_{\rm low}  \right)'
            \Big|  \lambda^0, f^0, w \right] \right\} .
   \end{align*}
    Assumption~\ref{ass:id}$(iv)$ guarantees
   ${\rm rank} \left( M_{w}   \lambda^0 f^{0 \prime} f^0 \lambda^{0 \prime} M_w \right)
    = {\rm rank} \left(  \lambda^0 f^{0 \prime} f^0 \lambda^{0 \prime} \right)
    = R$, that is, we must have
    \begin{align*}
         \mathbb{E} \left[
          \left(
           (\bar \beta^{\rm low}-\beta^{0, {\rm low}}) \cdot X_{\rm low}  \right)
           M_{f^0}
           \left(
              (\bar \beta^{\rm low}-\beta^{0,{\rm low}}) \cdot X_{\rm low}  \right)'
            \Big|  \lambda^0, f^0, w \right] =0.
    \end{align*}
    According to Assumption~\ref{ass:id}$(iii)$ this implies
    $\bar \beta^{\rm low} =\beta^{0,{\rm low}}$, i.e., we have
    $\bar \beta = \beta^0$.
    This also implies
    $Q^*(\bar \beta,\bar \lambda,\bar f) = \|  \lambda^0 f^{0 \prime} -  \bar \lambda \bar f'   \|_F^2
    =0$,
    and therefore $\bar \lambda \bar f' = \lambda^0 f^{0 \prime}$.
\end{proof}

\section{Examples of Error Distributions}
\label{app:error}

The following Lemma provides examples of error distributions that satisfy $\|e\|={\cal O}_p(\sqrt{\max(N,T)})$
as $N,T \rightarrow \infty$.
Example (i) is particularly relevant for us, because those
assumptions on $e_{it}$ are imposed in Assumption~\ref{ass:A5} in the main text, i.e., under 
those main text assumptions we indeed have $\|e\|={\cal O}_p(\sqrt{\max(N,T)})$.

\begin{lemma}
    \label{lemma:Enorm}
For each of the following distributional assumptions on the errors $e_{it}$, $i=1,\ldots,N$, $t=1,\ldots,T$,
we have $\|e\|={\cal O}_p(\sqrt{\max(N,T)})$.
   \begin{itemize}
      \item[(i)] The $e_{it}$ are independent across $i$ and $t$, conditional on ${\cal C}$, and
                 satisfy
                 $\mathbb{E}(e_{it} | {\cal C}) = 0$, and $\mathbb{E}(e_{it}^4 | {\cal C})$ is bounded uniformly by a non-random constant,
                 uniformly
                 over $i,t$ and $N,T$.
  Here ${\cal C}$ can be any conditioning sigma-field, including
  the empty one (corresponding to unconditional expectations).               
                                  
      \item[(ii)] The $e_{it}$ follow different ${\rm MA}(\infty)$ processes for each $i$, namely
                  \begin{align}
                     e_{it} &= \sum_{\tau=0}^\infty \, \psi_{i\tau} \, u_{i,t-\tau} \; , \qquad \text{for }
                     i=1\ldots N, \; t=1\ldots T \; ,  \label{errorMA}
                  \end{align}
                  where the $u_{it}$, $i=1\ldots N$, $t=-\infty \ldots T$ are independent random
                  variables with $\mathbb{E} u_{it} =0$ and $\mathbb{E} u_{it}^4$ uniformly bounded
                  across $i,t$ and $N,T$. The coefficients $\psi_{i\tau}$ satisfy
                  \begin{align}
                     \sum_{\tau=0}^\infty \, \tau \, \max_{i=1\ldots N} \, \psi_{i\tau}^2  \,
                            &< \, B \; , & \sum_{\tau=0}^\infty \, \max_{i=1\ldots N}  \left|
                      \psi_{i\tau} \right|  \, &< \, B \; ,
                      \label{MArestrictionsPSI}
                  \end{align}
                  for a finite constant $B$ which is independent of $N$ and $T$.
      \item[(iii)] The error matrix $e$ is generated as
         $e=\sigma^{1/2} \, u \, \Sigma^{1/2}$, where
                   $u$ is an $N \times T$ matrix with
                   independently distributed entries $u_{it}$ and $\mathbb{E} u_{it}=0$,
                   $\mathbb{E} u_{it}^2=1$, and $\mathbb{E} u_{it}^4$ is bounded uniformly across $i,t$
                   and $N,T$. Here $\sigma$ is the $N\times N$ cross-sectional covariance matrix,
                   and $\Sigma$ is the $T\times T$ time-serial covariance matrix, and they satisfy
                   \begin{align}
                      \max_{j=1\ldots N} \, \sum_{i=1}^N \, \left| \sigma_{ij} \right| \, &< \, B
                       \; , & \max_{\tau=1\ldots T} \, \sum_{t=1}^T \, \left| \Sigma_{t\tau}
                      \right| \, &< \, B \; ,
                   \end{align}
                   for some finite constant $B$ which is independent of $N$ and $T$.
                   In this example we have $\mathbb{E} e_{it} e_{j\tau} = \sigma_{ij} \Sigma_{t\tau}$.
\end{itemize}
\end{lemma}

\begin{proof}[\bf Proof of Lemma~\ref{lemma:Enorm}, Example (i)]
   Latala \cite*{Latala2006} showed that for a $N\times T$ matrix $e$ with independent entries,
   conditional on ${\cal C}$, we have
   \begin{align*}
      \mathbb{E}\left(\| e \| \, \big| {\cal C} \right)  \, \leq \, c \left\{ \max_i  \left[ \sum_t 
      \mathbb{E}\left( e_{it}^2  \, \big| {\cal C} \right) \right]^{1/2}
                                            +\max_j \left[ \sum_i \mathbb{E} \left( e_{it}^2  \, \big| {\cal C} \right) \right]^{1/2}
                                            + \left[ \sum_{i,t} \mathbb{E} \left( e_{it}^4 \, \big| {\cal C} \right) \right]^{1/4} 
                                            \right\} \; ,
   \end{align*}
   where $c$ is some universal constant.
   Because we assumed uniformly bounded $4$th conditional moments for $e_{it}$ we thus have
   $\|e\| = {\cal O}_P(\sqrt{T})+{\cal O}_P(\sqrt{N})+{\cal O}_P((TN)^{1/4}) = {\cal O}_p(\sqrt{\max(N,T)})$.
\end{proof}

\begin{proof}[\bf Example (ii)]

Let $\psi_j = (\psi_{1j}, \ldots , \psi_{Nj})$ be an $N \times 1$ vector for each $j \geq 0$.
Let $U_{-j}$ be an $N\times T$ sub-matrix of $(u_{it})$ consisting of $u_{it}$, $i=1\ldots N$, $t=1-j,\ldots,T-j$.
We can then write equation \eqref{errorMA} in matrix notation as
\begin{align*}
   e &= \sum_{j=0}^\infty \, \limfunc{diag}(\psi_j) \, U_{-j}
   \nonumber \\
     &= \sum_{j=0}^T \, \limfunc{diag}(\psi_j) \, U_{-j} + r_{NT} ,
\end{align*}
where we cut the sum at $T$, which results in the remainder
$r_{NT}= \sum_{j=T+1}^\infty \, \limfunc{diag}(\psi_j) \, U_{-j}$. When approximating
an ${\rm MA}(\infty)$ by a finite ${\rm MA}(T)$ process we have for the remainder
\begin{align*}
   \mathbb{E} \left(\| r_{NT} \|_{F}\right)^2 \,
     = \sum_{i=1}^N \, \sum_{t=1}^T \, \mathbb{E}  \left( r_{NT} \right)_{ij}^2 \,
      &\leq \, \sigma_u^2 \, \sum_{i=1}^N \, \sum_{t=1}^T \, \sum_{j=T+1}^\infty \, \psi_{ij}^2
      \nonumber \\
               &\leq  \sigma_u^2 \,  N \, T \, \sum_{j=T+1}^\infty \, \max_i\left( \psi_{ij}^2 \right)
      \nonumber \\
               &\leq  \sigma_u^2 \,  N \, \sum_{j=T+1}^\infty \, j \, \max_i\left( \psi_{ij}^2 \right) \;,
\end{align*}
where $\sigma_u^2$ is the variance of $u_{it}$.
Therefore, for $T \rightarrow \infty$ we have
\begin{align*}
   \mathbb{E} \left( \frac{ \left(\| r_{NT} \|_{F}\right)^2 } N \right) \, \longrightarrow \, 0 \; ,
\end{align*}
which implies $\left(\| r_{NT} \|_{F}\right)^2 = {\cal O}_p(N)$, and therefore
$\| r_{NT} \| \leq \| r_{NT} \|_{F} = {\cal O}_p(\sqrt{N})$.

Let $V$ be the $N\times 2T$ matrix consisting of $u_{it}$, $i=1\ldots N$, $t=1-T,\ldots,T$.
For $j=0\ldots T$ the matrices $U_{-j}$ are sub-matrices of $V$, and therefore $\| U_{-j} \| \leq \|V\|$.
From example (i) we know $\| V \| = {\cal O}_p(\sqrt{\max(N,2T)})$.
Furthermore, we know $\| \limfunc{diag}(\psi_j) \| \leq \max_i\left( \left| \psi_{ij} \right| \right)$.

Combining these results we find
\begin{align*}
    \left\| e  \right\|
       &\leq \, \sum_{j=0}^T \, \| \limfunc{diag}(\psi_j) \| \, \|U_{-j}\| + \|r_{NT}\|
     \nonumber \\
       &\leq \, \sum_{j=0}^T \, \max_i\left( \left| \psi_{ij} \right| \right) \|V\| + o_p(\sqrt{N})
     \nonumber \\
       &\leq \, \left[ \sum_{j=0}^\infty \, \max_i\left( \left| \psi_{ij} \right| \right) \right]
                      {\cal O}_p(\sqrt{\max(N,2T)})
                    + o_p(\sqrt{N})
     \nonumber \\
       &\leq \, {\cal O}_p(\sqrt{\max(N,T)}),
\end{align*}
as required for the proof.
\end{proof}

\begin{proof}[\bf Example (iii)]
  Because $\sigma$ and $\Sigma$ are positive definite, there exits a symmetric $N\times N$ matrix
  $\phi$ and a symmetric $T\times T$ matrix $\psi$ such that $\sigma=\phi^2$ and
  $\Sigma=\psi^2$. The error term can then be generated as $e= \phi u \psi$, where
  $u$ is an $N\times T$ matrix with iid entries $u_{it}$ such that
  $\mathbb{E}(u_{it})=0$ and $\mathbb{E}(u_{it}^4)<\infty$. Given this definition
  of $e$ we immediately have
  $\mathbb{E} e_{it} = 0$ and $\mathbb{E} e_{it} e_{j\tau} = \sigma_{ij} \Sigma_{t\tau}$.
  What is left to show is $\| e \|={\cal O}_p(\sqrt{\max(N,T)})$.
  From example (i) we know $\| u \|={\cal O}_p(\sqrt{\max(N,T)})$.
  Using the inequality
  $\| \sigma \| \leq \sqrt{ \| \sigma \|_1 \, \| \sigma \|_\infty} = \| \sigma \|_1$,
  where $\| \sigma \|_1 = \| \sigma \|_\infty$ because $\sigma$ is symmetric we find
  \begin{align*}
     \| \sigma \| \leq
     \| \sigma \|_1 \, \equiv \, \max_{j=1\ldots N} \, \sum_{i=1}^N \, \left| \sigma_{ij} \right| \, &< \, L \; , &
  \end{align*}
  and analogously $\| \Sigma \| < L$. Because $\| \sigma\| = \| \phi\|^2$ and $\|\Sigma\| = \| \psi \|^2$,
  we thus find $\|e \| \leq \| \phi \| \|u\| \|\psi \| \leq L {\cal O}_p(\sqrt{\max(N,T)})$,
  i.e., $\| e \|={\cal O}_p(\sqrt{\max(N,T)})$.
\end{proof}

\section{Comments on Assumption \ref{ass:A4} on the regressors}

Consistency of the LS estimator $\widehat \beta$ requires the regressors not only satisfy the standard non-collinearity
condition in assumption \ref{ass:A4}(i), but also the additional conditions on high- and low-rank regressors
in assumption \ref{ass:A4}(ii). Bai \cite*{Bai2009} considers the special cases of only high-rank and only low-rank
regressors. As low-rank regressors he considers only cross-sectional invariant and time-invariant regressors,
and he shows that if only these two types of regressors are present, one can show consistency
under the assumption $\limfunc{plim}_{N,T\rightarrow \infty} W_{NT} > 0$ on the regressors
(instead of assumption \ref{ass:A4}), where $W_{NT}$ is the $K\times K$ matrix
defined by $W_{NT,k_1 k_2} = (NT)^{-1} \, {\rm Tr}(M_{f^0} \, X^{\prime}_{k_1} \, M_{\lambda^0} \, X_{k_2})$.
This matrix appears as the approximate Hessian in the profile objective expansion in theorem \ref{th:ass_expand},
i.e., the condition $\limfunc{plim}_{N,T\rightarrow \infty} W_{NT} > 0$ is very natural in the context
of the interactive fixed effect models, and one may wonder whether also for the general case one can replace
assumption \ref{ass:A4} with this weaker condition and still obtain consistency of the LS estimator.
Unfortunately, this is not the case, and below we present two simple counter
examples that show this.

\begin{itemize}
   \item[(i)] Let there only be one factor ($R=1$) $f^0_t$ with corresponding factor loadings $\lambda^0_i$.
              Let there only be one regressor ($K=1$) of the form $X_{it}=w_i v_t + \lambda^0_i f^0_t$.
              Assume the $N\times 1$ vector $w=(w_1,\ldots,w_N)'$, and the $T\times 1$ vector
              $v=(v_1,\ldots,v_N)'$ are such that the $N\times 2$ matrix $\Lambda=(\lambda^0,w)$ and
              and the $T\times 2$ matrix $F=(f^0,v)$ satisfy
              $\limfunc{plim}_{N,T \rightarrow \infty}\left(\Lambda^{\prime} \Lambda/N\right) > 0$, and
              $\limfunc{plim}_{N,T \rightarrow \infty} \left( F^{\prime} F / T \right) > 0$.
              In this case, we have
            $W_{NT}=(NT)^{-1} \, {\rm Tr}(M_{f^0} \, v w' \, M_{\lambda^0} \, w v')$,
              and therefore $\limfunc{plim}_{N,T\rightarrow \infty} W_{NT}
               =\limfunc{plim}_{N,T\rightarrow \infty} (NT)^{-1} \, {\rm Tr}(M_{f^0} \, v w' \, M_{\lambda^0} \, w v')
                          > 0$. However, $\beta$ is not identified
              because $\beta^0 X + \lambda^0 f^{0\prime} = (\beta^0+1) X - w v'$, i.e., it is not possible
              to distinguish $(\beta,\lambda,f)=(\beta^0,\lambda^0,f^0)$ and
              $(\beta,\lambda,f)=(\beta^0+1,-w,v)$. This implies that the LS estimator is not consistent
              (both $\beta^0$ and $\beta^0+1$ could be the true parameter, but the LS estimator cannot be consistent for both).
   \item[(ii)] Let there only be one factor ($R=1$) $f^0_t$ with corresponding factor loadings $\lambda^0_i$.
               Let the $N\times 1$ vectors $\lambda^0$, $w_1$ and $w_2$ be such that
               $\Lambda=(\lambda^0,w_1,w_2)$ satisfies
               $\limfunc{plim}_{N,T \rightarrow \infty}\left(\Lambda^{\prime} \Lambda/N\right) > 0$.
               Let the $T\times 1$ vectors $f^0$, $v_1$ and $v_2$ be such that
               $F=(f^0,v_1,v_2)$ satisfies
               $\limfunc{plim}_{N,T \rightarrow \infty} \left( F^{\prime} F / T \right) > 0$.
               Let there be four regressors ($K=4$) defined by
               $X_1=w_1 v_1'$, $X_2=w_2 v_2'$,
               $X_3=(w_1+\lambda^0)(v_2+f^0)'$, $X_4=(w_2+\lambda^0)(v_1+f^0)'$.
               In this case, one can easily check $\limfunc{plim}_{N,T\rightarrow \infty} W_{NT} > 0$.
               However, again $\beta_k$ is not identified, because
               $\sum_{k=1}^4 \beta^0_k X_k + \lambda^0 f^{0\prime} =
                \sum_{k=1}^4 (\beta^0_k+1) X_k - (\lambda^0+w_1+w_2) (f^{0\prime}+v_1+v_2)'$,
               i.e., we cannot distinguish between the true parameters
               and $(\beta,\lambda,f)=(\beta^0+1,\,-\lambda^0-w_1-w_2,\,f^{0\prime}+v_1+v_2)$.
               Again, as a consequence the LS estimator is not consistent in this case.
\end{itemize}

In example (ii), there are only low-rank regressors with ${\rm rank}(X_l)=1$.
One can easily check assumption \ref{ass:A4} is not satisfied for this example.
In example (i) the regressor is a low-rank regressor with ${\rm rank}(X)=2$.
In our present version of assumption \ref{ass:A4} we only consider low-rank regressors with ${\rm rank}(X)=1$,
but (as already noted in a footnote in the main paper) it is straightforward to extend the assumption and the
consistency proof to low-rank regressors with rank larger than one. Independent of whether we extend the assumption
or not, the regressor $X$ of example (i) fails to satisfy assumption \ref{ass:A4}.
This justifies our formulation of assumption \ref{ass:A4}, because it shows in general
the assumption cannot be replaced by the
weaker condition $\limfunc{plim}_{N,T\rightarrow \infty} W_{NT} > 0$.

\section{Some Matrix Algebra (including Proof of Lemma~\ref{lemma:Optimization})}
\label{app:matrix}

The following statements are true for real matrices (throughout the whole paper and supplementary material
we never use complex numbers anywhere).
Let $A$ be an arbitrary $n\times m$ matrix.
In addition to the operator (or spectral) norm $\|A\|$ and to the Frobenius (or Hilbert-Schmidt) norm $\|A\|_{F}$,
it is also convenient to define the $1$-norm, the $\infty$-norm, and the $\max$-norm by
\begin{align*}
   \| A \|_1 \, &= \, \max_{j=1\ldots m} \, \sum_{i=1}^n \, \left| A_{ij} \right| \; , &
   \| A \|_\infty \, &= \, \max_{i=1\ldots n} \, \sum_{j=1}^m \, \left| A_{ij} \right| \; , &
   \| A \|_{\max} \, &= \, \max_{i=1\ldots n} \, \max_{j=1 \ldots m} \, \left| A_{ij} \right| \; . &
\end{align*}

\begin{lemma}[Some useful inequalities]
   \label{lemma:inequalities}
   Let A be an $n\times m$ matrix, $B$ be an $m\times p$ matrix, and $C$ and $D$ be $n\times n$ matrices.
   Then we have:
   \begin{align*}
       \text{(i)}& \qquad
             \left\| A\right\| \, \leq \, \left\| A\right\|_{F}
             \, \leq \, \left\| A\right\| \, \limfunc{rank}\left( A\right)^{1/2} \; ,
      \nonumber \\
       \text{(ii)}& \qquad
             \left\| AB \right\| \, \leq \, \left\| A\right\| \left\|B\right\|     \; ,
      \nonumber \\
       \text{(iii)}& \qquad
             \left\| AB \right\|_{F} \, \leq \, \left\| A\right\|_{F} \left\|B\right\|
                                          \, \leq \, \left\| A\right\|_{F} \left\|B\right\|_{F} \; ,
      \nonumber \\
       \text{(iv)}& \qquad
             |{\rm Tr}(AB)| \, \leq \, \left\| A\right\|_{F} \left\|B\right\|_{F} \; ,
                \qquad \text{for $n=p$,}
      \nonumber \\
       \text{(v)}& \qquad
                     \left| {\rm Tr}\left( C\right) \right|
                      \leq \left\| C\right\| \limfunc{rank}\left( C\right) \; ,
      \nonumber \\
       \text{(vi)}& \qquad
                     \left\| C\right\| \leq {\rm Tr}\left( C\right) \; ,
                \qquad \text{for $C$ symmetric and $C\geq0$,}
      \nonumber \\
       \text{(vii)}& \qquad
            \|A\|^2 \, \leq \, \|A\|_1 \, \|A\|_{\infty} \; ,
      \nonumber \\
       \text{(viii)}& \qquad
             \|A\|_{\max} \, \leq \, \|A\| \, \leq \, \sqrt{nm} \, \|A\|_{\max} \; ,
      \nonumber \\
       \text{(ix)}& \qquad
            \|A' C A \| \leq \|A' D A \| \; ,
                \qquad  \text{for $C$ symmetric and $C\leq D$.}
      \nonumber \\
      & \text{For $C$, $D$ symmetric, and $i=1,\ldots,n$ we have:}
      \nonumber \\
       \text{(x)} & \qquad   %
         {\mu}_i(C) + {\mu}_n(D) \, \leq \, {\mu}_i(C+D) \, \leq \,
        {\mu}_i(C) + {\mu}_1(D) \; ,
      \nonumber \\
       \text{(xi)} & \qquad
             {\mu}_i(C) \leq \, {\mu}_i(C+D) \; ,
                \qquad  \text{for $D\geq0$,}
      \nonumber \\
       \text{(xii)} & \qquad
             {\mu}_i(C) - \|D\| \, \leq \, {\mu}_i(C+D) \, \leq \,  {\mu}_i(C) + \|D\| \; .
   \end{align*}
\end{lemma}

\begin{proof}[\bf Proof]%
  Here we use notation $s_i(A)$ for the $i$th largest singular value of a matrix $A$.
\\
   (i) We have $\|A\|=s_1(A)$, and $\|A\|_F^2=\sum_{i=1}^{{\rm rank}(A)} (s_i(A))^2$.
       The inequalities follow directly from this representation.
   (ii) This inequality is true for all unitarily invariant norms, see, e.g., Bhatia \cite*{Bhatia97}.
   (iii) can be shown as follows
   \begin{align*}
       \left\| AB \right\|_{F}^2 &= {\rm Tr}(ABB'A')
               \nonumber \\
                        &= {\rm Tr}[\|B\|^2 \, AA'
                                    - A(\|B\|^2\mathbb{I}-BB')A']
               \nonumber \\
                        &\leq \|B\|^2 {\rm Tr}(AA') = \|B\|^2 \, \left\| A \right\|_{F}^2 \; ,
   \end{align*}
   where we used $A(\|B\|^2\mathbb{I}-BB')A'$ is positive definite.
   Relation (iv) is just the Cauchy Schwarz inequality.
   To show (v) we decompose $C=UDO'$ (singular value decomposition), where
   $U$ and $O$ are $n\times {\rm rank}(C)$ that satisfy $U'U=O'O=\mathbb{I}$
   and $D$ is a ${\rm rank}(C) \times {\rm rank}(C)$ diagonal matrix with entries
   $s_i(C)$. We then have $\|O\|=\|U\|=1$ and $\|D\|=\|C\|$ and therefore
   \begin{align*}
      |{\rm Tr}(C)| &= |{\rm Tr}(UDO')| = |{\rm Tr}(DO'U)|
               \nonumber \\
               &= \left|\sum_{i=1}^{{\rm rank}(C)} \, \eta_i' DO'U \eta_i \right|
               \nonumber \\
               &\leq \sum_{i=1}^{{\rm rank}(C)} \|D\| \|O'\| \|U\| = {\rm rank}(C) \|C\| \; .
   \end{align*}
   For (vi) let $e_1$ be a vector that satisfies $\|e_1\|=1$ and $\left\| C\right\| = e_1' C e_1$.
   Because $C$ is symmetric such an $e_1$ has to exist.
   Now choose $e_i$, $i=2,\ldots,n$, such that $e_i$, $i=1,\ldots,n$, becomes a orthonormal basis
   of the vector space of $n\times 1$ vectors. Because $C$ is positive semi definite we then have
   ${\rm Tr}\left( C\right) = \sum_{i} e_i' C e_i \geq e_1 C e_1 = \|C\|$, which is what we wanted to show.
   For (vii) we refer to Golub and van Loan \cite*{golubvanloan1996}, p.15.
   For (viii) let $e$ be the vector that satisfies $\|e\|=1$ and
   $\|A' C A\|=e' A' C A e$. Because $A' C A$ is symmetric such an $e$ has to exist.
   Because $C\leq D$ we then have $\|C\|= (e' A') C (A e) \leq (e' A') D (A e) \leq \|A'DA\|$.
   This is what we wanted to show.
   For inequality (ix) let $e_1$ be a vector that satisfied $\|e_1\|=1$ and $\left\| A'C A\right\| = e_1' A' C A e_1$.
   Then we have $\left\| A'C A\right\| = e_1' A' D A e_1 - e_1' A' (D-C) A e_1 \leq e_1' A' D A e_1 \leq \|A'DA\|$.
   Statement (x) is a special case of Weyl's inequality, see, e.g., Bhatia \cite*{Bhatia97}.
   The inequalities (xi) and (xii) follow directly from (ix) because ${\mu}_n(D)\geq 0$ for $D\geq 0$,
   and because $-\|D\| \leq {\mu}_i(D)\leq \|D\|$ for $i=1,\ldots,n$.
\end{proof}

\begin{definition}
   \label{def:angle}
   Let $A$ be an $n\times r_1$ matrix and $B$ be an $n \times r_2$ matrix
   with ${\rm rank}(A)=r_1$ and ${\rm rank}(B)=r_2$. The smallest principal angle
   $\theta_{A,B} \in [0,\pi/2]$ between the linear subspaces
   ${\rm span}(A)=\{A a | \, a \in \mathbb{R}^{r_1} \}$
   and ${\rm span}(B)=\{B b | \, b \in \mathbb{B}^{r_2} \}$ of $\mathbb{R}^n$ is defined by
   \begin{align*}
      \cos(\theta_{A,B}) &= \max_{0 \neq a \in \mathbb{R}^{r_1}} \max_{0\neq b \in \mathbb{R}^{r_2}}
                                  \frac{a' A' B b} {\|A a\| \|B b\|} \, .
   \end{align*}
\end{definition}

\begin{lemma}
   \label{lemma:angle}
   Let $A$ be an $n\times r_1$ matrix and $B$ be an $n \times r_2$ matrix
   with ${\rm rank}(A)=r_1$ and ${\rm rank}(B)=r_2$. Then we have the following alternative
   characterizations of the smallest principal angle between ${\rm span}(A)$ and ${\rm span}(B)$
   \begin{align*}
      \sin(\theta_{A,B}) &= \min_{0 \neq a \in \mathbb{R}^{r_1}} \, \frac{\| M_B \, A \, a \|} {\|A \, a\|}
         \nonumber \\
                        &= \min_{0 \neq b \in \mathbb{R}^{r_2}} \, \frac{\| M_A \, B \, b \|} {\|B \, b\|} \; .
   \end{align*}
\end{lemma}

\begin{proof}[\bf Proof]%
   Because $\| M_B \, A \, a \|^2 + \| P_B \, A \, a \|^2 = \|A\,a\|^2$ and
   $\sin(\theta_{A,B})^2 + \cos(\theta_{A,B})^2 = 1$, we find proving the theorem is equivalent to proving
   \begin{align*}
      \cos(\theta_{A,B}) &= \min_{0 \neq a \in \mathbb{R}^{r_1}} \, \frac{\| P_B \, A \, a \|} {\|A \, a\|}
                          = \min_{0 \neq b \in \mathbb{R}^{r_2}} \, \frac{\| P_A \, B \, b \|} {\|A \, b\|} \; .
   \end{align*}
   This last statement is theorem 8 in Galantai and Hegedus \cite*{GalantaiHegedus2006}, and the proof can be found there.
\end{proof}

\begin{proof}[\bf Proof of Lemma~\ref{lemma:Optimization}]
    Let
   \begin{align*}
      S_1(Z) &= \min_{f,\lambda} {\rm Tr}\left[ \left(Z-\lambda f'\right) \left(Z'-f \lambda'\right)\right] \; ,
 \nonumber \\
      S_2(Z) &= \min_f {\rm Tr}(Z \, M_f \, Z') \; ,
 \nonumber \\
      S_3(Z) &= \min_\lambda {\rm Tr}(Z' \, M_\lambda \, Z) \; ,
 \nonumber \\
      S_4(Z) &= \min_{\tilde \lambda,\tilde f} {\rm Tr}(M_{\widetilde \lambda} \, Z \, M_{\widetilde f} \, Z') \; ,
 \nonumber \\
      S_5(Z) &= \sum_{i=R+1}^{T} {\mu}_i(Z'Z) \; ,
 \nonumber \\
      S_6(Z) &= \sum_{i=R+1}^{N} {\mu}_i(ZZ') \; .
   \end{align*}
   The theorem claims
   \begin{align*}
      S_1(Z) \, &= \, S_2(Z) \, = \, S_3(Z) \, = \, S_4(Z) \, = \, S_5(Z) \, = \, S_6(Z) \; .
   \end{align*}
   We find:
   \begin{itemize}
      \item[(i)] The non-zero eigenvalues of $Z'Z$ and $ZZ'$ are identical, so in the sums in $S_5(Z)$ and
                 in $S_6(Z)$ we are summing over identical values, which shows $S_5(Z)=S_6(Z)$.
      \item[(ii)] Starting with $S_1(Z)$ and minimizing with respect to $f$ we obtain the first-order condition
                  \begin{align*}
                      \lambda^{\prime}\, Z &= \lambda^{\prime}\, \lambda \, f^{\prime} \; .
                  \end{align*}
                  Putting this into the objective function we can integrate out $f$, namely
                 \begin{align*}
{\rm Tr}\left[ \left( Z - \lambda f^{\prime }\right) ^{\prime
}\left( Z - \lambda f^{\prime }\right) \right] &= {\rm Tr}\left(
Z^{\prime} Z - Z^{\prime} \lambda f^{\prime}\right)
\notag \\
&= {\rm Tr}\left( Z^{\prime} Z - Z^{\prime} \lambda
(\lambda^{\prime}\lambda)^{-1} (\lambda^{\prime}\lambda) f^{\prime}\right)
\notag \\
&= {\rm Tr}\left( Z^{\prime} Z - Z^{\prime} \lambda
(\lambda^{\prime}\lambda)^{-1} (\lambda^{\prime}\lambda) \lambda^{\prime}\,
Z \right)  \notag \\
&= {\rm Tr}\left( Z^{\prime} \, M_\lambda \, Z \right) \; .
                 \end{align*}
                 This shows $S_1(Z)=S_3(Z)$.
                 Analogously, we can integrate out $\lambda$
                 to obtain $S_1(Z)=S_2(Z)$.
       \item[(iii)] Let $M_{\widehat \lambda}$ be the projector on the $N-R$ eigenspaces
                   corresponding to the $N-R$ smallest eigenvalues\footnote{
If an eigenvalue has multiplicity $m$, we count it $m$ times when finding
the $N-R$ smallest eigenvalues. In this terminology we always have exactly $N$
eigenvalues of $ZZ^{\prime}$, but some may appear multiple
times.} of $ZZ^{\prime}$, let $P_{\widehat \lambda} = \mathbb{I}_N
- M_{\widehat \lambda}$, and let $\omega_{R}$ be the $R$'th largest eigenvalue
of $ZZ^{\prime}$. We then know the matrix $P_{\widehat
\lambda} [Z Z^{\prime}-\omega_R \mathbb{I}_N ] P_{\widehat
\lambda} -M_{\widehat \lambda} [Z Z^{\prime}-\omega_R \mathbb{I}_N
] M_{\widehat \lambda}$ is positive semi-definite. Thus, for an arbitrary $%
N\times R$ matrix $\lambda$ with corresponding projector $M_\lambda$ we have
\begin{align*}
0 &\leq {\rm Tr} \left\{ \left( P_{\widehat \lambda} [Z
Z^{\prime}-\omega_R \mathbb{I}_N ] P_{\widehat \lambda} -M_{\widehat \lambda}
[Z Z^{\prime}-\omega_R \mathbb{I}_N ] M_{\widehat \lambda} \right)
\left( M_{\lambda} - M_{\widehat \lambda} \right)^2 \right\}  \notag \\
&= {\rm Tr} \left\{ \left( P_{\widehat \lambda} [Z
Z^{\prime}-\omega_R \mathbb{I}_N ] P_{\widehat \lambda} + M_{\widehat
\lambda} [Z Z^{\prime}-\omega_R \mathbb{I}_N ] M_{\widehat
\lambda} \right) \left( M_{\lambda} - M_{\widehat \lambda} \right) \right\}
\notag \\
&= {\rm Tr} \left[ Z^{\prime} \, M_\lambda \, Z \right] -%
{\rm Tr} \left[ Z^{\prime} \, M_{\widehat \lambda} \, Z \right]
+ \omega_R \, \left[ \limfunc{rank}(M_\lambda) - \limfunc{rank}(M_{\widehat
\lambda}) \right] \; ,
\end{align*}
and because $\limfunc{rank}(M_{\widehat \lambda}) = N-R$ and $\limfunc{rank}%
(M_\lambda) \leq N-R$ we have
\begin{align*}
{\rm Tr} \left[ Z^{\prime} \, M_{\widehat \lambda} \, Z \right]
&\leq {\rm Tr} \left[ Z^{\prime} \, M_\lambda \, Z \right]
\; .
\end{align*}
This shows $M_{\widehat \lambda}$ is the optimal choice in the minimization problem
of $S_3(Z)$, i.e., the optimal $\lambda=\widehat \lambda$ is chosen such
that the span of the $N$-dimensional vectors $\widehat \lambda_r$ ($r=1\ldots R$)
equals to the span of the $R$ eigenvectors that correspond to the $R$
largest eigenvalues of $ZZ^{\prime}$. This shows $S_3(Z)=S_6(Z)$.
Analogously one can show $S_2(Z)=S_5(Z)$.
      \item[(iv)]
         In the minimization problem in $S_4(Z)$ we can choose
         $\widetilde \lambda$ such that the span of the $N$-dimensional vectors $\widetilde \lambda_r$ ($r=1\ldots R_1$) is
         equal to the span of the $R_1$ eigenvectors that correspond to the $R_1$
         largest eigenvalues of $ZZ^{\prime}$. In addition, we can choose
         $\widetilde f$ such that the span of the $T$-dimensional vectors $\widetilde f_r$ ($r=1\ldots R_2$) is
         equal to the span of the $R_2$ eigenvectors that correspond to the $(R_1+1)$-largest up to the
         $R$-largest eigenvalue of $Z^{\prime}Z$. With this choice of $\widetilde \lambda$ and $\widetilde f$
         we actually project out all the $R$ largest eigenvalues of $Z'Z$ and $ZZ'$. This shows that
         $S_4(Z) \leq S_5(Z)$. (This result is actually best understood by using the singular value decomposition
         of $Z$.)

         We can write $M_{\widetilde \lambda} \, Z \, M_{\widetilde f}=Z-\widetilde Z$, where
         \begin{align*}
            \widetilde Z &= P_{\widetilde \lambda} \, Z \, M_{\widetilde f} + Z \, P_{\widetilde f} \; .
         \end{align*}
         Because ${\rm rank}(Z)\leq {\rm rank}(P_{\widetilde \lambda} \, Z \, M_{\widetilde f})
                                 +{\rm rank}(Z \, P_{\widetilde f}) = R_1 + R_2 = R$, we can
         always write $\widetilde Z=\lambda f'$ for some appropriate $N\times R$ and $T\times R$ matrices
         $\lambda$ and $f$. This shows that
         \begin{align*}
            S_4(Z) &=    \min_{\bar \lambda,\bar f} {\rm Tr}(M_{\widetilde \lambda} \, Z \, M_{\widetilde f} \, Z')
               \nonumber \\
                   &\geq \min_{\{\widetilde Z \;:\; {\rm rank}(\widetilde Z)\leq R\}} {\rm Tr}((Z-\widetilde Z)(Z-\widetilde Z)')
               \nonumber \\
                   &= \min_{f,\lambda} {\rm Tr}\left[ \left(Z-\lambda f'\right) \left(Z'-f \lambda'\right)\right]
                    = S_1(Z) \; .
         \end{align*}
         Thus we have shown here $S_1(Z) \leq S_4(Z) \leq S_5(Z)$, and this holds with equality
         because $S_1(Z)=S_5(Z)$ was already shown above.
   \end{itemize}
\end{proof}

\section{Supplement to the Consistency Proof (Appendix \ref{app:consistency})}

\begin{lemma}
   \label{lemma:wv}
   Under assumptions \ref{ass:A1} and \ref{ass:A4} there exists a constant $B_0>0$ such that
   for the matrices $w$ and $v$ introduced in assumption \ref{ass:A4} we have
   \begin{align*}
      w' \, M_{\lambda^0} \, w  \, - \, B_0 \, w' \, w &\geq 0 \; , \qquad \text{wpa1,}
    \nonumber \\
      v' \, M_{f^0} \, v  \, - \, B_0 \, v' \, v &\geq 0  \; , \qquad \text{wpa1.}
   \end{align*}
\end{lemma}

\begin{proof}[\bf Proof]
   We can decompose $w=\widetilde w \, \bar w$, where
   $\widetilde w$ is an $N \times {\rm rank}(w)$ matrix and $\bar w$ is a ${\rm rank}(w) \times K_1$ matrix.
   Note $\widetilde w$ has full rank, and $M_w=M_{\widetilde w}$.

   By assumption \ref{ass:A1}(i) we know $\lambda^{0\prime}\lambda^0/N$ has a probability limit,
   i.e., there exists some $B_1>0$ such that $\lambda^{0\prime}\lambda^0/N < B_1 \mathbb{I}_R$ wpa1.
   Using this and assumption \ref{ass:A4} we find for any $R\times 1$ vector $a\neq0$ we have
   \begin{align*}
      \frac{\|M_{v} \, \lambda^0 \, a\|^2}
          {\|\lambda^0 \, a\|^2} \, =   \,
      \frac{a' \, \lambda^{0\prime} \, M_{v} \, \lambda^0 \, a }
          {a' \, \lambda^{0\prime} \, \lambda^0 \, a} &> \frac{B}{B_1} \;,  \qquad \text{wpa1.}
   \end{align*}
   Applying Lemma~\ref{lemma:angle} we find
   \begin{align*}
      \min_{0\neq b \in \mathbb{R}^{{\rm rank}(w)}} \, \frac{b' \, \widetilde w' \, M_{\lambda^0} \, \widetilde w \, b }
          {b' \, \widetilde w' \, \widetilde w \, b} \, =   \,
      \min_{0\neq a \in \mathbb{R}^R} \, \frac{a' \, \lambda^{0\prime} \, M_{w} \, \lambda^0 \, a }
          {a' \, \lambda^{0\prime} \, \lambda^0 \, a} &> \frac{B}{B_1} \; ,  \; \qquad \text{wpa1.}
   \end{align*}
   Therefore we find for every ${\rm rank}(w) \times 1$ vector $b$ that
$b' \left( \widetilde w' \, M_{\lambda^0} \,  \widetilde w - (B/B_1) \widetilde w' \widetilde w \,\right) b > 0$,
wpa1.
   Thus
$\widetilde w' \, M_{\lambda^0} \, \widetilde w - (B/B_1) \, \widetilde w' \, \widetilde w  >  0$, wpa1.
Multiplying from the left with $\bar w'$ and from the right with $\bar w$ we obtain
$w' \, M_{\lambda^0} \, w - (B/B_1) \, w' \, w  \geq  0$, wpa1.
This is what we wanted to show.
Analogously we can show the statement for $v$.
\end{proof}

As a consequence of the this lemma we obtain some properties of the low-rank regressors summarized in the following
lemma.
\begin{lemma}
   \label{lemma:lowrankprop}
   Let the assumptions \ref{ass:A1} and \ref{ass:A4} be satisfied and let
   $X_{{\rm low},\alpha}=\sum_{l=1}^{K_1} \alpha_l X_l$ be a linear combination of the low-rank regressors.
   Then there exists some constant $B>0$ such that
   \begin{align*}
         \min_{\{\alpha \in \mathbb{R}^{K_1}, \|\alpha\|=1\}}
         \frac{\left\|X_{{\rm low},\alpha} \, M_{f^0} \, X_{{\rm low},\alpha}'\right\|}{NT}
         &> B \; , \qquad \text{wpa1,}
      \nonumber \\
         \min_{\{\alpha \in \mathbb{R}^{K_1}, \|\alpha\|=1\}}
  \frac{\left\|M_{\lambda^0} \, X_{{\rm low},\alpha} \, M_{f^0} \, X_{{\rm low},\alpha}' \, M_{\lambda^0} \right\|}{NT}
         &> B \; , \qquad \text{wpa1.}
   \end{align*}
\end{lemma}

\begin{proof}[\bf Proof]
   Note
   $\left\|M_{\lambda^0} \, X_{{\rm low},\alpha} \, M_{f^0} \, X_{{\rm low},\alpha}' \, M_{\lambda^0} \right\|
     \leq \left\|X_{{\rm low},\alpha} \, M_{f^0} \, X_{{\rm low},\alpha}'\right\|$,
   because $\|M_{\lambda^0}\|=1$, i.e., if we can show the second inequality of the lemma we have also shown
   the first inequality.

   We can write $X_{{\rm low},\alpha} = w \, {\rm diag}(\alpha') \, v'$.
   Using Lemma~\ref{lemma:wv} and part (v), (vi) and (ix) of Lemma~\ref{lemma:inequalities} we find
   \begin{align*}
      \left\|M_{\lambda^0} \, X_{{\rm low},\alpha} \, M_{f^0} \, X_{{\rm low},\alpha}' \, M_{\lambda^0} \right\|
      &=  \left\|M_{\lambda^0} \, w \, {\rm diag}(\alpha')
          \, v' \, M_{f^0} \, v \, {\rm diag}(\alpha') \, w'  M_{\lambda^0} \right\|
      \nonumber \\  &
      \geq  B_0 \, \left\|M_{\lambda^0} \, w \, {\rm diag}(\alpha')
          \, v' \, \, v \, {\rm diag}(\alpha') \, w'  M_{\lambda^0} \right\|
      \nonumber \\  &
      \geq  \frac{B_0}{K_1} \,  {\rm Tr} \left[ M_{\lambda^0} \, w \, {\rm diag}(\alpha')
          \, v' \, \, v \, {\rm diag}(\alpha') \, w'  M_{\lambda^0} \right]
      \nonumber \\  &
      =  \frac{B_0}{K_1} \,  {\rm Tr} \left[ v \, {\rm diag}(\alpha') \, w'  M_{\lambda^0}
                w \, {\rm diag}(\alpha') \, v' \right]
      \nonumber \\  &
      \geq  \frac{B_0}{K_1} \, \left\| v \, {\rm diag}(\alpha') \, w'  M_{\lambda^0}
                w \, {\rm diag}(\alpha') \, v' \right\|
      \nonumber \\  &
      \geq  \frac{B_0^2}{K_1} \, \left\| v \, {\rm diag}(\alpha') \, w' w \, {\rm diag}(\alpha') \, v' \right\|
      \nonumber \\  &
      \geq  \frac{B_0^2}{K_1^2} \, {\rm Tr}\left[ v \, {\rm diag}(\alpha') \, w' w \, {\rm diag}(\alpha') \, v' \right]
      \nonumber \\  &
      = \frac{B_0^2}{K_1^2} {\rm Tr}\left[ X_{{\rm low},\alpha} X'_{{\rm low},\alpha} \right] \; .
   \end{align*}
   Thus we have
   $ \left\|M_{\lambda^0} \, X_{{\rm low},\alpha} \, M_{f^0} \, X_{{\rm low},\alpha}' \, M_{\lambda^0} \right\|
      /(NT) \geq (B_0/K_1)^2 \, \alpha' \, W^{\rm low}_{NT} \, \alpha$ ,
   where the $K_1 \times K_1$ matrix $W^{\rm low}_{NT}$ is defined by
     $W^{\rm low}_{NT,l_1 l_2} = (NT)^{-1} {\rm Tr}\left( X_{l_1} X'_{l_2} \right)$, i.e., it is a submatrix of $W_{NT}$.
   Because $W_{NT}$ and thus $W^{\rm low}_{NT}$ converges to a positive definite matrix the lemma is proven
   by the inequality above.
\end{proof}

Using the above lemmas we can now prove the lower bound on $\widetilde S^{(2)}_{NT}(\beta,f)$ that was used
in the consistency proof. Remember
\begin{align*}
  \widetilde S^{(2)}_{NT}(\beta,f)
    &=   \frac{1}{NT} \; {\rm Tr}\left[ \left(
       \lambda^0 \, f^{0\prime} + \sum_{k=1}^{K} (\beta^0_k-\beta_k) X_{k} \right)
       \, M_f \, \left( \lambda^0 \, f^{0\prime} + \sum_{k=1}^{K}
       (\beta^0_k-\beta_k) X_{k} \right)^{\prime}\, P_{(\lambda^0,w)} \right] \; .
\end{align*}
We want to show under the assumptions of theorem \ref{th:consistency} there exist
finite positive constants $a_0$, $a_1$, $a_2$, $a_3$ and $a_4$ such that
\begin{align*}
  \widetilde S^{(2)}_{NT}(\beta,f)
    &\geq  \, \frac{a_0  \left\| \beta^{\rm low} - \beta^{0,{\rm low}} \right\|^2 }
                { \left\| \beta^{\rm low} - \beta^{0,{\rm low}} \right\|^2
                  + a_1 \left\| \beta^{\rm low} - \beta^{0,{\rm low}} \right\|
                  + a_2 }
        \nonumber \\ & \qquad \qquad
      -  a_3 \left\| \beta^{\rm high} - \beta^{0, {\rm high}} \right\|
  - a_4 \left\| \beta^{\rm high} - \beta^{0, {\rm high}} \right\| \, \left\| \beta^{\rm low} - \beta^{0,{\rm low}} \right\|
         \; , \qquad \text{wpa1.}
\end{align*}

\begin{proof}[\bf Proof of the lower bound on $\widetilde S^{(2)}_{NT}(\beta,f)$.]
Applying Lemma~\ref{lemma:Optimization} and part (xi) of Lemma~\ref{lemma:inequalities} we find
\begin{align*}
  \widetilde S^{(2)}_{NT}(\beta,f)
     &\geq \frac{1}{NT} \; {\mu}_{R+1} \left[
       \left( \lambda^0 \, f^{0\prime} + \sum_{k=1}^{K} (\beta^0_k-\beta_k) X_{k} \right)^{\prime}
        \, P_{(\lambda^0,w)} \,
        \left( \lambda^0 \, f^{0\prime} + \sum_{k=1}^{K} (\beta^0_k-\beta_k) X_{k} \right) \right]
     \nonumber \\
     &= \frac{1}{NT} \; {\mu}_{R+1} \Bigg[
       \left( \lambda^0 \, f^{0\prime} + \sum_{l=1}^{K_1} (\beta^0_l-\beta_l) w_l \, v_l' \right)^{\prime}
       \left( \lambda^0 \, f^{0\prime} + \sum_{l=1}^{K_1} (\beta^0_l-\beta_l) w_l \, v_l' \right)
     \nonumber \\
       &\qquad\qquad\qquad\qquad+ \left( \lambda^0 \, f^{0\prime} + \sum_{l=1}^{K_1} (\beta^0_l-\beta_l) w_l \, v_l' \right)^{\prime}
       P_{(\lambda^0,w)} \sum_{m=K_1}^{K} (\beta^0_m-\beta_m) X_m
     \nonumber \\
       &\qquad\qquad\qquad\qquad+ \sum_{m=K_1}^{K} (\beta^0_m-\beta_m) X_m' P_{(\lambda^0,w)}
         \left( \lambda^0 \, f^{0\prime} + \sum_{l=1}^{K_1} (\beta^0_l-\beta_l) w_l \, v_l' \right)
     \nonumber \\
       &\qquad\qquad\qquad\qquad+ \sum_{m=K_1}^{K} (\beta^0_m-\beta_m) X_m' P_{(\lambda^0,w)}
         \sum_{m=K_1}^{K} (\beta^0_m-\beta_m) X_m
     \Bigg]
     \nonumber \\
     &\geq \frac{1}{NT} \; {\mu}_{R+1} \Bigg[
       \left( \lambda^0 \, f^{0\prime} + \sum_{l=1}^{K_1} (\beta^0_l-\beta_l) w_l \, v_l' \right)^{\prime}
       \left( \lambda^0 \, f^{0\prime} + \sum_{l=1}^{K_1} (\beta^0_l-\beta_l) w_l \, v_l' \right)
     \nonumber \\
       &\qquad\qquad\qquad\qquad+ \left( \lambda^0 \, f^{0\prime} + \sum_{l=1}^{K_1} (\beta^0_l-\beta_l) w_l \, v_l' \right)^{\prime}
       P_{(\lambda^0,w)} \sum_{m=K_1}^{K} (\beta^0_m-\beta_m) X_m
     \nonumber \\
       &\qquad\qquad\qquad\qquad+ \sum_{m=K_1}^{K} (\beta^0_m-\beta_m) X_m' P_{(\lambda^0,w)}
         \left( \lambda^0 \, f^{0\prime} + \sum_{l=1}^{K_1} (\beta^0_l-\beta_l) w_l \, v_l' \right)
     \Bigg]
     \nonumber \\
     &\geq \frac{1}{NT} \; {\mu}_{R+1} \left[
       \left( \lambda^0 \, f^{0\prime} + \sum_{l=1}^{K_1} (\beta^0_l-\beta_l) w_l \, v_l' \right)^{\prime}
       \left( \lambda^0 \, f^{0\prime} + \sum_{l=1}^{K_1} (\beta^0_l-\beta_l) w_l \, v_l' \right) \right]
     \nonumber \\
          &\quad
          - a_3 \left\| \beta^{\rm high} - \beta^{0, {\rm high}} \right\|
          - a_4 \left\| \beta^{\rm high} - \beta^{0, {\rm high}} \right\|
                \left\| \beta^{\rm low} - \beta^{0,{\rm low}} \right\| \; , \qquad \text{wpa1,}
\end{align*}
where $a_3>0$ and $a_4>0$ are appropriate constants.
For the last step we used part (xii) of Lemma~\ref{lemma:inequalities} and the fact that
\begin{align*}
   & \frac 1 {NT} \left\| \sum_{m=K_1}^{K} (\beta^0_m-\beta_m) X_m' P_{(\lambda^0,w)}
         \left( \lambda^0 \, f^{0\prime} + \sum_{l=1}^{K_1} (\beta^0_l-\beta_l) w_l \, v_l' \right) \right\|
   \nonumber \\
    & \qquad \leq K  \, \left\| \beta^{\rm high} - \beta^{0, {\rm high}} \right\|
          \max_m \left\| \frac {X_m}{\sqrt{NT}}  \right\|
         \left( \left\| \frac {\lambda^0 \, f^{0\prime}}{\sqrt{NT}}  \right\|
                + K \, \left\| \beta^{\rm low} - \beta^{0,{\rm low}} \right\| \,
                      \max_l \left\| \frac {w_l v_l'}{\sqrt{NT}} \right\| \right) \; .
\end{align*}
Our assumptions guarantee the operator norms of $\lambda^0 \, f^{0\prime}/{\sqrt{NT}}$
and ${X_m}/{\sqrt{NT}}$ are bounded from above as $N,T \rightarrow \infty$, which results in
finite constants $a_3$ and $a_4$.

We write the above result as
$\widetilde S^{(2)}_{NT}(\beta,f)
 \geq {\mu}_{R+1}(A'A)/(NT) + \text{terms containing $\beta^{\rm high}$}$,
where we defined
$A=\lambda^0 \, f^{0\prime} + \sum_{l=1}^{K_1} (\beta^0_l-\beta_l) \, w_l \, v_l'$.
We also write $A = A_1 + A_2 + A_3$,
with $A_1 = M_w \, A \, P_{f^0} =  M_{w} \, \lambda^0 \, f^{0\prime}$,
$A_2 = P_w \, A \, M_{f^0} = \sum_{l=1}^{K_1} (\beta^0_l-\beta_l) \, w_l \, v_l' \, M_{f^0}$,
$A_3 = P_w \, A \, P_{f^0} = P_w \, \lambda^0 \, f^{0\prime}
       + \sum_{l=1}^{K_1} (\beta^0_l-\beta_l) \, w_l \, v_l' \, P_f$.
We then find $A'A=A_1' A_1  + (A'_2 + A_3') (A_2 + A_3)$ and
\begin{align*}
   A' A \, &\geq \, A' A  - (a^{1/2} A'_3 + a^{-1/2} A_2') (a^{1/2} A_3 + a^{-1/2} A_2)
      \nonumber \\
           &=  \left[ A'_1 A_1 - (a-1) \, A'_3 A_3 \right] \, + \, (1-a^{-1}) A'_2 A_2 \; ,
\end{align*}
where $\geq$ for matrices refers to the difference being positive definite, and $a$ is a
positive number. We choose
$a = 1 + {\mu}_{R}(A'_1 A_1) / (2 \, \|A_3\|^2)$.
The reason for this choice becomes clear below.

Note $\left[A'_1 A_1 - (a-1) \, A'_3 A_3 \right]$ has at most rank $R$
(asymptotically it has exactly rank $R$).
The non-zero eigenvalues of $A'A$ are therefore given by the (at most) $R$ non-zero eigenvalues of
$\left[ A'_1 A_1 - (a-1) \, A'_3 A_3 \right]$
and the non-zero eigenvalues of $(1-a^{-1}) A'_2 A_2$,
the largest one of the latter being given given by the operator norm
$(1-a^{-1}) \|A_2 \|^2$. We therefore find
\begin{align*}
  \frac{1}{NT} \; {\mu}_{R+1} \left( A'A \right)
        &\geq \frac{1}{NT} \; {\mu}_{R+1}
           \left[ \left(  A'_1 A_1 - (a-1) \, A'_3 A_3 \right) \, + \,  (1-a^{-1}) A'_2 A_2 \right]
       \nonumber \\
        &\geq \frac{1}{NT} \,
          \min\left\{ (1-a^{-1}) \|A_2\|^2 \; , \; \; {\mu}_{R}\left[A'_1 A_1 - (a-1) \, A'_3 A_3 \right]  \right\} \; .
\end{align*}
Using Lemma~\ref{lemma:inequalities}(xii) and our particular choice of $a$ we find
\begin{align*}
   {\mu}_{R} \, \left[ A'_1 A_1 - (a-1) \, A'_3 A_3 \right]
      &\geq \, {\mu}_{R}(A'_1 A_1) -  \left\| (a-1) A'_3 A_3 \right\|
   \nonumber \\
      &= \, \frac{1}{2} \, {\mu}_{R}(A'_1 A_1) \; .
\end{align*}
Therefore
\begin{align*}
   \frac 1 {NT} \, {\mu}_{R+1}(A'A)
        &\geq \frac{1}{2 \,NT} \, {\mu}_{R}(A'_1 A_1) \,
          \min\left\{ 1 \; , \; \; \frac{2 \, \|A_2\|^2}
                           {2 \, \|A_3\|^2 + {\mu}_{R}(A'_1 A_1)} \right\}
   \nonumber \\
        &\geq \frac{1}{NT} \, \frac{\|A_2\|^2 \, {\mu}_{R}(A'_1 A_1)}
                           {2 \, \|A\|^2 + {\mu}_{R}(A'_1 A_1)} \; ,
\end{align*}
where we used $\|A\|\geq\|A_3\|$ and $\|A\|\geq\|A_2\|$.

Our assumptions guarantee there exist positive constants $c_0$, $c_1$, $c_2$, and $c_3$ such that
\begin{align*}
   \frac {\|A\|} {\sqrt{NT}}
        &\leq  \frac {\|\lambda^0 \, f^{0\prime}\|} {\sqrt{NT}}
              + \sum_{l=1}^{K_1} |\beta^0_l-\beta_l| \frac {\| w_l \, v_l' \|} {\sqrt{NT}}
        \leq  c_0 + c_1 \left\| \beta^{\rm low} - \beta^{0,{\rm low}} \right\| \, , \quad \text{wpa1} \; ,
  \nonumber \\
   \frac {{\mu}_{R}(A'_1 A_1)} {NT}
        &= \frac{{\mu}_{R}\left( f^0 \, \lambda^{0\prime} \, M_{w} \, \lambda^0 \, f^{0\prime} \right)}
                {NT}
           \geq c_2 \, , \quad \text{wpa1} \; ,
  \nonumber \\
   \frac{\|A_2\|^2}{NT} &= {\mu}_{1}
                     \left[ \sum_{l_1=1}^{K_1} (\beta^0_{l_1}-\beta_{l_1}) \, w_{l_1} \, v_{l_1}' \, M_{f^0} \,
                            \sum_{l_2=1}^{K_1} (\beta^0_{l_2}-\beta_{l_2}) \, v_{l_2} \, w_{l_2}' \right]
                 \nonumber \\
                      &\geq c_3 \left\| \beta^{\rm low} - \beta^{0,{\rm low}} \right\|^2 \, , \quad \text{wpa1} \; ,
\end{align*}
were for the last inequality we used Lemma~\ref{lemma:lowrankprop}.

We thus have
\begin{align*}
  \frac{1}{NT} \; {\mu}_{R+1} \left( A'A \right)
        &\geq \frac{c_3 \left\| \beta^{\rm low} - \beta^{0,{\rm low}} \right\|^2}
                   {1 + \frac{2}  {c_2} \left(c_0 + c_1 \left\| \beta^{\rm low} - \beta^{0,{\rm low}} \right\|\right)^2
                            } \, , \quad \text{wpa1} \; .
\end{align*}
Defining $a_0=\frac{c_2 c_3}{2 c_1^2}$, $a_1=\frac{2 c_0}{c_1}$ and $a_2=\frac{c_2}{2 c_1^2}$ we thus obtain
\begin{align*}
  \frac{1}{NT} \; {\mu}_{R+1} \left( A'A \right)
        &\geq \frac{a_0  \left\| \beta^{\rm low} - \beta^{0,{\rm low}} \right\|^2 }
                { \left\| \beta^{\rm low} - \beta^{0,{\rm low}} \right\|^2
                  + a_1 \left\| \beta^{\rm low} - \beta^{0,{\rm low}} \right\|
                  + a_2 } \, , \quad \text{wpa1} \; ,
\end{align*}
i.e., we have shown the desired bound on $\widetilde S^{(2)}_{NT}(\beta,f)$.
\end{proof}

\section{Regarding the Proof of Corollary \ref{cor:limit}}

As discussed in the main text, the proof of Corollary \ref{cor:limit}
is provided in Moon and Weidner~\cite*{MoonWeidner2015}.
All that is left to show here
is the matrix $W_{NT}=W_{NT}(\lambda^0,\, f^0,\, X_{k})$ does not become singular as $N,T \rightarrow \infty$
under our assumptions.

\begin{proof}[\bf Proof]
  Remember
  \begin{align*}
      W_{NT} &= \frac 1 {NT} {\rm Tr}(M_{f^0}
      \, X^{\prime}_{k_1} \, M_{\lambda^0} \, X_{k_2}) \; .
  \end{align*}
  The smallest eigenvalue of the symmetric matrix $W(\lambda^0,\, f^0,\, X_{k})$ is given by
  \begin{align*}
     {\mu}_K \left( W_{NT} \right)
          &= \min_{\{a \in \mathbb{R}^K, \; a \neq 0\}}
               \frac{a' \,  W_{NT} \, a}  {\|a\|^2}
       \nonumber \\
          &= \min_{\{a \in \mathbb{R}^K, \; a \neq 0\}}
            \frac 1 {NT \, \|a\|^2}
            {\rm Tr}\left[ M_{f^0} \, \left( \sum_{k_1=1}^K \,  a_{k_1} \, X^{\prime}_{k_1} \right)
               \, M_{\lambda^0} \, \left(\sum_{k_2=1}^K \, a_{k_2} \, X_{k_2} \right) \right]
       \nonumber \\
          &= \min_{\begin{minipage}{2.8cm}\begin{center}\scriptsize
                   $\{\alpha \in \mathbb{R}^{K_1}, \; \varphi \in \mathbb{R}^{K_2}$\\
                   $\alpha \neq 0, \; \varphi\neq 0\}$\end{center}\end{minipage}}
            \frac {
            {\rm Tr}\left[ M_{f^0} \, \left( X'_{{\rm low},\varphi} + X'_{{\rm high},\alpha} \right)
               \, M_{\lambda^0} \,  \left( X_{{\rm low},\varphi} + X_{{\rm high},\alpha} \right) \right] }
               {NT \, \left( \|\alpha\|^2 + \|\varphi\|^2 \right)} \; ,
  \end{align*}
  where we decomposed $a=(\varphi',\alpha')'$, with $\varphi$ and $\alpha$ being vectors of length $K_1$ and $K_2$,
  respectively, and we defined linear combinations
  of high- and low-rank regressors\footnote{As in assumption \ref{ass:A4} the components of $\alpha$
  are denoted $\alpha_{K_1+1},\ldots,\alpha_{K}$ to simplify notation.}
  \begin{align*}
     X_{{\rm low},\varphi} &=  \sum_{l=1}^{K_1} \, \varphi_{l} \, X_{l}  \; , &
     X_{{\rm high},\alpha} &=  \sum_{m=K_1+1}^{K} \, \alpha_{m} \, X_{m}  \; .
  \end{align*}
  We have $M_{\lambda^0} =  M_{(\lambda^0,w)} + P_{(M_{\lambda^0} w)}$, where $w$ is the
  $N \times K_1$ matrix defined in assumption \ref{ass:A4}, i.e., $(\lambda^0,w)$
  is an $N \times (R+K_1)$ matrix, whereas $M_{\lambda^0} w$ is also an $N \times K_1$ matrix.
  Using this we obtain
  \begin{align}
      & {\mu}_K \left( W_{NT} \right)
      \nonumber \\
       & \quad = \min_{\begin{minipage}{2.8cm}\begin{center}\scriptsize
                   $\{\varphi \in \mathbb{R}^{K_1}, \; \alpha \in \mathbb{R}^{K_2}$\\
                   $\varphi \neq 0, \; \alpha\neq 0\}$\end{center}\end{minipage}}
            \frac 1 {NT \, \left( \|\varphi\|^2 + \|\alpha\|^2 \right)}
            \bigg\{
            {\rm Tr}\left[ M_{f^0} \, \left( X'_{{\rm low},\varphi} + X'_{{\rm high},\alpha} \right)
               \, M_{(\lambda^0,w)} \,  \left( X_{{\rm low},\varphi} + X_{{\rm high},\alpha} \right) \right]
         \nonumber \\ & \qquad \qquad \qquad \qquad \qquad \qquad \qquad\qquad
           + {\rm Tr}\left[ M_{f^0} \, \left( X'_{{\rm low},\varphi} + X'_{{\rm high},\alpha} \right)
       \, P_{(M_{\lambda^0} w)} \,  \left( X_{{\rm low},\varphi} + X_{{\rm high},\alpha} \right) \right] \bigg\}
      \nonumber \\
       & \quad = \min_{\begin{minipage}{2.8cm}\begin{center}\scriptsize
                   $\{\varphi \in \mathbb{R}^{K_1}, \; \alpha \in \mathbb{R}^{K_2}$\\
                   $\varphi \neq 0, \; \alpha\neq 0\}$\end{center}\end{minipage}}
            \frac 1 {NT \, \left( \|\varphi\|^2 + \|\alpha\|^2 \right)}
            \bigg\{
            {\rm Tr}\left[ M_{f^0} \, X'_{{\rm high},\alpha} \, M_{(\lambda^0,w)} \,  X_{{\rm high},\alpha} \right]
         \nonumber \\ & \qquad \qquad \qquad \qquad \qquad \qquad \qquad \qquad
           + {\rm Tr}\left[ M_{f^0} \, \left( X'_{{\rm low},\varphi} + X'_{{\rm high},\alpha} \right)
       \, P_{(M_{\lambda^0} w)} \,  \left( X_{{\rm low},\varphi} + X_{{\rm high},\alpha} \right) \right] \bigg\} \, .
       \label{eq:boundEK1}
  \end{align}
  We note there exists finite positive constants $c_1$, $c_2$, and $c_3$ such that
  \begin{align}
     \frac 1 {NT} {\rm Tr}\left[ M_{f^0} \, X'_{{\rm high},\alpha} \, M_{(\lambda^0,w)} \,  X_{{\rm high},\alpha} \right]
                &\geq  \, c_1 \| \alpha \|^2 \; , \quad \text{wpa1,}
   \nonumber \\
     \frac 1 {NT} {\rm Tr}\left[ M_{f^0} \, \left( X'_{{\rm low},\varphi} + X'_{{\rm high},\alpha} \right)
    \, P_{(M_{\lambda^0} w)} \,  \left( X_{{\rm low},\varphi} + X_{{\rm high},\alpha} \right) \right] &\geq 0 \; ,
   \nonumber \\
     \frac 1 {NT} {\rm Tr}\left[ M_{f^0} \,  X'_{{\rm low},\varphi}
    \, P_{(M_{\lambda^0} w)} \,  X_{{\rm low},\varphi}  \right] &\geq \, c_2 \,  \| \varphi \|^2 \; , \quad \text{wpa1,}
   \nonumber \\
     \frac 1 {NT} {\rm Tr}\left[ M_{f^0} \, X'_{{\rm low},\varphi}
    \, P_{(M_{\lambda^0} w)} \,   X_{{\rm high},\alpha}  \right] &\geq - \frac {c_3} 2 \, \| \varphi \| \| \alpha \| \; ,
        \quad \text{wpa1,}
   \nonumber \\
     \frac 1 {NT} {\rm Tr}\left[ M_{f^0} \,  X'_{{\rm high},\alpha}
    \, P_{(M_{\lambda^0} w)} \,   X_{{\rm high},\alpha} \right] &\geq 0 \; ,
    \label{inequ_highlow}
  \end{align}
  and we want to justify these inequalities now.
   The second and the last equation in \eqref{inequ_highlow} are true because,
  e.g., ${\rm Tr}\left[ M_{f^0} \,  X'_{{\rm high},\alpha} \, P_{(M_{\lambda^0} w)} \,   X_{{\rm high},\alpha} \right]
={\rm Tr}\left[ M_{f^0} \,  X'_{{\rm high},\alpha} \, P_{(M_{\lambda^0} w)} \,   X_{{\rm high},\alpha} \, M_{f^0} \right]$,
  and the trace of a symmetric positive semi-definite matrix is non-negative.
  The first inequality in \eqref{inequ_highlow} is true because ${\rm rank}(f^0)+{\rm rank}(\lambda^0,w)=2R+K_1$
  and using Lemma~\ref{lemma:Optimization} and assumption \ref{ass:A4} we have
  \begin{align*}
     \frac 1 {NT\|\alpha\|^2}
     {\rm Tr}\left[ M_{f^0} \, X'_{{\rm high},\alpha} \, M_{(\lambda^0,w)} \,  X_{{\rm high},\alpha} \right]
         \geq
     \frac 1 {NT\|\alpha\|^2}
     {\mu}_{2R+K_1+1}\left[  X_{{\rm high},\alpha} \, X'_{{\rm high},\alpha} \right]
         &> b \; , \quad \text{wpa1},
  \end{align*}
  i.e., we can set $c_1=b$.
  The third inequality in \eqref{inequ_highlow} is true because according Lemma~\ref{lemma:inequalities}(v) we have
  \begin{align*}
     \frac 1 {NT} {\rm Tr}\left[ M_{f^0} \, X'_{{\rm low},\varphi}
    \, P_{(M_{\lambda^0} w)} \,   X_{{\rm high},\alpha}  \right]
       &\geq - \, \frac {K_1} {NT} \, \left\| X_{{\rm low},\varphi} \right\| \left\| X_{{\rm high},\alpha} \right\|
   \nonumber \\
       &\geq - \, \frac {K_1} {NT} \, \left\| X_{{\rm low},\varphi} \right\|_F \left\| X_{{\rm high},\alpha} \right\|_F
   \nonumber \\
       &\geq - \, K_1 \, K_1 \, K_2 \, \|\varphi\| \, \|\alpha\| \,
              \, \max_{k_1=1\ldots K_1} \left\| \frac{X_{k_1}} {\sqrt{NT}} \right\|_F
              \, \max_{k_2=K_1+1\ldots K} \left\| \frac{X_{k_2}} {\sqrt{NT}} \right\|_F
   \nonumber \\
       &\geq - \frac {c_3} 2 \, \|\varphi\| \, \|\alpha\| \; ,
  \end{align*}
  where we used that assumption \ref{ass:A4} implies $\left\| X_{k} / \sqrt{NT} \right\|_F < C$
  holds wpa1 for some constant $C$ as, and we set $c_3 = K_1 \, K_1 \, K_2 \, C^2$.
  Finally, we have to argue that the third inequality in \eqref{inequ_highlow} holds.
  Note $X'_{{\rm low},\varphi} \, P_{(M_{\lambda^0} w)} \,  X_{{\rm low},\varphi}
  = X'_{{\rm low},\varphi} \, M_{\lambda^0} \,  X_{{\rm low},\varphi}$, i.e., we need to show
  \begin{align*}
      \frac 1 {NT} {\rm Tr}\left[ M_{f^0} \,  X'_{{\rm low},\varphi}
      \, M_{\lambda^0} \,  X_{{\rm low},\varphi}  \right] &\geq \, c_2 \,  \| \varphi \|^2 \; .
  \end{align*}
  Using part (vi) of Lemma~\ref{lemma:inequalities} we find
  \begin{align*}
     \frac 1 {NT} {\rm Tr}\left[ M_{f^0} \,  X'_{{\rm low},\varphi}
      \, M_{\lambda^0} \,  X_{{\rm low},\varphi}  \right]
   &= \frac 1 {NT}
   {\rm Tr}\left[ M_{\lambda^0} \,  X_{{\rm low},\varphi} \, M_{f^0} \,  X'_{{\rm low},\varphi} \, M_{\lambda^0} \right]
   \nonumber \\
   &\geq \frac 1 {NT}
         \left\| M_{\lambda^0} \,  X_{{\rm low},\varphi} \, M_{f^0} \,  X'_{{\rm low},\varphi} \, M_{\lambda^0} \right\| \; ,
  \end{align*}
  and according to Lemma~\ref{lemma:lowrankprop} this expression is bounded by some positive constant times
  $\| \varphi \|^2$ (in the lemma we have $\| \varphi \|=1$, but all expressions are homogeneous in $\|\varphi\|$).

  Using the inequalities \eqref{inequ_highlow} in equation \eqref{eq:boundEK1} we obtain
  \begin{align*}
     {\mu}_K \left( W_{NT} \right)
                  &\geq \min_{\begin{minipage}{2.8cm}\begin{center}\scriptsize
                   $\{\varphi \in \mathbb{R}^{K_1}, \; \alpha \in \mathbb{R}^{K_2}$\\
                   $\varphi \neq 0, \; \alpha\neq 0\}$\end{center}\end{minipage}}
            \frac 1 {\|\varphi\|^2 + \|\alpha\|^2}
               \left\{ c_1 \| \alpha \|^2 + \max\left[ 0, \, c_2 \| \varphi \|^2 - c_3 \|\varphi\| \|\alpha\| \right] \right\}
      \nonumber \\
      & \geq \min\left( \frac {c_2} 2 , \,  \frac{c_1 c_2^2} {c_2^2+c_3^2} \right)   \, , \quad \text{wpa1}.
  \end{align*}
  Thus, the smallest eigenvalue of $W_{NT}$ is bounded from below by a positive constant as $N,T \rightarrow \infty$,
  i.e., $W_{NT}$ is non-degenerate and invertible.
\end{proof}

\section{Proof of Examples for Assumption~\ref{ass:A5}}

\begin{proof}[\bf Proof of Example 1.]

We want to show the conditions of Assumption~\ref{ass:A5} are satisfied.
Conditions (i)-(iii) are satisfied by the assumptions of the example.

For condition (iv), notice ${\rm Cov} \left( X_{it}, X_{is} | \mathcal{C} \right) = \mathbb{E} \left( U_{it} U_{is} \right)$. Because $|\beta^0| < 1$ and $\sup_{it} \mathbb{E}(e_{it}^2) < \infty$, it follows

\begin{eqnarray*}
    \frac{1}{NT} \sum_{i=1}^{N} \sum_{t,s=1}^T \left| {\rm Cov} \left( X_{it}, X_{is} | \mathcal{C} \right) \right| &=&
     \frac{1}{NT} \sum_{i=1}^{N} \sum_{t,s=1}^T \left| \mathbb{E} \left( U_{it} U_{is} \right) \right| \\
    &=&  \frac{1}{NT} \sum_{i=1}^{N} \sum_{t,s=1}^T \sum_{p,q = 0}^{\infty}  \left| (\beta^0)^{p+q} \mathbb{E} \left( e_{it-p} e_{is-q} \right) \right|
    < \infty .
\end{eqnarray*}

For condition (v), notice by the independence between the sigma field $%
\mathcal{C}$ and the error terms $\left\{ e_{it}\right\} $ that we have for
some finite constant $M,$
\begin{eqnarray*}
    && \frac{1}{NT^{2}} \sum_{i=1}^{N} \sum_{t,s,u,v=1}^{T} \left\vert {\rm Cov}\left(
e_{it}\widetilde{X}_{is},e_{iu}\widetilde{X}_{iv}|\mathcal{C}\right) \right\vert  \\
   &=& \frac{1}{NT^{2}}\sum_{i=1}^{N}\sum_{t,s,u,v=1}^{T}\left\vert {\rm Cov}\left(
e_{it}U_{is},e_{iu}U_{iv}\right) \right\vert  \\
   &=& \frac{1}{NT^{2}}\sum_{i=1}^{N}\sum_{t,s,u,v=1}^{T}\sum_{p,q=0}^{\infty
}\left\vert \left( \beta ^{0}\right) ^{p+q}\mathbb{E}\left(
e_{it}e_{is-p}e_{iu}e_{iv-q}\right) -\left( \beta ^{0}\right) ^{p}\mathbb{E}%
\left( e_{it}e_{is-p}\right) \left( \beta ^{0}\right) ^{q}\mathbb{E}\left(
e_{iu}e_{iv-q}\right) \right\vert  \\
  &\leq & \frac{M}{T^{2}}\sum_{t,s,u,v=1}^{T}\sum_{p,q=0}^{\infty }\left\vert
\beta ^{0}\right\vert ^{p+q}\left[ \mathbb{I}\left\{ t=u\right\} \mathbb{I}%
\left\{ s-p=v-q\right\} +\mathbb{I}\left\{ t=v-q\right\} \mathbb{I}\left\{
s-p=u\right\} \right]  \\
  &=& \frac{M}{T^{2}}\sum_{t,u,s,v=1}^{T}\sum_{k=-\infty }^{s}\sum_{l=-\infty
}^{v}\left\vert \beta ^{0}\right\vert ^{s-k+v-l}\mathbb{I}\left\{
t=u\right\} \mathbb{I}\left\{ k=l\right\} +M\left( \frac{1}{T}\sum
_{\substack{ s,u=1 \\ s-u\geq 0}}^{T}\left\vert \beta ^{0}\right\vert
^{s-u}\right) \left( \frac{1}{T}\sum_{\substack{ v,t=1 \\ v-t\geq 0}}%
^{T}\left\vert \beta ^{0}\right\vert ^{v-t}\right)  \\
  &=& \frac{M}{T}\sum_{s,v=1}^{T}\sum_{k=-\infty }^{\min \left\{ s,v\right\}
}\left\vert \beta ^{0}\right\vert ^{s+v-2k}+M\left( \frac{1}{T}\sum
_{\substack{ s,u=1 \\ s-u\geq 0}}^{T}\left\vert \beta ^{0}\right\vert
^{s-u}\right) \left( \frac{1}{T}\sum_{\substack{ v,t=1 \\ v-t\geq 0}}%
^{T}\left\vert \beta ^{0}\right\vert ^{v-t}\right) .
\end{eqnarray*}%
Notice
\begin{eqnarray*}
&&\frac{1}{T}\sum_{s,v=1}^{T}\sum_{k=-\infty }^{\min \left\{ s,v\right\}
}\left\vert \beta ^{0}\right\vert ^{s+v-2k} \\
&=&\frac{2}{T}\sum_{s=2}^{T}\sum_{v=1}^{s}\sum_{k=-\infty }^{v}\left\vert
\beta ^{0}\right\vert ^{s-v+2(v-k)}+\frac{2}{T}\sum_{s=1}^{T}\sum_{k=-\infty
}^{s}\left\vert \beta ^{0}\right\vert ^{2(s-k)} \\
&=&\frac{2}{T}\sum_{s=2}^{T}\sum_{v=1}^{s}\left\vert \beta ^{0}\right\vert
^{s-v}\sum_{l=0}^{\infty }\left\vert \beta ^{0}\right\vert ^{2l}+\frac{2}{T}%
\sum_{s=1}^{T}\sum_{l=0}^{\infty }\left\vert \beta ^{0}\right\vert ^{2l} \\
&=&\frac{2}{1-\left\vert \beta ^{0}\right\vert ^{2}}\frac{1}{T}%
\sum_{s=2}^{T}\sum_{v=1}^{s}\left\vert \beta ^{0}\right\vert ^{s-v}+\frac{2}{%
1-\left\vert \beta ^{0}\right\vert ^{2}} \\
&=&\left( \frac{2}{1-\left\vert \beta ^{0}\right\vert ^{2}}\right)
\sum_{l=1}^{T-1}\left\vert \beta ^{0}\right\vert ^{l}\left( 1-\frac{l}{T}%
\right) +\frac{2}{1-\left\vert \beta ^{0}\right\vert ^{2}} \\
&=&O\left( 1\right) ,
\end{eqnarray*}%
and
\begin{equation*}
\frac{1}{T}\sum_{\substack{ s,u=1 \\ s-u\geq 0}}^{T}\left\vert \beta
^{0}\right\vert ^{s-u}=\frac{1}{T}\sum_{s=1}^{T}\sum_{u=1}^{s}\left\vert
\beta ^{0}\right\vert ^{s-u}=\sum_{l=0}^{T-1}\left\vert \beta
^{0}\right\vert ^{l}\left( 1-\frac{l}{T}\right) =O\left( 1\right) .
\end{equation*}%
Therefore, we have the desired result
\begin{equation*}
\frac{1}{NT^{2}}\sum_{i=1}^{N}\sum_{t,s,u,v=1}^{T}\left\vert {\rm Cov}\left( e_{it}%
\widetilde{X}_{is},e_{iu}\widetilde{X}_{iv}|\mathcal{C}\right) \right\vert
={\cal O}_{p}\left( 1\right) .
\end{equation*}
\end{proof}

\textsc{Preliminaries for Proof of Example 2}

\begin{itemize}
\item Although we observe $X_{it}$ for $1\leq t\leq T,$ here we treat $%
Z_{it}=\left( e_{it},X_{it}\right) $ as having an infinite past and future.
Define
\begin{equation*}
\mathcal{G}_{\tau }^{t}\left( i\right) =
{\cal C} \vee \sigma \left( \left\{ X_{is}:\tau
\leq s\leq t\right\}  \right) \text{ and }\mathcal{H}_{\tau
}^{t}\left( i\right) =
{\cal C} \vee  \sigma \left( \left\{ Z_{it}:\tau \leq s\leq t\right\}\right) .
\end{equation*}%
Then, by definition, we have $\mathcal{G}_{\tau }^{t}\left( i\right) ,%
\mathcal{H}_{\tau }^{t}\left( i\right) \subset \mathcal{F}_{\tau }^{t}\left(
i\right) $ for all $\tau ,t,i.$ By Assumption (iv) of Example 2, the time
series of $\left\{ X_{it}:-\infty <t<\infty \right\} $ and $\left\{
Z_{it}:-\infty <t<\infty \right\} $ are conditional $\alpha$-mixing
conditioning on $\mathcal{C}$ uniformly in $i.$

\item Mixing inequality: The following inequality is a conditional version
of the $\alpha$-mixing inequality of Hall and Heyde~\cite*{HallHeyde1980}, p.~278. Suppose
 $X_{it}$ is a $\mathcal{F}_{t}$-measurable  random
variable with $\mathbb{E}\left( \left\vert X_{it}\right\vert ^{\max \left\{
p,q\right\} }|\mathcal{C}\right) <\infty ,$ where $p,q>1$ with $1/p+1/q<1.$
Denote $\left\Vert X_{it}\right\Vert _{\mathcal{C},p}=\left( \mathbb{E}%
\left( \left\vert X_{it}\right\vert ^{p}|\mathcal{C}\right) \right)
^{1/p}. $ Then, for each $i,$ we have
\begin{equation}
\left\vert {\rm Cov}\left( X_{it},X_{it+m}|\mathcal{C}\right) \right\vert \leq
8\left\Vert X_{it}\right\Vert _{\mathcal{C},p}\left\Vert
X_{it+m}\right\Vert _{\mathcal{C},q}\alpha _{m}^{1-\frac{1}{p}-\frac{1}{q}%
}\left( i\right) .  \label{eq:mixing inequality}
\end{equation}
\end{itemize}

\begin{proof}[\bf Proof of Example 2.]
Again, we want to show the conditions of Assumption~\ref{ass:A5} are satisfied.
Conditions (i)-(iii) are satisfied by the assumptions of the example.

For condition (iv), we apply the mixing inequality $\left( \ref{eq:mixing
inequality}\right) $ with $p=q>4$. Then, we have%
\begin{eqnarray*}
&&\frac{1}{NT}\sum_{i=1}^{N}\sum_{t,s=1}^{T}\left\vert {\rm Cov}\left(
X_{it},X_{is}|\mathcal{C}\right) \right\vert  \\
&\leq &\frac{2}{NT}\sum_{i=1}^{N}\sum_{t=1}^{T}\sum_{m=0}^{T-t}\left\vert
{\rm Cov}\left( X_{it},X_{it+m}|\mathcal{C}\right) \right\vert =\frac{2}{NT}%
\sum_{i=1}^{N}\sum_{m=0}^{T-1}\sum_{t=1}^{T-m}\left\vert {\rm Cov}\left(
X_{it},X_{it+m}|\mathcal{C}\right) \right\vert  \\
&=&\frac{16}{NT}\sum_{i=1}^{N}\sum_{m=0}^{T-1}\sum_{t=1}^{T-m}\left\Vert
X_{it}\right\Vert _{\mathcal{C},p}\left\Vert X_{it+m}\right\Vert _{\mathcal{C%
},p}\alpha _{m}\left( i\right) ^{\frac{p-2}{P}} \\
&\leq &16\left( \sup_{i,t}\left\Vert X_{it}\right\Vert _{\mathcal{C}%
,p}^{2}\right) \sum_{m=0}^{\infty }\alpha _{m}^{\frac{p-2}{P}} \\
&\leq &{\cal O}_{p}\left( 1\right),
\end{eqnarray*}%
where the last line holds because $\sup_{i,t}\left\Vert X_{it}\right\Vert _{\mathcal{C}%
,p}^{2}={\cal O}_{p}\left( 1\right) $ for some $p>4$ as assumed in the example (2), and $\sum_{m=0}^{\infty
}\alpha _{m}^{\frac{p-2}{P}}= \sum_{m=0}^{\infty} m^{-\zeta\frac{p-2}{P}} = {\cal O}\left( 1\right)$ because of $\zeta > 3\frac{4p}{4p-1}$ and $p>4$.

For condition (v), we need to show
\begin{equation*}
\frac{1}{NT^{2}}\sum_{i=1}^{N}\sum_{t,s,u,v=1}^{T}\left\vert {\rm Cov}\left( e_{it}%
\widetilde{X}_{is},e_{iu}\widetilde{X}_{iv}|\mathcal{C}\right) \right\vert
={\cal O}_{p}\left( 1\right) .
\end{equation*}%
Notice
\begin{eqnarray*}
&&\frac{1}{NT^{2}}\sum_{i=1}^{N}\sum_{t,s,u,v=1}^{T}\left\vert {\rm Cov}\left(
e_{it}\widetilde{X}_{is},e_{iu}\widetilde{X}_{iv}|\mathcal{C}\right) \right\vert  \\
&=&\frac{1}{NT^{2}}\sum_{i=1}^{N}\sum_{t,s,u,v=1}^{T}\left\vert \mathbb{E}%
\left( e_{it}\widetilde{X}_{is}e_{iu}\widetilde{X}_{iv}|\mathcal{C}\right) -\mathbb{E%
}\left( e_{it}\widetilde{X}_{is}|\mathcal{C}\right) \mathbb{E}\left( e_{iu}%
\widetilde{X}_{iv}|\mathcal{C}\right) \right\vert  \\
&\leq &\frac{1}{NT^{2}}\sum_{i=1}^{N}\sum_{t,s,u,v=1}^{T}\left\vert \mathbb{E%
}\left( e_{it}\widetilde{X}_{is}e_{iu}\widetilde{X}_{iv}|\mathcal{C}\right)
\right\vert +\frac{1}{N}\sum_{i=1}^{N}\left( \frac{1}{T}\sum_{t,s=1}^{T}%
\mathbb{E}\left( e_{it}\widetilde{X}_{is}|\mathcal{C}\right) \right) ^{2} \\
&=&I+II,\text{ say.}
\end{eqnarray*}%
First, for term $I,$ there are a finite number of different orderings among
the indices $t,s,u,v.$ We consider the case $t\leq s\leq u\leq v$ and
establish the desired result. The other cases can be shown analogously. Note
\begin{eqnarray*}
&&\frac{1}{NT^{2}}\sum_{i=1}^{N}\sum_{t=1}^{T}\sum_{k=0}^{T-t}%
\sum_{l=0}^{T-k}\sum_{m=0}^{T-l}\left\vert \mathbb{E}\left( e_{it}\widetilde{X}%
_{it+k}e_{it+k+l}\widetilde{X}_{it+k+l+m}|\mathcal{C}\right) \right\vert  \\
&\leq &\frac{1}{N}\sum_{i=1}^{N}\frac{1}{T^{2}}\sum_{t=1}^{T}\sum_{\substack{
0\leq l,m\leq k \\ 0\leq k+l+m\leq T-t}}\left\vert \mathbb{E}\left(
e_{it}\left( \widetilde{X}_{it+k}e_{it+k+l}\widetilde{X}_{it+k+l+m}\right) |\mathcal{%
C}\right) \right\vert  \\
&&+\frac{1}{N}\sum_{i=1}^{N}\frac{1}{T^{2}}\sum_{t=1}^{T}\sum_{\substack{ %
0\leq k,m\leq l \\ 0\leq k+l+m\leq T-t}}\bigg\vert \mathbb{E}\left[ \left(
e_{it}\widetilde{X}_{it+k}\right) \left( e_{it+k+l}\widetilde{X}_{it+k+l+m}\right) |%
\mathcal{C}\right] 
\\
&& \qquad \qquad \qquad \qquad  \qquad \qquad  \qquad \qquad 
 -\mathbb{E}\left( e_{it}\widetilde{X}_{it+k}|\mathcal{C}%
\right) \mathbb{E}\left( e_{it+k+l}\widetilde{X}_{it+k+l+m}|\mathcal{C}\right)
\bigg\vert  \\
&&+\frac{1}{N}\sum_{i=1}^{N}\frac{1}{T^{2}}\sum_{t=1}^{T}\sum_{\substack{ %
0\leq k,m\leq l \\ 0\leq k+l+m\leq T-t}}\mathbb{E}\left( e_{it}\widetilde{X}%
_{it+k}|\mathcal{C}\right) \mathbb{E}\left( e_{it+k+l}\widetilde{X}_{it+k+l+m}|%
\mathcal{C}\right)  \\
&&+\frac{1}{N}\sum_{i=1}^{N}\frac{1}{T^{2}}\sum_{t=1}^{T}\sum_{\substack{ %
0\leq p,l\leq m \\ 0\leq k+l+m\leq T-t}}\left\vert \mathbb{E}\left[ \left(
e_{it}\widetilde{X}_{it+k}e_{it+k+l}\right) \widetilde{X}_{it+k+l+m}|\mathcal{C}%
\right] \right\vert  \\
&=&I_{1}+I_{2}+I_{3}+I_{4},\text{ say.}
\end{eqnarray*}%
By applying the mixing inequality $\left( \ref{eq:mixing inequality}\right) $
to $\left\vert \mathbb{E}\left( e_{it}\left( \widetilde{X}_{it+k}e_{it+k+l}%
\widetilde{X}_{it+k+l+m}\right) |\mathcal{C}\right) \right\vert $ with $e_{it}$
and $\widetilde{X}_{it+k}e_{it+k+l}\widetilde{X}_{it+k+l+m},$ we have%
\begin{eqnarray*}
&&\left\vert \mathbb{E}\left( e_{it}\left( \widetilde{X}_{it+k}e_{it+k+l}\widetilde{X%
}_{it+k+l+m}\right) |\mathcal{C}\right) \right\vert  \\
&\leq &8\left\Vert e_{it}\right\Vert _{\mathcal{C},p}\left\Vert \widetilde{X}%
_{it+k}e_{it+k+l}\widetilde{X}_{it+k+l+m}\right\Vert _{\mathcal{C},q}\alpha
_{k}^{1-\frac{1}{p}-\frac{1}{q}}\left( i\right)  \\
&\leq &8\left\Vert e_{it}\right\Vert _{\mathcal{C},p}\left\Vert \widetilde{X}%
_{it+k}\right\Vert _{\mathcal{C},3q}\left\Vert e_{it+k+l}\right\Vert _{%
\mathcal{C},3q}\left\Vert \widetilde{X}_{it+k+l+m}\right\Vert _{\mathcal{C}%
,3q}\alpha _{k}^{1-\frac{1}{p}-\frac{1}{q}}\left( i\right) ,
\end{eqnarray*}%
where the last inequality follows by the generalized Holder's inequality.
Choose $p=3q>4.$ Then,
\begin{eqnarray*}
I_{1} &\leq &\frac{8}{N}\sum_{i=1}^{N}\frac{1}{T^{2}}\sum_{t=1}^{T}\sum
_{\substack{ 0\leq l,m\leq k \\ 0\leq k+l+m\leq T-t}}\left\Vert
e_{it}\right\Vert _{\mathcal{C},p}\left\Vert \widetilde{X}_{it+k}\right\Vert _{%
\mathcal{C},p}\left\Vert e_{it+k+l}\right\Vert _{\mathcal{C},p}\left\Vert
\widetilde{X}_{it+k+l+m}\right\Vert _{\mathcal{C},p}\alpha _{k}^{1-\frac{1}{4p}%
}\left( i\right)  \\
&\leq &8\left( \sup_{i,t}\left\Vert e_{it}\right\Vert _{\mathcal{C}%
,p}^{2}\right) \left( \sup_{i,t}\left\Vert \widetilde{X}_{it+k}\right\Vert _{%
\mathcal{C},p}^{2}\right) \frac{1}{T^{2}}\sum_{t=1}^{T}\sum_{\substack{ %
0\leq l,m\leq k \\ 0\leq k+l+m\leq T-t}}\alpha _{k}^{1-\frac{1}{4p}} \\
&\leq &8\left( \sup_{i,t}\left\Vert e_{it}\right\Vert _{\mathcal{C}%
,p}^{2}\right) \left( \sup_{i,t}\left\Vert \widetilde{X}_{it+k}\right\Vert _{%
\mathcal{C},p}^{2}\right) \sum_{k=0}^{\infty }k^{2}\alpha _{k}^{1-\frac{1}{4p%
}} \\
&\leq &{\cal O}_{p}\left( 1\right),
\end{eqnarray*}%
where the last line holds because we assume in example (2) that 
$\left(\sup_{i,t}\left\Vert e_{it}\right\Vert _{\mathcal{C},p}^{2}\right) \left(
\sup_{i,t}\left\Vert \widetilde{X}_{it+k}\right\Vert _{\mathcal{C},p}^{2}\right)
={\cal O}_{p}\left( 1\right) $ for some $p>4,$, and $\sum_{m=0}^{\infty }m^{2}\alpha _{m}^{1-\frac{1%
}{4p}}= \sum_{m=0}^{\infty }m^{2-\zeta\frac{4p-1}{4p}}= O\left( 1\right)$ because of $\zeta > 3\frac{4p}{4p-1}$ and $p>4$.

By applying similar arguments, we can also show
\begin{equation*}
I_{2},I_{3},I_{4}={\cal O}_{p}\left( 1\right) .
\end{equation*}
\end{proof}

\section{Supplement to the Proof of Theorem \ref{th:limdis}}

\paragraph{\underline{Notation $\mathbb{E}_{\cal C}$ 
and ${\rm Var}_{\cal C}$ and ${\rm Cov}_{\cal C}$:}}
 In the remainder of this supplementary file we
write $\mathbb{E}_{\cal C}$, ${\rm Var}_{\cal C}$ and ${\rm Cov}_{\cal C}$
 for the expectation, variance and covariance operators conditional
on ${\cal C}$,~i.e., $\mathbb{E}_{\cal C}(A)=\mathbb{E}(A|{\cal C})$,
${\rm Var}_{\cal C}(A) = {\rm Var}(A|{\cal C})$
and ${\rm Cov}_{\cal C}(A,B) = {\rm Cov}(A,B |{\cal C})$.

\bigskip

What is left to show to complete the proof of Theorem \ref{th:limdis} is  
that Lemma~\ref{lemma:vanishing} and Lemma~\ref{lemma:denCLT} in the main text appendix hold.
Before showing this, we first present two further intermediate lemmas.

\begin{lemma}
   \label{lemma:normXweak}
   Under the assumptions of Theorem~\ref{th:limdis} we have for $k=1,\ldots ,K$,
   \begin{align*}
      \qquad && (a) && \| P_{\lambda^0} \widetilde X_k  \| &= o_p(\sqrt{NT})  \; ,
    \nonumber \\
     \qquad && (b) &&  \|  \widetilde X_k P_{f^0}\| &= o_p(\sqrt{NT}) \; ,
    \nonumber \\
      \qquad && (c) && \|P_{\lambda^0} e  X^{\prime}_k \| &= o_p(N^{3/2})          \, ,
    \nonumber \\
      \qquad && (d) &&  \|P_{\lambda^0} e P_{f^0} \| &= {\cal O}_p(1)  \; .
      && \qquad
   \end{align*}
\end{lemma}

\begin{proof}[\bf Proof of Lemma~\ref{lemma:normXweak}]
     \# Part (a): We have
       \begin{align*}
        \|P_{\lambda^0} \widetilde X_k \| &=
        \|\lambda^0 (\lambda^{0\prime}\lambda^0)^{-1} \lambda^{0\prime}  \widetilde X_k \|
        \nonumber \\
            &\leq
           \|\lambda^0 (\lambda^{0\prime}\lambda^0)^{-1}\| \| \lambda^{0\prime}  \widetilde X_k \|
        \nonumber \\
            &\leq
           \|\lambda^0 \| \| (\lambda^{0\prime}\lambda^0)^{-1}\| \| \lambda^{0\prime}  \widetilde X_k \|_F
           = {\cal O}_p(N^{-1/2})  \| \lambda^{0\prime}  \widetilde X_k \|_F \; ,
     \end{align*}
     where we used part (i) and (ii) of Lemma~\ref{lemma:inequalities}
     and Assumption~\ref{ass:A1}.
     We have
     \begin{align*}
        \mathbb{E}\left\{  \mathbb{E}_{\cal C}\left[  \| \lambda^{0\prime}  \widetilde X_k \|_F^2  \right] \right\}
          &=  \mathbb{E}\left\{  \sum_{r=1}^R \sum_{t=1}^T
          \mathbb{E}_{\cal C} \left[    \left(  \sum_{i=1}^N  \lambda^0_{ir}
          \widetilde X_{k,it} \right)^2
          \right] \right\}
        \\
             &=    \mathbb{E}\left\{  \sum_{r=1}^R \sum_{t=1}^T   \sum_{i=1}^N
               (\lambda^0_{ir})^2
          \mathbb{E}_{\cal C} \left(   \widetilde X_{k,it}^2  \right) \right\}
       \\
             &=     \sum_{r=1}^R \sum_{t=1}^T   \sum_{i=1}^N
               \mathbb{E}\left[ (\lambda^0_{ir})^2
           {\rm Var}_{\cal C} \left(  X_{k,it}   \right) \right]
      \\
           &=   {\cal O}_p( NT ) ,
     \end{align*}
     where we used $\widetilde X_{k,it}$ is mean zero and independent across $i$,
     conditional on ${\cal C}$, and our bounds on the moments
     of $\lambda^0_{ir}$ and $X_{k,it}$. We therefore have $ \| \lambda^{0\prime}  \widetilde X_k \|_F =
     {\cal O}_p(\sqrt{NT})$ and    the above inequality thus
     gives $\|P_{\lambda^0} \widetilde X_k \|  = {\cal O}_p(\sqrt{T} ) =  o_p(\sqrt{NT})$.

     \# The proof for part (b) is similar. As above we first obtain
     $ \|  \widetilde X_k P_{f^0}\| =   \|P_{f^0} \widetilde X_k'  \| \leq
       {\cal O}_p(T^{-1/2})  \| f^{0\prime}  \widetilde X_k' \|_F$.
    Next, we have
      \begin{align*}
          \mathbb{E}_{\cal C}\left[  \| f^{0\prime}  \widetilde X_k' \|_F^2 \right]
          &= \sum_{r=1}^R \sum_{i=1}^N
          \mathbb{E}_{\cal C} \left[    \left(  \sum_{t=1}^T  f^0_{tr}
          \widetilde X_{k,it} \right)^2  \right]
        \\
             &=
             \sum_{r=1}^R \sum_{i=1}^N   \sum_{t,s=1}^T f^0_{tr} f^0_{sr}
          \mathbb{E}_{\cal C} \left(
          \widetilde X_{k,it}    \widetilde X_{k,is}  \right)
        \\
            &\leq \left[ \sum_{r=1}^R
            \left(  \max_t | f^0_{tr} | \right)^2 \right]
             \sum_{i=1}^N   \sum_{t,s=1}^T
        \left|  {\rm Cov}_{\cal C} \left(
          X_{k,it} , X_{k,is} \right)  \right|
       \\
            &= {\cal O}_p(T^{2/(4+\epsilon)}) \, {\cal O}_p(NT) = o_p(N T^2) ,
     \end{align*}
     where we used that uniformly bounded $\mathbb{E} \| f^0_t \|^{4+\epsilon}$
     implies $\max_t | f^0_{tr} | = {\cal O}_p( T^{1/(4+\epsilon)} )$.
     We thus have $ \| f^{0\prime}  \widetilde X_k' \|_F^2 = o_p(T \sqrt{N})$
     and therefore $\|  \widetilde X_k P_{f^0}\|   = o_p(\sqrt{NT})$.

     \# Next, we show part (c).
      First, we  have
     \begin{align*}
        \mathbb{E} \left\{  \mathbb{E}_{\cal C} \left[ \left(  \|\lambda^{0\,\prime} e X'_k \|_F \right)^2  \right]
        \right\}
          &=   \mathbb{E} \left\{  \mathbb{E}_{\cal C}
      \left[ \sum_{r=1}^R \sum_{j=1}^N
           \left( \sum_{i=1}^N \sum_{t=1}^T \lambda^{0}_{ir} e_{it} X_{k,jt} \right)^2    \right] \right\}
        \nonumber \\
           &=   \mathbb{E} \left\{   \sum_{r=1}^R \sum_{i,j,l=1}^N
            \sum_{t,s=1}^T
           \lambda^{0}_{ir}    \lambda^{0}_{lr}
            \mathbb{E}_{\cal C}  \left( e_{it} e_{ls} X_{k,jt}   X_{k,js}      \right) \right\}
        \nonumber \\
          &=  \sum_{r=1}^R \sum_{i,j=1}^N    \sum_{t=1}^T
            \mathbb{E} \left[ (\lambda^{0}_{ir})^2
             \mathbb{E}_{\cal C} \left( e_{it}^2  X_{k,jt}^2     \right)
             \right]
         = {\cal O}(N^2 T) \; ,
     \end{align*}
     where we used that $ \mathbb{E}_{\cal C}  \left( e_{it} e_{ls} X_{k,jt}   X_{k,js}     \right)$
     is only non-zero if $i=l$ (because of cross-sectional independence conditional on ${\cal C}$)
     and $t=s$ (because regressors are pre-determined).
     We can thus conclude $\|\lambda^{0\,\prime} e X'_k \|_F = {\cal O}_p(N \sqrt{T})$.
    Using this we find
     \begin{align*}
        \|P_{\lambda^0} e X'_k \| &=
        \|\lambda^0 (\lambda^{0\prime}\lambda^0)^{-1} \lambda^{0\prime}  e X'_k \|
        \nonumber \\
            &\leq
           \|\lambda^0 (\lambda^{0\prime}\lambda^0)^{-1}\| \| \lambda^{0\prime}  e X'_k \|
        \nonumber \\
            &\leq
           \|\lambda^0 \|  \| (\lambda^{0\prime}\lambda^0)^{-1}\| \| \lambda^{0\prime}  e X'_k \|_F
           = {\cal O}_p(N^{-1/2})  {\cal O}_p(N \sqrt{T}) =  {\cal O}_p(\sqrt{NT}) \; .
     \end{align*}
     This is what we wanted to show.

      \# For part (d), we first find
      $\frac{1}{\sqrt{NT}}\left\Vert f^{0\prime}e\lambda^0 \right\Vert_{F}={\cal O}_{p}\left( 1\right)$, because
     \begin{eqnarray*}
        \mathbb{E} \left\{  \mathbb{E}_{\cal C} \left[ \left( \frac{\left\Vert f^{0\prime }e\lambda^0 \right\Vert _{F}}{\sqrt{NT}}%
            \right) ^{2}  \right] \right\}
            &=&\mathbb{E} \left\{  \frac{1}{NT}\mathbb{E}_{\cal C} \left[ \left(
        \sum_{i=1}^{N}\sum_{t=1}^{T}e_{it}f_{t}^{0\prime }\lambda^0_{i}\right) ^{2}  \right]  \right\}
           \nonumber \\
          &=& \mathbb{E} \left\{ \frac{1}{NT}\sum_{i=1}^{N}\sum_{j=1}^{N}\sum_{t=1}^{T}\sum_{s=1}^{T}\mathbb{E}_{\cal C}%
        \left( e_{it}e_{js}   \right) f_{t}^{0\prime }\lambda_{i}^0\lambda_{j}^{0\prime}f^0_{s} \right\}
            \nonumber \\
         &=&\frac{1}{NT}\sum_{i=1}^{N}\sum_{t=1}^{T}
         \mathbb{E}\left[ \mathbb{E}_{\cal C} \left( e_{it}^2  \right)
           f_{t}^{0\prime
       }\lambda^0_{i}\lambda_{i}^{0\prime }f^0_{t}  \right]
     \nonumber \\
        &=&{\cal O}\left( 1\right) ,
     \end{eqnarray*}
     where we used $e_{it}$ is independent across $i$ and over $t$,
     conditional on ${\cal C}$.
     Thus we obtain
     \begin{align*}
        \|P_{\lambda^0} e P_{f^0} \|
              &= \| \lambda^0 (\lambda^{0\prime}\lambda^0)^{-1} \lambda^{0\prime}
                     e f^0 (f^{0\prime}f^0)^{-1} \, f^{0\prime} \|
         \nonumber \\
              &\leq \| \lambda^0 \|  \left\| (\lambda^{0\prime}\lambda^0)^{-1} \right\|
                    \| \lambda^{0\prime} e f^0 \| \left\| (f^{0\prime}f^0)^{-1} \right\| \| f^{0\prime} \|
         \nonumber \\
              &\leq {\cal O}_p(N^{1/2})  {\cal O}_p(N^{-1})
                    \| \lambda^{0\prime} e f^0 \|_F {\cal O}_p(T^{-1}) {\cal O}_p(T^{1/2}) = {\cal O}_p(1) \;,
     \end{align*}
     where we used part (i) and (ii) of Lemma~\ref{lemma:inequalities}.
 \end{proof}

\begin{lemma}
  \label{lemma:eeterms}
Suppose $A$ and $B$ are  $T\times T$ and  $N\times N$ matrices that
are independent of $e$, conditional on ${\cal C}$, such that
$\mathbb{E}_{\cal C}\left( \left\Vert A\right\Vert _{F}^{2}   \right)={\cal O}_p\left(
NT\right) $ and $\mathbb{E}_{\cal C} \left( \left\Vert B\right\Vert _{F}^{2} \right)={\cal O}_p\left( NT\right)$,
and let Assumption~\ref{ass:A5} be satisfied.
Then there exists a finite non-random constant $c_0$ such that
\begin{align*}
  (a) &&  \mathbb{E}_{\cal C}\left( \left\{ \limfunc{Tr}\left[ \left( e^{\prime }e-\mathbb{E}_{\cal C} \left( e^{\prime
}e\right) \right) A\right]  \right\}^2  \right) &\leq c_0 \, N \,  \mathbb{E}_{\cal C} \left( \left\Vert
A\right\Vert _{F}^{2}   \right) \; ,
\nonumber \\
  (b) && \mathbb{E}_{\cal C}\left( \left\{ \limfunc{Tr}\left[ \left( ee^{\prime }-\mathbb{E}_{\cal C} \left( ee^{\prime
}\right) \right) B \right] \right\}^{2}  \right)  &\leq
c_0 \, T \,  \mathbb{E}_{\cal C} \left( \left\Vert B\right\Vert _{F}^{2}  \right) \; .
\end{align*}
\end{lemma}

\begin{proof}[\bf Proof]
\# Part (a): Denote $A_{ts}$ to be the $\left( t,s\right) ^{th}$ element of $A$.
We have
\begin{align*}
\limfunc{Tr}\left\{ \left( e^{\prime }e-\mathbb{E}_{\cal C}\left( e^{\prime }e\right)
\right) A\right\}
&= \sum_{t=1}^{T}\sum_{s=1}^{T}\left( e^{\prime }e-\mathbb{E}_{\cal C}\left( e^{\prime
}e\right) \right) _{ts}A_{ts} \nonumber \\
&= \sum_{t=1}^{T}\sum_{s=1}^{T}\left( \sum_{i=1}^{N}\left(
e_{it}e_{is}-\mathbb{E}_{\cal C}\left( e_{it}e_{is}\right) \right) \right) A_{ts}.
\end{align*}%
Therefore,
\begin{align*}
&\mathbb{E}_{\cal C}\left( \limfunc{Tr}\left\{ \left( e^{\prime }e-\mathbb{E}_{\cal C}\left( e^{\prime
}e\right) \right) A\right\} \right) ^{2} \nonumber \\
&\qquad  = \sum_{t=1}^{T}\sum_{s=1}^{T}\sum_{p=1}^{T}\sum_{q=1}^{T}\mathbb{E}_{\cal C}\left[ \left(
\sum_{i=1}^{N}\left( e_{it}e_{is}-\mathbb{E}_{\cal C}\left( e_{it}e_{is}\right) \right)
\right) \left( \sum_{j=1}^{N}\left( e_{jp}e_{jq}-\mathbb{E}_{\cal C}\left( e_{jp}e_{jq}\right)
\right) \right) \right]
\mathbbm{E}_{\cal C} \left( A_{ts}A_{pq}  \right).
\end{align*}%
Let $\Sigma_{it}=\mathbb{E}_{\cal C}(e_{it}^2)$. Then we find
\begin{align*}
&\mathbb{E}_{\cal C}\left\{ \left( \sum_{i=1}^{N}\left( e_{it}e_{is}-\mathbb{E}_{\cal C}\left(
e_{it}e_{is}\right) \right) \right) \left( \sum_{j=1}^{N}\left(
e_{jp}e_{jq}-\mathbb{E}_{\cal C}\left( e_{jp}e_{jq}\right) \right) \right) \right\}
  \nonumber \\ &\qquad\qquad\qquad\qquad
=\sum_{i=1}^{N}\sum_{j=1}^{N}\left\{ \mathbb{E}_{\cal C}\left(
e_{it}e_{is}e_{jp}e_{jq}\right) -\mathbb{E}_{\cal C}\left( e_{it}e_{is}\right) \mathbb{E}_{\cal C}\left(
e_{jp}e_{jq}\right) \right\}
  \nonumber \\ &\qquad\qquad\qquad\qquad
=\left\{
  \begin{array}{l@{\quad}l}
   \Sigma _{it}\Sigma _{is}  & \text{ if }\left( t=p\right) \neq \left(
s=q\right) \text{ and }\left( i=j\right)   \\
   \Sigma _{it}\Sigma _{is}  & \text{ if }\left( t=q\right) \neq \left(
s=p\right) \text{ and }\left( i=j\right)   \\
   \mathbb{E}_{\cal C}\left( e_{it}^{4}\right) -\Sigma _{it}^{2}   &   \text{ if }\left(
t=s=p=q\right) \text{ and }\left( i=j\right)  \\
    0 & \text{ otherwise.}
  \end{array} \right.
\end{align*}%
Therefore,
\begin{align*}
&\mathbb{E}_{\cal C}\left( \limfunc{Tr}\left\{ \left( e^{\prime }e-\mathbb{E}_{\cal C}\left( e^{\prime
}e\right) \right) A\right\} \right) ^{2}
\\
& \qquad
\leq \sum_{t=1}^{T}\sum_{s=1}^{T}\sum_{i=1}^{N}\Sigma _{it}\Sigma
_{is}\left( \mathbb{E}_{\cal C} \left( A_{ts}^{2}\right) +\mathbb{E}_{\cal C} \left( A_{ts}A_{st}\right) \right)
+\sum_{t=1}^{T}\sum_{i=1}^{N}\left( \mathbb{E}_{\cal C}\left( e_{it}^{4}\right) -\Sigma
_{it}^{2}\right) \mathbb{E}_{\cal C}  A_{tt}^{2}.
\end{align*}%
Define $\Sigma ^{i}={\rm diag}\left(  \Sigma _{i1},...,\Sigma _{iT} \right) .$ Then, we have%
\begin{eqnarray}
\sum_{t=1}^{T}\sum_{s=1}^{T}\sum_{i=1}^{N}\Sigma _{it}\Sigma _{is}\left(
\mathbb{E}_{\cal C}  A_{ts}^{2}\right) &=&\mathbb{E}_{\cal C} \left( \sum_{i=1}^{N}\limfunc{Tr}\left( A^{\prime
}\Sigma ^{i}A\Sigma ^{i}\right) \right) \nonumber \\
&\leq &\sum_{i=1}^{N}\mathbb{E}_{\cal C} \left\Vert A\Sigma ^{i}\right\Vert _{F}^{2}\leq
\sum_{i=1}^{N}\left\Vert \Sigma ^{i}\right\Vert ^{2}\mathbb{E}_{\cal C} \left\Vert A\right\Vert
_{F}^{2} \nonumber \\
&\leq &N\left( \sup_{it}\Sigma _{it}^{2}\right) \mathbb{E}_{\cal C} \left\Vert A\right\Vert
_{F}^{2}.
\end{eqnarray}%
Also,%
\begin{eqnarray}
\sum_{t=1}^{T}\sum_{s=1}^{T}\sum_{i=1}^{N}\Sigma _{it}\Sigma _{is}\mathbb{E}_{\cal C} \left(
A_{ts}A_{st}\right) &=&\mathbb{E}_{\cal C} \left[ \sum_{i=1}^{N}\limfunc{Tr}\left( \Sigma
^{i}AA\Sigma ^{i}\right) \right] \nonumber \\
&\leq &\sum_{i=1}^{N}\mathbb{E}_{\cal C} \left\Vert \Sigma ^{i}A\right\Vert _{F}\left\Vert
A\Sigma ^{i}\right\Vert _{F}\leq \sum_{i=1}^{N}\left\Vert \Sigma
^{i}\right\Vert ^{2}\mathbb{E}_{\cal C} \left\Vert A\right\Vert _{F}^{2} \nonumber \\
&\leq &N\left( \sup_{it}\Sigma _{it}^{2}\right) \mathbb{E}_{\cal C} \left\Vert A\right\Vert
_{F}^{2} \; .
\end{eqnarray}%
Finally,
\begin{eqnarray}
\sum_{t=1}^{T}\sum_{i=1}^{N}\left( \mathbb{E}_{\cal C}\left( e_{it}^{4}\right) -\Sigma
_{it}^{2}\right) \mathbb{E}_{\cal C}  A_{tt}^{2}
&\leq &N\left( \sup_{it}\mathbb{E}_{\cal C}\left( e_{it}^{4}\right) \right) \mathbb{E}_{\cal C} \left\Vert
A\right\Vert _{F}^{2},
\end{eqnarray}
and $\sup_{it}\mathbb{E}_{\cal C}\left( e_{it}^{4}\right)$ is assumed
bounded by Assumption~\ref{ass:A5}$(vi)$.

\# Part (b): The proof is analogous to the proof of part (a).
\end{proof}

\begin{proof}[\bf Proof of Lemma~\ref{lemma:vanishing}]
   \# For part (a) we have
     \begin{align*}
           \left|  \frac 1 {\sqrt{NT}} {\rm Tr} \left( P_{f^0} \, e^{\prime}\, P_{\lambda^0} \, \widetilde  X_k \right) \right|
             &=
             \left| \frac 1 {\sqrt{NT}} {\rm Tr} \left( P_{f^0} \, e^{\prime}\, P_{\lambda^0}  P_{\lambda^0} \widetilde  X_k P_{f^0} \right) \right|
            \nonumber \\
          & \leq  \frac R {\sqrt{NT}}
            \left\| P_{\lambda^0}  \, e \,  P_{f^0} \right\|
            \left\| P_{\lambda^0} \widetilde  X_k \right\|
            \left\| P_{f^0} \right\|
        \nonumber \\
          &=    \frac{1}{\sqrt{NT}} \; {\cal O}_p(1) \, o_p(\sqrt{NT}) \, {\cal O}_p(1)
        \nonumber \\
          &= o_p(1),
     \end{align*}
where the second-last equality follows by Lemma~\ref{lemma:normXweak} (a) and (d).

  \# To show statement (b) we define $\zeta_{k,ijt} = e_{it} \widetilde X_{k,jt}$. We then have
    \begin{align*}
          \frac 1 {\sqrt{NT}} {\rm Tr} \left(  P_{\lambda^0} \, e \,\widetilde X_k' \right)
             &=  \sum_{r,q=1}^R  \left[\left( \frac{\lambda^{0 \prime}  \lambda^{0}} N \right)^{-1} \right]_{rq}
              \underbrace{  \frac 1 {N \sqrt{NT}}
                \sum_{t=1}^T \sum_{i,j=1}^N  \lambda_{ir}^{0} \lambda_{jq}^{0} \zeta_{k,ijt}
                }_{ \equiv A_{k,rq} } .
    \end{align*}
    We only have $ \mathbb{E_{\cal C}}\left( \zeta_{k,ijt} \zeta_{k,lm s}  \right) \neq 0$
    if $t=s$ (because regressors are pre-determined)
    and $i=l$ and $j=m$ (because of cross-sectional independence). Therefore
    \begin{align*}
        \mathbb{E}\left\{  \mathbb{E}_{\cal C}\left( A_{k,rq}^2   \right) \right\}
         &=   \mathbb{E}\left\{ \frac 1 {N^3 T}    \sum_{t,s=1}^T \sum_{i,j,l,m=1}^N
         \lambda_{ir} \lambda_{jq} \lambda_{lr} \lambda_{mq}
         \, \mathbb{E}_{\cal C}\left( \zeta_{k,ijt} \zeta_{k,lm s}   \right) \right\}
       \\
          &=   \frac 1 {N^3 T}    \sum_{t=1}^T \sum_{i,j=1}^N
        \mathbb{E}\left[ \lambda_{ir}^2 \lambda_{jq}^2
         \, \mathbb{E}_{\cal C}\left( \zeta_{k,ijt}^2   \right) \right] = {\cal O}(1/N) = o_p(1).
    \end{align*}
    We thus have $A_{k,rq} = o_p(1)$ and
    therefore also $ \frac 1 {\sqrt{NT}} {\rm Tr} \left(  P_{\lambda^0} \, e \,\widetilde X_k' \right)
    = o_p(1)$.

 \# The proof for statement (c) is similar to the proof of statement (b).
 Define    $\xi_{k,its} = e_{it} \widetilde X_{k,is} - \mathbb{E}_{\cal C}\left( e_{it} \widetilde X_{k,is}  \right)$.
 We then have
 \begin{align*}
   \frac 1 {\sqrt{NT}} {\rm Tr} \left\{ P_{f^0} \, \left[ e^{\prime} \, \widetilde  X_k
                        - \mathbb{E}_{\cal C}\left( e^{\prime} \, \widetilde  X_k  \right) \right] \right\}
       &= \sum_{r,q=1}^R  \left[\left( \frac{f' f} T \right)^{-1} \right]_{rq}
              \underbrace{  \frac 1 {T \sqrt{NT}}
                \sum_{i=1}^N \sum_{t,s=1}^T  f_{tr} f_{sq} \xi_{k,its}
                }_{ \equiv B_{k,rq} } .
 \end{align*}
 Therefore
 \begin{align*}
      \mathbb{E}_{\cal C}\left( B_{k,rq}^2  \right)
         &=   \frac 1 {T^3 N}
          \sum_{i,j=1}^N \sum_{t,s,u,v=1}^T  f_{tr} f_{sq} f_{ur} f_{vq}  \mathbb{E}_{\cal C}\left(  \xi_{k,its} \xi_{k,juv}   \right)
    \\
        &\leq
         \left( \max_{t,\widetilde r} | f_{t \widetilde r} | \right)^4
         \frac 1 {T^3 N}
          \sum_{i,j=1}^N \sum_{t,s,u,v=1}^T
         \left|  {\rm Cov}_{\cal C}\left(  e_{it} \widetilde X_{k,is}  , e_{ju} \widetilde X_{k,jv}   \right) \right|
    \\
        &=
         \left( \max_{t,\widetilde r} | f_{t \widetilde r} | \right)^4
         \frac 1 {T^3 N}
          \sum_{i=1}^N \sum_{t,s,u,v=1}^T
         \left|  {\rm Cov}_{\cal C}\left(  e_{it} \widetilde X_{k,is}  , e_{iu} \widetilde X_{k,iv}    \right) \right|
  \\
       &= {\cal O}_p( T^{4/(4+\epsilon)} ) {\cal O}_p( 1/T )
   \\
      &= o_p(1),
 \end{align*}
 where we used uniformly bounded $\mathbb{E} \| f^0_t \|^{4+\epsilon}$
     implies $\max_t | f^0_{tr} | = {\cal O}_p( T^{1/(4+\epsilon)} )$.

  \#   Part (d) and (e):
  We have
   $\|\lambda^0 \, (\lambda^{0\prime}\lambda^0)^{-1} \, (f^{0\prime}f^0)^{-1} \, f^{0\prime}\| = {\cal O}_p((NT)^{-1/2})$,
    $\|e\|={\cal O}_p(N^{1/2})$,
     $\|X_k\|={\cal O}_p(\sqrt{NT})$ and
     $\|P_{\lambda^0} e P_{f^0} \| = {\cal O}_p(1)$, which was shown in Lemma~\ref{lemma:normXweak}.
     Therefore:
     \begin{align*}
         &\frac 1 {\sqrt{NT}} {\rm Tr}\left(e P_{f^0} \, e' \, M_{\lambda^0} \, X_k \,
              f^0 \, (f^{0\prime}f^0)^{-1} \, (\lambda^{0\prime}\lambda^0)^{-1} \, \lambda^{0\prime} \right)
        \nonumber \\
         &\qquad \qquad =
           \frac 1 {\sqrt{NT}} {\rm Tr}\left(P_{\lambda^0} e P_{f^0} \, e' \, M_{\lambda^0} \, X_k \,
              f^0 \, (f^{0\prime}f^0)^{-1} \, (\lambda^{0\prime}\lambda^0)^{-1} \, \lambda^{0\prime} \right)
        \nonumber \\
         &\qquad \qquad \leq
           \frac R {\sqrt{NT}} \left\|P_{\lambda^0} e P_{f^0}\right\| \|e\| \|X_k\|
              \left\| f^0 \, (f^{0\prime}f^0)^{-1} \, (\lambda^{0\prime}\lambda^0)^{-1} \, \lambda^{0\prime} \right\|
            = {\cal O}_p(N^{-1/2}) = o_p(1) \; .
     \end{align*}
     which shows statement (d). The proof for part (e) is analogous.

   \# To prove statement (f) we need to use in addition
        $ \| P_{\lambda^0} \, e  \, X'_k \| = o_p(N^{3/2})$, which was also
      shown in Lemma~\ref{lemma:normXweak}. We find
      \begin{align*}
         & \frac 1 {\sqrt{NT}} {\rm Tr}\left(e^{\prime}M_{\lambda^0} \, X_k \, M_{f^0} \, e^{\prime}
                \, \lambda^0 \, (\lambda^{0\prime}\lambda^0)^{-1} \, (f^{0\prime}f^0)^{-1} \, f^{0\prime} \right)
       \nonumber \\
        & \qquad\qquad =
          \frac 1 {\sqrt{NT}} {\rm Tr}\left(e^{\prime}M_{\lambda^0} \, X_k \, e^{\prime} \, P_{\lambda^0}
                \, \lambda^0 \, (\lambda^{0\prime}\lambda^0)^{-1} \, (f^{0\prime}f^0)^{-1} \, f^{0\prime} \right)
       \nonumber \\
        & \qquad\qquad \qquad -
          \frac 1 {\sqrt{NT}} {\rm Tr}\left(e^{\prime}M_{\lambda^0} \, X_k \, P_{f^0} \, e^{\prime}\, P_{\lambda^0}
                \, \lambda^0 \, (\lambda^{0\prime}\lambda^0)^{-1} \, (f^{0\prime}f^0)^{-1} \, f^{0\prime} \right)
       \nonumber \\
        & \qquad\qquad \leq
         \frac R {\sqrt{NT}}  \| e \|  \| P_{\lambda^0} \, e  \, X'_k \|
                \, \| \lambda^0 \, (\lambda^{0\prime}\lambda^0)^{-1} \, (f^{0\prime}f^0)^{-1} \, f^{0\prime} \|
       \nonumber \\
        & \qquad\qquad \qquad -
         \frac R {\sqrt{NT}} \| e   \| \| X_k \| \|   P_{\lambda^0}  \, e \, P_{f^0} \|
                \| \lambda^0 \, (\lambda^{0\prime}\lambda^0)^{-1} \, (f^{0\prime}f^0)^{-1} \, f^{0\prime} \|
        \nonumber \\
          &\qquad\qquad =  o_p(1) \; .
      \end{align*}

\# Now we want to prove part (g) and (h) of the present lemma.
   For part (g) we have
   \begin{align*}
         & \frac 1 {\sqrt{NT}} {\rm Tr}\left\{ \left[ e e' - \mathbb{E}_{\cal C} \left( e e' \right) \right]  \, M_{\lambda^0} \,   X_k \, f^0 \, (f^{0\prime}f^0)^{-1} \, (\lambda^{0\prime}\lambda^0)^{-1} \, \lambda^{0\prime} \right\}
          \nonumber \\
          &=     \frac 1 {\sqrt{NT}} {\rm Tr}\left\{ \left[ e e' - \mathbb{E}_{\cal C} \left( e e' \right) \right]  \, M_{\lambda^0} \,   \overline X_k \, f^0 \, (f^{0\prime}f^0)^{-1} \, (\lambda^{0\prime}\lambda^0)^{-1} \, \lambda^{0\prime} \right\}
          \\ & \qquad
            +
            \frac 1 {\sqrt{NT}} {\rm Tr}\left\{ \left[ e e' - \mathbb{E}_{\cal C} \left( e e' \right) \right]  \, M_{\lambda^0} \,   \widetilde X_k
            P_{f^0} \, f^0 \, (f^{0\prime}f^0)^{-1} \, (\lambda^{0\prime}\lambda^0)^{-1} \, \lambda^{0\prime} \right\}
        \nonumber \\
        &= \frac 1 {\sqrt{NT}} {\rm Tr}\left\{ \left[ e e' - \mathbb{E}_{\cal C} \left( e e' \right) \right]  \, M_{\lambda^0} \,   \overline X_k \, f^0 \, (f^{0\prime}f^0)^{-1} \, (\lambda^{0\prime}\lambda^0)^{-1} \, \lambda^{0\prime} \right\}
          \\ & \qquad
          +  \frac 1 {\sqrt{NT}}
            \left\| e e' - \mathbb{E}_{\cal C} \left( e e' \right) \right\|
            \left\| \widetilde X_k P_{f^0} \right\|
            \left\|  f^0 \, (f^{0\prime}f^0)^{-1} \, (\lambda^{0\prime}\lambda^0)^{-1} \, \lambda^{0\prime}  \right\|
        \nonumber \\
        &= \frac 1 {\sqrt{NT}} {\rm Tr}\left\{ \left[ e e' - \mathbb{E}_{\cal C} \left( e e' \right) \right]  \, M_{\lambda^0} \,   \overline X_k \, f^0 \, (f^{0\prime}f^0)^{-1} \, (\lambda^{0\prime}\lambda^0)^{-1} \, \lambda^{0\prime} \right\}
         + o_p(1) .
   \end{align*}
   Thus, what is left to prove is
   $ \frac 1 {\sqrt{NT}} {\rm Tr}\left\{ \left[ e e' - \mathbb{E}_{\cal C} \left( e e' \right) \right]  \, M_{\lambda^0} \,   \overline X_k \, f^0 \, (f^{0\prime}f^0)^{-1} \, (\lambda^{0\prime}\lambda^0)^{-1} \, \lambda^{0\prime} \right\} = o_p(1)$.
   For this we define
   \begin{align*}
      B_k &=  M_{\lambda^0} \, \overline X_k \, f^0 \, (f^{0\prime}f^0)^{-1} \, (\lambda^{0\prime}\lambda^0)^{-1} \, \lambda^{0\prime} \; .
   \end{align*}
   Using part (i) and (ii) of Lemma~\ref{lemma:inequalities} we find
   \begin{align*}
      \| B_k \|_F &\leq R^{1/2} \|B_k\|
        \nonumber \\
         &\leq R^{1/2} \| \overline X_k \|
           \left\| f^0 \, (f^{0\prime}f^0)^{-1} \, (\lambda^{0\prime}\lambda^0)^{-1} \, \lambda^{0\prime} \right\|
        \nonumber \\
         &\leq R^{1/2} \| \overline X_k \|_F
           \left\| f^0 \, (f^{0\prime}f^0)^{-1} \, (\lambda^{0\prime}\lambda^0)^{-1} \, \lambda^{0\prime} \right\| \; .
   \end{align*}
   and therefore
   \begin{align*}
      \mathbb{E}_{\cal C} \left( \| B_k \|_F^2   \right)
            &\leq R
                \left\| f^0 \, (f^{0\prime}f^0)^{-1} \, (\lambda^{0\prime}\lambda^0)^{-1} \, \lambda^{0\prime} \right\|^2
                 \mathbb{E}_{\cal C} \left( \| \overline X_k \|_F^2  \right)
       \nonumber \\
             &= {\cal O}(1) \; ,
   \end{align*}
   where we used $ \mathbb{E}_{\cal C}\left( \| \overline X_k \|_F^2  \right) = {\cal O}(NT)$, which is true because we assumed uniformly
   bounded moments of $\overline X_{k,it}$.
   Applying Lemma~\ref{lemma:eeterms} we therefore find
   \begin{align*}
      \mathbb{E}_{\cal C} \left(
        \frac 1 {\sqrt{NT}} {\rm Tr}\left\{ \left[ e e' - \mathbb{E}_{\cal C} \left( e e' \right) \right]  B_k \right\}  
      \right)^2
        &\leq c_0 \, \frac T {NT} \, \mathbb{E}_{\cal C} \left( \| B_k \|_F^2   \right)  = o(1) \; ,
   \end{align*}
   and thus
   \begin{align*}
      \frac 1 {\sqrt{NT}} {\rm Tr}\left\{ \left[ e e' - \mathbb{E}_{\cal C} \left( e e' \right) \right]  B_k \right\} &= o_p(1) \; ,
   \end{align*}
   which is what we wanted to show.
   The proof for part (h) is analogous.

   \#  Part (i): Conditional on ${\cal C}$ the expression
   $e_{it}^2  \mathfrak{X}_{it} \, \mathfrak{X}_{it}'
             -  \mathbb{E}_{\cal C} \left(  e_{it}^2 \, \mathfrak{X}_{it} \, \mathfrak{X}_{it}'  
                  \right)$
        is mean zero,
        and it is also uncorrelated across $i$.
        This together with the bounded moments that we assume implies
      \begin{align*}
          {\rm Var}_{\cal C}\left\{
          \frac 1 {NT} \, \sum_{i=1}^N \, \sum_{t=1}^T
             \left[  e_{it}^2  \, \mathfrak{X}_{it} \, \mathfrak{X}_{it}'
             -  \mathbb{E}_{\cal C} \left(  e_{it}^2 \, \mathfrak{X}_{it} \, \mathfrak{X}_{it}'  
                  \right)    \right]    \right\}  &= {\cal O}_p(1/N) = o_p(1) ,
      \end{align*}
      which shows the required result.

   \# Part (j):
     Define the $K \times K$ matrix $A =   \frac 1 {NT} \, \sum_{i=1}^N \, \sum_{t=1}^T
            \,    e_{it}^2 
            \left(  \mathfrak{X}_{it}   +  {\cal X}_{it}   \right)
            \left( \mathfrak{X}_{it}   -  {\cal X}_{it}  \right)'$.
      Then we have
      \begin{align*}
           \frac 1 {NT} \, \sum_{i=1}^N \, \sum_{t=1}^T
            \,   e_{it}^2  
            \left(  \mathfrak{X}_{it} \, \mathfrak{X}_{it}'   -  {\cal X}_{it} \, {\cal X}_{it}' \right)
             &=  \frac 1 2 \left( A + A' \right).
       \end{align*}
       Let $B_k$ be the $N \times T$ matrix with
        elements $B_{k,it} =  e_{it}^2  \left(  \mathfrak{X}_{k,it}   +  {\cal X}_{k,it}   \right)$.
       We have $\| B_k \| \leq \|B_k\|_F = {\cal O}_p( \sqrt{NT} )$, because the moments
       of  $B_{k,it}$ are uniformly bounded.
       The components of $A$ can be written as
       $A_{l k} = \frac 1 {NT} {\rm Tr}[ B_l  ( \mathfrak{X}_{k}   -  {\cal X}_{k})' ]$. We therefore have
       \begin{align*}
            | A_{l k} |
            \leq  \frac 1 {NT} {\rm rank}( \mathfrak{X}_{k}   -  {\cal X}_{k} )
                  \| B_l \| \left\|  \mathfrak{X}_{k}   -  {\cal X}_{k} \right\|   .
       \end{align*}
       We have $ \mathfrak{X}_{k}   -  {\cal X}_{k}  =
          \widetilde X_k \, P_{f^0} +  P_{\lambda^0} \, \widetilde X_k \, M_{f^0}$.
        Therefore   ${\rm rank}( \mathfrak{X}_{k}   -  {\cal X}_{k} ) \leq 2 R$ and
        \begin{align*}
             | A_{l k} |
           & \leq \frac{2 R} {NT}
                  \| B_l \|
                 \left(  \left\|   \widetilde X_k \, P_{f^0}
                  \right| + \left\| P_{\lambda^0} \, \widetilde X_k \, M_{f^0} \right\|   \right)
           \\
             &\leq  \frac{2 R} {NT}
                               \| B_l \|
                 \left(  \left\|   \widetilde X_k \, P_{f^0}
                  \right| + \left\| P_{\lambda^0} \, \widetilde X_k  \right\|   \right)
              = \frac{2 R} {NT}     {\cal O}_p( \sqrt{NT} )  o_p( \sqrt{NT} ) = o_p(1),
        \end{align*}
        where we used Lemma~\ref{lemma:normXweak}.
        This shows the desired result.
\end{proof}

\begin{proof}[\bf Proof of Lemma~\ref{lemma:denCLT}]
Let $c$ be a $K$-vector such that $\left\Vert c\right\Vert =1.$ The
required result follows by the Cramer-Wold device, if we show
\begin{align*}
\frac{1}{\sqrt{NT}}\sum_{i=1}^{N}\sum_{t=1}^{T}e_{it}\mathfrak{X}_{it}^{\prime }c \, \Rightarrow \, {\cal N}\left( 0,c^{\prime }\Omega c\right) \, .
\end{align*}
For this, define $\xi_{it} =e_{it}\mathfrak{X}_{it}^{\prime }c$.
Furthermore
define $\xi_m = \xi_{M,m} = \xi_{NT,it}$, with $M = NT$ and $m = T(i-1)+t \in \{1,\ldots,M\}$.
We then have the following:

\begin{itemize}
\item[(i)]
  Under Assumption~\ref{ass:A5}$(i)$, $(ii)$, $(iii)$ the sequence
 $\{ \xi_m, \, m=1,\ldots,M \}$  is a martingale difference sequence
 under the filtration 
 ${\cal F}_m = {\cal C} \vee \sigma(\{ \xi_{n}: n < m \})$.

\item[(ii)] $\mathbb{E}(\xi_{it}^4 )$ is uniformly bounded,
because by Assumption~\ref{ass:A5}$(vi)$
  $\mathbb{E}_{\cal C} e_{it}^8$ and $\mathbb{E}_{\cal C}\left( \| X_{it} \|^{8+\epsilon}   \right)$
  are uniformly bounded by a non-random constant
  (applying Cauchy-Schwarz and the law of iterated expectations).

\item[(iii)]
$\frac 1 {M}  \sum_{m=1}^M \xi_{m}^2 = c' \Omega c + o_p(1)$. \\
This is true, because firstly under our assumptions we have
$\mathbb{E}_{\cal C}
\left\{ \left[ \frac 1 {M}  \sum_{m=1}^M 
\left( \xi_{m}^2 - \mathbb{E}_{\cal C}( \xi_{m}^2) \right) \right]^2
\right\}
=
\mathbb{E}_{\cal C}
\left\{ \frac 1 {M^2}  \sum_{m=1}^M 
\left( \xi_{m}^2 - \mathbb{E}_{\cal C}( \xi_{m}^2) \right)^2
\right\} = {\cal O}_P(1/M) =  o_P(1)$, implying we have
$\frac 1 {M}  \sum_{m=1}^M \xi_{m}^2
=  \frac 1 {M}  \sum_{m=1}^M \mathbb{E}_{\cal C}( \xi_{m}^2  )
    + o_p(1)$.
We furthermore have
$\frac 1 {M}  \sum_{m=1}^M \mathbb{E}_{\cal C}( \xi_{m}^2  ) 
  = {\rm Var}_{\cal C}( M^{-1/2} \sum_{m=1}^M \xi_{m} )$,
and using the result in equation \eqref{VarEqOmega} of the main text
we find   
  ${\rm Var}_{\cal C}( M^{-1/2} \sum_{m=1}^M \xi_{m} )
  ={\rm Var}_{\cal C}( (NT)^{-1/2} \sum_{i=1}^N \sum_{t=1}^T \xi_{it} )
  = c' \Omega c + o_p(1)$.

\end{itemize}
These three properties of $\{ \xi_m, \, m=1,\ldots,M \}$ allow us to
apply Corollary~5.26 in White~\cite*{White2001}, which is based on
    Theorem~2.3 in Mcleish~\cite*{Mcleish1974},  to obtain
    $\frac 1 {  \sqrt{M}} \sum_{m=1}^M \xi_{m}
         \to_d {\cal N}(0,c' \Omega c)$.
    This concludes the proof, because
   $ \frac 1 {  \sqrt{M}} \sum_{m=1}^M \xi_{m}  =     \frac{1}{\sqrt{NT}}\sum_{i=1}^{N}\sum_{t=1}^{T}e_{it}\mathfrak{X}_{it}^{\prime }c$.
\end{proof}

\section{Expansions of Projectors and Residuals}

The incidental parameter estimators  $\widehat f$ and $\widehat \lambda$
as well as the residuals $\widehat e$ enter into the asymptotic bias and variance
estimators for the LS estimator $\widehat \beta$.
To describe the properties of $\widehat f$, $\widehat \lambda$ and $\widehat e$,
it is convenient to have asymptotic  expansions of the projectors
$M_{\widehat \lambda}(\beta)$ and $M_{\widehat f}(\beta)$
that correspond to the minimizing parameters $\widehat \lambda(\beta)$
and $\widehat f(\beta)$ in equation \eqref{LNT123}. Note the minimizing $\widehat \lambda(\beta)$
and $\widehat f(\beta)$ can be defined for all values of $\beta$, not only for the
optimal value $\beta=\widehat \beta$.
The corresponding residuals are
$\widehat e(\beta) = Y \, - \, \beta \cdot X \, - \, \widehat \lambda(\beta) \, \widehat f'(\beta)$.

\begin{theorem}
   \label{theorem:expansions}
   Under Assumptions~\ref{ass:A1}, \ref{ass:A3}, and \ref{ass:A4}(i) we have the following expansions
  \begin{align*}
     M_{\widehat \lambda}(\beta) &= M_{\lambda^0} + M_{\widehat \lambda,e}^{(1)}
                                     + M_{\widehat \lambda,e}^{(2)}
                                     - \sum_{k=1}^K \left( \beta_k - \beta^0_k \right)  M_{\widehat \lambda,k}^{(1)}
                                     + M_{\widehat \lambda}^{({\rm rem})}(\beta) \; ,
    \nonumber \\
     M_{\widehat f}(\beta) &= M_{f^0} + M_{\widehat f,e}^{(1)}
                                     + M_{\widehat f,e}^{(2)}
                                     - \sum_{k=1}^K \left( \beta_k - \beta^0_k \right)  M_{\widehat f,k}^{(1)}
                                     + M_{\widehat f}^{({\rm rem})}(\beta) \; ,
    \nonumber \\
    \widehat e(\beta) &=  M_{\lambda^0} \, e \, M_{f^0}
            + \widehat e^{(1)}_e - \sum_{k=1}^K \left( \beta_k - \beta^0_k \right) \widehat e^{(1)}_k
            + \widehat e^{({\rm rem})}(\beta) \; ,
  \end{align*}
  where the spectral norms of the remainders satisfy for any series $\eta_{NT} \rightarrow 0$:
  \begin{align*}
     \sup_{\{\beta :\left\| \beta -\beta^{0} \right\| \leq \eta_{NT}\}}
        \frac{\left\| M_{\widehat \lambda}^{({\rm rem})}(\beta) \right\|}
        { \|\beta - \beta^0\|^2 + (NT)^{-1/2} \, \|e\| \, \|\beta - \beta^0\| \,  \,  + (NT)^{-3/2} \, \|e\|^3}
            &= {\cal O}_p\left(1\right) \, ,
     \nonumber \\
     \sup_{\{\beta :\left\| \beta -\beta^{0} \right\| \leq \eta_{NT}\}}
        \frac{\left\| M_{\widehat f}^{({\rm rem})}(\beta) \right\|}
        { \|\beta - \beta^0\|^2 + (NT)^{-1/2} \, \|e\| \, \|\beta - \beta^0\| \,  \,  + (NT)^{-3/2} \, \|e\|^3}
            &= {\cal O}_p\left(1\right) \, ,
     \nonumber \\
     \sup_{\{\beta :\left\| \beta -\beta^{0} \right\| \leq \eta_{NT}\}}
       \frac{ \left\| \widehat e^{({\rm rem})}(\beta) \right\| }
         { (NT)^{1/2} \|\beta - \beta^0\|^2 + \|e\| \, \|\beta - \beta^0\| + (NT)^{-1} \|e\|^3 }
           &= {\cal O}_p\left(1\right)  \; ,
  \end{align*}
  and we have ${\rm rank}(\widehat e^{({\rm rem})}(\beta)) \leq 7R$,
  and the expansion coefficients are given by
  \begin{align*}
     M^{(1)}_{\widehat \lambda,e} &= - \, M_{\lambda^0} \, e \,
      f^0 \, (f^{0\prime}f^0)^{-1} \, (\lambda^{0\prime}\lambda^0)^{-1} \lambda^{0\prime}
       \, - \, \lambda^0 \, (\lambda^{0\prime}\lambda^0)^{-1} \, (f^{0\prime}f^0)^{-1} \, f^{0\prime}
              \, e' \, M_{\lambda^0}  \; ,
    \nonumber \\
   M^{(1)}_{\widehat \lambda,k} &= - \, M_{\lambda^0} \, X_k  \,
      f^0 \, (f^{0\prime}f^0)^{-1} \, (\lambda^{0\prime}\lambda^0)^{-1} \lambda^{0\prime}
       \, - \, \lambda^0 \, (\lambda^{0\prime}\lambda^0)^{-1} \, (f^{0\prime}f^0)^{-1} \, f^{0\prime}
              \,  X'_k  \, M_{\lambda^0}  \; ,
    \nonumber \\
   M^{(2)}_{\widehat \lambda,e} &=
        M_{\lambda^0} \, e \, f^0 \, (f^{0\prime}f^0)^{-1} \, (\lambda^{0\prime}\lambda^0)^{-1} \lambda^{0\prime}
             \, e \, f^0 \, (f^{0\prime}f^0)^{-1} \, (\lambda^{0\prime}\lambda^0)^{-1} \lambda^{0\prime}
       \nonumber \\ & \qquad
       +\lambda^0 \, (\lambda^{0\prime}\lambda^0)^{-1} \, (f^{0\prime}f^0)^{-1} \, f^{0\prime}
             \, e' \, \lambda^0 \, (\lambda^{0\prime}\lambda^0)^{-1} \, (f^{0\prime}f^0)^{-1} \, f^{0\prime}
             \, e' \, M_{\lambda^0}
       \nonumber \\ & \qquad
          - M_{\lambda^0} \, e \, M_{f^0} \, e' \,
\lambda^0\,(\lambda^{0\prime}\lambda^0)^{-1}\,(f^{0\prime}f^0)^{-1}\,(\lambda^{0\prime}\lambda^0)^{-1}\,\lambda^{0\prime}
       \nonumber \\ & \qquad
-\lambda^0\,(\lambda^{0\prime}\lambda^0)^{-1}\,(f^{0\prime}f^0)^{-1}\,(\lambda^{0\prime}\lambda^0)^{-1}\,\lambda^{0\prime}
   \, e \, M_{f^0} \, e' \, M_{\lambda^0}
       \nonumber \\ & \qquad
          - M_{\lambda^0} \, e \,
      f^0 \, (f^{0\prime}f^0)^{-1} \, (\lambda^{0\prime}\lambda^0)^{-1} \, (f^{0\prime}f^0)^{-1} \, f^{0\prime}
         \, e' \, M_{\lambda^0}
       \nonumber \\ & \qquad
        + \lambda^0 \, (\lambda^{0\prime}\lambda^0)^{-1} \, (f^{0\prime}f^0)^{-1} \, f^{0\prime}
            \, e' \, M_{\lambda^0} \, e \,
              f^0 \, (f^{0\prime}f^0)^{-1} \, (\lambda^{0\prime}\lambda^0)^{-1} \lambda^{0\prime} \, ,
  \end{align*}
  analogously
  \begin{align*}
   M^{(1)}_{\widehat f,e} &= \, - \, M_{f^0} \, e' \,
      \lambda^0 \, (\lambda^{0\prime}\lambda^0)^{-1} \, (f^{0\prime}f^0)^{-1} f^{0\prime}
       \, - \, f^0 \, (f^{0\prime}f^0)^{-1} \, (\lambda^{0\prime}\lambda^0)^{-1} \, \lambda^{0\prime}
              \, e \, M_{f^0}   \; ,
    \nonumber \\
  M^{(1)}_{\widehat f,k} &=
       \, - \, M_{f^0} \,  X'_k  \,
      \lambda^0 \, (\lambda^{0\prime}\lambda^0)^{-1} \, (f^{0\prime}f^0)^{-1} f^{0\prime}
       \, - \, f^0 \, (f^{0\prime}f^0)^{-1} \, (\lambda^{0\prime}\lambda^0)^{-1} \, \lambda^{0\prime}
              \,  X_k \,M_{f^0}  \; ,
    \nonumber \\
  M^{(2)}_{\widehat f,e} &=
       M_{f^0} \, e' \, \lambda^0 \, (\lambda^{0\prime}\lambda^0)^{-1} \, (f^{0\prime}f^0)^{-1} f^{0\prime}
             \, e' \, \lambda^0 \, (\lambda^{0\prime}\lambda^0)^{-1} \, (f^{0\prime}f^0)^{-1} f^{0\prime}
       \nonumber \\ & \qquad
       +f^0 \, (f^{0\prime}f^0)^{-1} \, (\lambda^{0\prime}\lambda^0)^{-1} \, \lambda^{0\prime}
             \, e \, f^0 \, (f^{0\prime}f^0)^{-1} \, (\lambda^{0\prime}\lambda^0)^{-1} \, \lambda^{0\prime}
             \, e \, M_{f^0}
       \nonumber \\ & \qquad
          - M_{f^0} \, e' \, M_{\lambda^0} \, e \,
f^0\,(f^{0\prime}f^0)^{-1}\,(\lambda^{0\prime}\lambda^0)^{-1}\,(f^{0\prime}f^0)^{-1}\,f^{0\prime}
       \nonumber \\ & \qquad
-f^0\,(f^{0\prime}f^0)^{-1}\,(\lambda^{0\prime}\lambda^0)^{-1}\,(f^{0\prime}f^0)^{-1}\,f^{0\prime}
   \, e' \, M_{\lambda^0} \, e \, M_{f^0}
       \nonumber \\ & \qquad
          - M_{f^0} \, e' \,
      \lambda^0 \, (\lambda^{0\prime}\lambda^0)^{-1} \, (f^{0\prime}f^0)^{-1} \, (\lambda^{0\prime}\lambda^0)^{-1} \, \lambda^{0\prime}
         \, e \, M_{f^0}
       \nonumber \\ & \qquad
        + f^0 \, (f^{0\prime}f^0)^{-1} \, (\lambda^{0\prime}\lambda^0)^{-1} \, \lambda^{0\prime}
            \, e \, M_{f^0} \, e' \,
              \lambda^0 \, (\lambda^{0\prime}\lambda^0)^{-1} \, (f^{0\prime}f^0)^{-1} f^{0\prime} \, ,
  \end{align*}
  and finally
  \begin{align*}
   \widehat e^{(1)}_k &= M_{\lambda^0} \, X_k \, M_{f^0} \; ,
             \nonumber \\
   \widehat e^{(1)}_e &= - M_{\lambda^0} \, e \, M_{f^0} \, e' \,
   \lambda^0\,(\lambda^{0\prime}\lambda^0)^{-1}\,(f^{0\prime}f^0)^{-1}\,f^{0\prime}
             \nonumber \\ & \qquad
  - \lambda^0\,(\lambda^{0\prime}\lambda^0)^{-1}\,(f^{0\prime}f^0)^{-1}\,f^{0\prime}
   \, e' \, M_{\lambda^0} \, e \, M_{f^0}
             \nonumber \\ & \qquad
  - M_{\lambda^0} \, e \,
      f^0 \, (f^{0\prime}f^0)^{-1} \, (\lambda^{0\prime}\lambda^0)^{-1} \, \lambda^{0\prime}
              \, e \, M_{f^0} \; .
  \end{align*}
\end{theorem}

\begin{proof}[\bf Proof]
   The general expansion of $M_{\widehat \lambda}(\beta)$ is given in
   Moon and Weidner \cite*{MoonWeidner2015}, and in the theorem we just make
   this expansion explicit up to a particular order.
   The result for $M_{\widehat f}(\beta)$
   is just obtained by symmetry ($N \leftrightarrow T$, $\lambda \leftrightarrow f$, $e \leftrightarrow e'$,
   $X_k \leftrightarrow X_k'$). For the residuals $\widehat e$ we have
   \begin{align*}
     \widehat e &= M_{\widehat \lambda} \, \left( Y - \sum_{k=1} \, \widehat \beta_k \, X_k \right)
   = M_{\widehat \lambda} \, \left[ e - 
     \left( \widehat \beta - \beta^0 \right) \cdot X
                                 +  \lambda^0 f^{0\prime} \right]  \; ,
   \end{align*}
   and plugging in the expansion of $M_{\widehat \lambda}$ gives the expansion of $\widehat e$.
   We have $\widehat e(\beta)
        = A_0  + \lambda^0 f^{0\prime} - \widehat \lambda(\beta) \widehat f'(\beta)$,
   where $A_0=e - \sum_k (\beta_k-\beta^0_k) X_k$.
   Therefore $\widehat e^{({\rm rem})}(\beta)=A_1+A_2+A_3$ with
   $A_1 = A_0 -  M_{\lambda^0} \, A_0 \, M_{f^0}$,
   $A_2 =\lambda^0 f^{0\prime} - \widehat \lambda(\beta) \widehat f'(\beta)$, and
   $A_3 =-\widehat e^{(1)}_e$.
   We find ${\rm rank}(A_1)\leq 2R$, ${\rm rank}(A_2)\leq 2R$, ${\rm rank}(A_3)\leq 3R$, and thus
   ${\rm rank}(\widehat e^{({\rm rem})}(\beta)) \leq 7R$, as stated in the theorem.
\end{proof}

Having expansions for $M_{\widehat \lambda}(\beta)$ and $M_{\widehat f}(\beta)$, we also have expansions
for $P_{\widehat \lambda}(\beta)=\mathbb{I}_N-M_{\widehat \lambda}(\beta)$ and $P_{\widehat f}(\beta)=\mathbb{I}_T-M_{\widehat f}(\beta)$.
The reason why we give expansions of the projectors and not expansions of $\widehat \lambda(\beta)$ and $\widehat f(\beta)$
directly is for the latter we would need to specify a normalization, whereas the projectors
are independent of any normalization choice. An expansion for $\widehat \lambda(\beta)$ can, for example, be defined by
$\widehat \lambda(\beta) = P_{\widehat \lambda}(\beta) \lambda^0$, in which case the normalization of $\widehat \lambda(\beta)$
is implicitly defined by the normalization of $\lambda^0$.

\section{Consistency Proof for Bias and Variance Estimators
  (Proof of Theorem~\ref{th:biascorrection})}
\label{app:ProofThBias}  

It is convenient to introduce some alternative notation for
Definition~\ref{def:estimators}
in section~\ref{sec:BiasCorrection} of the main text.

\begin{paragraph}{\bf Definition}
\it
Let $\Gamma: \mathbb{R} \rightarrow \mathbb{R}$ be the truncation kernel defined by $\Gamma(x)=1$ for $|x|\leq 1$, and $\Gamma(x)=0$
otherwise. Let $M$ be a bandwidth parameter that depends on $N$ and $T$.
For an $N\times N$ matrix $A$ with elements $A_{ij}$ and a $T\times T$ matrix $B$ with elements
$B_{ts}$ we define
\begin{itemize}
   \item[(i)] the diagonal truncations
     $A^{\rm truncD} = {\rm diag}[ (A_{ii})_{i=1,\ldots,N} ]$
     and
     $B^{\rm truncD} =  {\rm diag}[ (B_{tt})_{t=1,\ldots,T} ]$.

   \item[(ii)] the right-sided  Kernel truncation of $B$,
          which is a $T \times T$ matrix $B^{\rm truncR}$
          with elements $B^{\rm truncR}_{ts} = \Gamma\left( \frac{s-t} M \right) B_{ts}$ for $t<s$,
          and $B^{\rm truncR}_{ts}=0$ otherwise.

\end{itemize}
\end{paragraph}

Here, we suppress the dependence of $B^{\rm truncR}$ on the bandwidth parameter $M$. Using this notation we can represent the 
estimators for the bias in Definition~\ref{def:estimators} as follows:
    \begin{align*}
       \widehat B_{1,k} &= \frac 1 N \, {\rm Tr}\left[ P_{\widehat f} \left( \widehat e' X_k \right)^{\rm truncR} \right]  \; ,
       \nonumber \\
       \widehat B_{2,k} &=
       \frac 1 T
         {\rm Tr}\left[ \left( \widehat e \, {\widehat e}^{\prime} \right)^{\rm truncD} \, M_{\widehat \lambda} \, X_k \,
              \widehat f \, (\widehat f^{\prime} \widehat f)^{-1} \,
                             (\widehat \lambda^{\prime}\widehat \lambda)^{-1} \, \widehat \lambda^{\prime} \right] \; ,
       \nonumber \\
       \widehat B_{3,k} &=
       \frac 1 N
       {\rm Tr}\left[ \left( \widehat e' \, \widehat e \right)^{\rm truncD} \, M_{\widehat f} \, X^{\prime}_k \,
              \widehat \lambda \, (\widehat \lambda^{\prime} \widehat \lambda)^{-1} \,
                             (\widehat f^{\prime}\widehat f)^{-1} \, \widehat f^{\prime} \right]
                   \; .
    \end{align*}

Before proving Theorem~\ref{th:biascorrection} we establish some
preliminary results.

\begin{corollary}
   \label{lemma:sqrtNTcons}
   Under the Assumptions of Theorem~\ref{th:limdis}
    we have $\sqrt{NT} \left( \widehat \beta - \beta^0 \right) = {\cal O}_p(1)$.
\end{corollary}
This corollary directly follows from Theorem~\ref{th:limdis}.

\begin{corollary}
   \label{lemma:Pfhat}
   Under the Assumptions of Theorem~\ref{th:biascorrection} we have
  \begin{align*}
      \left\| P_{\widehat \lambda} - P_{\lambda^0} \right\| &= \left\| M_{\widehat \lambda} - M_{\lambda^0} \right\|
              = {\cal O}_p(N^{-1/2}) \; ,
      \nonumber \\
      \left\| P_{\widehat f} - P_{f^0} \right\| &= \left\| M_{\widehat f} - M_{f^0} \right\| = {\cal O}_p(T^{-1/2}) \; .
   \end{align*}
\end{corollary}

\begin{proof}[\bf Proof]
   Using $\|e\|={\cal O}_p(N^{1/2})$ and $\|X_k\|={\cal O}_p(N)$ we find
   the expansion terms in Theorem~\ref{theorem:expansions} satisfy
   \begin{align*}
      \left\|M_{\widehat \lambda,e}^{(1)}\right\| &= {\cal O}_p(N^{-1/2}) \; , &
      \left\|M_{\widehat \lambda,e}^{(2)}\right\| &= {\cal O}_p(N^{-1}) \; , &
      \left\|M_{\widehat \lambda,k}^{(1)}\right\| &= {\cal O}_p(1) \; .
   \end{align*}
   Together with corollary \ref{lemma:sqrtNTcons} the result for $\left\| M_{\widehat \lambda} - M_{\lambda^0} \right\|$
   immediately follows. In addition we have $P_{\widehat \lambda} - P_{\lambda^0}=-M_{\widehat \lambda}+M_{\lambda^0}$.
   The proof for $M_{\widehat f}$ and $P_{\widehat f}$ is analogous.
\end{proof}

\begin{lemma}
   \label{lemma:A1A2}
   Under the Assumptions of Theorem~\ref{th:biascorrection} we have
     \begin{align*}
         A_1 &\equiv \frac 1 {NT} \sum_{i=1}^N \sum_{t=1}^T e_{it}^2
               \left( {\cal X}_{it} {\cal X}_{it}' - \widehat {\cal X}_{it} \widehat {\cal X}_{it}' \right)
                     = o_p(1)  \; ,
        \\
         A_2 &\equiv \frac 1 {NT} \sum_{i=1}^N \sum_{t=1}^T \left( e_{it}^2 - \widehat e_{it}^2 \right)
                                               \widehat {\cal X}_{it} \widehat {\cal X}_{it}' = o_p(1) \; .
     \end{align*}
\end{lemma}

\begin{lemma}
   \label{lemma:lambdafINV}
   Let $\widehat f$ and $f^0$ be normalized as $\widehat f' \widehat f / T = \mathbb{I}_R$
   and $f^{0\prime} f^0 / T = \mathbb{I}_R$.
   Then, under the assumptions of Theorem~\ref{th:biascorrection}, there exists an $R\times R$ matrix $H=H_{NT}$
   such that\footnote{We consider a limit $N,T\rightarrow \infty$ and for different $N,T$ different $H$-matrices
   can be chosen, but we write $H$ instead of $H_{NT}$ to keep notation simple.}
   \begin{align*}
      \left\| \widehat f - f^0 \, H \right\| &= O_p\left(1\right) \; , &
      \left\| \widehat \lambda - \lambda^0 \, \left(H'\right)^{-1} \right\| &= O_p\left(1\right) \; .
   \end{align*}
   Furthermore
   \begin{align*}
     \left\| \widehat \lambda \, (\widehat \lambda^{\prime}\widehat\lambda)^{-1} \, (\widehat f^{\prime}\widehat f)^{-1} \, \widehat f^{\prime}
        -\lambda^0 \, (\lambda^{0\prime}\lambda^0)^{-1} \, (f^{0\prime}f^0)^{-1} \, f^{0\prime} \right\|
    &= O_p\left( N^{-3/2} \right) \; .
   \end{align*}
\end{lemma}

\begin{lemma}
   \label{lemma:exp}
   Under the Assumptions of Theorem~\ref{th:biascorrection} we have
   \begin{align*}
      {\rm(i)}&& N^{-1} \, \left\| \mathbb{E}_{\cal C}(e'   X_k    ) -
       \left( \widehat e' \, X_k \right)^{\rm truncR} \right\|
                         &= o_p(1) \; ,
     \nonumber \\
     {\rm(ii)}&& N^{-1} \, \left\| \mathbb{E}_{\cal C}(e'   e) -
         \left( \widehat e' \, \widehat e \right)^{\rm truncD} \right\|
                         &= o_p(1) \; ,
     \nonumber \\
     {\rm(iii)}&& T^{-1} \, \left\| \mathbb{E}_{\cal C}(e   e')
           - \left( \widehat e \, \widehat e' \right)^{\rm truncD}  \right\|
                         &= o_p(1) \; .
   \end{align*}
\end{lemma}

\begin{lemma}
   \label{lemma:normsTrunc}
   Under the Assumptions of Theorem~\ref{th:biascorrection} we have
   \begin{align*}
      {\rm(i)}&& N^{-1} \, \left\| \left( \widehat e' \, X_k \right)^{\rm truncR} \right\|
                         &= {\cal O}_p(M T^{1/8}) \; ,
     \nonumber \\
     {\rm(ii)}&& N^{-1} \, \left\| \left( \widehat e' \, \widehat e \right)^{\rm truncD} \right\|
                         &= {\cal O}_p(1) \; ,
     \nonumber \\
     {\rm(iii)}&& T^{-1} \, \left\| \left( \widehat e \, \widehat e' \right)^{\rm truncD}  \right\|
                         &= {\cal O}_p(1) \; .
   \end{align*}
\end{lemma}

The proof of the above lemmas is given section~\ref{sec:Intermed} below.
Using these lemmas we can now prove Theorem~\ref{th:biascorrection}.

\begin{proof}[\bf Proof of Theorem \ref{th:biascorrection}, Part I: show $\widehat W=W+o_p(1)$.]
   $\phantom{a}$
   \\
   Using $\left| {\rm Tr}\left( C\right) \right| \leq \left\| C\right\| \limfunc{rank}\left( C\right)$
   and corollary \ref{lemma:Pfhat}
   we find:
   \begin{align*}
      \big| \widehat W_{k_1 k_2} - & W_{NT,k_1 k_2} \big|
           \nonumber \\
         & = \bigg| (NT)^{-1} {\rm Tr}\left[ \left( M_{\widehat \lambda} - M_{\lambda^0} \right)
                                      \, X_{k_1} \, M_{\widehat f} \,  X_{k_2}' \right]
            + (NT)^{-1} {\rm Tr}\left[ M_{\lambda^0} \, X_{k_1} \,
                                       \left( M_{\widehat f} - M_{f^0} \right) \,  X_{k_2}' \right]
           \nonumber \\
         & \leq \frac{2R}{NT} \left\| M_{\widehat \lambda} - M_{\lambda^0} \right\|
                                      \| X_{k_1} \| \| X_{k_2} \|
                \frac{2R}{NT} \left\| M_{\widehat f} - M_{f^0} \right\|
                                      \| X_{k_1} \| \| X_{k_2} \|
           \nonumber \\
         & = \frac{2R}{NT} {\cal O}_p(N^{-1}) {\cal O}_p(NT)
             +\frac{2R}{NT} {\cal O}_p(T^{-1}) {\cal O}_p(NT)
           \nonumber \\
         &= o_p(1) \; .
   \end{align*}
   Thus we have $\widehat W=W_{NT}+o_p(1)=W+o_p(1)$.
\end{proof}

\begin{proof}[\bf Proof of Theorem \ref{th:biascorrection}, Part II: show $\widehat \Omega=\Omega+o_p(1)$.]
   $\phantom{a}$
   \\
    Let
    $\Omega_{NT} \equiv  \frac 1 {NT} \, \sum_{i=1}^N \, \sum_{t=1}^T
            \,   e_{it}^2   \, {\cal X}_{it} \, {\cal X}_{it}'$.
     We have
     $\Omega = \Omega_{NT} + o_P(1) 
     = \widehat \Omega + A_1 + A_2 + o_p(1) = \widehat \Omega + o_P(1)$, 
     where $A_1$ and $A_2$ are defined in Lemma \ref{lemma:A1A2},
     and the lemma states $A_1$ and $A_2$ are $o_p(1)$.
\end{proof}

\begin{proof}[\bf Proof of Theorem \ref{th:biascorrection}, Part III: show $\widehat B_1=B_1+o_p(1)$.]
  $\phantom{a}$
  \\
   Let $B_{1,k,NT}=N^{-1} \, {\rm Tr}\left[ P_{f^0} \, \mathbb{E}_{\cal C} \left( e' \, X_k   \right) \right]$.
   According to Assumption~\ref{ass:A6} we have $B_{1,k}=B_{1,k,NT}+o_p(1)$.
   What is left to show is $ B_{1,k,NT}=\widehat B_{1,k} + o_p(1)$.
   Using $\left| {\rm Tr}\left( C\right) \right| \leq \left\| C\right\| \limfunc{rank}\left( C\right)$ we find
   \begin{align*}
     \left|  B_{1,k,NT} - \widehat B_1 \right|
   &= \left| \mathbb{E}_{\cal C} \left[ \frac{1}{N} {\rm Tr}(P_{f^0} \, e^{\prime}\,   \, X_k)    \right]
         - \frac{1}{N} {\rm Tr}\left[ P_{\widehat f} \left( \widehat e' X_k \right)^{\rm truncR} \right] \right|
      \nonumber \\
   &\leq \left| \frac{1}{N} {\rm Tr}\left[
                \left( P_{f^0} - P_{\widehat f} \right) \left( \widehat e' X_k \right)^{\rm truncR} \right] \right|
      \nonumber \\ & \quad
    + \left| \frac{1}{N} {\rm Tr}\left\{
 P_{f^0} \left[ \mathbb{E}_{\cal C} \left(e^{\prime}\,  \, X_k  \right) - \left( \widehat e' X_k \right)^{\rm truncR} \right]
       \right\} \right|
      \nonumber \\
        &\leq \frac{2R}{N} \left\| P_{f^0} - P_{\widehat f} \right\| \,
                                   \left\| \left(\widehat e' X_k \right)^{\rm truncR} \right\|
      \nonumber \\ & \quad
             +\frac{R}{N}  \left\| P_{f^0} \right\| \,
             \left\| \mathbb{E}_{\cal C} \left(e^{\prime}\,  \, X_k   \right) - \left( \widehat e' X_k \right)^{\rm truncR}
             \right\| \; .
   \end{align*}
   We have $\left\| P_{f^0} \right\| = 1$.
   We now apply Lemmas \ref{lemma:exp}, \ref{lemma:Pfhat} and \ref{lemma:normsTrunc} to find
   \begin{align*}
      \left|  B_{1,k,NT} - \widehat B_1 \right|
        &=  N^{-1} \left( {\cal O}_p(N^{-1/2}) {\cal O}_p(M N T^{1/8}) + o_p(N) \right) = o_p(1) \; .
   \end{align*}
   This is what we wanted to show.
\end{proof}

\begin{proof}[\bf Proof of Theorem \ref{th:biascorrection}, final part: show $\widehat B_{2}=B_{2}+o_p(1)$ 
and $\widehat B_{3}=B_{3}+o_p(1)$.]
  $\phantom{a}$
  \\
   Define
   \begin{align*}
      B_{2,k,NT} &= \frac 1 T
         {\rm Tr}\left[ \mathbb{E}_{\cal C}\left( e e' \right) \, M_{\lambda^0} \, X_k \,
              f^0 \, (f^{0\prime}f^0)^{-1} \, (\lambda^{0\prime}\lambda^0)^{-1} \, \lambda^{0\prime} \right] \; .
   \end{align*}
   According to Assumption~\ref{ass:A6} we have $B_{2,k}=B_{2,k,NT}+o_p(1)$.
   What is left to show is $B_{2,k,NT}=\widehat B_{2,k} + o_p(1)$.
    We have
  \begin{align*}
      B_{2,k} - \widehat B_{2,k}
      =&
       \frac 1 T {\rm Tr}\left[ \mathbb{E}_{\cal C}\left( e   e' \right) \, M_{\lambda^0} \, X_k \,
              f^0 \, (f^{0\prime}f^0)^{-1} \, (\lambda^{0\prime}\lambda^0)^{-1} \, \lambda^{0\prime} \right]
       \nonumber \\ & \qquad
       -
       \frac 1 T
         {\rm Tr}\left[ \left( \widehat e \, {\widehat e}^{\prime} \right)^{\rm truncD} \, M_{\widehat \lambda} \, X_k \,
              \widehat f \, (\widehat f^{\prime} \widehat f)^{-1} \,
                             (\widehat \lambda^{\prime}\widehat \lambda)^{-1} \, \widehat \lambda^{\prime} \right]
      \nonumber \\
      =& \frac 1 T
         {\rm Tr}\left[ \left( \widehat e \, {\widehat e}^{\prime} \right)^{\rm truncD} \, M_{\widehat \lambda} \, X_k \,
           \left( f^0 \, (f^{0\prime}f^0)^{-1} \, (\lambda^{0\prime}\lambda^0)^{-1} \, \lambda^{0\prime}
           - \widehat f \, (\widehat f^{\prime} \widehat f)^{-1} \,
                             (\widehat \lambda^{\prime}\widehat \lambda)^{-1} \, \widehat \lambda^{\prime} \right) \right]
      \nonumber \\
       & + \frac 1 T
         {\rm Tr}\left[ \left( \widehat e \, {\widehat e}^{\prime} \right)^{\rm truncD} \,
                 \left( M_{\lambda^0} - M_{\widehat \lambda} \right) \, X_k \,
            f^0 \, (f^{0\prime}f^0)^{-1} \, (\lambda^{0\prime}\lambda^0)^{-1} \, \lambda^{0\prime}
           \right]
      \nonumber \\
       & + \frac 1 T
         {\rm Tr}\left\{ \left[ \mathbb{E}_{\cal C}\left( e   e' \right)
                           - \left( \widehat e \, {\widehat e}^{\prime} \right)^{\rm truncD} \right]\,
                  M_{\lambda^0}  \, X_k \,
            f^0 \, (f^{0\prime}f^0)^{-1} \, (\lambda^{0\prime}\lambda^0)^{-1} \, \lambda^{0\prime}
           \right\} \; .
   \end{align*}
   Using $\left| {\rm Tr}\left( C\right) \right| \leq \left\| C\right\| \limfunc{rank}\left( C\right)$
   (which is true for every square matrix $C$) we find
\begin{align*}
   \left|   B_{2,k} - \widehat B_{2,k} \right|
   \leq &
       \frac R T
         \left\| \left( \widehat e \, {\widehat e}^{\prime} \right)^{\rm truncD} \right\|
         \left\|  X_k \right\|
          \left\| f^0 \, (f^{0\prime}f^0)^{-1} \, (\lambda^{0\prime}\lambda^0)^{-1} \, \lambda^{0\prime}
           - \widehat f \, (\widehat f^{\prime} \widehat f)^{-1} \,
                             (\widehat \lambda^{\prime}\widehat \lambda)^{-1} \, \widehat \lambda^{\prime} \right\|
      \nonumber \\
       & + \frac R T
        \left\| \left( \widehat e \, {\widehat e}^{\prime} \right)^{\rm truncD} \right\|
                 \left\| M_{\lambda^0} - M_{\widehat \lambda} \right\| \left\| X_k \right\|
            \left\| f^0 \, (f^{0\prime}f^0)^{-1} \, (\lambda^{0\prime}\lambda^0)^{-1} \, \lambda^{0\prime} \right\|
      \nonumber \\
       & + \frac R T
            \left\| \mathbb{E}_{\cal C}\left( e  e' \right)
                           - \left( \widehat e \, {\widehat e}^{\prime} \right)^{\rm truncD} \right\| \,
                 \left\| X_k \right\| \,
            \left\| f^0 \, (f^{0\prime}f^0)^{-1} \, (\lambda^{0\prime}\lambda^0)^{-1} \, \lambda^{0\prime} \right\|
       \, .
\end{align*}
Here we used $\left\| M_{f^0} \right\| = \left\| M_{\widehat f} \right\| = 1$.
Using $\left\| X_k \right\|={\cal O}_p(\sqrt{NT})$, and applying
Lemmas~\ref{lemma:Pfhat}, \ref{lemma:lambdafINV}, \ref{lemma:exp} and \ref{lemma:normsTrunc}, we now find
\begin{align*}
   \left|   B_{2,k} - \widehat B_{2,k} \right|
      &= T^{-1}
       \bigg[ {\cal O}_p(T) \, {\cal O}_p((NT)^{1/2}) \, {\cal O}_p(N^{-3/2})
   \nonumber \\ & \qquad \qquad \qquad
            +  {\cal O}_p(T) \, {\cal O}_p(N^{-1/2}) \, {\cal O}_p((NT)^{1/2}) \, {\cal O}_p( (NT)^{-1/2})
   \nonumber \\ & \qquad \qquad \qquad
          +  o_p(T) \,  {\cal O}_p((NT)^{1/2}) \, {\cal O}_p( (NT)^{-1/2}) \bigg] = o_p(1) \; .
\end{align*}
This is what we wanted to show. The proof of $\widehat B_3=B_3+o_p(1)$ is analogous.
\end{proof}

\section{Proof of Intermediate Lemma}
\label{sec:Intermed}

Here we provide the proof of some intermediate lemmas that were stated and
used in section~\ref{app:ProofThBias}.
The following lemma gives a useful bound on the maximum of (correlated) random variables
\begin{lemma}
   \label{lemma:maxRV}
   Let $Z_i$, $i=1,2,\ldots,n$, be $n$ real valued random variables, and
   let $\gamma \geq 1$ and $B>0$ be finite constants (independent of $n$).
   Assume $\max_i \, \mathbb{E}_{\cal C} |Z_i|^{\gamma} \leq B$,
   i.e., the $\gamma$'th moment of the $Z_i$ are finite and uniformly bounded.
   For $n \rightarrow \infty$ we then have
   \begin{align}
      \max_i |Z_i| &= {\cal O}_p\left( n^{1/\gamma} \right) \; .
      \label{eq:Zbound}
   \end{align}
\end{lemma}

\begin{proof}[\bf Proof]
   Using Jensen's inequality one obtains
   $\mathbb{E}_{\cal C} \max_i |Z_i| \leq \left( \mathbb{E}_{\cal C} \max_i |Z_i|^\gamma \right)^{1/\gamma}
    \leq \left( \mathbb{E}_{\cal C} \sum_{i=1}^n |Z_i|^\gamma \right)^{1/\gamma}
    \leq \left( n \, \max_i \mathbb{E}_{\cal C} |Z_i|^\gamma \right)^{1/\gamma}
    \leq n^{1/\gamma} \, B^{1/\gamma}$.
   Markov's inequality then gives equation \eqref{eq:Zbound}.
\end{proof}

\begin{lemma}
   \label{lemma:barZ}
   Let
   \begin{align*}
      \bar Z^{(1)}_{k,t\tau}
            &= N^{-1/2} \sum_{i=1}^N \left[ e_{it} X_{k,i\tau} - \mathbb{E}_{\cal C} \left( e_{it} X_{k,i\tau} \right) \right] \; ,
      \nonumber \\
      \bar Z^{(2)}_{t}
            &= N^{-1/2} \sum_{i=1}^N \left[ e_{it}^2 - \mathbb{E}_{\cal C} \left( e_{it}^2 \right) \right] \; ,
      \nonumber \\
      \bar Z^{(3)}_{i}
            &= T^{-1/2} \sum_{t=1}^T \left[ e_{it}^2 - \mathbb{E}_{\cal C} \left( e_{it}^2 \right) \right] \; .
   \end{align*}
   Under assumption \ref{ass:A5} we have
   \begin{align*}
      \mathbb{E}_{\cal C} \left| \bar Z^{(1)}_{k,t\tau} \right|^4 &\leq B \; ,
      \nonumber \\
      \mathbb{E}_{\cal C} \left| \bar Z^{(2)}_{t\tau} \right|^4 &\leq B \; ,
      \nonumber \\
      \mathbb{E}_{\cal C} \left| \bar Z^{(3)}_{i} \right|^4 &\leq B \; ,
   \end{align*}
   for some $B>0$, i.e., the 
   conditional expectations $\bar Z^{(1)}_{k,t\tau}$, $\bar Z^{(2)}_{t\tau}$, and $\bar Z^{(3)}_{i}$
   are uniformly bounded over $t,\tau$, or $i$, respectively.
\end{lemma}

\begin{proof}[\bf Proof]
   \# We start with the proof for $\bar Z^{(1)}_{k,t\tau}$.
      Define
      $Z^{(1)}_{k,t\tau,i} = e_{it} X_{k,i\tau} - \mathbb{E}_{\cal C} \left( e_{it} X_{k,i\tau} \right)$.
      By assumption we have finite 8th moments for $e_{it}$ and $X_{k,i\tau}$ uniformly
      across $k,i,t,\tau$, and thus (using Cauchy Schwarz inequality) we have finite 4th moment
      of $Z^{(1)}_{k,t\tau,i}$ uniformly across $k,i,t,\tau$. For ease of notation we now
      fix $k,t,\tau$ and write $Z_i=Z^{(1)}_{k,t\tau,i}$.
      We have $\mathbb{E}_{\cal C}(Z_i)=0$ and
      $\mathbb{E}_{\cal C}(Z_{i} Z_{j} Z_{k} Z_{l})=0$ if $i \notin \{j,k,l\}$
      (and the same holds for permutations of $i,j,k,l$).
      Using this we compute
     \begin{align*}
        \mathbb{E}_{\cal C} \left( \sum_{i=1}^N \, Z_i \right)^4
         &= \sum_{i,j,k,l=1}^N \, \mathbb{E}_{\cal C}\left( Z_i Z_j Z_k Z_l\right)
        \nonumber \\
         &= 3 \, \sum_{i\neq j} \, \mathbb{E}_{\cal C}\left( Z_i^2 \, Z_j^2\right)
            + \sum_{i} \mathbb{E}_{\cal C}\left( Z_i^4 \right)
        \nonumber \\
         &= 3 \, \sum_{i,j=1}^N \, \mathbb{E}_{\cal C}\left( Z_i^2\right) \, \mathbb{E}_{\cal C}\left( Z_j^2\right)
    + \sum_{i=1}^N \, \left\{ \mathbb{E}_{\cal C}\left( Z_i^4 \right) - 3 \left[ \mathbb{E}_{\cal C}\left(Z_i^2\right) \right]^2 \right\}
       \; ,
     \end{align*}
     Because we argued $\mathbb{E}_{\cal C}\left( Z_i^4 \right)$ is bounded uniformly, the last equation shows
     $\bar Z^{(1)}_{k,t\tau} = N^{-1/2} \sum_{i=1}^N \, Z^{(1)}_{k,t\tau,i}$ is bounded uniformly
     across $k,t,\tau$. This is what we wanted to show.

   \# The proofs for $\bar Z^{(2)}_{t}$ and $\bar Z^{(3)}_{i}$ are analogous.
 \end{proof}

\begin{lemma}
   \label{Lemma:Trunc}
   For a $T\times T$ matrix $A$ we have\footnote{For the boundaries of $\tau$ we could write
   $\max(1,t-M)$ instead of $t-M$, and $\min(T,t+M)$ instead of $t+M$, to guarantee $1\leq \tau \leq T$.
   Since this would complicate notation, we prefer the convention $A_{t\tau}=0$ for $t<1$ or $\tau<1$
   of $t>T$ or $\tau>T$.}
   \begin{align*}
      \left\| A^{\rm truncR} \right\| \, &\leq  \, M \left\| A^{\rm truncR} \right\|_{\max}
         \, \equiv \, M \, \max_t \, \max_{t<\tau\leq t+M} |A_{t\tau}| \, ,
   \end{align*}
\end{lemma}

\begin{proof}[\bf Proof]
   For the $1$-norm of $A^{\rm truncR}$ we find
   \begin{align*}
     \left\| A^{\rm truncR} \right\|_1  &= \max_{t=1\ldots T} \, \sum_{\tau=t+1}^{t+M} \, |A_{t\tau}|
       \nonumber \\
            &\leq M \,  \max_{t<\tau\leq t+M} \, |A_{t\tau}| = M \left\| A^{\rm truncR} \right\|_{\max}  \; ,
   \end{align*}
   and analogously we find the same bound for the $\infty$-norm $\left\| A^{\rm truncR} \right\|_\infty$.
   Applying part (vii) of Lemma~\ref{lemma:inequalities} we therefore also get this bound for the operator norm
   $\| A^{\rm truncR} \|$.
\end{proof}

\begin{proof}[\bf Proof of Lemma~\ref{lemma:A1A2}]
     \# We first show $A_1\equiv (NT)^{-1} \sum_{i=1}^N \sum_{t=1}^T e_{it}^2
               \left( {\cal X}_{it} {\cal X}_{it}' - \widehat {\cal X}_{it} \widehat {\cal X}_{it}' \right)
                     = o_p(1)$.
     Let $B_{1,it} =  {\cal X}_{it}- \widehat {\cal X}_{it}$,
     $B_{2,it} = e_{it}^2  {\cal X}_{it}$, and
     $B_{3,it} = e_{it}^2  \widehat {\cal X}_{it}$.
     Note $B_1$, $B_2$, and $B_3$ can either be viewed as $K$-vectors for each
     pair $(i,t)$, or equivalently as $N\times T$ matrices $B_{1,k}$, $B_{2,k}$, and $B_{3,k}$ for each $k=1,\ldots,K$.
     We have $A_1 = (NT)^{-1} \sum_i \sum_t \left( B_{1,it} B_{2,it}' +  B_{3,it} B_{1,it}' \right)$,
     or equivalently
     \begin{align*}
        A_{1,k_1 k_2} &= \frac 1 {NT} {\rm Tr}\left( B_{1,k_1} B_{3,k_2}' + B_{2,k_1} B_{1,k_2}' \right) \; .
     \end{align*}
     Using $\|M_{\widehat \lambda} - M_{\lambda^0}\| = {\cal O}_p(N^{-1/2})$,
     $\|M_{\widehat f} - M_{f^0}\| = {\cal O}_p(N^{-1/2})$,
     $\|X_k\| = {\cal O}_p(\sqrt{NT}) = {\cal O}_p(N)$,
     we find for $B_{1,k}  =
                     (M_{\lambda^0} - M_{\widehat \lambda}) X_k M_{f^0}
                    + M_{\widehat \lambda} X_k (M_{f^0} - M_{\widehat f}) $
    that $\|B_{1,k}\| = {\cal O}_p(N^{1/2})$.
     In addition we have ${\rm rank}(B_{1,k}) \leq 4R$.
      We also have
      \begin{align*}
        \| B_{2,k} \|^4 &\leq \| B_{2,k} \|_F^4
                   \nonumber \\
            &= \left( \sum_{i=1}^N \sum_{t=1}^T e_{it}^4  {\cal X}_{k,it}^2 \right)^2
                   \nonumber \\
            &\leq \left( \sum_{i=1}^N \sum_{t=1}^T e_{it}^8 \right)
                  \left( \sum_{i=1}^N \sum_{t=1}^T {\cal X}_{k,it}^4 \right)
             = {\cal O}_p(NT) \, {\cal O}_p(NT) \, ,
     \end{align*}
     which implies $\| B_{2,k} \|={\cal O}_p(\sqrt{NT})$, and analogously we find $\| B_{3,k} \|={\cal O}_p(\sqrt{NT})$.
     Therefore
     \begin{align*}
        | A_{1,k_1 k_2} | &\leq \frac {4R} {NT}
                      \left( \| B_{1,k_1} \| \| B_{3,k_2}\| + \|B_{2,k_1}\| \| B_{1,k_2}\| \right)
        \nonumber \\
           &= \frac {4R} {NT}
                      \left( {\cal O}_p(N^{1/2}) {\cal O}_p(\sqrt{NT}) + {\cal O}_p(\sqrt{NT}) {\cal O}_p(N^{1/2}) \right)
            = o_p(1) \; .
     \end{align*}
     This is what we wanted to show.

     \# Finally, we want to show $A_2 \equiv (NT)^{-1} \sum_{i=1}^N \sum_{t=1}^T \left( e_{it}^2 - \widehat e_{it}^2 \right)
                                               \widehat {\cal X}_{it} \widehat {\cal X}_{it}' = o_p(1) $.
     According to theorem \ref{theorem:expansions} we have $e - \widehat e = C_1 + C_2$,
     where we defined $C_1 = - \sum_{k=1}^K \left( \widehat \beta_k - \beta^0_k \right) \, X_k $,
     and $C_2=\sum_{k=1}^K \left( \widehat \beta_k - \beta^0_k \right)
               \left( P_{\lambda^0} \, X_k \, M_{f^0} +  X_k \, P_{f^0} \right)
             + P_{\lambda^0} \, e \, M_{f^0} + e \, P_{f^0}
            - \widehat e^{(1)}_e - \widehat e^{({\rm rem})}$,
     which satisfies $\|C_2\|={\cal O}_p(N^{1/2})$, and ${\rm rank}(C_2) \leq 11 R$
     (actually, one can easily prove $\leq 5R$, but this does not follow from theorem \ref{theorem:expansions}).
     Using this notation we have
     \begin{align*}
         A_2 &= \frac 1 {NT} \sum_{i=1}^N \sum_{t=1}^T
                 (e_{it}+\widehat e_{it}) (C_{1,it} + C_{2,it})
                                               \widehat {\cal X}_{it} \widehat {\cal X}_{it}' \; ,
     \end{align*}
     which can also be written as
     \begin{align*}
        A_{2,k_1 k_2} &= - \,\sum_{k_3=1}^K \left( \widehat \beta_{k_3} - \beta^0_{k_3} \right)
                        \left( C_{5,k_1 k_2 k_3} + C_{6,k_1 k_2 k_3} \right)
                         + \frac 1 {NT} {\rm Tr}\left( C_2 \, C_{3,k_1 k_2}  \right)
                         + \frac 1 {NT} {\rm Tr}\left( C_2 \, C_{4,k_1 k_2}  \right) \; ,
     \end{align*}
     where we defined
     \begin{align*}
        C_{3,k_1 k_2,it} &= e_{it} \widehat {\cal X}_{k_1,it} \widehat {\cal X}_{k_2,it} \; ,
        \nonumber \\
        C_{4,k_1 k_2,it} &=\widehat e_{it} \widehat {\cal X}_{k_1,it} \widehat {\cal X}_{k_2,it} \; ,
        \nonumber \\
        C_{5,k_1 k_2 k_3}
       &=
 \frac 1 {NT} \sum_{i=1}^N \sum_{t=1}^T \, e_{it} \widehat {\cal X}_{k_1,it} \widehat {\cal X}_{k_2,it} X_{k_3,it} \; ,
        \nonumber \\
    C_{6,k_1 k_2 k_3}
   &=\frac 1 {NT} \sum_{i=1}^N \sum_{t=1}^T \, \widehat e_{it} \widehat {\cal X}_{k_1,it} \widehat {\cal X}_{k_2,it} X_{k_3,it} \; .
     \end{align*}
     Again, because we have uniformly bounded $8$th moments for $e_{it}$ and $X_{k,it}$, we find
     \begin{align*}
        \| C_{3,k_1 k_2} \|^4 &\leq \| C_{3,k_1 k_2} \|_F^4
                  \nonumber \\
                &= \left( \sum_{i=1}^N \sum_{t=1}^T e_{it}^2 \widehat {\cal X}_{k_1,it}^2 \widehat {\cal X}_{k_2,it}^2
                   \right)^2
                  \nonumber \\
                &\leq  \left( \sum_{i=1}^N \sum_{t=1}^T e_{it}^4 \right)
             \left( \sum_{i=1}^N \sum_{t=1}^T \widehat {\cal X}_{k_1,it}^4 \widehat {\cal X}_{k_2,it}^4 \right)
                  \nonumber \\
                &= {\cal O}_p(N^2 T^2) \; ,
     \end{align*}
     i.e., $\| C_{3,k_1 k_2} \|={\cal O}_p(\sqrt{NT})$. Furthermore
     \begin{align*}
        \| C_{4,k_1 k_2} \|^2 &\leq \| C_{3,k_1 k_2} \|_F^2
                  \nonumber \\
                &= \sum_{i=1}^N \sum_{t=1}^T \widehat e_{it}^2 \widehat {\cal X}_{k_1,it}^2 \widehat {\cal X}_{k_2,it}^2
                  \nonumber \\
                &\leq \left( \sum_{i=1}^N \sum_{t=1}^T \widehat e_{it}^2 \right)
                       \max_{i=1\ldots N} \max_{t=1\ldots T}
                       \left( \widehat {\cal X}_{k_1,it}^2 \widehat {\cal X}_{k_2,it}^2 \right)
                  \nonumber \\
                &\leq \left( \sum_{i=1}^N \sum_{t=1}^T e_{it}^2 \right)
                       \max_{i=1\ldots N} \max_{t=1\ldots T}
                       \left( \widehat {\cal X}_{k_1,it}^2 \widehat {\cal X}_{k_2,it}^2 \right)
                  \nonumber \\
                &= {\cal O}_p(NT) {\cal O}_p((NT)^{(4/(8+\epsilon))}) = o_p((NT)^{(3/4)}) \; .
     \end{align*}
     Here we used the assumption that $X_k$ has uniformly bounded moments of order
     $8+\epsilon$ for some $\epsilon>0$. We also used
         $\sum_{i=1}^N \sum_{t=1}^T \widehat e_{it}^2 \leq \sum_{i=1}^N \sum_{t=1}^T e_{it}^2$.

     For $C_5$ we find
     \begin{align*}
        C_{5,k_1 k_2 k_3}^2
       &\leq
       \left( \frac 1 {NT} \sum_{i=1}^N \sum_{t=1}^T \, e_{it}^2 \right)
       \left( \frac 1 {NT} \widehat {\cal X}_{k_1,it}^2 \widehat {\cal X}_{k_2,it}^2 X_{k_3,it}^2 \right)
       \nonumber \\
       &= {\cal O}_p(1) \; ,
     \end{align*}
     i.e., $C_{5,k_1 k_2 k_3} = {\cal O}_p(1)$, and analogously $C_{6,k_1 k_2 k_3} = {\cal O}_p(1)$, because
     $\sum_{i=1}^N \sum_{t=1}^T \widehat e_{it}^2 \leq \sum_{i=1}^N \sum_{t=1}^T e_{it}^2$.

     Using these results we obtain
     \begin{align*}
        | A_{2,k_1 k_2} | &\leq - \,\sum_{k_3=1}^K \left\| \widehat \beta_{k_3} - \beta^0_{k_3} \right\|
                        \left| C_{5,k_1 k_2 k_3} + C_{6,k_1 k_2 k_3} \right|
                         + \frac {11 R} {NT} \|C_2\| \|C_{3,k_1 k_2} \|
                         + \frac {11 R} {NT} \| C_2\| \| C_{4,k_1 k_2} \|
               \nonumber \\
                      &= {\cal O}_p((NT)^{-1/2}) {\cal O}_p(1)
                           + \frac {11 R} {NT} {\cal O}_p(N^{1/2}) {\cal O}_p(\sqrt{NT})
                           + \frac {11 R} {NT} {\cal O}_p(N^{1/2}) o_p((NT)^{3/4})
                       = o_p(1) \; .
     \end{align*}
     This is what we wanted to show.
\end{proof}

Remember, the truncation Kernel $\Gamma(.)$ is defined by $\Gamma(x)=1$ for $|x|\leq 1$ and $\Gamma(x)=0$ otherwise.
Without loss of generality we assume in the following the bandwidth parameter $M$ is a positive integer
(without this assumption, one needs to replace $M$ everywhere below by the largest integer contained in $M$,
but nothing else changes).

\begin{proof}[\bf Proof of Lemma \ref{lemma:lambdafINV}]
  By Lemma~\ref{lemma:Pfhat} we know asymptotically $P_{\widehat f}$ is close to $P_{f^0}$ and therefore
  ${\rm rank}(P_{\widehat f}P_{f^0})={\rm rank}(P_{f^0}P_{f^0})=R$ , i.e.,
  ${\rm rank}(P_{\widehat f}f^0)=R$ asymptotically. We can therefore write
  $\widehat f = P_{\widehat f} f^0 H$, where $H=H_{NT}$ is a non-singular $R\times R$ matrix.

  We now want to show $\|H\|=O_p(1)$ and $\|H^{-1}\|=O_p(1)$.
  Because of our normalization of $\widehat f$ and $f^0$ we have $H=(\widehat f' P_{\widehat f} f^0/T)^{-1}=(\widehat f' f^0/T)^{-1}$,
  and therefore $\|H^{-1}\|\leq \|\widehat f\| \|f^0\| /T ={\cal O}_p(1)$.
  We also have $\widehat f = f^0 H + (P_{\widehat f}-P_{f^0}) f^0 H$,
  and thus $H=f^{0\prime} \widehat f/T - f^{0\prime} (P_{\widehat f}-P_{f^0}) f^0 H /T$,
  i.e., $\|H\| \leq {\cal O}_p(1) + \|H\| {\cal O}_p\left(T^{-1/2}\right)$
  which shows $\|H\|={\cal O}_p(1)$.
  Note all the following results only require $\|H\|={\cal O}_p(1)$ and $\|H^{-1}\|={\cal O}_p(1)$, but apart from that
  are independent of the choice of normalization.

  The advantage of expressing $\widehat f$ in terms of $P_{\widehat f}$ as above is that
  the result $\left\| P_{\widehat f} - P_{f^0} \right\| = {\cal O}_p\left(T^{-1/2}\right)$
  of Lemma~\ref{lemma:Pfhat} immediately implies
  \begin{align*}
     \left\| \widehat f - f^0 \, H \right\| &= {\cal O}_p\left(1\right) \; .
  \end{align*}
  The FOC wrt $\lambda$ in the minimization of the first line in equation \eqref{LNT123} reads
  \begin{align}
     \widehat \lambda \, \widehat f' \widehat f &= \left(Y-\sum_{k=1}^{K} \widehat \beta_{k} X_{k} \right) \widehat f \; ,
     \label{hatf_close}
  \end{align}
  which yields
  \begin{align*}
    \widehat \lambda &= \left[ \lambda^0 f^{0\prime} - \sum_{k=1}^{K} \left( \widehat \beta_{k} - \beta^0_k \right) X_{k} \right]
                   \widehat f \left(\widehat f' \widehat f\right)^{-1}
      \nonumber \\
                &= \left[ \lambda^0 f^{0\prime} + \sum_{k=1}^{K} \left( \beta^0_k - \widehat \beta_{k} \right) X_{k} + e\right]
                   P_{\widehat f} f^0 \left(f^{0\prime} P_{\widehat f} f^0\right)^{-1} \, \left(H'\right)^{-1}
      \nonumber \\
                &= \lambda^0 \, \left(H'\right)^{-1}
            + \lambda^0 f^{0\prime} \left( P_{\widehat f}-P_{f^0} \right) f^0
                                    \left(f^{0\prime} P_{\widehat f} f^0\right)^{-1} \, \left(H'\right)^{-1}
      \nonumber \\ & \qquad \qquad \qquad
            + \lambda^0 f^{0\prime} f^0
                 \left[ \left(f^{0\prime} P_{\widehat f} f^0\right)^{-1} - \left(f^{0\prime} f^0\right)^{-1} \right]
                       \, \left(H'\right)^{-1}
      \nonumber \\ & \qquad \qquad \qquad
            + \left[ \sum_{k=1}^{K} \left( \beta^0_k - \widehat \beta_{k} \right) X_{k} + e\right]
                   P_{\widehat f} f^0 \left(f^{0\prime} P_{\widehat f} f^0\right)^{-1} \, \left(H'\right)^{-1} \; .
  \end{align*}
  We have
  $\left(f^{0\prime} P_{\widehat f} f^0\right/T)^{-1} - \left(f^{0\prime} f^0/T\right)^{-1}={\cal O}_p(T^{-1/2})$, because
  $\left\| P_{\widehat f} - P_{f^0} \right\| = {\cal O}_p\left(T^{-1/2}\right)$
  and $f^{0\prime} f^0/T$ by assumption is converging to a positive definite matrix (or given our particular choice of
  normalization is just the identity matrix $\mathbb{I}_R$).
  In addition, we have $\|e\|={\cal O}_p(\sqrt{T})$, $\|X_k\|={\cal O}_p(\sqrt{NT})$ and
  by corollary \ref{lemma:sqrtNTcons} also $\|\widehat \beta - \beta^0\|={\cal O}_p(1/\sqrt{NT})$. Therefore
  \begin{align}
     \left\| \widehat \lambda - \lambda^0 \, \left(H'\right)^{-1} \right\| &= {\cal O}_p\left(1\right) \; ,
     \label{hatlambda_close}
  \end{align}
  which is what we wanted to prove.

  Next, we want to show
  \begin{align}
     \label{lambdafsquare}
     \left\| \left( \frac{\widehat \lambda^{\prime} \, \widehat \lambda} N \right)^{-1}
   - \left( \frac{ \left(H\right)^{-1} \,\lambda^{0\prime} \, \lambda^0 \, \left(H'\right)^{-1}} N \right)^{-1} \right\|
       = {\cal O}_p\left(N^{-1/2}\right) \; ,
     \nonumber \\
     \left\| \left( \frac{\widehat f^{\prime} \, \widehat f} T \right)^{-1}
   -  \left( \frac{H' \,f^{0\prime} \, f^0 \,  H} T \right)^{-1} \right\|
       = {\cal O}_p\left(T^{-1/2}\right) \; .
  \end{align}
  Let $A=N^{-1} \, \widehat \lambda^{\prime} \, \widehat \lambda$ and
  $B=N^{-1} \, \left(H\right)^{-1} \, \lambda^{0\prime} \, \lambda^0 \, \left(H'\right)^{-1}$.
  Using \eqref{hatlambda_close} we find
  \begin{align*}
     \| A-B \| &= \frac 1 {2N} \left\|
                         \left[ \widehat \lambda^{\prime} + \left(H\right)^{-1} \, \lambda^{0\prime} \right]
                         \left[ \widehat \lambda - \lambda^0 \, \left(H'\right)^{-1} \right]
                        +\left[ \widehat \lambda^{\prime} - \left(H\right)^{-1} \, \lambda^{0\prime} \right]
                         \left[ \widehat \lambda + \lambda^0 \, \left(H'\right)^{-1} \right] \right\|
           \nonumber \\
               &= N^{-1} \, {\cal O}_p(N^{1/2}) \, {\cal O}_p(1) = {\cal O}_p \left( N^{-1/2} \right) \; .
  \end{align*}
  By assumption \ref{ass:A1} we know
  \begin{align*}
      \left\| \left( \frac{\lambda^{0\prime} \, \lambda^0} N \right)^{-1} \right\| &= {\cal O}_p(1) \; ,
  \end{align*}
  and thus also $\left\|B^{-1}\right\|= {\cal O}_p(1)$, and therefore $\left\|A^{-1}\right\|= {\cal O}_p(1)$
  (using $\| A-B \|=o_p(1)$ and applying Weyl's inequality to the smallest eigenvalue of $B$).
  Because $A^{-1} - B^{-1} = A^{-1} (B-A) B^{-1}$ we find
  \begin{align*}
     \left\| A^{-1} - B^{-1} \right\| &\leq  \left\|A^{-1}\right\| \, \left\|B^{-1}\right\| \, \left\|A-B\right\|
                        \nonumber \\
                                      &= {\cal O}_p\left( N^{-1/2} \right) \; .
  \end{align*}
  Thus, we have shown the first statement of \eqref{lambdafsquare}, and analogously one can show the second one.
  Combining \eqref{hatlambda_close}, \eqref{hatf_close} and \eqref{lambdafsquare} we obtain
  \begin{align*}
     & \left\| \frac{\widehat \lambda}{\sqrt{N}} \, \left( \frac{\widehat \lambda^{\prime}\widehat\lambda}{N} \right)^{-1}
            \, \left(\frac{\widehat f^{\prime}\widehat f}T\right)^{-1} \, \frac{\widehat f^{\prime}}{\sqrt{T}}
        - \frac{\lambda^0}{\sqrt{N}} \, \left(\frac{\lambda^{0\prime}\lambda^0}N\right)^{-1} \,
             \left(\frac{f^{0\prime}f^0}T\right)^{-1} \, \frac{f^{0\prime}} {\sqrt{T}} \right\|
    \nonumber \\  & =
       \left\| \frac{\widehat \lambda}{\sqrt{N}} \, \left( \frac{\widehat \lambda^{\prime}\widehat\lambda}{N} \right)^{-1}
            \, \left(\frac{\widehat f^{\prime}\widehat f}T\right)^{-1} \, \frac{\widehat f^{\prime}}{\sqrt{T}}
        - \frac{\lambda^0 \left(H'\right)^{-1}}{\sqrt{N}} \,
                \left(\frac{\left(H\right)^{-1}\lambda^{0\prime}\lambda^0 \left(H'\right)^{-1}}N\right)^{-1} \,
             \left(\frac{H' f^{0\prime}f^0 H}T\right)^{-1} \, \frac{H' f^{0\prime}} {\sqrt{T}} \right\|
    \nonumber \\  & \qquad \qquad =
       {\cal O}_p\left( N^{-1/2} \right) \; ,
  \end{align*}
  which is equivalent to the statement in the lemma.
  Note also
  $\widehat \lambda \, (\widehat \lambda^{\prime}\widehat\lambda)^{-1} \, (\widehat f^{\prime}\widehat f)^{-1} \, \widehat f^{\prime}$
  is independent of $H$, i.e., independent of the choice of normalization.
\end{proof}

\begin{proof}[\bf Proof of Lemma~\ref{lemma:exp}]
   \# Part A of the proof: We start by showing
   \begin{align}
       N^{-1} \, \left\|
           \mathbb{E}_{\cal C}\left[ e'  X_k - \left( e'  X_k \right)^{\rm truncR} \right] \right\|
                         &= o_p(1) \; .
       \label{exp_proof_part1}
   \end{align}
   Let $A=e'  X_k$ and $B=A-A^{\rm truncR}$.
   By definition of the left-sided truncation (using the truncation kernel $\Gamma(.)$ defined above)
   we have $B_{t\tau}=0$ for $t<\tau\leq t+M$ and $B_{t\tau}=A_{t\tau}$ otherwise.
   By assumption~\ref{ass:A5} we have $\mathbb{E}_{\cal C}(A_{t\tau})=0$ for $t \geq \tau$.
   For $t<\tau$ we have
   $\mathbb{E}_{\cal C}(A_{t\tau}) = \sum_{i=1}^N  \mathbb{E}_{\cal C}(e_{it}  X_{k,i\tau})$.
   We thus have $\mathbb{E}_{\cal C}(B_{t\tau})=0$ for $\tau\leq t+M$,
   and $\mathbb{E}_{\cal C} B_{t\tau} =  \sum_{i=1}^N  \mathbb{E}_{\cal C}(e_{it}  X_{k,i\tau}) $ for $\tau > t+M$.
   Therefore
   \begin{align*}
     \left\| \mathbb{E}_{\cal C}(B) \right\|_1  &= \max_{t=1\ldots T} \, \sum_{\tau=1}^T |\mathbb{E}_{\cal C}(B_{t\tau})|
                         \nonumber \\
                           &\leq    \max_{t=1\ldots T} \, \sum_{\tau=t+M+1}^T \,
                           \left|   \sum_{i=1}^N  \mathbb{E}_{\cal C}(e_{it}  X_{k,i\tau}) \right|
                     \leq     N  \max_{t=1\ldots T} \, \sum_{\tau=t+M+1}^T    c \, (\tau - t)^{- (1+ \epsilon)}
                            = o_p(N) \; ,
   \end{align*}
   where we used $M\rightarrow \infty$.
   Analogously we can show $\left\| \mathbb{E}_{\cal C}(B) \right\|_\infty  = o_p(N)$.
   Using part (vii) of Lemma~\ref{lemma:inequalities} we therefore also find
   $\left\| \mathbb{E}_{\cal C}(B) \right\| = o_p(N)$,
   which is equivalent to equation \eqref{exp_proof_part1} we wanted to show in this part of the proof.
   Analogously we can show
   \begin{align*}
       N^{-1} \, \left\|
           \mathbb{E}_{\cal C}\left[ e'  e - \left( e'  e \right)^{\rm truncD} \right] \right\|
                         &= o_p(1) \; ,
       \nonumber \\
       T^{-1} \, \left\|
           \mathbb{E}_{\cal C}\left[ e   e' - \left( e   e' \right)^{\rm truncD} \right] \right\|
                         &= o_p(1) \; .
   \end{align*}

   \# Part B of the proof: Next, we want to show
   \begin{align}
       N^{-1} \, \left\|
            \left[ e'  X_k \, - \,
                    \mathbb{E}_{\cal C}\left(e'  X_k\right) \right]^{\rm truncR} \right\|
                         &= o_p(1) \; .
       \label{exp_proof_part2}
   \end{align}
   Using Lemma~\ref{Lemma:Trunc} we have
   \begin{align*}
       N^{-1} \left\| \left[ e'  X_k \, - \,
                    \mathbb{E}_{\cal C}\left(e'  X_k\right) \right]^{\rm truncR} \right\|
          &\leq M \, \max_t \, \max_{t<\tau\leq t+M} \, N^{-1} \, \left| e_t'  X_{k,\tau}
                                                - \mathbb{E}_{\cal C} \left( e_t'  X_{k,\tau} \right) \right|
         \nonumber \\
          &\leq M \, \max_t \, \max_{t<\tau\leq t+M} \, N^{-1} \, \left| \sum_{i=1}^N \left[ e_{it} X_{k,i\tau}
                                       - \mathbb{E}_{\cal C} \left( e_{it} X_{k,i\tau} \right) \right] \right|
         \nonumber \\
          &\leq M \, N^{-1/2} \, \max_t \, \max_{t<\tau\leq t+M} \, \left| \bar Z^{(1)}_{k,t\tau} \right| .
   \end{align*}
   According to Lemma~\ref{lemma:barZ} we know $\mathbb{E}_{\cal C} \left| \bar Z^{(1)}_{k,t\tau} \right|^4$
   is bounded uniformly across $t$ and $\tau$. Applying Lemma~\ref{lemma:maxRV} we therefore find
   $\max_t \, \max_{t<\tau\leq t+M} \bar Z^{(1)}_{t\tau} = {\cal O}_p((MT)^{1/4})$. Thus we have
   \begin{align*}
      M \, N^{-1/2} \, \max_t \, \max_{t<\tau\leq t+M} \, \left| \bar Z^{(1)}_{t\tau} \right|
           &= {\cal O}_p\left(M \, N^{-1/2} \, (MT)^{1/4}\right) \, = \, o_p(1) \; .
   \end{align*}
   Here we used $M^5/T \rightarrow 0$.
   Analogously we can show
   \begin{align*}
       N^{-1} \, \left\|
            \left[ e'  e \, - \,
                    \mathbb{E}_{\cal C} \left(e'  e\right) \right]^{\rm truncD} \right\|
                         &= o_p(1) \; ,
       \nonumber \\
       T^{-1} \, \left\|
            \left[ e  e' \, - \,
                    \mathbb{E}_{\cal C}  \left(e  e'\right) \right]^{\rm truncD} \right\|
                         &= o_p(1) \; .
   \end{align*}

   \# Part C of the proof: Finally, we want to show
   \begin{align}
       N^{-1} \, \left\|
            \left[ e'  X_k \, - \,
                   \widehat e' \, X_k  \right]^{\rm truncR} \right\|
                         &= o_p(1) \; .
       \label{exp_proof_part3}
   \end{align}
   According to theorem \ref{theorem:expansions} we have
   $\widehat e =M_{\lambda^0}  e  M_{f^0}   + e_{\rm rem}$,
   where $e_{\rm rem} \equiv \widehat e^{(1)}_e
            - \sum_{k=1}^K \left( \widehat \beta_k - \beta^0_k \right) \widehat e^{(1)}_k
            + \widehat e^{({\rm rem})}$.
   We then have
   \begin{align*}
      & N^{-1} \, \left\|
            \left[ e'  X_k \, - \,
                   \widehat e' \, X_k  \right]^{\rm truncR} \right\|
      \\
       &\leq             N^{-1} \, \left\|
            \left[ e_{\rm rem}'  X_k    \right]^{\rm truncR} \right\|
          +   N^{-1} \, \left\|
            \left[ P_{f^0} e'  M_{\lambda^0}  X_k   \right]^{\rm truncR} \right\|
          +     N^{-1} \, \left\|
            \left[ e'  P_{\lambda^0} X_k   \right]^{\rm truncR} \right\| .
   \end{align*}
  Using
   corollary \ref{lemma:sqrtNTcons} we find the remainder term satisfies
   $\| e_{\rm rem} \| = {\cal O}_p(1)$.     Using Lemma~\ref{Lemma:Trunc} we find
   \begin{align*}
       N^{-1} \, \left\|
            \left[  e'_{\rm rem} \, X_k  \right]^{\rm truncR} \right\|
          &= \frac M N \, \max_{t,\tau} \, \widehat e'_{{\rm rem},t} \, X_{k,\tau}
      \nonumber \\
          &\leq \frac M N \, \max_{t,\tau} \, \|  e_{{\rm rem},t} \| \,  \| X_{k,\tau} \|
      \nonumber \\
          &\leq \frac M N \, \|   e_{\rm rem} \| \, \max_{\tau}  \| X_{k,\tau} \|
      \nonumber \\
          &\leq \frac M N  {\cal O}_p(1) {\cal O}_p(N^{1/2} T^{1/8}) = o_p(1) \; ,
   \end{align*}
   where we used the fact that the norm of each column $ e_{{\rm rem},t}$
   is smaller than the operator norm of the whole matrix $ e_{\rm rem}$.
   In addition we used Lemma~\ref{lemma:maxRV}
   and the fact that $N^{-1/2} \, \| X_{k,\tau} \| = \sqrt{ N^{-1} \sum_{i=1}^N X_{k,i\tau}^2 }$
   has finite 8'th moment to show $\max_{\tau}  \| X_{k,\tau} \| = {\cal O}_p(N^{1/2} T^{1/8})$.
      Using again Lemma~\ref{Lemma:Trunc} we find
   \begin{align*}
      N^{-1} \, \left\| \left[ P_{f^0} e' M_{\lambda^0} X_k  \right]^{\rm truncR} \right\|
        &\leq  N^{-1} \,  M \, \max_{t,\tau=1\ldots T}  \,
              \left| f^0_t \, (f^{0\prime} \, f^0)^{-1} \, f^{0\prime} \, e' M_{\lambda^0} X_{k,\tau} \right|
       \nonumber \\
        &\leq  N^{-1} \, M \, \|e\| \, \|f^0\| \, \left\|(f^{0\prime} \, f^0)^{-1} \right\| \,
                \max_t \, \| f^0_t \| \, \max_{\tau} \| X_{k,\tau} \|
       \nonumber \\
        &=  N^{-1} \, M \, {\cal O}_p(N^{1/2}) \, {\cal O}_p(T^{1/2}) \, {\cal O}_p(T^{-1})
               \, {\cal O}_p(N^{1/2} T^{1/8}) = o_p(1) \; ,
    \end{align*}
    and
    \begin{align*}
          \, \left\|
            \left[   e'  P_{\lambda^0} X_k   \right]^{\rm truncR} \right\|
            &\leq N^{-1/2} M \max_{t =1\ldots T}
             \left( N^{-1/2} \sum_i e_{it} \lambda^0_i \right)
             (N^{-1} \lambda^{0\prime} \, \lambda^0)^{-1}
            \max_{\tau =1\ldots T}    \left(  N^{-1} \sum_j \lambda^{0 \prime}_j  X_{jt}  \right)
         \\
           &= N^{-1/2} M {\cal O}_p(T^{1/8})   {\cal O}_p(1)  {\cal O}_p(T^{1/8})    = o_p(1).
    \end{align*}
   Thus, we proved equation \eqref{exp_proof_part3}.
    Analogously we obtain
   \begin{align*}
       N^{-1} \, \left\|
            \left[ e'  e \, - \,
                   \widehat e' \, \widehat e  \right]^{\rm truncD} \right\|
                         &= o_p(1) \; ,
      \nonumber \\
       T^{-1} \, \left\|
            \left[ e  e' \, - \,
                   \widehat e \, \widehat e'  \right]^{\rm truncD} \right\|
                         &= o_p(1) \; .
   \end{align*}

  \# Combining \eqref{exp_proof_part1}, \eqref{exp_proof_part2}, and \eqref{exp_proof_part3},
  we obtain
 $N^{-1} \, \left\| \mathbb{E}_{\cal C}(e'   X_k    ) -
       \left( \widehat e' \, X_k \right)^{\rm truncR} \right\|
                         = o_p(1)$. The proof of the other two statements of the lemma is analogous.
\end{proof}

\begin{proof}[\bf Proof of Lemma~\ref{lemma:normsTrunc}]
   Using theorem \ref{theorem:expansions} and \ref{lemma:sqrtNTcons} we find
   $\| \widehat e \| = {\cal O}_p(N^{1/2})$.
   Applying Lemma~\ref{Lemma:Trunc} we therefore find
   \begin{align*}
      N^{-1} \, \left\| \left( \widehat e' \, X_k \right)^{\rm truncR} \right\|
            &\leq \frac M N \, \max_{t,\tau} \, \left| \widehat e'_t \, X_{k,\tau} \right|
      \nonumber \\
          &\leq \frac M N \, \max_{t,\tau} \, \| \widehat e_{t} \| \,  \| X_{k,\tau} \|
      \nonumber \\
          &\leq \frac M N \, \| \widehat e \| \, \max_{\tau}  \| X_{k,\tau} \|
      \nonumber \\
          &\leq \frac M N  {\cal O}_p(N^{1/2}) {\cal O}_p(N^{1/2} T^{1/8}) = {\cal O}_p(M T^{1/8}) \; ,
   \end{align*}
   where we used the result $\max_{\tau}  \| X_{k,\tau} \| = {\cal O}_p(N^{1/2} T^{1/8})$ that was
   already obtained in the proof of the last theorem.

   The proof for the statement (ii) and (iii) is analogous.
\end{proof}

\section{Proofs for Section \ref{sec:testing} (Testing)}

\begin{proof}[\bf Proof of Theorem \ref{th:gradient}]
   Using the expansion for $L_{NT}(\beta)$ in Lemma S.1
   in the supplementary material 
   of Moon and Weidner~\cite*{MoonWeidner2015} we find for the derivative
   (the sign convention $\epsilon_k=\beta^0_k - \beta_k$ results in the minus sign below)
   \begin{align*}
      \frac{\partial L_{NT}}{\partial \beta_k}
       &= \, - \, \frac{1} {NT} \, \sum_{g=2}^\infty \, g \,
      \sum_{\kappa_1=0}^K \, \sum_{\kappa_2=0}^K \, \ldots \sum_{\kappa_{g-1}=0}^K \,
     \epsilon_{\kappa_1} \, \epsilon_{\kappa_2} \, \ldots \, \epsilon_{\kappa_{g-1}}
      \, L^{(g)}\left(\lambda^0,\, f^0,\, X_{k},\, X_{\kappa_1}, \ldots
     ,X_{\kappa_{g-1}}\right)
      \nonumber \\
      &=  \left[ 2 W_{NT} (\beta-\beta^0) \right]_k  - \, \frac 2 {\sqrt{NT}} C_{NT,k}
           + \frac 1 {NT} \nabla R_{1,NT,k} + \frac 1 {NT} \nabla R_{2,NT,k} \; ,
   \end{align*}
   where
   \begin{align*}
      W_{NT,k_1 k_2} &= \frac{1} {NT} \, L^{(2)}\left(\lambda^0,\, f^0,\, X_{k_1},\, X_{k_2} \right)  \; ,
     \nonumber \\
      C_{NT,k} &= \frac 1 {2\sqrt{NT}} \, \sum_{g=2}^{G_e} \, g \, (\epsilon_{0})^{g-1}
                         \, L^{(g)}\left(\lambda^0,\, f^0,\, X_{k},\, X_{0}, \ldots ,X_{0}\right)
     \nonumber \\
               &=  \sum_{g=2}^{G_e} \, \frac g {2\sqrt{NT}} \,
                         \, L^{(g)}\left(\lambda^0,\, f^0,\, X_{k},\, e, \ldots ,e\right) \; ,
   \end{align*}
   and
   \begin{align*}
      \nabla R_{1,NT,k}
               &= \, - \, \sum_{g=G_e+1}^\infty \, g \,
        (\epsilon_{0})^{g-1}
      \, L^{(g)}\left(\lambda^0,\, f^0,\, X_{k},\, X_{0}, \ldots ,X_{0}\right)  \; ,
     \nonumber \\
               &= \, - \, \sum_{g=G_e+1}^\infty \, g \,
      \, L^{(g)}\left(\lambda^0,\, f^0,\, X_{k},\, e, \ldots ,e\right)  \; ,
     \nonumber \\
      \nabla R_{2,NT,k}
          &= \, - \, \sum_{g=3}^\infty \, g \, \sum_{r=1}^{g-1} \, {g-1 \choose r}
      \sum_{k_1=1}^K \, \ldots \sum_{{k_r}=1}^K \,
      \epsilon_{k_1} \, \ldots \,\epsilon_{k_r} \, (\epsilon_0)^{g-r-1} \,
    \nonumber \\
        & \qquad \qquad \qquad \qquad \qquad \qquad
      \, L^{(g)}\left(\lambda^0,\, f^0,\, X_{k},\, X_{k_1}, \ldots
     ,X_{k_r},X_0,\ldots,X_0\right)  \; .
    \nonumber \\
          &= \, - \, \sum_{g=3}^\infty \, g \, \sum_{r=1}^{g-1} \, {g-1 \choose r}
      \sum_{k_1=1}^K \, \ldots \sum_{{k_r}=1}^K \,
      \,(\beta^0_{k_1}-\beta_{k_1}) \, \ldots \, (\beta^0_{k_r}-\beta_{k_r})
    \nonumber \\
        & \qquad \qquad \qquad \qquad \qquad \qquad
      \, L^{(g)}\left(\lambda^0,\, f^0,\, X_{k},\, X_{k_1}, \ldots
     ,X_{k_r},e,\ldots,e\right)  \; .
   \end{align*}
   The above expressions for $W_{NT}$ and $C_{NT}$ are equivalent to their definitions given in theorem
   \ref{th:ass_expand}. Using the bound on $L^{(g)}$
   we find\footnote{Here we use ${{n \choose k}} \leq 4^n$.}
   \begin{align*}
      |\nabla R_{1,NT,k}|
               &\leq  c_0 \, NT \,  \frac{\|X_k\|}{\sqrt{NT}} \, \sum_{g=G_e+1}^\infty \, g^2 \,
                                \left(\frac{c_1 \|e\|}{\sqrt{NT}} \right)^{g-1}
    \nonumber \\
               &\leq 2 \, c_0 \, (1+G_e)^2 \, NT \,  \frac{\|X_k\|}{\sqrt{NT}} \,
                         \left(\frac{c_1 \|e\|}{\sqrt{NT}} \right)^{G_e} \,
                          \left[1-\left(\frac{c_1 \|e\|}{\sqrt{NT}} \right)\right]^{-3}
               = o_p(\sqrt{NT}) \; ,
    \nonumber \\
      |\nabla R_{2,NT,k}|
          &\leq c_0 \, NT \, \frac{\|X_{k}\|}{\sqrt{NT}} \,
                \sum_{g=3}^\infty \, g^2 \, \sum_{r=1}^{g-1} \, {g-1 \choose r}
                        \, c_1^{g-1} \,
                   \left( \sum_{\widetilde k=1}^K |\beta_{\widetilde k}-\beta^0_k|  \frac{\|X_{\widetilde k}\|}{\sqrt{NT}}
                             \right)
              \nonumber \\ & \qquad\qquad\qquad\qquad\qquad\qquad\qquad\qquad \times
                    \left( \sum_{\widetilde k=1}^K |\beta_{\widetilde k}-\beta^0_k|  \frac{\|X_{\widetilde k}\|}{\sqrt{NT}}
                            + \frac{\|e\|}{\sqrt{NT}} \right)^{g-2}
        \nonumber \\
          &\leq c_0 \, NT \, \frac{\|X_{k}\|}{\sqrt{NT}} \,
                \sum_{g=3}^\infty \, g^3
                        \, (4c_1)^{g-1} \,
                   \left( \sum_{\widetilde k=1}^K |\beta_{\widetilde k}-\beta^0_k|  \frac{\|X_{\widetilde k}\|}{\sqrt{NT}}
                             \right)
                    \left( \sum_{\widetilde k=1}^K |\beta_{\widetilde k}-\beta^0_{\widetilde k}|  \frac{\|X_{\widetilde k}\|}{\sqrt{NT}}
                            + \frac{\|e\|}{\sqrt{NT}} \right)^{g-2}
        \nonumber \\
          &\leq c_2 \, NT \, \frac{\|X_{k}\|}{\sqrt{NT}} \,
                    \left( \sum_{\widetilde k=1}^K |\beta_{\widetilde k}-\beta^0_k|  \frac{\|X_{\widetilde k}\|}{\sqrt{NT}}
                             \right)
                     \left( \sum_{\widetilde k=1}^K |\beta_{\widetilde k}-\beta^0_{\widetilde k}|  \frac{\|X_{\widetilde k}\|}{\sqrt{NT}}
                            + \frac{\|e\|}{\sqrt{NT}} \right) \; ,
  \end{align*}
  where $c_0=8 R d_{\max}(\lambda^0,f^0)/2$ and $c_1=16 d_{\max}(\lambda^0,f^0)/d_{\min}^2(\lambda^0,f^0)$
  both converge to a constants as $N,T \rightarrow \infty$, and the very last inequality is only true
  if $4 c_1 \left( \sum_{\widetilde k=1}^K |\beta_{\widetilde k}-\beta^0_{\widetilde k}|  \frac{\|X_{\widetilde k}\|}{\sqrt{NT}}
                            + \frac{\|e\|}{\sqrt{NT}} \right)<1$, and $c_2>0$ is an appropriate positive constant.
  To show $\nabla R_{1,NT,k}=o_p(NT)$ we used Assumption~\ref{ass:A3}$^*$.
  From the above inequalities we find for $\eta_{NT} \rightarrow \infty$
  \begin{align*}
      \sup_{\{\beta :\left\| \beta -\beta^{0} \right\| \leq \eta_{NT}\}}
                  \frac{ \left\| \nabla R_{1,NT}(\beta) \right\| }
       { \sqrt{NT} } = o_{p}\left( 1 \right) ,
     \nonumber \\
      \sup_{\{\beta :\left\| \beta -\beta^{0} \right\| \leq \eta_{NT}\}}
                  \frac{ \left\| \nabla R_{2,NT}(\beta) \right\| }
       { NT \, \left\| \beta -\beta^{0} \right\|  } = o_{p}\left( 1 \right) .
  \end{align*}
  Thus $R_{NT}(\beta)=R_{1,NT}(\beta)+R_{2,NT}(\beta)$ satisfies the bound
  in the theorem.
\end{proof}
\begin{proof}[\bf Proof of Theorem \ref{th:testing}]
   Using Theorem~\ref{th:limdis} it is straightforward
   to show $WD_{NT}^*$ has limiting distribution $\chi^2_r$.

    For the LR test we have to show the estimator
   $\widehat c = (NT)^{-1} {\rm Tr} (\widehat e(\widehat \beta) \, \widehat e'(\widehat \beta))$
   is consistent for  $c=\mathbb{E}_{\cal C} e_{it}^2$.
   As already noted in the main text we have $\widehat c = L_{NT}\left( \widehat \beta \right)$,
   and using our expansion and $\sqrt{NT}$-consistency of $\widehat \beta$ we immediately obtain
   \begin{align*}
      \widehat c &= \frac 1 {NT} \,  {\rm Tr} (M_{\lambda^0} e M_{f^0} e') + o_p(1) \; .
   \end{align*}
   Alternatively, one could use the expansion of $\widehat e$ in Theorem~\ref{theorem:expansions} to show this.
   From the above result we find
   \begin{align*}
      \left| \widehat c - \frac 1 {NT} {\rm Tr} (ee') \right|
          &= \frac 1 {NT} \left|
                   {\rm Tr} (P_{\lambda^0} e M_{f^0} e')
                   + {\rm Tr} (e P_{f^0} e') \right| + o_p(1)
       \nonumber \\
             &\leq \frac {2R} {NT} \, \|e\|^2 + o_p(1) = o_p(1)  \; .
   \end{align*}
   By the weak law of large numbers we thus have
   \begin{align*}
      \widehat c = \frac 1 {NT} \sum_{i=1}^N \sum_{t=1}^T e_{it}^2 + o_p(1) = c + o_p(1) \; ,
   \end{align*}
   i.e., $\widehat c$ is indeed consistent for $c$.
   Having this one immediately obtains the result
   for the limiting distribution of $LR_{NT}^*$.

    For the LM test we first want to show equation \eqref{EquivGrads} holds.
      Using the expansion of $\widehat e$ in Theorem~\ref{theorem:expansions} one obtains
      \begin{align*}
         \sqrt{NT} (\widetilde \nabla {\cal L}_{NT})_k \, &= \, - \, \frac 2 {\sqrt{NT}} \,  {\rm Tr}\left(X'_k \widetilde e\right)
             \nonumber \\
                     &= \left[ 2 \, \sqrt{NT}
                \, W_{NT} \, \left( \widetilde \beta - \beta^0\right)  \right]_k + \frac 2 {NT} C^{(1)}(\lambda^0,f^0,X_k,e)
                                  + \frac 2 {NT} C^{(2)}(\lambda^0,f^0,X_k,e)
                       \nonumber \\ & \qquad \qquad \qquad \qquad \qquad \qquad \qquad \qquad
                                     \qquad \qquad \qquad \qquad
                         \, - \, \frac 2 {\sqrt{NT}} \,  {\rm Tr}\left(X'_k \widetilde e^{(\rm rem)}\right)
                \nonumber \\
                     &= \left[ 2 \,\sqrt{NT} \, W_{NT} \, \left( \widetilde \beta - \beta^0\right) + \frac 2 {NT} C_{NT}
                        \right]_k + o_p(1)
                \nonumber \\
                     &= \sqrt{NT}  \left[ \nabla L_{NT}(\widetilde \beta) \right]_k  + o_p(1) \, ,
      \end{align*}
      which is what we wanted to show.
      Here we used $|{\rm Tr}\left(X'_k \widetilde e^{(\rm rem)}\right)| \leq 7 R \|X_k\| \|\widetilde e^{(\rm rem)}\|
                                               = {\cal O}_p(N^{3/2})$.
      Note that $\|X_k\|={\cal O}_p(N)$, and Theorem~\ref{theorem:expansions}, and $\sqrt{NT}$-consistency of
       $\widetilde \beta$, together imply $\|\widetilde e^{(\rm rem)}\| = {\cal O}_p(\sqrt{N})$.
       We also used the expression for $\nabla L_{NT}(\widetilde \beta)$ given in Theorem~\ref{th:gradient},
       and the bound on $\nabla R_{NT}(\beta)$ given there.

    We now use
    equation \eqref{limNablaL} and $\widetilde W=W+o_p(1)$, $\widetilde \Omega=\Omega+o_p(1)$, and $\widetilde B=B+o_p(1)$
   to obtain
   \begin{align*}
      LM^*_{NT} \; \;  \limfunc{\longrightarrow}_d \; \;
                  (C - B)'
                   W^{-1} H' (H W^{-1} \Omega W^{-1} H')^{-1} H W^{-1}
                      (C - B) \; .
   \end{align*}
   Under $H_0$ we thus find $LM^*_{NT} \; \limfunc{\rightarrow}_d \; \chi^2_r$.
\end{proof}

\section{Additional Monte Carlo Results}

We consider an ${\rm AR}(1)$ model with $R$ factors
\begin{align*}
   Y_{it} \, &= \, \rho^0 \, Y_{i,t-1} \, + \, \sum_{r=1}^R \lambda^0_{ir} \, f^0_{tr} \, + \, e_{it} \; .
\end{align*}
We draw the $e_{it}$ independently and identically distributed from a t-distribution with
five degrees of freedom.
The   $\lambda^0_{ir}$ are  independently distributed as ${\cal N}(1,1)$, and we generate the factors
from an ${\rm AR}(1)$ specification, namely $f^0_{tr}=\rho_f \, f^0_{t-1,r} + u_{tr}$, for each $r=1,\ldots,R$,
where $u_{tr} \sim {\rm iid} {\cal N}(0,(1-\rho_f^2)\sigma_f^2)$.   For all simulations we generate 1,000 initial time periods for $f^0_t$ and $Y_{it}$ that are not used for estimation.
  This guarantees the simulated data used for estimation are distributed according to the stationary distribution
  of the model.
  
For $R=1$ this is exactly the simulation design used in the main text Monte Carlo section, but DGPs with $R>1$
were not considered in the main text.
Table~\ref{tab:extra1} reports results for which $R=1$ is used both in the DGP and for the LS estimation.
Table~\ref{tab:extra2} reports results for which $R=1$ is used in the DGP, but $R=2$ is used for the LS estimation.
Table~\ref{tab:extra3} reports results for which $R=2$ is used both in the DGP and for the LS estimation.
The results in Table~\ref{tab:extra1} and~\ref{tab:extra2} are identical to those reported in the main text
Table~\ref{tab:T1} and~\ref{tab:T2}, except we also report results for the CCE estimator.
The results in Table~\ref{tab:extra3} are not contained in the main text.

The CCE estimator is obtained by using $\widehat f^{\rm proxy}_t =  N^{-1} \sum_i ( Y_{it}, \,   Y_{i,t-1} )'$ as a proxy for the
factors and then estimating the parameters $\rho$, $\lambda_{i1}$, $\lambda_{i2}$, $i=1,\ldots,N$,
via OLS in the linear regression model
$ Y_{it}  =  \rho   Y_{i,t-1}  +   \lambda_{i1}  \widehat f^{\rm proxy}_{t1}
   +   \lambda_{i2}  \widehat f^{\rm proxy}_{t2}  +  e_{it} $.
   
The performance of the CCE estimator in Table~\ref{tab:extra1} and~\ref{tab:extra2} are identical
(up to random MC noise), because the number of factors need not be specified for the CCE estimator, and the DGPs
in Table~\ref{tab:extra1} and~\ref{tab:extra2} are identical. These tables show for $R=1$ in the DGP,
the CCE estimator performs very well. From Chudik and Pesaran~\cite*{ChudikPesaran2015} we expect the
CCE estimator to have a bias of order $1/T$ in a dynamic model, which is confirmed in the simulations: 
the bias of the CCE estimator shrinks roughly in inverse proportion to $T$, as $T$ becomes larger.
The $1/T$ bias of the CCE estimator could be corrected for, and we would expect the bias-corrected CCE estimator
to perform similarly to the bias-corrected LS estimator.

However, if there are $R=2$ factors in the true DGP, then it turns out the proxies $\widehat f^{\rm proxy}_t$
do not pick those up correctly.
Table~\ref{tab:extra3} shows for some parameter values and sample sizes (e.g., $\rho^0=0.3$ and $T=10$,
or $\rho^0 = 0.9$ and $T=40$) the CCE estimator is almost unbiased, but for other values, including $T=80$,
the CCE estimator is heavily biased if $R=2$. In particular,   the bias of the CCE estimator does not seem to converge
to zero as $T$ becomes large in this case. By contrast, the correctly specified  LS estimators (i.e., correctly 
using $R=2$ factors in the estimation)
performs very well according to Table~\ref{tab:extra3}. However, an incorrectly specified
LS estimator, which would underestimate the number of factors (e.g., using $R=1$ factors in estimation 
instead of the correct number $R=2$)
would probably perform similarly to the CCE estimator, because not all factors would be corrected for.
Overestimating the number of factors (i.e., using $R=3$ factors in estimation 
instead of the correct number $R=2$) should, however, not pose a problem for the LS estimator, according to Moon and Weidner~\cite*{MoonWeidner2015}.

\newpage

\section*{Tables with Simulation Results}



\newcommand{\bOLSEbiasAA}{0.1232}
\newcommand{\bOLSEstdeAA}{0.1444}
\newcommand{\bOLSErmseAA}{0.1898}
\newcommand{\bQMLEbiasAA}{-0.1419}
\newcommand{\bQMLEstdeAA}{0.1480}
\newcommand{\bQMLErmseAA}{0.2050}
\newcommand{\bBCQMbiasAA}{-0.0713}
\newcommand{\bBCQMstdeAA}{0.0982}
\newcommand{\bBCQMrmseAA}{0.1213}
\newcommand{\bCCEPbiasAA}{-0.1755}
\newcommand{\bCCEPstdeAA}{0.1681}
\newcommand{\bCCEPrmseAA}{0.2430}
\newcommand{\bOLSEbiasAB}{0.0200}
\newcommand{\bOLSEstdeAB}{0.0723}
\newcommand{\bOLSErmseAB}{0.0750}
\newcommand{\bQMLEbiasAB}{-0.3686}
\newcommand{\bQMLEstdeAB}{0.1718}
\newcommand{\bQMLErmseAB}{0.4067}
\newcommand{\bBCQMbiasAB}{-0.2330}
\newcommand{\bBCQMstdeAB}{0.1301}
\newcommand{\bBCQMrmseAB}{0.2669}
\newcommand{\bCCEPbiasAB}{-0.3298}
\newcommand{\bCCEPstdeAB}{0.2203}
\newcommand{\bCCEPrmseAB}{0.3966}
\newcommand{\bOLSEbiasAC}{0.1339}
\newcommand{\bOLSEstdeAC}{0.1148}
\newcommand{\bOLSErmseAC}{0.1764}
\newcommand{\bQMLEbiasAC}{-0.0542}
\newcommand{\bQMLEstdeAC}{0.0596}
\newcommand{\bQMLErmseAC}{0.0806}
\newcommand{\bBCQMbiasAC}{-0.0201}
\newcommand{\bBCQMstdeAC}{0.0423}
\newcommand{\bBCQMrmseAC}{0.0469}
\newcommand{\bCCEPbiasAC}{-0.0819}
\newcommand{\bCCEPstdeAC}{0.0593}
\newcommand{\bCCEPrmseAC}{0.1011}
\newcommand{\bOLSEbiasAD}{0.0218}
\newcommand{\bOLSEstdeAD}{0.0513}
\newcommand{\bOLSErmseAD}{0.0557}
\newcommand{\bQMLEbiasAD}{-0.1019}
\newcommand{\bQMLEstdeAD}{0.1094}
\newcommand{\bQMLErmseAD}{0.1495}
\newcommand{\bBCQMbiasAD}{-0.0623}
\newcommand{\bBCQMstdeAD}{0.0747}
\newcommand{\bBCQMrmseAD}{0.0973}
\newcommand{\bCCEPbiasAD}{-0.1436}
\newcommand{\bCCEPstdeAD}{0.0972}
\newcommand{\bCCEPrmseAD}{0.1734}
\newcommand{\bOLSEbiasAE}{0.1441}
\newcommand{\bOLSEstdeAE}{0.0879}
\newcommand{\bOLSErmseAE}{0.1687}
\newcommand{\bQMLEbiasAE}{-0.0264}
\newcommand{\bQMLEstdeAE}{0.0284}
\newcommand{\bQMLErmseAE}{0.0388}
\newcommand{\bBCQMbiasAE}{-0.0070}
\newcommand{\bBCQMstdeAE}{0.0240}
\newcommand{\bBCQMrmseAE}{0.0250}
\newcommand{\bCCEPbiasAE}{-0.0405}
\newcommand{\bCCEPstdeAE}{0.0277}
\newcommand{\bCCEPrmseAE}{0.0491}
\newcommand{\bOLSEbiasAF}{0.0254}
\newcommand{\bOLSEstdeAF}{0.0353}
\newcommand{\bOLSErmseAF}{0.0434}
\newcommand{\bQMLEbiasAF}{-0.0173}
\newcommand{\bQMLEstdeAF}{0.0299}
\newcommand{\bQMLErmseAF}{0.0345}
\newcommand{\bBCQMbiasAF}{-0.0085}
\newcommand{\bBCQMstdeAF}{0.0219}
\newcommand{\bBCQMrmseAF}{0.0235}
\newcommand{\bCCEPbiasAF}{-0.0617}
\newcommand{\bCCEPstdeAF}{0.0406}
\newcommand{\bCCEPrmseAF}{0.0739}
\newcommand{\bOLSEbiasAG}{0.1517}
\newcommand{\bOLSEstdeAG}{0.0657}
\newcommand{\bOLSErmseAG}{0.1654}
\newcommand{\bQMLEbiasAG}{-0.0130}
\newcommand{\bQMLEstdeAG}{0.0170}
\newcommand{\bQMLErmseAG}{0.0214}
\newcommand{\bBCQMbiasAG}{-0.0021}
\newcommand{\bBCQMstdeAG}{0.0160}
\newcommand{\bBCQMrmseAG}{0.0161}
\newcommand{\bCCEPbiasAG}{-0.0200}
\newcommand{\bCCEPstdeAG}{0.0166}
\newcommand{\bCCEPrmseAG}{0.0260}
\newcommand{\bOLSEbiasAH}{0.0294}
\newcommand{\bOLSEstdeAH}{0.0250}
\newcommand{\bOLSErmseAH}{0.0386}
\newcommand{\bQMLEbiasAH}{-0.0057}
\newcommand{\bQMLEstdeAH}{0.0105}
\newcommand{\bQMLErmseAH}{0.0119}
\newcommand{\bBCQMbiasAH}{-0.0019}
\newcommand{\bBCQMstdeAH}{0.0089}
\newcommand{\bBCQMrmseAH}{0.0091}
\newcommand{\bCCEPbiasAH}{-0.0281}
\newcommand{\bCCEPstdeAH}{0.0162}
\newcommand{\bCCEPrmseAH}{0.0324}
\newcommand{\bOLSEbiasAI}{0.1552}
\newcommand{\bOLSEstdeAI}{0.0487}
\newcommand{\bOLSErmseAI}{0.1627}
\newcommand{\bQMLEbiasAI}{-0.0066}
\newcommand{\bQMLEstdeAI}{0.0112}
\newcommand{\bQMLErmseAI}{0.0130}
\newcommand{\bBCQMbiasAI}{-0.0007}
\newcommand{\bBCQMstdeAI}{0.0109}
\newcommand{\bBCQMrmseAI}{0.0109}
\newcommand{\bCCEPbiasAI}{-0.0100}
\newcommand{\bCCEPstdeAI}{0.0111}
\newcommand{\bCCEPrmseAI}{0.0149}
\newcommand{\bOLSEbiasAJ}{0.0326}
\newcommand{\bOLSEstdeAJ}{0.0179}
\newcommand{\bOLSErmseAJ}{0.0372}
\newcommand{\bQMLEbiasAJ}{-0.0026}
\newcommand{\bQMLEstdeAJ}{0.0056}
\newcommand{\bQMLErmseAJ}{0.0062}
\newcommand{\bBCQMbiasAJ}{-0.0006}
\newcommand{\bBCQMstdeAJ}{0.0053}
\newcommand{\bBCQMrmseAJ}{0.0053}
\newcommand{\bCCEPbiasAJ}{-0.0136}
\newcommand{\bCCEPstdeAJ}{0.0073}
\newcommand{\bCCEPrmseAJ}{0.0154}



\newcommand{\bbOLSEbiasAA}{0.1239}
\newcommand{\bbOLSEstdeAA}{0.1454}
\newcommand{\bbOLSErmseAA}{0.1910}
\newcommand{\bbQMLEbiasAA}{-0.5467}
\newcommand{\bbQMLEstdeAA}{0.1528}
\newcommand{\bbQMLErmseAA}{0.5676}
\newcommand{\bbBCQMbiasAA}{-0.3721}
\newcommand{\bbBCQMstdeAA}{0.1299}
\newcommand{\bbBCQMrmseAA}{0.3942}
\newcommand{\bbCCEPbiasAA}{-0.1767}
\newcommand{\bbCCEPstdeAA}{0.1678}
\newcommand{\bbCCEPrmseAA}{0.2437}
\newcommand{\bbOLSEbiasAB}{0.0218}
\newcommand{\bbOLSEstdeAB}{0.0731}
\newcommand{\bbOLSErmseAB}{0.0763}
\newcommand{\bbQMLEbiasAB}{-0.9716}
\newcommand{\bbQMLEstdeAB}{0.1216}
\newcommand{\bbQMLErmseAB}{0.9792}
\newcommand{\bbBCQMbiasAB}{-0.7490}
\newcommand{\bbBCQMstdeAB}{0.1341}
\newcommand{\bbBCQMrmseAB}{0.7609}
\newcommand{\bbCCEPbiasAB}{-0.3289}
\newcommand{\bbCCEPstdeAB}{0.2203}
\newcommand{\bbCCEPrmseAB}{0.3958}
\newcommand{\bbOLSEbiasAC}{0.1343}
\newcommand{\bbOLSEstdeAC}{0.1145}
\newcommand{\bbOLSErmseAC}{0.1765}
\newcommand{\bbQMLEbiasAC}{-0.1874}
\newcommand{\bbQMLEstdeAC}{0.1159}
\newcommand{\bbQMLErmseAC}{0.2203}
\newcommand{\bbBCQMbiasAC}{-0.1001}
\newcommand{\bbBCQMstdeAC}{0.0758}
\newcommand{\bbBCQMrmseAC}{0.1256}
\newcommand{\bbCCEPbiasAC}{-0.0816}
\newcommand{\bbCCEPstdeAC}{0.0592}
\newcommand{\bbCCEPrmseAC}{0.1008}
\newcommand{\bbOLSEbiasAD}{0.0210}
\newcommand{\bbOLSEstdeAD}{0.0518}
\newcommand{\bbOLSErmseAD}{0.0559}
\newcommand{\bbQMLEbiasAD}{-0.4923}
\newcommand{\bbQMLEstdeAD}{0.1159}
\newcommand{\bbQMLErmseAD}{0.5058}
\newcommand{\bbBCQMbiasAD}{-0.3271}
\newcommand{\bbBCQMstdeAD}{0.0970}
\newcommand{\bbBCQMrmseAD}{0.3412}
\newcommand{\bbCCEPbiasAD}{-0.1414}
\newcommand{\bbCCEPstdeAD}{0.0971}
\newcommand{\bbCCEPrmseAD}{0.1715}
\newcommand{\bbOLSEbiasAE}{0.1451}
\newcommand{\bbOLSEstdeAE}{0.0879}
\newcommand{\bbOLSErmseAE}{0.1696}
\newcommand{\bbQMLEbiasAE}{-0.0448}
\newcommand{\bbQMLEstdeAE}{0.0469}
\newcommand{\bbQMLErmseAE}{0.0648}
\newcommand{\bbBCQMbiasAE}{-0.0168}
\newcommand{\bbBCQMstdeAE}{0.0320}
\newcommand{\bbBCQMrmseAE}{0.0362}
\newcommand{\bbCCEPbiasAE}{-0.0407}
\newcommand{\bbCCEPstdeAE}{0.0277}
\newcommand{\bbCCEPrmseAE}{0.0492}
\newcommand{\bbOLSEbiasAF}{0.0255}
\newcommand{\bbOLSEstdeAF}{0.0354}
\newcommand{\bbOLSErmseAF}{0.0436}
\newcommand{\bbQMLEbiasAF}{-0.1822}
\newcommand{\bbQMLEstdeAF}{0.0820}
\newcommand{\bbQMLErmseAF}{0.1999}
\newcommand{\bbBCQMbiasAF}{-0.1085}
\newcommand{\bbBCQMstdeAF}{0.0528}
\newcommand{\bbBCQMrmseAF}{0.1207}
\newcommand{\bbCCEPbiasAF}{-0.0618}
\newcommand{\bbCCEPstdeAF}{0.0404}
\newcommand{\bbCCEPrmseAF}{0.0739}
\newcommand{\bbOLSEbiasAG}{0.1511}
\newcommand{\bbOLSEstdeAG}{0.0663}
\newcommand{\bbOLSErmseAG}{0.1650}
\newcommand{\bbQMLEbiasAG}{-0.0161}
\newcommand{\bbQMLEstdeAG}{0.0209}
\newcommand{\bbQMLErmseAG}{0.0264}
\newcommand{\bbBCQMbiasAG}{-0.0038}
\newcommand{\bbBCQMstdeAG}{0.0177}
\newcommand{\bbBCQMrmseAG}{0.0181}
\newcommand{\bbCCEPbiasAG}{-0.0199}
\newcommand{\bbCCEPstdeAG}{0.0167}
\newcommand{\bbCCEPrmseAG}{0.0260}
\newcommand{\bbOLSEbiasAH}{0.0300}
\newcommand{\bbOLSEstdeAH}{0.0250}
\newcommand{\bbOLSErmseAH}{0.0390}
\newcommand{\bbQMLEbiasAH}{-0.0227}
\newcommand{\bbQMLEstdeAH}{0.0342}
\newcommand{\bbQMLErmseAH}{0.0410}
\newcommand{\bbBCQMbiasAH}{-0.0128}
\newcommand{\bbBCQMstdeAH}{0.0225}
\newcommand{\bbBCQMrmseAH}{0.0258}
\newcommand{\bbCCEPbiasAH}{-0.0282}
\newcommand{\bbCCEPstdeAH}{0.0164}
\newcommand{\bbCCEPrmseAH}{0.0326}
\newcommand{\bbOLSEbiasAI}{0.1550}
\newcommand{\bbOLSEstdeAI}{0.0488}
\newcommand{\bbOLSErmseAI}{0.1625}
\newcommand{\bbQMLEbiasAI}{-0.0072}
\newcommand{\bbQMLEstdeAI}{0.0123}
\newcommand{\bbQMLErmseAI}{0.0143}
\newcommand{\bbBCQMbiasAI}{-0.0011}
\newcommand{\bbBCQMstdeAI}{0.0115}
\newcommand{\bbBCQMrmseAI}{0.0116}
\newcommand{\bbCCEPbiasAI}{-0.0100}
\newcommand{\bbCCEPstdeAI}{0.0111}
\newcommand{\bbCCEPrmseAI}{0.0149}
\newcommand{\bbOLSEbiasAJ}{0.0325}
\newcommand{\bbOLSEstdeAJ}{0.0182}
\newcommand{\bbOLSErmseAJ}{0.0372}
\newcommand{\bbQMLEbiasAJ}{-0.0030}
\newcommand{\bbQMLEstdeAJ}{0.0064}
\newcommand{\bbQMLErmseAJ}{0.0071}
\newcommand{\bbBCQMbiasAJ}{-0.0010}
\newcommand{\bbBCQMstdeAJ}{0.0057}
\newcommand{\bbBCQMrmseAJ}{0.0058}
\newcommand{\bbCCEPbiasAJ}{-0.0136}
\newcommand{\bbCCEPstdeAJ}{0.0074}
\newcommand{\bbCCEPrmseAJ}{0.0155}



\newcommand{\xBCfractionAA}{0.889}
\newcommand{\xBCfractionAB}{0.832}
\newcommand{\xBCfractionAC}{0.791}
\newcommand{\xBCfractionAD}{0.754}
\newcommand{\xBCfractionAE}{0.720}
\newcommand{\xBCfractionAF}{0.689}
\newcommand{\xBCfractionAG}{0.660}
\newcommand{\xBCfractionAH}{0.633}
\newcommand{\xBCfractionAI}{0.752}
\newcommand{\xBCfractionAJ}{0.806}
\newcommand{\xBCfractionAK}{0.778}
\newcommand{\xBCfractionAL}{0.742}
\newcommand{\xBCfractionAM}{0.708}
\newcommand{\xBCfractionAN}{0.677}
\newcommand{\xBCfractionAO}{0.648}
\newcommand{\xBCfractionAP}{0.621}
\newcommand{\xBCfractionAQ}{0.589}
\newcommand{\xBCfractionAR}{0.718}
\newcommand{\xBCfractionAS}{0.728}
\newcommand{\xBCfractionAT}{0.704}
\newcommand{\xBCfractionAU}{0.674}
\newcommand{\xBCfractionAV}{0.644}
\newcommand{\xBCfractionAW}{0.616}
\newcommand{\xBCfractionAX}{0.590}
\newcommand{\xBCfractionAY}{0.299}
\newcommand{\xBCfractionAZ}{0.428}
\newcommand{\xBCfractionBA}{0.486}
\newcommand{\xBCfractionBB}{0.510}
\newcommand{\xBCfractionBC}{0.519}
\newcommand{\xBCfractionBD}{0.516}
\newcommand{\xBCfractionBE}{0.508}
\newcommand{\xBCfractionBF}{0.495}



\newcommand{\bbbOLSEbiasAA}{0.1452}
\newcommand{\bbbOLSEstdeAA}{0.0870}
\newcommand{\bbbOLSErmseAA}{0.1693}
\newcommand{\bbbQMLEbiasAA}{-0.0268}
\newcommand{\bbbQMLEstdeAA}{0.0286}
\newcommand{\bbbQMLErmseAA}{0.0392}
\newcommand{\bbbBCQMbiasAA}{-0.0056}
\newcommand{\bbbBCQMstdeAA}{0.0239}
\newcommand{\bbbBCQMrmseAA}{0.0245}
\newcommand{\bbbCCEPbiasAA}{-0.0408}
\newcommand{\bbbCCEPstdeAA}{0.0278}
\newcommand{\bbbCCEPrmseAA}{0.0494}
\newcommand{\bbbOLSEbiasAB}{0.0254}
\newcommand{\bbbOLSEstdeAB}{0.0351}
\newcommand{\bbbOLSErmseAB}{0.0433}
\newcommand{\bbbQMLEbiasAB}{-0.0173}
\newcommand{\bbbQMLEstdeAB}{0.0305}
\newcommand{\bbbQMLErmseAB}{0.0351}
\newcommand{\bbbBCQMbiasAB}{-0.0100}
\newcommand{\bbbBCQMstdeAB}{0.0253}
\newcommand{\bbbBCQMrmseAB}{0.0272}
\newcommand{\bbbCCEPbiasAB}{-0.0618}
\newcommand{\bbbCCEPstdeAB}{0.0406}
\newcommand{\bbbCCEPrmseAB}{0.0739}
\newcommand{\bbbOLSEbiasAC}{0.1446}
\newcommand{\bbbOLSEstdeAC}{0.0884}
\newcommand{\bbbOLSErmseAC}{0.1694}
\newcommand{\bbbQMLEbiasAC}{-0.0267}
\newcommand{\bbbQMLEstdeAC}{0.0284}
\newcommand{\bbbQMLErmseAC}{0.0390}
\newcommand{\bbbBCQMbiasAC}{-0.0082}
\newcommand{\bbbBCQMstdeAC}{0.0241}
\newcommand{\bbbBCQMrmseAC}{0.0255}
\newcommand{\bbbCCEPbiasAC}{-0.0407}
\newcommand{\bbbCCEPstdeAC}{0.0276}
\newcommand{\bbbCCEPrmseAC}{0.0492}
\newcommand{\bbbOLSEbiasAD}{0.0255}
\newcommand{\bbbOLSEstdeAD}{0.0351}
\newcommand{\bbbOLSErmseAD}{0.0434}
\newcommand{\bbbQMLEbiasAD}{-0.0172}
\newcommand{\bbbQMLEstdeAD}{0.0300}
\newcommand{\bbbQMLErmseAD}{0.0346}
\newcommand{\bbbBCQMbiasAD}{-0.0083}
\newcommand{\bbbBCQMstdeAD}{0.0212}
\newcommand{\bbbBCQMrmseAD}{0.0228}
\newcommand{\bbbCCEPbiasAD}{-0.0615}
\newcommand{\bbbCCEPstdeAD}{0.0403}
\newcommand{\bbbCCEPrmseAD}{0.0735}
\newcommand{\bbbOLSEbiasAE}{0.1434}
\newcommand{\bbbOLSEstdeAE}{0.0865}
\newcommand{\bbbOLSErmseAE}{0.1675}
\newcommand{\bbbQMLEbiasAE}{-0.0261}
\newcommand{\bbbQMLEstdeAE}{0.0285}
\newcommand{\bbbQMLErmseAE}{0.0387}
\newcommand{\bbbBCQMbiasAE}{-0.0100}
\newcommand{\bbbBCQMstdeAE}{0.0247}
\newcommand{\bbbBCQMrmseAE}{0.0266}
\newcommand{\bbbCCEPbiasAE}{-0.0401}
\newcommand{\bbbCCEPstdeAE}{0.0276}
\newcommand{\bbbCCEPrmseAE}{0.0487}
\newcommand{\bbbOLSEbiasAF}{0.0255}
\newcommand{\bbbOLSEstdeAF}{0.0358}
\newcommand{\bbbOLSErmseAF}{0.0440}
\newcommand{\bbbQMLEbiasAF}{-0.0175}
\newcommand{\bbbQMLEstdeAF}{0.0307}
\newcommand{\bbbQMLErmseAF}{0.0353}
\newcommand{\bbbBCQMbiasAF}{-0.0089}
\newcommand{\bbbBCQMstdeAF}{0.0208}
\newcommand{\bbbBCQMrmseAF}{0.0227}
\newcommand{\bbbCCEPbiasAF}{-0.0614}
\newcommand{\bbbCCEPstdeAF}{0.0410}
\newcommand{\bbbCCEPrmseAF}{0.0738}
\newcommand{\bbbOLSEbiasAG}{0.1506}
\newcommand{\bbbOLSEstdeAG}{0.0655}
\newcommand{\bbbOLSErmseAG}{0.1642}
\newcommand{\bbbQMLEbiasAG}{-0.0132}
\newcommand{\bbbQMLEstdeAG}{0.0170}
\newcommand{\bbbQMLErmseAG}{0.0216}
\newcommand{\bbbBCQMbiasAG}{-0.0017}
\newcommand{\bbbBCQMstdeAG}{0.0159}
\newcommand{\bbbBCQMrmseAG}{0.0160}
\newcommand{\bbbCCEPbiasAG}{-0.0202}
\newcommand{\bbbCCEPstdeAG}{0.0166}
\newcommand{\bbbCCEPrmseAG}{0.0262}
\newcommand{\bbbOLSEbiasAH}{0.0293}
\newcommand{\bbbOLSEstdeAH}{0.0250}
\newcommand{\bbbOLSErmseAH}{0.0385}
\newcommand{\bbbQMLEbiasAH}{-0.0057}
\newcommand{\bbbQMLEstdeAH}{0.0103}
\newcommand{\bbbQMLErmseAH}{0.0118}
\newcommand{\bbbBCQMbiasAH}{-0.0024}
\newcommand{\bbbBCQMstdeAH}{0.0095}
\newcommand{\bbbBCQMrmseAH}{0.0098}
\newcommand{\bbbCCEPbiasAH}{-0.0279}
\newcommand{\bbbCCEPstdeAH}{0.0159}
\newcommand{\bbbCCEPrmseAH}{0.0321}
\newcommand{\bbbOLSEbiasAI}{0.1510}
\newcommand{\bbbOLSEstdeAI}{0.0658}
\newcommand{\bbbOLSErmseAI}{0.1647}
\newcommand{\bbbQMLEbiasAI}{-0.0131}
\newcommand{\bbbQMLEstdeAI}{0.0169}
\newcommand{\bbbQMLErmseAI}{0.0214}
\newcommand{\bbbBCQMbiasAI}{-0.0023}
\newcommand{\bbbBCQMstdeAI}{0.0159}
\newcommand{\bbbBCQMrmseAI}{0.0161}
\newcommand{\bbbCCEPbiasAI}{-0.0201}
\newcommand{\bbbCCEPstdeAI}{0.0165}
\newcommand{\bbbCCEPrmseAI}{0.0260}
\newcommand{\bbbOLSEbiasAJ}{0.0295}
\newcommand{\bbbOLSEstdeAJ}{0.0248}
\newcommand{\bbbOLSErmseAJ}{0.0385}
\newcommand{\bbbQMLEbiasAJ}{-0.0057}
\newcommand{\bbbQMLEstdeAJ}{0.0104}
\newcommand{\bbbQMLErmseAJ}{0.0119}
\newcommand{\bbbBCQMbiasAJ}{-0.0019}
\newcommand{\bbbBCQMstdeAJ}{0.0089}
\newcommand{\bbbBCQMrmseAJ}{0.0091}
\newcommand{\bbbCCEPbiasAJ}{-0.0280}
\newcommand{\bbbCCEPstdeAJ}{0.0162}
\newcommand{\bbbCCEPrmseAJ}{0.0324}
\newcommand{\bbbOLSEbiasAK}{0.1512}
\newcommand{\bbbOLSEstdeAK}{0.0655}
\newcommand{\bbbOLSErmseAK}{0.1648}
\newcommand{\bbbQMLEbiasAK}{-0.0131}
\newcommand{\bbbQMLEstdeAK}{0.0168}
\newcommand{\bbbQMLErmseAK}{0.0214}
\newcommand{\bbbBCQMbiasAK}{-0.0030}
\newcommand{\bbbBCQMstdeAK}{0.0159}
\newcommand{\bbbBCQMrmseAK}{0.0162}
\newcommand{\bbbCCEPbiasAK}{-0.0201}
\newcommand{\bbbCCEPstdeAK}{0.0166}
\newcommand{\bbbCCEPrmseAK}{0.0260}
\newcommand{\bbbOLSEbiasAL}{0.0292}
\newcommand{\bbbOLSEstdeAL}{0.0251}
\newcommand{\bbbOLSErmseAL}{0.0385}
\newcommand{\bbbQMLEbiasAL}{-0.0055}
\newcommand{\bbbQMLEstdeAL}{0.0103}
\newcommand{\bbbQMLErmseAL}{0.0117}
\newcommand{\bbbBCQMbiasAL}{-0.0018}
\newcommand{\bbbBCQMstdeAL}{0.0085}
\newcommand{\bbbBCQMrmseAL}{0.0087}
\newcommand{\bbbCCEPbiasAL}{-0.0277}
\newcommand{\bbbCCEPstdeAL}{0.0161}
\newcommand{\bbbCCEPrmseAL}{0.0321}



\newcommand{\aOLSEbiasAA}{-0.0007}
\newcommand{\aOLSEstdeAA}{0.0182}
\newcommand{\aOLSErmseAA}{0.0182}
\newcommand{\aQMLEbiasAA}{-0.0076}
\newcommand{\aQMLEstdeAA}{0.0332}
\newcommand{\aQMLErmseAA}{0.0340}
\newcommand{\aBCQMbiasAA}{-0.0043}
\newcommand{\aBCQMstdeAA}{0.0243}
\newcommand{\aBCQMrmseAA}{0.0247}
\newcommand{\aOLSEbiasAB}{-0.0004}
\newcommand{\aOLSEstdeAB}{0.0178}
\newcommand{\aOLSErmseAB}{0.0178}
\newcommand{\aQMLEbiasAB}{-0.0074}
\newcommand{\aQMLEstdeAB}{0.0331}
\newcommand{\aQMLErmseAB}{0.0339}
\newcommand{\aBCQMbiasAB}{-0.0041}
\newcommand{\aBCQMstdeAB}{0.0242}
\newcommand{\aBCQMrmseAB}{0.0245}
\newcommand{\aOLSEbiasAC}{0.0153}
\newcommand{\aOLSEstdeAC}{0.0251}
\newcommand{\aOLSErmseAC}{0.0294}
\newcommand{\aQMLEbiasAC}{-0.0113}
\newcommand{\aQMLEstdeAC}{0.0303}
\newcommand{\aQMLErmseAC}{0.0323}
\newcommand{\aBCQMbiasAC}{-0.0032}
\newcommand{\aBCQMstdeAC}{0.0229}
\newcommand{\aBCQMrmseAC}{0.0231}
\newcommand{\aOLSEbiasAD}{0.0474}
\newcommand{\aOLSEstdeAD}{0.0382}
\newcommand{\aOLSErmseAD}{0.0609}
\newcommand{\aQMLEbiasAD}{-0.0291}
\newcommand{\aQMLEstdeAD}{0.0387}
\newcommand{\aQMLErmseAD}{0.0484}
\newcommand{\aBCQMbiasAD}{-0.0071}
\newcommand{\aBCQMstdeAD}{0.0272}
\newcommand{\aBCQMrmseAD}{0.0281}
\newcommand{\aOLSEbiasAE}{0.0567}
\newcommand{\aOLSEstdeAE}{0.0633}
\newcommand{\aOLSErmseAE}{0.0850}
\newcommand{\aQMLEbiasAE}{-0.0137}
\newcommand{\aQMLEstdeAE}{0.0260}
\newcommand{\aQMLErmseAE}{0.0294}
\newcommand{\aBCQMbiasAE}{-0.0041}
\newcommand{\aBCQMstdeAE}{0.0207}
\newcommand{\aBCQMrmseAE}{0.0211}
\newcommand{\aOLSEbiasAF}{0.1491}
\newcommand{\aOLSEstdeAF}{0.0763}
\newcommand{\aOLSErmseAF}{0.1675}
\newcommand{\aQMLEbiasAF}{-0.0403}
\newcommand{\aQMLEstdeAF}{0.0298}
\newcommand{\aQMLErmseAF}{0.0501}
\newcommand{\aBCQMbiasAF}{-0.0126}
\newcommand{\aBCQMstdeAF}{0.0226}
\newcommand{\aBCQMrmseAF}{0.0259}



\newcommand{\cWDsizeAA}{0.219}
\newcommand{\cLRsizeAA}{0.214}
\newcommand{\cLMsizeAA}{0.192}
\newcommand{\dWDsizeAA}{0.066}
\newcommand{\dLRsizeAA}{0.062}
\newcommand{\dLMsizeAA}{0.056}
\newcommand{\cWDlpowerAA}{0.094}
\newcommand{\cLRlpowerAA}{0.089}
\newcommand{\cLMlpowerAA}{0.076}
\newcommand{\cWDrpowerAA}{0.526}
\newcommand{\cLRrpowerAA}{0.515}
\newcommand{\cLMrpowerAA}{0.487}
\newcommand{\dWDlpowerAA}{0.128}
\newcommand{\dLRlpowerAA}{0.123}
\newcommand{\dLMlpowerAA}{0.121}
\newcommand{\dWDrpowerAA}{0.235}
\newcommand{\dLRrpowerAA}{0.227}
\newcommand{\dLMrpowerAA}{0.206}
\newcommand{\ccWDlpowerAA}{0.010}
\newcommand{\ccLRlpowerAA}{0.011}
\newcommand{\ccLMlpowerAA}{0.010}
\newcommand{\ccWDrpowerAA}{0.211}
\newcommand{\ccLRrpowerAA}{0.208}
\newcommand{\ccLMrpowerAA}{0.206}
\newcommand{\ddWDlpowerAA}{0.105}
\newcommand{\ddLRlpowerAA}{0.104}
\newcommand{\ddLMlpowerAA}{0.112}
\newcommand{\ddWDrpowerAA}{0.199}
\newcommand{\ddLRrpowerAA}{0.197}
\newcommand{\ddLMrpowerAA}{0.193}
\newcommand{\cWDsizeAB}{0.199}
\newcommand{\cLRsizeAB}{0.198}
\newcommand{\cLMsizeAB}{0.195}
\newcommand{\dWDsizeAB}{0.055}
\newcommand{\dLRsizeAB}{0.054}
\newcommand{\dLMsizeAB}{0.054}
\newcommand{\cWDlpowerAB}{0.066}
\newcommand{\cLRlpowerAB}{0.064}
\newcommand{\cLMlpowerAB}{0.063}
\newcommand{\cWDrpowerAB}{0.549}
\newcommand{\cLRrpowerAB}{0.545}
\newcommand{\cLMrpowerAB}{0.540}
\newcommand{\dWDlpowerAB}{0.154}
\newcommand{\dLRlpowerAB}{0.151}
\newcommand{\dLMlpowerAB}{0.153}
\newcommand{\dWDrpowerAB}{0.194}
\newcommand{\dLRrpowerAB}{0.191}
\newcommand{\dLMrpowerAB}{0.190}
\newcommand{\ccWDlpowerAB}{0.008}
\newcommand{\ccLRlpowerAB}{0.008}
\newcommand{\ccLMlpowerAB}{0.008}
\newcommand{\ccWDrpowerAB}{0.236}
\newcommand{\ccLRrpowerAB}{0.237}
\newcommand{\ccLMrpowerAB}{0.235}
\newcommand{\ddWDlpowerAB}{0.143}
\newcommand{\ddLRlpowerAB}{0.143}
\newcommand{\ddLMlpowerAB}{0.145}
\newcommand{\ddWDrpowerAB}{0.181}
\newcommand{\ddLRrpowerAB}{0.182}
\newcommand{\ddLMrpowerAB}{0.181}
\newcommand{\cWDsizeAC}{0.560}
\newcommand{\cLRsizeAC}{0.556}
\newcommand{\cLMsizeAC}{0.532}
\newcommand{\dWDsizeAC}{0.089}
\newcommand{\dLRsizeAC}{0.088}
\newcommand{\dLMsizeAC}{0.076}
\newcommand{\cWDlpowerAC}{0.306}
\newcommand{\cLRlpowerAC}{0.305}
\newcommand{\cLMlpowerAC}{0.284}
\newcommand{\cWDrpowerAC}{0.791}
\newcommand{\cLRrpowerAC}{0.787}
\newcommand{\cLMrpowerAC}{0.769}
\newcommand{\dWDlpowerAC}{0.100}
\newcommand{\dLRlpowerAC}{0.097}
\newcommand{\dLMlpowerAC}{0.096}
\newcommand{\dWDrpowerAC}{0.309}
\newcommand{\dLRrpowerAC}{0.305}
\newcommand{\dLMrpowerAC}{0.279}
\newcommand{\ccWDlpowerAC}{0.008}
\newcommand{\ccLRlpowerAC}{0.008}
\newcommand{\ccLMlpowerAC}{0.009}
\newcommand{\ccWDrpowerAC}{0.187}
\newcommand{\ccLRrpowerAC}{0.185}
\newcommand{\ccLMrpowerAC}{0.181}
\newcommand{\ddWDlpowerAC}{0.055}
\newcommand{\ddLRlpowerAC}{0.052}
\newcommand{\ddLMlpowerAC}{0.062}
\newcommand{\ddWDrpowerAC}{0.210}
\newcommand{\ddLRrpowerAC}{0.208}
\newcommand{\ddLMrpowerAC}{0.208}
\newcommand{\cWDsizeAD}{0.593}
\newcommand{\cLRsizeAD}{0.591}
\newcommand{\cLMsizeAD}{0.586}
\newcommand{\dWDsizeAD}{0.056}
\newcommand{\dLRsizeAD}{0.055}
\newcommand{\dLMsizeAD}{0.055}
\newcommand{\cWDlpowerAD}{0.254}
\newcommand{\cLRlpowerAD}{0.253}
\newcommand{\cLMlpowerAD}{0.248}
\newcommand{\cWDrpowerAD}{0.871}
\newcommand{\cLRrpowerAD}{0.869}
\newcommand{\cLMrpowerAD}{0.866}
\newcommand{\dWDlpowerAD}{0.128}
\newcommand{\dLRlpowerAD}{0.127}
\newcommand{\dLMlpowerAD}{0.129}
\newcommand{\dWDrpowerAD}{0.225}
\newcommand{\dLRrpowerAD}{0.224}
\newcommand{\dLMrpowerAD}{0.224}
\newcommand{\ccWDlpowerAD}{0.005}
\newcommand{\ccLRlpowerAD}{0.005}
\newcommand{\ccLMlpowerAD}{0.005}
\newcommand{\ccWDrpowerAD}{0.226}
\newcommand{\ccLRrpowerAD}{0.227}
\newcommand{\ccLMrpowerAD}{0.225}
\newcommand{\ddWDlpowerAD}{0.119}
\newcommand{\ddLRlpowerAD}{0.119}
\newcommand{\ddLMlpowerAD}{0.120}
\newcommand{\ddWDrpowerAD}{0.213}
\newcommand{\ddLRrpowerAD}{0.213}
\newcommand{\ddLMrpowerAD}{0.212}
\newcommand{\cWDsizeAE}{0.326}
\newcommand{\cLRsizeAE}{0.311}
\newcommand{\cLMsizeAE}{0.272}
\newcommand{\dWDsizeAE}{0.098}
\newcommand{\dLRsizeAE}{0.091}
\newcommand{\dLMsizeAE}{0.077}
\newcommand{\cWDlpowerAE}{0.192}
\newcommand{\cLRlpowerAE}{0.180}
\newcommand{\cLMlpowerAE}{0.147}
\newcommand{\cWDrpowerAE}{0.619}
\newcommand{\cLRrpowerAE}{0.605}
\newcommand{\cLMrpowerAE}{0.563}
\newcommand{\dWDlpowerAE}{0.184}
\newcommand{\dLRlpowerAE}{0.171}
\newcommand{\dLMlpowerAE}{0.171}
\newcommand{\dWDrpowerAE}{0.335}
\newcommand{\dLRrpowerAE}{0.318}
\newcommand{\dLMrpowerAE}{0.294}
\newcommand{\ccWDlpowerAE}{0.014}
\newcommand{\ccLRlpowerAE}{0.014}
\newcommand{\ccLMlpowerAE}{0.016}
\newcommand{\ccWDrpowerAE}{0.196}
\newcommand{\ccLRrpowerAE}{0.193}
\newcommand{\ccLMrpowerAE}{0.196}
\newcommand{\ddWDlpowerAE}{0.114}
\newcommand{\ddLRlpowerAE}{0.115}
\newcommand{\ddLMlpowerAE}{0.127}
\newcommand{\ddWDrpowerAE}{0.233}
\newcommand{\ddLRrpowerAE}{0.234}
\newcommand{\ddLMrpowerAE}{0.231}
\newcommand{\cWDsizeAF}{0.260}
\newcommand{\cLRsizeAF}{0.255}
\newcommand{\cLMsizeAF}{0.248}
\newcommand{\dWDsizeAF}{0.056}
\newcommand{\dLRsizeAF}{0.053}
\newcommand{\dLMsizeAF}{0.057}
\newcommand{\cWDlpowerAF}{0.081}
\newcommand{\cLRlpowerAF}{0.079}
\newcommand{\cLMlpowerAF}{0.076}
\newcommand{\cWDrpowerAF}{0.680}
\newcommand{\cLRrpowerAF}{0.675}
\newcommand{\cLMrpowerAF}{0.668}
\newcommand{\dWDlpowerAF}{0.200}
\newcommand{\dLRlpowerAF}{0.195}
\newcommand{\dLMlpowerAF}{0.200}
\newcommand{\dWDrpowerAF}{0.268}
\newcommand{\dLRrpowerAF}{0.262}
\newcommand{\dLMrpowerAF}{0.267}
\newcommand{\ccWDlpowerAF}{0.005}
\newcommand{\ccLRlpowerAF}{0.005}
\newcommand{\ccLMlpowerAF}{0.005}
\newcommand{\ccWDrpowerAF}{0.288}
\newcommand{\ccLRrpowerAF}{0.288}
\newcommand{\ccLMrpowerAF}{0.288}
\newcommand{\ddWDlpowerAF}{0.185}
\newcommand{\ddLRlpowerAF}{0.187}
\newcommand{\ddLMlpowerAF}{0.184}
\newcommand{\ddWDrpowerAF}{0.248}
\newcommand{\ddLRrpowerAF}{0.252}
\newcommand{\ddLMrpowerAF}{0.247}
\newcommand{\cWDsizeAG}{0.591}
\newcommand{\cLRsizeAG}{0.582}
\newcommand{\cLMsizeAG}{0.552}
\newcommand{\dWDsizeAG}{0.174}
\newcommand{\dLRsizeAG}{0.167}
\newcommand{\dLMsizeAG}{0.136}
\newcommand{\cWDlpowerAG}{0.421}
\newcommand{\cLRlpowerAG}{0.412}
\newcommand{\cLMlpowerAG}{0.378}
\newcommand{\cWDrpowerAG}{0.792}
\newcommand{\cLRrpowerAG}{0.787}
\newcommand{\cLMrpowerAG}{0.765}
\newcommand{\dWDlpowerAG}{0.163}
\newcommand{\dLRlpowerAG}{0.160}
\newcommand{\dLMlpowerAG}{0.150}
\newcommand{\dWDrpowerAG}{0.438}
\newcommand{\dLRrpowerAG}{0.426}
\newcommand{\dLMrpowerAG}{0.399}
\newcommand{\ccWDlpowerAG}{0.013}
\newcommand{\ccLRlpowerAG}{0.016}
\newcommand{\ccLMlpowerAG}{0.015}
\newcommand{\ccWDrpowerAG}{0.128}
\newcommand{\ccLRrpowerAG}{0.127}
\newcommand{\ccLMrpowerAG}{0.126}
\newcommand{\ddWDlpowerAG}{0.040}
\newcommand{\ddLRlpowerAG}{0.039}
\newcommand{\ddLMlpowerAG}{0.051}
\newcommand{\ddWDrpowerAG}{0.206}
\newcommand{\ddLRrpowerAG}{0.201}
\newcommand{\ddLMrpowerAG}{0.209}
\newcommand{\cWDsizeAH}{0.666}
\newcommand{\cLRsizeAH}{0.663}
\newcommand{\cLMsizeAH}{0.656}
\newcommand{\dWDsizeAH}{0.060}
\newcommand{\dLRsizeAH}{0.058}
\newcommand{\dLMsizeAH}{0.059}
\newcommand{\cWDlpowerAH}{0.318}
\newcommand{\cLRlpowerAH}{0.314}
\newcommand{\cLMlpowerAH}{0.307}
\newcommand{\cWDrpowerAH}{0.912}
\newcommand{\cLRrpowerAH}{0.911}
\newcommand{\cLMrpowerAH}{0.908}
\newcommand{\dWDlpowerAH}{0.171}
\newcommand{\dLRlpowerAH}{0.169}
\newcommand{\dLMlpowerAH}{0.172}
\newcommand{\dWDrpowerAH}{0.322}
\newcommand{\dLRrpowerAH}{0.316}
\newcommand{\dLMrpowerAH}{0.320}
\newcommand{\ccWDlpowerAH}{0.005}
\newcommand{\ccLRlpowerAH}{0.005}
\newcommand{\ccLMlpowerAH}{0.005}
\newcommand{\ccWDrpowerAH}{0.236}
\newcommand{\ccLRrpowerAH}{0.236}
\newcommand{\ccLMrpowerAH}{0.238}
\newcommand{\ddWDlpowerAH}{0.153}
\newcommand{\ddLRlpowerAH}{0.153}
\newcommand{\ddLMlpowerAH}{0.154}
\newcommand{\ddWDrpowerAH}{0.291}
\newcommand{\ddLRrpowerAH}{0.291}
\newcommand{\ddLMrpowerAH}{0.291}



\newcommand{\bbbbOLSEbiasAA}{0.1861}
\newcommand{\bbbbOLSEstdeAA}{0.1562}
\newcommand{\bbbbOLSErmseAA}{0.2429}
\newcommand{\bbbbQMLEbiasAA}{-0.4968}
\newcommand{\bbbbQMLEstdeAA}{0.1910}
\newcommand{\bbbbQMLErmseAA}{0.5322}
\newcommand{\bbbbBCQMbiasAA}{-0.3323}
\newcommand{\bbbbBCQMstdeAA}{0.1580}
\newcommand{\bbbbBCQMrmseAA}{0.3680}
\newcommand{\bbbbCCEPbiasAA}{-0.1002}
\newcommand{\bbbbCCEPstdeAA}{0.2063}
\newcommand{\bbbbCCEPrmseAA}{0.2294}
\newcommand{\bbbbOLSEbiasAB}{0.0309}
\newcommand{\bbbbOLSEstdeAB}{0.0801}
\newcommand{\bbbbOLSErmseAB}{0.0859}
\newcommand{\bbbbQMLEbiasAB}{-0.9305}
\newcommand{\bbbbQMLEstdeAB}{0.1644}
\newcommand{\bbbbQMLErmseAB}{0.9449}
\newcommand{\bbbbBCQMbiasAB}{-0.7057}
\newcommand{\bbbbBCQMstdeAB}{0.1754}
\newcommand{\bbbbBCQMrmseAB}{0.7272}
\newcommand{\bbbbCCEPbiasAB}{-0.2750}
\newcommand{\bbbbCCEPstdeAB}{0.2302}
\newcommand{\bbbbCCEPrmseAB}{0.3586}
\newcommand{\bbbbOLSEbiasAC}{0.1989}
\newcommand{\bbbbOLSEstdeAC}{0.1185}
\newcommand{\bbbbOLSErmseAC}{0.2315}
\newcommand{\bbbbQMLEbiasAC}{-0.1569}
\newcommand{\bbbbQMLEstdeAC}{0.1018}
\newcommand{\bbbbQMLErmseAC}{0.1870}
\newcommand{\bbbbBCQMbiasAC}{-0.0758}
\newcommand{\bbbbBCQMstdeAC}{0.0700}
\newcommand{\bbbbBCQMrmseAC}{0.1031}
\newcommand{\bbbbCCEPbiasAC}{0.0036}
\newcommand{\bbbbCCEPstdeAC}{0.1074}
\newcommand{\bbbbCCEPrmseAC}{0.1074}
\newcommand{\bbbbOLSEbiasAD}{0.0326}
\newcommand{\bbbbOLSEstdeAD}{0.0543}
\newcommand{\bbbbOLSErmseAD}{0.0633}
\newcommand{\bbbbQMLEbiasAD}{-0.4209}
\newcommand{\bbbbQMLEstdeAD}{0.1607}
\newcommand{\bbbbQMLErmseAD}{0.4505}
\newcommand{\bbbbBCQMbiasAD}{-0.2732}
\newcommand{\bbbbBCQMstdeAD}{0.1235}
\newcommand{\bbbbBCQMrmseAD}{0.2998}
\newcommand{\bbbbCCEPbiasAD}{-0.1040}
\newcommand{\bbbbCCEPstdeAD}{0.1070}
\newcommand{\bbbbCCEPrmseAD}{0.1492}
\newcommand{\bbbbOLSEbiasAE}{0.2096}
\newcommand{\bbbbOLSEstdeAE}{0.0884}
\newcommand{\bbbbOLSErmseAE}{0.2274}
\newcommand{\bbbbQMLEbiasAE}{-0.0592}
\newcommand{\bbbbQMLEstdeAE}{0.0377}
\newcommand{\bbbbQMLErmseAE}{0.0702}
\newcommand{\bbbbBCQMbiasAE}{-0.0185}
\newcommand{\bbbbBCQMstdeAE}{0.0287}
\newcommand{\bbbbBCQMrmseAE}{0.0341}
\newcommand{\bbbbCCEPbiasAE}{0.0520}
\newcommand{\bbbbCCEPstdeAE}{0.0711}
\newcommand{\bbbbCCEPrmseAE}{0.0881}
\newcommand{\bbbbOLSEbiasAF}{0.0366}
\newcommand{\bbbbOLSEstdeAF}{0.0356}
\newcommand{\bbbbOLSErmseAF}{0.0511}
\newcommand{\bbbbQMLEbiasAF}{-0.0741}
\newcommand{\bbbbQMLEstdeAF}{0.0859}
\newcommand{\bbbbQMLErmseAF}{0.1134}
\newcommand{\bbbbBCQMbiasAF}{-0.0406}
\newcommand{\bbbbBCQMstdeAF}{0.0552}
\newcommand{\bbbbBCQMrmseAF}{0.0686}
\newcommand{\bbbbCCEPbiasAF}{-0.0310}
\newcommand{\bbbbCCEPstdeAF}{0.0512}
\newcommand{\bbbbCCEPrmseAF}{0.0599}
\newcommand{\bbbbOLSEbiasAG}{0.2174}
\newcommand{\bbbbOLSEstdeAG}{0.0649}
\newcommand{\bbbbOLSErmseAG}{0.2269}
\newcommand{\bbbbQMLEbiasAG}{-0.0275}
\newcommand{\bbbbQMLEstdeAG}{0.0192}
\newcommand{\bbbbQMLErmseAG}{0.0335}
\newcommand{\bbbbBCQMbiasAG}{-0.0054}
\newcommand{\bbbbBCQMstdeAG}{0.0170}
\newcommand{\bbbbBCQMrmseAG}{0.0179}
\newcommand{\bbbbCCEPbiasAG}{0.0759}
\newcommand{\bbbbCCEPstdeAG}{0.0500}
\newcommand{\bbbbCCEPrmseAG}{0.0908}
\newcommand{\bbbbOLSEbiasAH}{0.0404}
\newcommand{\bbbbOLSEstdeAH}{0.0239}
\newcommand{\bbbbOLSErmseAH}{0.0469}
\newcommand{\bbbbQMLEbiasAH}{-0.0134}
\newcommand{\bbbbQMLEstdeAH}{0.0166}
\newcommand{\bbbbQMLErmseAH}{0.0214}
\newcommand{\bbbbBCQMbiasAH}{-0.0047}
\newcommand{\bbbbBCQMstdeAH}{0.0122}
\newcommand{\bbbbBCQMrmseAH}{0.0131}
\newcommand{\bbbbCCEPbiasAH}{-0.0012}
\newcommand{\bbbbCCEPstdeAH}{0.0281}
\newcommand{\bbbbCCEPrmseAH}{0.0281}
\newcommand{\bbbbOLSEbiasAI}{0.2232}
\newcommand{\bbbbOLSEstdeAI}{0.0472}
\newcommand{\bbbbOLSErmseAI}{0.2281}
\newcommand{\bbbbQMLEbiasAI}{-0.0134}
\newcommand{\bbbbQMLEstdeAI}{0.0118}
\newcommand{\bbbbQMLErmseAI}{0.0179}
\newcommand{\bbbbBCQMbiasAI}{-0.0016}
\newcommand{\bbbbBCQMstdeAI}{0.0113}
\newcommand{\bbbbBCQMrmseAI}{0.0114}
\newcommand{\bbbbCCEPbiasAI}{0.0873}
\newcommand{\bbbbCCEPstdeAI}{0.0364}
\newcommand{\bbbbCCEPrmseAI}{0.0946}
\newcommand{\bbbbOLSEbiasAJ}{0.0433}
\newcommand{\bbbbOLSEstdeAJ}{0.0164}
\newcommand{\bbbbOLSErmseAJ}{0.0463}
\newcommand{\bbbbQMLEbiasAJ}{-0.0052}
\newcommand{\bbbbQMLEstdeAJ}{0.0066}
\newcommand{\bbbbQMLErmseAJ}{0.0084}
\newcommand{\bbbbBCQMbiasAJ}{-0.0012}
\newcommand{\bbbbBCQMstdeAJ}{0.0058}
\newcommand{\bbbbBCQMrmseAJ}{0.0059}
\newcommand{\bbbbCCEPbiasAJ}{0.0125}
\newcommand{\bbbbCCEPstdeAJ}{0.0176}
\newcommand{\bbbbCCEPrmseAJ}{0.0216}

\begin{table}[H]
   \begin{center}
   \caption{\label{tab:T1} %
                           Simulation results for the ${\rm AR}(1)$ model described in the main text with
                           $N=100$, $\rho_f=0.5$, $\sigma_f=0.5$, and different values
                           of $T$ (with corresponding bandwidth $M$) and true AR(1) coefficient $\rho^0$.
                           }
                         
           \begin{tabular}{ll@{\qquad}lll@{\qquad \;}lll}
    \hline 
                  &   &     & $\rho^0=0.3$ & & & $\rho^0=0.9$ \\[0.1cm]
                  &   & OLS & FLS & BC-FLS & OLS & FLS & BC-FLS  \\
      \hline 
    $T=5$ &  bias &   \bOLSEbiasAA &  \bQMLEbiasAA & \bBCQMbiasAA &  \bOLSEbiasAB &  \bQMLEbiasAB & \bBCQMbiasAB \\
      $(M=2)$             &  std  &   \bOLSEstdeAA &  \bQMLEstdeAA & \bBCQMstdeAA &  \bOLSEstdeAB &  \bQMLEstdeAB & \bBCQMstdeAB \\
                  &  rmse &   \bOLSErmseAA &  \bQMLErmseAA & \bBCQMrmseAA &  \bOLSErmseAB &  \bQMLErmseAB & \bBCQMrmseAB \\[8pt]
   $T=10$  &  bias &   \bOLSEbiasAC &  \bQMLEbiasAC & \bBCQMbiasAC &  \bOLSEbiasAD &  \bQMLEbiasAD & \bBCQMbiasAD \\
       $(M=3)$           &  std  &   \bOLSEstdeAC &  \bQMLEstdeAC & \bBCQMstdeAC &  \bOLSEstdeAD &  \bQMLEstdeAD & \bBCQMstdeAD \\
                  &  rmse &   \bOLSErmseAC &  \bQMLErmseAC & \bBCQMrmseAC &  \bOLSErmseAD &  \bQMLErmseAD & \bBCQMrmseAD \\[8pt]
   $T=20$  &  bias &   \bOLSEbiasAE &  \bQMLEbiasAE & \bBCQMbiasAE &  \bOLSEbiasAF &  \bQMLEbiasAF & \bBCQMbiasAF \\
   $(M=4)$               &  std  &   \bOLSEstdeAE &  \bQMLEstdeAE & \bBCQMstdeAE &  \bOLSEstdeAF &  \bQMLEstdeAF & \bBCQMstdeAF \\
                  &  rmse &   \bOLSErmseAE &  \bQMLErmseAE & \bBCQMrmseAE &  \bOLSErmseAF &  \bQMLErmseAF & \bBCQMrmseAF \\[8pt]
   $T=40$  &  bias &   \bOLSEbiasAG &  \bQMLEbiasAG & \bBCQMbiasAG &  \bOLSEbiasAH &  \bQMLEbiasAH & \bBCQMbiasAH \\
     $(M=5)$             &  std  &   \bOLSEstdeAG &  \bQMLEstdeAG & \bBCQMstdeAG &  \bOLSEstdeAH &  \bQMLEstdeAH & \bBCQMstdeAH \\
                  &  rmse &   \bOLSErmseAG &  \bQMLErmseAG & \bBCQMrmseAG &  \bOLSErmseAH &  \bQMLErmseAH & \bBCQMrmseAH \\[8pt]
   $T=80$ &  bias &   \bOLSEbiasAI &  \bQMLEbiasAI & \bBCQMbiasAI &  \bOLSEbiasAJ &  \bQMLEbiasAJ & \bBCQMbiasAJ \\
      $(M=6)$            &  std  &   \bOLSEstdeAI &  \bQMLEstdeAI & \bBCQMstdeAI &  \bOLSEstdeAJ &  \bQMLEstdeAJ & \bBCQMstdeAJ \\
                  &  rmse &   \bOLSErmseAI &  \bQMLErmseAI & \bBCQMrmseAI &  \bOLSErmseAJ &  \bQMLErmseAJ & \bBCQMrmseAJ \\
      \hline
   \end{tabular}
   \end{center}
\end{table}

\begin{table}[H]
   \begin{center}
   \caption{\label{tab:T2} %
                   Same DGP as Table~\ref{tab:T1},
                   but misspecification in number of factors $R$ is present.
                   The true number of factors is $R=1$,
                   but the FLS and BC-FLS are calculated with $R=2$.
                    }
         \begin{tabular}{ll@{\qquad}lll@{\qquad \;}lll}
    \hline 
                  &   &     & $\rho^0=0.3$ & & & $\rho^0=0.9$ \\[0.1cm]
                  &   & OLS & FLS & BC-FLS & OLS & FLS & BC-FLS  \\
      \hline 
    $T=5$ &  bias &   \bbOLSEbiasAA &  \bbQMLEbiasAA & \bbBCQMbiasAA &  \bbOLSEbiasAB &  \bbQMLEbiasAB & \bbBCQMbiasAB \\
      $(M=2)$             &  std  &   \bbOLSEstdeAA &  \bbQMLEstdeAA & \bbBCQMstdeAA &  \bbOLSEstdeAB &  \bbQMLEstdeAB & \bbBCQMstdeAB \\
                  &  rmse &   \bbOLSErmseAA &  \bbQMLErmseAA & \bbBCQMrmseAA &  \bbOLSErmseAB &  \bbQMLErmseAB & \bbBCQMrmseAB \\[8pt]
   $T=10$  &  bias &   \bbOLSEbiasAC &  \bbQMLEbiasAC & \bbBCQMbiasAC &  \bbOLSEbiasAD &  \bbQMLEbiasAD & \bbBCQMbiasAD \\
       $(M=3)$           &  std  &   \bbOLSEstdeAC &  \bbQMLEstdeAC & \bbBCQMstdeAC &  \bbOLSEstdeAD &  \bbQMLEstdeAD & \bbBCQMstdeAD \\
                  &  rmse &   \bbOLSErmseAC &  \bbQMLErmseAC & \bbBCQMrmseAC &  \bbOLSErmseAD &  \bbQMLErmseAD & \bbBCQMrmseAD \\[8pt]
   $T=20$  &  bias &   \bbOLSEbiasAE &  \bbQMLEbiasAE & \bbBCQMbiasAE &  \bbOLSEbiasAF &  \bbQMLEbiasAF & \bbBCQMbiasAF \\
   $(M=4)$               &  std  &   \bbOLSEstdeAE &  \bbQMLEstdeAE & \bbBCQMstdeAE &  \bbOLSEstdeAF &  \bbQMLEstdeAF & \bbBCQMstdeAF \\
                  &  rmse &   \bbOLSErmseAE &  \bbQMLErmseAE & \bbBCQMrmseAE &  \bbOLSErmseAF &  \bbQMLErmseAF & \bbBCQMrmseAF \\[8pt]
   $T=40$  &  bias &   \bbOLSEbiasAG &  \bbQMLEbiasAG & \bbBCQMbiasAG &  \bbOLSEbiasAH &  \bbQMLEbiasAH & \bbBCQMbiasAH \\
     $(M=5)$             &  std  &   \bbOLSEstdeAG &  \bbQMLEstdeAG & \bbBCQMstdeAG &  \bbOLSEstdeAH &  \bbQMLEstdeAH & \bbBCQMstdeAH \\
                  &  rmse &   \bbOLSErmseAG &  \bbQMLErmseAG & \bbBCQMrmseAG &  \bbOLSErmseAH &  \bbQMLErmseAH & \bbBCQMrmseAH \\[8pt]
   $T=80$ &  bias &   \bbOLSEbiasAI &  \bbQMLEbiasAI & \bbBCQMbiasAI &  \bbOLSEbiasAJ &  \bbQMLEbiasAJ & \bbBCQMbiasAJ \\
      $(M=6)$            &  std  &   \bbOLSEstdeAI &  \bbQMLEstdeAI & \bbBCQMstdeAI &  \bbOLSEstdeAJ &  \bbQMLEstdeAJ & \bbBCQMstdeAJ \\
                  &  rmse &   \bbOLSErmseAI &  \bbQMLErmseAI & \bbBCQMrmseAI &  \bbOLSErmseAJ &  \bbQMLErmseAJ & \bbBCQMrmseAJ \\
      \hline
   \end{tabular}
   \end{center}
\end{table}

\begin{table}[H]
   \begin{center}
   \caption{\label{tab:T3} %
                           Simulation results for the ${\rm AR}(1)$ model with
                           $N=100$, $T=20$, $\rho_f=0.5$, and $\sigma_f=0.5$. For different values of the ${\rm AR}(1)$
                           coefficient $\rho^0$ and of the bandwidth $M$, we give the fraction of the LS estimator bias
                           that is accounted for by the bias correction,
                           i.e. the fraction
              $\sqrt{NT} \, \mathbb{E}(\widehat \beta-\beta)/\mathbb{E}(\widehat W^{-1}\widehat B)$,
                           computed over 10,000 simulation runs. Here and in all following tables it is assumed that $R=1$
                           is correctly specified. }
          \begin{tabular}{l@{\qquad}l@{\; \;}l@{\; \;}l@{\; \;}l@{\; \;}l@{\; \;}l@{\; \;}l@{\; \;}l}
     \hline 
                  &   $M=1$ &  $M=2$ & $M=3$  & $M=4$  & $M=5$  & $M=6$ & $M=7$ & $M=8$\\
      \hline
      $\rho^0=0$   &  \xBCfractionAA & \xBCfractionAB & \xBCfractionAC & \xBCfractionAD & \xBCfractionAE & \xBCfractionAF & \xBCfractionAG & \xBCfractionAH \\
      $\rho^0=0.3$ &  \xBCfractionAI & \xBCfractionAJ & \xBCfractionAK & \xBCfractionAL &  \xBCfractionAM & \xBCfractionAN & \xBCfractionAO & \xBCfractionAP \\
      $\rho^0=0.6$ &  \xBCfractionAQ & \xBCfractionAR & \xBCfractionAS & \xBCfractionAT & \xBCfractionAU & \xBCfractionAV & \xBCfractionAW & \xBCfractionAX \\
      $\rho^0=0.9$ &   \xBCfractionAY & \xBCfractionAZ & \xBCfractionBA & \xBCfractionBB & \xBCfractionBC & \xBCfractionBD & \xBCfractionBE & \xBCfractionBF
      \\  \hline
   \end{tabular}
   \end{center}
\end{table}

\begin{table}[H]
   \begin{center}
   \caption{\label{tab:T4} %
                   Same specification as Table~\ref{tab:T1}.
                   We only report the properties of the bias-corrected LS estimator, but
                   for multiple values of the bandwidth parameter $M$ and two different values for $T$.
                   Results were obtained using 10,000 simulation runs. }
         \begin{tabular}{ll@{\quad \;}lll@{\qquad}lll}
    \hline 
                  &   &   \multicolumn{3}{c}{BC-FLS for $\rho^0=0.3$}    & \multicolumn{3}{c}{ BC-FLS for $\rho^0=0.9$}  \\[0.1cm]
       & & M=2 & M=5 & M=8 & M=2 & M=5 & M=8
   \\
      \hline
    $T=20$  &  bias &  \bbbBCQMbiasAA & \bbbBCQMbiasAC & \bbbBCQMbiasAE &  \bbbBCQMbiasAB & \bbbBCQMbiasAD & \bbbBCQMbiasAF \\
                  &  std  &  \bbbBCQMstdeAA & \bbbBCQMstdeAC & \bbbBCQMstdeAE &    \bbbBCQMstdeAB &    \bbbBCQMstdeAD &    \bbbBCQMstdeAF \\
                  &  rmse &   \bbbBCQMrmseAA &  \bbbBCQMrmseAC &  \bbbBCQMrmseAE &   \bbbBCQMrmseAB &   \bbbBCQMrmseAD &  \bbbBCQMrmseAF
            \\[8pt]
    $T=40$  &  bias &  \bbbBCQMbiasAG & \bbbBCQMbiasAI & \bbbBCQMbiasAK &  \bbbBCQMbiasAH & \bbbBCQMbiasAJ & \bbbBCQMbiasAL \\
                  &  std  &  \bbbBCQMstdeAG & \bbbBCQMstdeAI & \bbbBCQMstdeAK &    \bbbBCQMstdeAH &    \bbbBCQMstdeAJ &    \bbbBCQMstdeAL \\
                  &  rmse &   \bbbBCQMrmseAG &  \bbbBCQMrmseAI &  \bbbBCQMrmseAK &   \bbbBCQMrmseAH &   \bbbBCQMrmseAJ &  \bbbBCQMrmseAL
     \\ \hline
   \end{tabular}
   \end{center}
\end{table}

\begin{table}[H]
   \begin{center}
   \caption{\label{tab:T5} %
                           Simulation results for the ${\rm AR}(1)$ model with
                           $N=100$, $T=20$, $M=4$, and $\rho^0=0.6$.
                           The three different estimators were computed for 10,000 simulation runs, and
                           the mean bias, standard deviation (std), and root mean square error (rmse)
                           are reported. }
              \begin{tabular}{l@{\; \;}l@{\quad \;}lll@{\qquad}lll}
    \hline
                  &   &     & $\rho_f=0.3$ & & & $\rho_f=0.7$ \\[0.1cm]
                  &   & OLS & FLS & BC-FLS & OLS & FLS & BC-FLS  \\
      \hline 
      $\sigma_f=0$    &  bias &   \aOLSEbiasAA &  \aQMLEbiasAA & \aBCQMbiasAA &  \aOLSEbiasAB &  \aQMLEbiasAB & \aBCQMbiasAB \\
                  &  std  &   \aOLSEstdeAA &  \aQMLEstdeAA & \aBCQMstdeAA &  \aOLSEstdeAB &  \aQMLEstdeAB & \aBCQMstdeAB \\
                  &  rmse &   \aOLSErmseAA &  \aQMLErmseAA & \aBCQMrmseAA &  \aOLSErmseAB &  \aQMLErmseAB & \aBCQMrmseAB 
                   \\[8pt]
    $\sigma_f=0.2$    &  bias &   \aOLSEbiasAC &  \aQMLEbiasAC & \aBCQMbiasAC &  \aOLSEbiasAD &  \aQMLEbiasAD & \aBCQMbiasAD \\
                  &  std  &   \aOLSEstdeAC &  \aQMLEstdeAC & \aBCQMstdeAC &  \aOLSEstdeAD &  \aQMLEstdeAD & \aBCQMstdeAD \\
                  &  rmse &   \aOLSErmseAC &  \aQMLErmseAC & \aBCQMrmseAC &  \aOLSErmseAD &  \aQMLErmseAD & \aBCQMrmseAD
                   \\[8pt]
$\sigma_f=0.5$    &  bias &   \aOLSEbiasAE &  \aQMLEbiasAE & \aBCQMbiasAE &  \aOLSEbiasAF &  \aQMLEbiasAF & \aBCQMbiasAF \\
                  &  std  &   \aOLSEstdeAE &  \aQMLEstdeAE & \aBCQMstdeAE &  \aOLSEstdeAF &  \aQMLEstdeAF & \aBCQMstdeAF \\
                  &  rmse &   \aOLSErmseAE &  \aQMLErmseAE & \aBCQMrmseAE &  \aOLSErmseAF &  \aQMLErmseAF & \aBCQMrmseAF
                  \\ \hline
   \end{tabular}
   \end{center}
\end{table}

\begin{table}[H]
   \begin{center}
         \caption{\label{tab:T6} %
                           Simulation results for the ${\rm AR}(1)$ model with
                           $\rho_f=0.5$ and $\sigma_f=0.5$. For the different values of $\rho^0$, $N$, $T$, and $M$,
                           we test the hypothesis $H_0:\rho=\rho^0$ using the uncorrected and
                           bias-corrected Wald, LR, and LM test, and nominal size $5\%$.
                           The bias-corrected tests are indicated by an asterisk superscript.
                           The size of the different tests is reported, based on 10,000 simulation runs.
                         }
          \begin{tabular}{l@{\qquad}lll@{\qquad}lll}
      \hline
     \multicolumn{1}{c}{\phantom{a}} &  \multicolumn{3}{c}{size} &  \multicolumn{3}{c}{size}  \\[0.1cm]
                  & $WD$ & $LR$ & $LM$ & $WD^*$ & $LR^*$ & $LM^*$  \\
      \hline 
            $\rho^0=0$ \\
        $N=100$, $T=20$, $M=4$ &
             \cWDsizeAA & \cLRsizeAA & \cLMsizeAA & \dWDsizeAA & \dLRsizeAA & \dLMsizeAA \\
                $N=400$, $T=80$, $M=6$ &
             \cWDsizeAB & \cLRsizeAB & \cLMsizeAB & \dWDsizeAB & \dLRsizeAB & \dLMsizeAB \\
                $N=400$, $T=20$, $M=4$ &
             \cWDsizeAC & \cLRsizeAC & \cLMsizeAC & \dWDsizeAC & \dLRsizeAC & \dLMsizeAC \\
                $N=1600$, $T=80$, $M=6$ &
             \cWDsizeAD & \cLRsizeAD & \cLMsizeAD & \dWDsizeAD & \dLRsizeAD & \dLMsizeAD 
             \\[9pt]
      $\rho^0=0.6$ \\
       $N=100$, $T=20$, $M=4$ &
             \cWDsizeAE & \cLRsizeAE & \cLMsizeAE & \dWDsizeAE & \dLRsizeAE & \dLMsizeAE \\
                $N=400$, $T=80$, $M=6$ &
             \cWDsizeAF & \cLRsizeAF & \cLMsizeAF & \dWDsizeAF & \dLRsizeAF & \dLMsizeAF \\
                $N=400$, $T=20$, $M=4$ &
             \cWDsizeAG & \cLRsizeAG & \cLMsizeAG & \dWDsizeAG & \dLRsizeAG & \dLMsizeAG \\
                $N=1600$, $T=80$, $M=6$ &
             \cWDsizeAH & \cLRsizeAH & \cLMsizeAH & \dWDsizeAH & \dLRsizeAH & \dLMsizeAH
          \\
      \hline       
   \end{tabular}
   \end{center}
\end{table}

\begin{table}[H]
   \begin{center}
   \caption{\label{tab:T7} %
                           As Table~\ref{tab:T6}, but we report the power for testing the alternatives
                           $H^{\rm left}_{a}: \rho=\rho^0-(NT)^{-1/2}$ and
                           $H^{\rm right}_{a}: \rho=\rho^0+(NT)^{-1/2}$.
                           The bias-corrected tests are indicated by an asterisk superscript.
                         }
              \begin{tabular}{l@{\; \;}l@{\quad \;}lll@{\qquad}lll}
     \multicolumn{2}{c}{\phantom{a}} &  \multicolumn{3}{c}{power} &  \multicolumn{3}{c}{power}
        \\[0.1cm]
                  & & $WD$ & $LR$ & $LM$ & $WD^*$ & $LR^*$ & $LM^*$  \\
      \hline  
      $\rho^0=0$ \\[8pt]
       $N=100$, $T=20$, $M=4$
       & $H^{\rm left}_{a}$ & \cWDlpowerAA & \cLRlpowerAA & \cLMlpowerAA & \dWDlpowerAA & \dLRlpowerAA & \dLMlpowerAA \\
     & $H^{\rm right}_{a}$ & \cWDrpowerAA & \cLRrpowerAA & \cLMrpowerAA & \dWDrpowerAA & \dLRrpowerAA & \dLMrpowerAA \\
                $N=400$, $T=80$, $M=6$
       & $H^{\rm left}_{a}$ & \cWDlpowerAB & \cLRlpowerAB & \cLMlpowerAB & \dWDlpowerAB & \dLRlpowerAB & \dLMlpowerAB \\
     & $H^{\rm right}_{a}$ & \cWDrpowerAB & \cLRrpowerAB & \cLMrpowerAB & \dWDrpowerAB & \dLRrpowerAB & \dLMrpowerAB \\
                $N=400$, $T=20$, $M=4$
       & $H^{\rm left}_{a}$ & \cWDlpowerAC & \cLRlpowerAC & \cLMlpowerAC & \dWDlpowerAC & \dLRlpowerAC & \dLMlpowerAC \\
     & $H^{\rm right}_{a}$ & \cWDrpowerAC & \cLRrpowerAC & \cLMrpowerAC & \dWDrpowerAC & \dLRrpowerAC & \dLMrpowerAC \\
                $N=1600$, $T=80$, $M=6$
       & $H^{\rm left}_{a}$ & \cWDlpowerAD & \cLRlpowerAD & \cLMlpowerAD & \dWDlpowerAD & \dLRlpowerAD & \dLMlpowerAD \\
     & $H^{\rm right}_{a}$ & \cWDrpowerAD & \cLRrpowerAD & \cLMrpowerAD & \dWDrpowerAD & \dLRrpowerAD & \dLMrpowerAD 
     \\[8pt]
      $\rho^0=0.6$ \\[8pt]
       $N=100$, $T=20$, $M=4$
       & $H^{\rm left}_{a}$ & \cWDlpowerAE & \cLRlpowerAE & \cLMlpowerAE & \dWDlpowerAE & \dLRlpowerAE & \dLMlpowerAE \\
     & $H^{\rm right}_{a}$ & \cWDrpowerAE & \cLRrpowerAE & \cLMrpowerAE & \dWDrpowerAE & \dLRrpowerAE & \dLMrpowerAE \\
                $N=400$, $T=80$, $M=6$
       & $H^{\rm left}_{a}$ & \cWDlpowerAF & \cLRlpowerAF & \cLMlpowerAF & \dWDlpowerAE & \dLRlpowerAF & \dLMlpowerAF \\
     & $H^{\rm right}_{a}$ & \cWDrpowerAF & \cLRrpowerAF & \cLMrpowerAF & \dWDrpowerAE & \dLRrpowerAF & \dLMrpowerAF \\
                $N=400$, $T=20$, $M=4$
       & $H^{\rm left}_{a}$ & \cWDlpowerAG & \cLRlpowerAG & \cLMlpowerAG & \dWDlpowerAE & \dLRlpowerAG & \dLMlpowerAG \\
     & $H^{\rm right}_{a}$ & \cWDrpowerAG & \cLRrpowerAG & \cLMrpowerAG & \dWDrpowerAE & \dLRrpowerAG & \dLMrpowerAG \\
                $N=1600$, $T=80$, $M=6$
       & $H^{\rm left}_{a}$ & \cWDlpowerAH & \cLRlpowerAH & \cLMlpowerAH & \dWDlpowerAF & \dLRlpowerAH & \dLMlpowerAH \\
     & $H^{\rm right}_{a}$ & \cWDrpowerAH & \cLRrpowerAH & \cLMrpowerAH & \dWDrpowerAF & \dLRrpowerAH & \dLMrpowerAH \\
     \hline
   \end{tabular}
   \end{center}
\end{table}

\begin{table}[H]
   \begin{center}
   \caption{\label{tab:T8} %
                           As Table~\ref{tab:T7}, but we report the size-corrected power.
                         }
              \begin{tabular}{l@{\; \;}l@{\quad \;}lll@{\qquad}lll}
     \multicolumn{2}{c}{\phantom{a}} &  \multicolumn{3}{c}{size-corrected power} &  \multicolumn{3}{c}{size-corrected power}
        \\[0.1cm]
                  & & $WD$ & $LR$ & $LM$ & $WD^*$ & $LR^*$ & $LM^*$  \\
      \hline  
      $\rho^0=0$ \\[8pt]
       $N=100$, $T=20$, $M=4$
       & $H^{\rm left}_{a}$ & \ccWDlpowerAA & \ccLRlpowerAA & \ccLMlpowerAA & \ddWDlpowerAA & \ddLRlpowerAA & \ddLMlpowerAA \\
     & $H^{\rm right}_{a}$ & \ccWDrpowerAA & \ccLRrpowerAA & \ccLMrpowerAA & \ddWDrpowerAA & \ddLRrpowerAA & \ddLMrpowerAA \\
                $N=400$, $T=80$, $M=6$
       & $H^{\rm left}_{a}$ & \ccWDlpowerAB & \ccLRlpowerAB & \ccLMlpowerAB & \ddWDlpowerAB & \ddLRlpowerAB & \ddLMlpowerAB \\
     & $H^{\rm right}_{a}$ & \ccWDrpowerAB & \ccLRrpowerAB & \ccLMrpowerAB & \ddWDrpowerAB & \ddLRrpowerAB & \ddLMrpowerAB \\
                $N=400$, $T=20$, $M=4$
       & $H^{\rm left}_{a}$ & \ccWDlpowerAC & \ccLRlpowerAC & \ccLMlpowerAC & \ddWDlpowerAC & \ddLRlpowerAC & \ddLMlpowerAC \\
     & $H^{\rm right}_{a}$ & \ccWDrpowerAC & \ccLRrpowerAC & \ccLMrpowerAC & \ddWDrpowerAC & \ddLRrpowerAC & \ddLMrpowerAC \\
                $N=1600$, $T=80$, $M=6$
       & $H^{\rm left}_{a}$ & \ccWDlpowerAD & \ccLRlpowerAD & \ccLMlpowerAD & \ddWDlpowerAD & \ddLRlpowerAD & \ddLMlpowerAD \\
     & $H^{\rm right}_{a}$ & \ccWDrpowerAD & \ccLRrpowerAD & \ccLMrpowerAD & \ddWDrpowerAD & \ddLRrpowerAD & \ddLMrpowerAD 
     \\[8pt]
      $\rho^0=0.6$ \\[8pt]
       $N=100$, $T=20$, $M=4$
       & $H^{\rm left}_{a}$ & \ccWDlpowerAE & \ccLRlpowerAE & \ccLMlpowerAE & \ddWDlpowerAE & \ddLRlpowerAE & \ddLMlpowerAE \\
     & $H^{\rm right}_{a}$ & \ccWDrpowerAE & \ccLRrpowerAE & \ccLMrpowerAE & \ddWDrpowerAE & \ddLRrpowerAE & \ddLMrpowerAE \\
                $N=400$, $T=80$, $M=6$
       & $H^{\rm left}_{a}$ & \ccWDlpowerAF & \ccLRlpowerAF & \ccLMlpowerAF & \ddWDlpowerAE & \ddLRlpowerAF & \ddLMlpowerAF \\
     & $H^{\rm right}_{a}$ & \ccWDrpowerAF & \ccLRrpowerAF & \ccLMrpowerAF & \ddWDrpowerAE & \ddLRrpowerAF & \ddLMrpowerAF \\
                $N=400$, $T=20$, $M=4$
       & $H^{\rm left}_{a}$ & \ccWDlpowerAG & \ccLRlpowerAG & \ccLMlpowerAG & \ddWDlpowerAE & \ddLRlpowerAG & \ddLMlpowerAG \\
     & $H^{\rm right}_{a}$ & \ccWDrpowerAG & \ccLRrpowerAG & \ccLMrpowerAG & \ddWDrpowerAE & \ddLRrpowerAG & \ddLMrpowerAG \\
                $N=1600$, $T=80$, $M=6$
       & $H^{\rm left}_{a}$ & \ccWDlpowerAH & \ccLRlpowerAH & \ccLMlpowerAH & \ddWDlpowerAF & \ddLRlpowerAH & \ddLMlpowerAH \\
     & $H^{\rm right}_{a}$ & \ccWDrpowerAH & \ccLRrpowerAH & \ccLMrpowerAH & \ddWDrpowerAF & \ddLRrpowerAH & \ddLMrpowerAH \\
     \hline
   \end{tabular}
   \end{center}

\end{table}

\begin{table}[H]
   \begin{center}
   \caption{\label{tab:extra1} %
                      Same as Table~\ref{tab:T1} in main paper, but also reporting pooled CCE estimator of
                      Pesaran~(2006). }
      \begin{tabular}{l@{\;\,}l@{\quad}l@{\;\;}l@{\;\;}l@{\;\;}l@{\quad}l@{\;\;}l@{\;\;}l@{\;\;}l}
     \hline 
                  &   &  \multicolumn{4}{c}{$\rho^0=0.3$}    & \multicolumn{4}{c}{$\rho^0=0.9$}    \\[0.1cm]
                  &   & OLS & FLS & BC-FLS & CCE &  OLS & FLS & BC-FLS  & CCE \\
      \hline 
    $T=5$  &  bias &   \bOLSEbiasAA &  \bQMLEbiasAA & \bBCQMbiasAA & \bCCEPbiasAA &  \bOLSEbiasAB &  \bQMLEbiasAB & \bBCQMbiasAB & \bCCEPbiasAB \\
     $(M=2)$             &  std  &   \bOLSEstdeAA &  \bQMLEstdeAA & \bBCQMstdeAA & \bCCEPstdeAA &  \bOLSEstdeAB &  \bQMLEstdeAB & \bBCQMstdeAB & \bCCEPstdeAB \\
                  &  rmse &   \bOLSErmseAA &  \bQMLErmseAA & \bBCQMrmseAA & \bCCEPrmseAA &  \bOLSErmseAB &  \bQMLErmseAB & \bBCQMrmseAB & \bCCEPrmseAB
                   \\[8pt]
   $T=10$  &  bias &   \bOLSEbiasAC &  \bQMLEbiasAC & \bBCQMbiasAC & \bCCEPbiasAC &  \bOLSEbiasAD &  \bQMLEbiasAD & \bBCQMbiasAD & \bCCEPbiasAD \\
   $(M=3)$               &  std  &   \bOLSEstdeAC &  \bQMLEstdeAC & \bBCQMstdeAC & \bCCEPstdeAC &  \bOLSEstdeAD &  \bQMLEstdeAD & \bBCQMstdeAD & \bCCEPstdeAD \\
                  &  rmse &   \bOLSErmseAC &  \bQMLErmseAC & \bBCQMrmseAC & \bCCEPrmseAC &  \bOLSErmseAD &  \bQMLErmseAD & \bBCQMrmseAD & \bCCEPrmseAD 
                   \\[8pt]
   $T=20$  &  bias &   \bOLSEbiasAE &  \bQMLEbiasAE & \bBCQMbiasAE & \bCCEPbiasAE &  \bOLSEbiasAF &  \bQMLEbiasAF & \bBCQMbiasAF & \bCCEPbiasAF \\
    $(M=4)$              &  std  &   \bOLSEstdeAE &  \bQMLEstdeAE & \bBCQMstdeAE & \bCCEPstdeAE &  \bOLSEstdeAF &  \bQMLEstdeAF & \bBCQMstdeAF & \bCCEPstdeAF \\
                  &  rmse &   \bOLSErmseAE &  \bQMLErmseAE & \bBCQMrmseAE & \bCCEPrmseAE &  \bOLSErmseAF &  \bQMLErmseAF & \bBCQMrmseAF & \bCCEPrmseAF 
                   \\[8pt]
        $T=40$  &  bias &   \bOLSEbiasAG &  \bQMLEbiasAG & \bBCQMbiasAG & \bCCEPbiasAG &  \bOLSEbiasAH &  \bQMLEbiasAH & \bBCQMbiasAH & \bCCEPbiasAH \\
    $(M=5)$              &  std  &   \bOLSEstdeAG &  \bQMLEstdeAG & \bBCQMstdeAG & \bCCEPstdeAG &  \bOLSEstdeAH &  \bQMLEstdeAH & \bBCQMstdeAH & \bCCEPstdeAH \\
                  &  rmse &   \bOLSErmseAG &  \bQMLErmseAG & \bBCQMrmseAG & \bCCEPrmseAG &  \bOLSErmseAH &  \bQMLErmseAH & \bBCQMrmseAH & \bCCEPrmseAH 
             \\[8pt]
                $T=80$ &  bias &   \bOLSEbiasAI &  \bQMLEbiasAI & \bBCQMbiasAI & \bCCEPbiasAI &  \bOLSEbiasAJ &  \bQMLEbiasAJ & \bBCQMbiasAJ & \bCCEPbiasAJ \\
    $(M=6)$              &  std  &   \bOLSEstdeAI &  \bQMLEstdeAI & \bBCQMstdeAI & \bCCEPstdeAI &  \bOLSEstdeAJ &  \bQMLEstdeAJ & \bBCQMstdeAJ & \bCCEPstdeAJ \\
                  &  rmse &   \bOLSErmseAI &  \bQMLErmseAI & \bBCQMrmseAI & \bCCEPrmseAI &  \bOLSErmseAJ &  \bQMLErmseAJ & \bBCQMrmseAJ & \bCCEPrmseAJ \\
                  \hline
   \end{tabular}
   \end{center}
\end{table}

\begin{table}[H]
   \begin{center}
   \caption{\label{tab:extra2} %
                      Same as Table~\ref{tab:T2} in main paper, but also reporting pooled CCE estimator of
                      Pesaran~(2006). }
      \begin{tabular}{l@{\;\,}l@{\quad}l@{\;\;}l@{\;\;}l@{\;\;}l@{\quad}l@{\;\;}l@{\;\;}l@{\;\;}l}
     \hline 
                  &   &  \multicolumn{4}{c}{$\rho^0=0.3$}    & \multicolumn{4}{c}{$\rho^0=0.9$}    \\[0.1cm]
                  &   & OLS & FLS & BC-FLS & CCE &  OLS & FLS & BC-FLS  & CCE \\
      \hline 
    $T=5$  &  bias &   \bbOLSEbiasAA &  \bbQMLEbiasAA & \bbBCQMbiasAA & \bbCCEPbiasAA &  \bbOLSEbiasAB &  \bbQMLEbiasAB & \bbBCQMbiasAB & \bbCCEPbiasAB \\
     $(M=2)$             &  std  &   \bbOLSEstdeAA &  \bbQMLEstdeAA & \bbBCQMstdeAA & \bbCCEPstdeAA &  \bbOLSEstdeAB &  \bbQMLEstdeAB & \bbBCQMstdeAB & \bbCCEPstdeAB \\
                  &  rmse &   \bbOLSErmseAA &  \bbQMLErmseAA & \bbBCQMrmseAA & \bbCCEPrmseAA &  \bbOLSErmseAB &  \bbQMLErmseAB & \bbBCQMrmseAB & \bbCCEPrmseAB 
                   \\[8pt]
                      $T=10$  &  bias &   \bbOLSEbiasAC &  \bbQMLEbiasAC & \bbBCQMbiasAC & \bbCCEPbiasAC &  \bbOLSEbiasAD &  \bbQMLEbiasAD & \bbBCQMbiasAD & \bbCCEPbiasAD \\
   $(M=3)$               &  std  &   \bbOLSEstdeAC &  \bbQMLEstdeAC & \bbBCQMstdeAC & \bbCCEPstdeAC &  \bbOLSEstdeAD &  \bbQMLEstdeAD & \bbBCQMstdeAD & \bbCCEPstdeAD \\
                  &  rmse &   \bbOLSErmseAC &  \bbQMLErmseAC & \bbBCQMrmseAC & \bbCCEPrmseAC &  \bbOLSErmseAD &  \bbQMLErmseAD & \bbBCQMrmseAD & \bbCCEPrmseAD 
                   \\[8pt]
                      $T=20$  &  bias &   \bbOLSEbiasAE &  \bbQMLEbiasAE & \bbBCQMbiasAE & \bbCCEPbiasAE &  \bbOLSEbiasAF &  \bbQMLEbiasAF & \bbBCQMbiasAF & \bbCCEPbiasAF \\
    $(M=4)$              &  std  &   \bbOLSEstdeAE &  \bbQMLEstdeAE & \bbBCQMstdeAE & \bbCCEPstdeAE &  \bbOLSEstdeAF &  \bbQMLEstdeAF & \bbBCQMstdeAF & \bbCCEPstdeAF \\
                  &  rmse &   \bbOLSErmseAE &  \bbQMLErmseAE & \bbBCQMrmseAE & \bbCCEPrmseAE &  \bbOLSErmseAF &  \bbQMLErmseAF & \bbBCQMrmseAF & \bbCCEPrmseAF 
                   \\[8pt]
                      $T=40$  &  bias &   \bbOLSEbiasAG &  \bbQMLEbiasAG & \bbBCQMbiasAG & \bbCCEPbiasAG &  \bbOLSEbiasAH &  \bbQMLEbiasAH & \bbBCQMbiasAH & \bbCCEPbiasAH \\
    $(M=5)$              &  std  &   \bbOLSEstdeAG &  \bbQMLEstdeAG & \bbBCQMstdeAG & \bbCCEPstdeAG &  \bbOLSEstdeAH &  \bbQMLEstdeAH & \bbBCQMstdeAH & \bbCCEPstdeAH \\
                  &  rmse &   \bbOLSErmseAG &  \bbQMLErmseAG & \bbBCQMrmseAG & \bbCCEPrmseAG &  \bbOLSErmseAH &  \bbQMLErmseAH & \bbBCQMrmseAH & \bbCCEPrmseAH
                   \\[8pt]
                      $T=80$ &  bias &   \bbOLSEbiasAI &  \bbQMLEbiasAI & \bbBCQMbiasAI & \bbCCEPbiasAI &  \bbOLSEbiasAJ &  \bbQMLEbiasAJ & \bbBCQMbiasAJ & \bbCCEPbiasAJ \\
    $(M=6)$              &  std  &   \bbOLSEstdeAI &  \bbQMLEstdeAI & \bbBCQMstdeAI & \bbCCEPstdeAI &  \bbOLSEstdeAJ &  \bbQMLEstdeAJ & \bbBCQMstdeAJ & \bbCCEPstdeAJ \\
                  &  rmse &   \bbOLSErmseAI &  \bbQMLErmseAI & \bbBCQMrmseAI & \bbCCEPrmseAI &  \bbOLSErmseAJ &  \bbQMLErmseAJ & \bbBCQMrmseAJ & \bbCCEPrmseAJ \\
                  \hline
   \end{tabular}
   \end{center}
\end{table}

\begin{table}[H]
   \begin{center}
   \caption{\label{tab:extra3} %
                      Analogous to Table~\ref{tab:T2} in main paper, but with $R=2$ correctly specified, and
                      also reporting pooled CCE estimator of
                      Pesaran~(2006). }
          \begin{tabular}{l@{\;\,}l@{\quad}l@{\;\;}l@{\;\;}l@{\;\;}l@{\quad}l@{\;\;}l@{\;\;}l@{\;\;}l}
     \hline 
                  &   &  \multicolumn{4}{c}{$\rho^0=0.3$}    & \multicolumn{4}{c}{$\rho^0=0.9$}    \\[0.1cm]
                  &   & OLS & FLS & BC-FLS & CCE &  OLS & FLS & BC-FLS  & CCE \\
      \hline 
    $T=5$  &  bias &   \bbbbOLSEbiasAA &  \bbbbQMLEbiasAA & \bbbbBCQMbiasAA & \bbbbCCEPbiasAA &  \bbbbOLSEbiasAB &  \bbbbQMLEbiasAB & \bbbbBCQMbiasAB & \bbbbCCEPbiasAB \\
     $(M=2)$             &  std  &   \bbbbOLSEstdeAA &  \bbbbQMLEstdeAA & \bbbbBCQMstdeAA & \bbbbCCEPstdeAA &  \bbbbOLSEstdeAB &  \bbbbQMLEstdeAB & \bbbbBCQMstdeAB & \bbbbCCEPstdeAB \\
                  &  rmse &   \bbbbOLSErmseAA &  \bbbbQMLErmseAA & \bbbbBCQMrmseAA & \bbbbCCEPrmseAA &  \bbbbOLSErmseAB &  \bbbbQMLErmseAB & \bbbbBCQMrmseAB & \bbbbCCEPrmseAB 
                   \\[8pt]
                      $T=10$  &  bias &   \bbbbOLSEbiasAC &  \bbbbQMLEbiasAC & \bbbbBCQMbiasAC & \bbbbCCEPbiasAC &  \bbbbOLSEbiasAD &  \bbbbQMLEbiasAD & \bbbbBCQMbiasAD & \bbbbCCEPbiasAD \\
   $(M=3)$               &  std  &   \bbbbOLSEstdeAC &  \bbbbQMLEstdeAC & \bbbbBCQMstdeAC & \bbbbCCEPstdeAC &  \bbbbOLSEstdeAD &  \bbbbQMLEstdeAD & \bbbbBCQMstdeAD & \bbbbCCEPstdeAD \\
                  &  rmse &   \bbbbOLSErmseAC &  \bbbbQMLErmseAC & \bbbbBCQMrmseAC & \bbbbCCEPrmseAC &  \bbbbOLSErmseAD &  \bbbbQMLErmseAD & \bbbbBCQMrmseAD & \bbbbCCEPrmseAD 
                   \\[8pt]
                      $T=20$  &  bias &   \bbbbOLSEbiasAE &  \bbbbQMLEbiasAE & \bbbbBCQMbiasAE & \bbbbCCEPbiasAE &  \bbbbOLSEbiasAF &  \bbbbQMLEbiasAF & \bbbbBCQMbiasAF & \bbbbCCEPbiasAF \\
    $(M=4)$              &  std  &   \bbbbOLSEstdeAE &  \bbbbQMLEstdeAE & \bbbbBCQMstdeAE & \bbbbCCEPstdeAE &  \bbbbOLSEstdeAF &  \bbbbQMLEstdeAF & \bbbbBCQMstdeAF & \bbbbCCEPstdeAF \\
                  &  rmse &   \bbbbOLSErmseAE &  \bbbbQMLErmseAE & \bbbbBCQMrmseAE & \bbbbCCEPrmseAE &  \bbbbOLSErmseAF &  \bbbbQMLErmseAF & \bbbbBCQMrmseAF & \bbbbCCEPrmseAF 
                   \\[8pt]
                      $T=40$  &  bias &   \bbbbOLSEbiasAG &  \bbbbQMLEbiasAG & \bbbbBCQMbiasAG & \bbbbCCEPbiasAG &  \bbbbOLSEbiasAH &  \bbbbQMLEbiasAH & \bbbbBCQMbiasAH & \bbbbCCEPbiasAH \\
    $(M=5)$              &  std  &   \bbbbOLSEstdeAG &  \bbbbQMLEstdeAG & \bbbbBCQMstdeAG & \bbbbCCEPstdeAG &  \bbbbOLSEstdeAH &  \bbbbQMLEstdeAH & \bbbbBCQMstdeAH & \bbbbCCEPstdeAH \\
                  &  rmse &   \bbbbOLSErmseAG &  \bbbbQMLErmseAG & \bbbbBCQMrmseAG & \bbbbCCEPrmseAG &  \bbbbOLSErmseAH &  \bbbbQMLErmseAH & \bbbbBCQMrmseAH & \bbbbCCEPrmseAH
                   \\[8pt]
                      $T=80$ &  bias &   \bbbbOLSEbiasAI &  \bbbbQMLEbiasAI & \bbbbBCQMbiasAI & \bbbbCCEPbiasAI &  \bbbbOLSEbiasAJ &  \bbbbQMLEbiasAJ & \bbbbBCQMbiasAJ & \bbbbCCEPbiasAJ \\
    $(M=6)$              &  std  &   \bbbbOLSEstdeAI &  \bbbbQMLEstdeAI & \bbbbBCQMstdeAI & \bbbbCCEPstdeAI &  \bbbbOLSEstdeAJ &  \bbbbQMLEstdeAJ & \bbbbBCQMstdeAJ & \bbbbCCEPstdeAJ \\
                  &  rmse &   \bbbbOLSErmseAI &  \bbbbQMLErmseAI & \bbbbBCQMrmseAI & \bbbbCCEPrmseAI &  \bbbbOLSErmseAJ &  \bbbbQMLErmseAJ & \bbbbBCQMrmseAJ & \bbbbCCEPrmseAJ \\
                  \hline
   \end{tabular}
      \end{center}
\end{table}
\end{document}